\documentclass[twocolumn,journal]{IEEEtran}
\usepackage{latexsym}
\usepackage{amssymb}
\usepackage{amsmath}

\usepackage[T1]{fontenc}
\usepackage{color}
\usepackage[dvips]{graphicx}

\usepackage{txfonts}

\usepackage{cancel}
\usepackage[normalem]{ulem}

\newcommand{\qed}{\hspace*{\fill}\rule{1ex}{1ex}}

\newtheorem{definition}{Definition}
\newtheorem{proposition}[definition]{Proposition}
\newtheorem{theorem}[definition]{Theorem}
\newtheorem{corollary}[definition]{Corollary}
\newtheorem{remark}[definition]{Remark}
\newtheorem{lemma}[definition]{Lemma}
\newtheorem{condition}[definition]{Condition}
\newtheorem{ensemble}{Code Ensemble}
\newtheorem{const}[ensemble]{Code Construction}
\newtheorem{example}[definition]{Example}

 \newenvironment{proofof}[1]{\vspace*{5mm} \par \noindent
         \quad{\it Proof of #1:\hspace{2mm}}}{\endproof
}

\def\cX{{\cal X}}
\def\cY{{\cal Y}}

\def\rE{{\mathbf{E}}}

\newcommand{\bF}{\mathbb{F}}
\newcommand{\bR}{\mathbf{R}}

\def\mix{\mathop{\rm mix}}
\def\supp{\mathop{\rm supp}}

\def\im{\mathop{\rm Im}}
\def\inn{\mathop{\rm inn}}

\def\BCC{\mathop{\rm BCC}}

\def\argmax{\mathop{\rm argmax}}
\def\argmin{\mathop{\rm argmin}}

\def\Label#1{\label{#1}\ [\ \text{#1}\ ]\ }
\def\Label{\label}

\begin{document}
\title{Secure Multiplex Coding with Dependent and
Non-Uniform Multiple Messages}
\author{
Masahito Hayashi,~\IEEEmembership{Senior Member,~IEEE,}%
\thanks{%
This research was partially supported by 
the MEXT Grant-in-Aid for Young Scientists (A) No.\ 20686026,
(B) No.\ 22760267, Grant-in-Aid for Scientific Research (A) No.\ 23246071,
and the ImPACT Program of Council for Science, Technology and
Innovation (Cabinet Office, Government of Japan),
the Villum Foundation through their VELUX Visiting Professor Programme
2011--2012.
The Centre for Quantum Technologies is funded
by the Singapore Ministry of Education and the National Research
Foundation as part of the Research Centres of Excellence programme.
This paper was
presented 
in part
at 2010 IEEE International Symposium on Information Theory, 
Austin, Texas, USA, June, 2010 \cite{hayashimatsumoto10},
in part
at 2011 IEEE International Symposium on Information Theory, Saint
Petersburg, Russia, August 2011 \cite{matsumotohayashi2011isit},
in part
at 49th Annual Allerton Conference, University of Illinois at Urbana-Champaign, IL, USA, September 2011 \cite{hayashimatsumoto2011allerton},
and
in part
at 50th Annual Allerton Conference, University of Illinois at Urbana-Champaign, IL, USA, October 2012 \cite{2012allerton}.}%
\thanks{
M. Hayashi is with Graduate School of Mathematics, Nagoya University, 
Furocho, Chikusaku, Nagoya, 464-8602, Japan, 
and
Centre for Quantum Technologies, National University of Singapore, 3 Science Drive 2, Singapore 117542.
(e-mail: masahito@math.nagoya-u.ac.jp)},
Ryutaroh Matsumoto,~\IEEEmembership{Member,~IEEE,}\thanks{
R. Matsumoto is with
Department of Communications and Computer Engineering,
Tokyo Institute of Technology, 152-8550 Japan}
}
\date{}

\maketitle
\begin{abstract}
The secure multiplex coding (SMC) is a technique to
remove rate loss in the coding for wire-tap
channels  and broadcast channels with
confidential messages
caused by the inclusion of random bits
into transmitted signals.
SMC replaces the random bits by other
meaningful secret messages, and
a collection of secret messages serves
as the random bits to hide the rest of messages.
In the previous researches, multiple secret messages
were assumed to have independent and uniform
distributions, which is difficult to be ensured
in practice. We remove this restrictive assumption
by a generalization of the channel resolvability technique.

We also give practical construction techniques for SMC by using
an arbitrary given error-correcting code as an ingredient,
and channel-universal coding of SMC.
By using the same principle as the channel-universal SMC,
we give coding for the broadcast channel with confidential
messages universal to both channel and source distributions.
\end{abstract}

\section{Introduction}\Label{sec1}
\subsection{Overview}
Recently, the security of personal information is demanded much more.
The wire-tap model is a typical secure message transmission model 
with the presence of an eavesdropper.
Specially, there are the legitimate sender called Alice,
the legitimate receiver called Bob, and
the eavesdropper Eve.
There is also a noisy broadcast channel from Alice to Bob and Eve.
Alice wants to send secret messages reliably to Bob and secretly
from Eve. This problem was first formulated by Wyner \cite{wyner75}.
Csisz\'ar and K\"orner generalized Wyner's original problem to
include common messages from Alice to both Bob and Eve,
and determined the optimal 
information rate tuples of 
the secret message and the common message, and the information leakage rate
of the secret message to Eve, which is measured by the conditional
entropy of the secret message given Eve's received signal
\cite{csiszar78}.
They called their generalized problem as the broadcast channel
with confidential messages, hereafter abbreviated as BCC.
The secrecy of messages over the wire-tap channel and
the BCC is realized by including meaningless random variable, 
which is called the dummy message,
into Alice's transmitted signal. This decreases the information rate.

In order to get rid of this information rate loss,
Yamamoto et al.\ \cite{yamamoto05} proposed the secure multiplex coding, hereafter abbreviated as SMC,
as a generalization of the wire-tap channel coding.
The SMC can be used, for example, in the following case.
When a company treats a collection of personal information,
it is required to keep the secrecy of the respective personal information.
However, it 
may not be required to keep the secrecy of the relation among several personal information.
For example, when all of personal information are subject to the uniform distribution of the same length bit sequence,
the secrecy of their exclusive OR may not be required.
Consider the case when the sender Alice sends the collection 
of $T$ persons' personal information
$S_1, \ldots, S_T$ via the channel partially leaked to Eve.
It is required that the receiver Bob can decode all of $S_1, \ldots, S_T$,
and that Eve cannot obtain any information of the respective personal information.
In order to keep the secrecy of the message $S_i$ from Eve,
Yamamoto et al. \cite{yamamoto05} proposed to use 
the remaining information $S_1, \ldots, S_{i-1}, S_{i+1}, \ldots,S_T$ as 
the dummy message for the message $S_i$.
Then, they realized the secrecy of the message $S_i$ without loss of
the information rate.
This type of coding problem is called the SMC.
It is known that 
the application of the channel resolvability \cite{han93} 
yields the security of the wire-tap channel model \cite{hay-wire}.
Hence, employing this method, Yamamoto et al. \cite{yamamoto05} proved 
the security of SMC.

On the other hand,
since $S_1, \ldots, S_{i-1}, S_{i+1}, \ldots,S_T$ are personal information,
they are not necessarily uniform random bits and might be dependent,
while the existing papers \cite{matsumotohayashi2011isit,yamamoto05}
assumed their uniformity and independence.
Such assumption is difficult to be ensured in practice.
Unfortunately, the application of the original channel resolvability 
can prove the security only when 
the messages $S_1, \ldots, S_{i-1}, S_{i+1}, \ldots,S_T$
are conditionally uniform and independent of $S_i$
because it treats the approximation of the channel output distribution 
with the uniform input random variable.
One may consider that 
the compressed data satisfies that assumption
so that the removal of that assumption is not needed. 
However, as is shown in \cite{han-folklore,hayashi08},
the compressed data is not uniform in the 
sense of the variational distance nor the divergence.
That is, the uniformity assumption does not hold for such compressed data.
Hence, the removal of the assumption is essential for non-uniform information source.

The reader might also conceive that this problem could
be solved by a straightforward combination of the coding for
intrinsic randomness \cite{Vem} and that for the original
secure multiplex coding \cite{yamamoto05,matsumotohayashi2011isit}.
We emphasize that this is false.
We cannot recover the original secret messages from
a codeword generated by an intrinsic randomness encoder, and
a new technique must be deployed to remove the
independence and uniform assumption on the multiple secret messages.
One of the main contributions of this paper is to remove that assumption.
In order to treat the non-uniform and dependent case, 
we need a generalization of the channel resolvability.
Hence, this paper also studies a generalization of the channel resolvability problem
\cite{han93,hay-wire}.

Even after we solve the above problem by a generalization of the channel resolvability problem,
the security of $S_i$ depends on the randomness and the dependence of 
the remaining messages $S_1, \ldots, S_{i-1}, S_{i+1}, \ldots,S_T$
on $S_i$.
This dependence causes 
another difficulty in the asymptotic formulation of SMC. 
That is, we need to characterize the randomness and the dependence in the asymptotic setting.
For this purpose, we introduce several kinds of asymptotic conditional uniformity
conditions and study their properties.
In addition to this, for the case when the channel is unknown,
we also treat 
universal coding for the secure multiplex coding \cite{yamamoto05}.
Further, 
as a byproduct, we obtain source-channel universal coding for the broadcast channel with confidential messages
\cite{csiszar78}.
We divide the introductory section to six subsections.

Finally, we should explain the assumptions for our probability spaces.
In the main body, we assume that all of probability spaces
are finite sets.
However, our result can be extended to the case of measurable spaces
except for the contents in Sections \ref{s6-2-2}, \ref{s9}, and \ref{s10}.
This generalization contains the case of continuous sets.
In Appendix \ref{as14}, we summarize how to generalize our results to the case of measurable spaces.
As a byproduct, we show the strong security for the Gaussian channel. 

\subsection{Generalization of the Channel Resolvability}
For a given channel $W$ with input alphabet
$\mathcal{X}$ and output alphabet $\mathcal{Y}$,
and given information source $X$ on $\mathcal{X}$,
Han and Verd\'u \cite{han93} considered to find a coding
$f: \mathcal{A} \rightarrow \mathcal{X}$ and a random variable
$A$ such that the distributions of $W(f(A))$ is close to $W(X)$
with respect to the variational distance or the normalized
divergence, and evaluated the minimum resolution of $A$ to
make the variational distance or the normalized divergence
asymptotically zero.
In their problem formulation,
one can choose the randomness $A$ used to simulate the channel output distribution.

In this paper, we shall consider the situation in which
we are given a channel $W$, an information source $X$, and
randomness $A$ and asked to find coding $f: \mathcal{A} \rightarrow
\mathcal{X}$ such that $W(f(A))$ is as close as possible to $W(X)$
with respect to unnormalized divergence.
We shall study how close $W(f(A))$ can be to $W(X)$
in Theorems \ref{lem-01} and \ref{Lee3} in Section \ref{s4}.
Hence, this problem can be regarded as a generalization of
channel resolvability
because this problem contains the original channel resolvability as a special case in the above sense.

\subsection{Asymptotic Conditional Uniformity}
In Subsection \ref{s6-2-2},  
in order to characterize the randomness and the dependence of 
the messages $S_1, \ldots, S_{i-1}, S_{i+1}, \ldots,S_T$ on
the other message $S_i$ asymptotically,
we introduce three asymptotic conditional uniformity conditions.
Then, we can characterize what a conditional distribution of 
the messages $S_1, \ldots, S_{i-1}, S_{i+1}, \ldots,S_T$
has a similar performance to the conditionally uniform distribution 
when we apply SMC.
We summarize the relations among those conditions as Theorem \ref{th-12-26-1}.
In particular, in Appendix \ref{as1}, we show that two introduced asymptotic conditional uniformity conditions are equivalent. 
Hence, we essentially have two different conditional uniformity conditions, 
namely, the weaker and the stronger asymptotic conditional uniformity conditions.

In Subsection \ref{as2}, we give sufficient conditions for the Slepian-Wolf compression so that the compressed data satisfies these asymptotic conditional uniformity conditions.
For the stationary ergodic sources,
we show the existence of a sequence of Slepian-Wolf codes whose compressed data 
satisfies the weaker asymptotic conditional uniformity conditions 
(Theorem \ref{th-12-23-3} and Remark \ref{rem1}).
Also for the i.i.d. sources,
we show the existence of a sequence of Slepian-Wolf codes whose compressed data 
satisfies the stronger asymptotic conditional uniformity conditions
(Theorem \ref{th-12-23-1} and Remark \ref{rem2}).

\subsection{Secure Multiplex Coding}\Label{s0-2}
Here, we explain the detail of our contributions to SMC.
As is explained above,
we have to realize the security of $S_i$ 
when the remaining messages $S_1, \ldots, S_{i-1}, S_{i+1}, \ldots,S_T$ are not uniform and are dependent on the message $S_i$.
In order to solve this problem,
we employ our generalized channel resolvability coding in
Theorems \ref{lem-01} and \ref{Lee3}.
Then, we can construct 
coding for a wire-tap channel that can ensure the secrecy of
message against the eavesdropper Eve 
when the dummy message
used by the encoder is non-uniform and statistically dependent
on the secret message that has to be kept secret from Eve.
We apply our generalized channel resolvability coding
to the above SMC case. 
Hence, we can remove the independence and uniform assumption on the multiple secret messages
while the original paper \cite{yamamoto05} by Yamamoto et al. and 
the previous paper \cite{matsumotohayashi2011isit} by the present authors 
assumed the independence and the uniformity of the multiple secret messages.

Indeed, Yamamoto et al.\ \cite{yamamoto05} treated only the secrecy of each message $S_i$,
and did not evaluate the information leakage of
multiple messages $S_{i_1}$, \ldots, $S_{i_n}$ to Eve, and the present authors
analyzed such information leakage in \cite{matsumotohayashi2011isit}.
The present authors also generalized coding in \cite{matsumotohayashi2011isit} so that
Alice's encoder can support the common message $S_0$
to both Bob and Eve.
The present authors also characterized the achievable 
information leakage rate in \cite{matsumotohayashi2011isit}.
Those enhancements are retained in this paper.

In Section \ref{s5}, we shall give two code constructions for SMC.
The first construction given in Subsection \ref{s5-3} is a simple application of 
channel resolvability coding in Theorem \ref{lem-01}.
Although it achieves the capacity region when there is no common message,
it is insufficient to fully prove the capacity region.
In Subsection \ref{s5-2},
to overcome this defect, we propose
the second construction given in Theorem \ref{Lee3}, which is based on 
another type of the channel resolvability coding. 
By using these constructions,
we shall evaluate the decoding error probability and
the mutual information to Eve in Section \ref{s5}
in single-shot setting in the sense of \cite{single}.

In Section \ref{s6} we formulate the capacity region of SMC,
analyze the asymptotic performance of two constructions, and
prove that the second construction achieves the capacity
region of SMC.
The capacity region is defined based on 
the weaker asymptotic conditional uniformity condition
given in Definition \ref{D12-18-1}.
In Section \ref{s7}, we shall prove that the mutual information to
Eve converges to zero when the normalized mutual information to Eve
converges to zero
under the stronger asymptotic conditional uniformity
given in Definition \ref{D12-18-2}.
The convergence
is so-called the strong security \cite{maurer94}.
In Subsection \ref{s7-3},
we also derive the exponent of the mutual information to Eve. 
The relation between our results and the paper \cite{yamamoto05} is explained as \eqref{11-14-1d}.

Section \ref{s8} addresses a more practical issue.
In Theorem \ref{lem2} of Section \ref{s5},
we show that we can have an upper bound of mutual information
between multiple secret messages and Eve's received signal,
by attaching randomly chosen group homomorphisms satisfying Condition \ref{C2-b}
to \emph{any} given error-correcting code for channels with single
sender and single receiver or the broadcast channel with
degraded message sets \cite{korner77}.
However, the upper bound in Theorem \ref{lem2} becomes difficult to be
computed when the error-correcting code is not given by the standard
random coding in information theory.
In Section \ref{s8},
we shall construct more practical codes
by combining the construction of Section \ref{s5} with an arbitrary given error-correcting code.
Under these codes, we shall give two upper bounds on the leaked mutual information that can be computed easily in practice.
Section \ref{s8} gives enhancement of our earlier proceeding paper
\cite{hayashimatsumoto10}.

\subsection{Universal Coding}

Universal coding is construction of encoder and decoder that do not
use the statistical knowledge on the underlying information system
(usually channel and/or source) \cite{csiszarbook}.
In Section \ref{s9} we shall give a construction of SMC universal to
channel. The basic idea in Section \ref{s9} is to combine the construction
in Section \ref{s5} with the universal coding using constant-type codes
for the broadcast
channel with degraded messages sets (BCD) in \cite{korner80},
while in Sections \ref{s5}--\ref{s7} the superposition random coding in
\cite{korner77} is used as their error-correcting mechanism.
The exponent given in Section \ref{s9} is better than
that given in our earlier proceeding paper
\cite{hayashimatsumoto2011allerton}.

Channel-universal coding for 
BCC had not been studied before
\cite{hayashimatsumoto2011allerton}, 
and coding for BCC can be
regarded as a special case of SMC
while Muramatsu et al. \cite{Muramatsu} treat channel-universal coding for wire-tap channel independently of \cite{hayashimatsumoto2011allerton}.
In Section \ref{s9} and \cite{hayashimatsumoto2011allerton} we consider
SMC universal to channel, but its universality to the source
is not considered.
In Section \ref{s10} we give a coding for BCC universal to both
channel and source. Its channel-universality is realized by the
same principle as Section \ref{s9} and \cite{hayashimatsumoto2011allerton}.
The exponent given in Section \ref{s9} is also greater than
that given in our earlier proceeding paper
\cite{hayashimatsumoto2011allerton}.

In Section \ref{s12}, we compare the exponent of leaked information
given in Sections \ref{s9} and \ref{s10} and 
that given in Subsection \ref{s7-3}.
As a result, we show that the exponent in Sections \ref{s9} and \ref{s10}
is greater than one of exponents in Subsection \ref{s7-3},
which is the same as that in \cite{hayashimatsumoto2011allerton}.
We also derive the equality condition.

\subsection{Organization of This Paper}
The outline of this paper is given as follows.
First, we prepare notations used in this paper in Section \ref{s0}.
Second, we prepare information quantities and their properties 
used in this paper in Section \ref{s1-2}.
Then, we review the formulation and existing results of 
BCC in Subsection \ref{s2-1}.
We give its reformulation for the dependent and non-uniform messages case
in Subsection \ref{s2-2}.
This new formulation is essential in the later discussion for 
SMC with dependent and non-uniform multiple messages.
In Subsection \ref{s4-2}, we review the formulation and existing results of BCD
as a special case of BCC, which will be used for our codes of SMC.
In Subsection \ref{s4-3}, we review 
K\"orner and Sgarro \cite{korner80}'s result for universal code for BCD,
which will be used for our construction of universal codes for SMC and BCC.
In Section \ref{s4}, we proceed to 
generalization of channel resolvability, which is a key idea of the paper 
and is used for codes of SMC and universal codes for SMC and BCC.
Section \ref{s5} introduces SMC with the single-shot setting.
Section \ref{s6-2-} introduces three asymptotic conditional uniformity conditions.
Based on these conditions,
Sections \ref{s6}--\ref{s8} treats SMC with the asymptotic setting, as is explained in Subsection \ref{s0-2}.
In Section \ref{s9}, 
combining the discussion of Subsections \ref{s4-2} and \ref{s5-4},
we propose universal coding for SMC by using 
K\"orner and Sgarro \cite{korner80}'s universal coding for BCD.
In Section \ref{s10}, we propose source-channel universal coding for BCC.
Appendices are devoted for several additionally required discussions
for asymptotic conditional uniformity conditions.
This paper contains two types of descriptions for each topics, i.e.,
the single-shot description \cite{single} and the $n$-fold description. 
Formulations and many coding theorems are given with the single-shot description.
The definitions of capacity regions are given in the $n$-fold description.


\section{Notation in This Paper}\Label{s0}

$\mathcal{X}$ denotes the channel input alphabet and
$\mathcal{Y}$ (resp.\ $\mathcal{Z}$) denotes
the channel output alphabet to Bob (resp.\ Eve).
We assume that $\mathcal{X}$, $\mathcal{Y}$, and
$\mathcal{Z}$ are finite unless otherwise stated.
We denote the 
conditional probability of the channel
to Bob and Eve by $P_{YZ|X}$.
Then, taking the marginal distribution,
we denote the conditional probability of the channel
to Bob (resp.\ Eve) by $P_{Y|X}$ (resp.\ $P_{Z|X}$).
Also, we denote the distribution of the random variable
$X$ by $P_X$.

We denote the uniform distribution on $\Omega$ by $P_{\mix,\Omega}$.
When $\Omega$ is a subset of $\mathcal{X} \times \mathcal{Y}$,
$P_{\mix,\Omega}$ is a joint distribution for the random variables $X$ and $Y$.
We denote the marginal distribution of $P_{\mix,\Omega}$
for the random variable $X$ and the random variable $Y$ by
$P_{X, \mix,\Omega}$ and $P_{Y, \mix,\Omega}$, respectively.
Further, the conditional distribution on the random variable $X$ 
conditioned to the other random variable $Y$ is denoted by 
$P_{X|Y, \mix,\Omega}$, i.e., 
\begin{align}
P_{X|Y, \mix,\Omega}(x|y)
=P_{X|Y=y, \mix,\Omega}(x)
:=\frac{P_{\mix,\Omega}(x,y)}{P_{Y,\mix,\Omega}(y)} 
\Label{12-19-10}
\end{align}
for $x\in \mathcal{X}$ and $y\in \mathcal{Y}$.
We denote the support of the distribution $P_X$ by $\supp(P_X)$.
Given a joint distribution $P_{XY}$, we define the distribution 
$P_{X|Y=y}$
on $\mathcal{X}$ by $P_{X|Y=y}(x):=P_{X|Y}(x|y)$.
When we need to treat another distribution
of the same random variables $X$ and $Y$, 
we denote it by $Q_{XY}$.
This is because
it is crucial to consider several distributions on the same probability space in this paper\footnote{Recently, the meta converse theorem was introduced for the channel coding in \cite{Nag,Pol}.
In the meta converse theorem,
it is the key point to optimize the choice of the distribution on 
the output alphabet
and we usually denote the distribution different from the marginal distribution by $Q$\cite{Hsec,Pol}.
Also, another recent paper \cite{T-H} adopts this notation for optimizing the distribution.
This kind notation becomes more popular, recently.}.
In this case, we denote the marginal distribution over $\mathcal{X}$ by $Q_X$,
and the conditional distribution by $Q_{X|Y}$.
We also define the distribution 
$Q_{X|Y=y}$ on $\mathcal{X}$ by $Q_{X|Y=y}(x):=Q_{X|Y}(x|y)$.

When we have to treat more than two distributions on 
$\mathcal{X}$, $\mathcal{Y}$, and $\mathcal{Z}$,
the above notation is not useful.
In this case, we consider 
the set $\mathcal{P}(\mathcal{X})$
of probability distributions on $\mathcal{X}$
or
the set $\mathcal{W}(\mathcal{X}$, $\mathcal{Y})$
of conditional probability distributions
from $\mathcal{X}$ to $\mathcal{Y}$,
which are mathematically equivalent to probability transition matrices.
When the output alphabet of the channel 
is given as a product set $\mathcal{Y}\times \mathcal{Z}$,
the alphabet is written by $\mathcal{W}(\mathcal{X}$, $\mathcal{Y}\times \mathcal{Z})$.
For any probability transition matrix $W \in 
\mathcal{W}(\mathcal{X}$, $\mathcal{Y}\times \mathcal{Z})$,
$W_x$ expresses the output distribution when the input $X$ is $x$.
When we focus on the random variable $Y$, 
we use the notation 
$W^{Y}_{x}(y):=\sum_{z \in \mathcal{Z}}W_x(y,z)$.

In the following, we treat an arbitrary 
probability transition matrix $W \in \mathcal{W}(\mathcal{X}$, $\mathcal{Y})$.
Given a subset $\Omega \subset \mathcal{X}$,
we define the restriction $W|_{\Omega} \in \mathcal{W}(\Omega$, $\mathcal{Y})$
by $W|_{\Omega}(y|x)=W(y|x)$ for $x\in \Omega$ and $y\in \mathcal{Y}$.
We often employ another probability transition matrix 
$\Xi$ 
from $\mathcal{V}$ to $\mathcal{X}$.
We define the probability transition matrix 
from $\mathcal{V}$ to $\mathcal{Y}$
by $W \circ \Xi_v (y):=
\sum_{x \in \mathcal{X}}W_x(y) \Xi_v (x) $ for 
$v \in \mathcal{V}$ and $y\in \mathcal{Y}$.
When a probability distribution $P$ on $\mathcal{X}$ is given,
we define the distribution on $\mathcal{Y}$
by $W \circ P(y):= \sum_{x\in \mathcal{X}}W_x(y)P(x)$ 
for $y \in \mathcal{Y}$.
When we need the joint distribution on 
$\mathcal{X}\times\mathcal{Y}$,
we use the notation
$W \times P(x,y):= W_x(y)P(x)$ 
for $x\in \mathcal{X}$ and $y\in \mathcal{Y}$
as \cite{Csiszar}.
Similarly, when a distribution $P_{XV}$ on $\mathcal{X}\times \mathcal{V}$
is given,
we use the notation
$W \times P_{XV}(v,x,y):= W_x(y)P_{XV}(x,v)$
for 
$v\in \mathcal{V}$, $x\in \mathcal{X}$, and $y\in \mathcal{Y}$.

When 
a function $f:\mathcal{V} \to \mathcal{X}$ is given
and a random variable $V$ taking the values in $\mathcal{V}$
obeys the distribution $P_V$,
we can define the random variable $f(V)$ taking the values in $\mathcal{X}$.
The random variable $f(V)$ takes the value $x$ with probability
$\sum_{v \in f^{-1}(x)}P_V(v)$.
We also use the same symbol $f:\mathcal{V} \to \mathcal{X}$ 
to denote the probability transition matrix 
from $\mathcal{V}$ to $\mathcal{X}$, in which, 
the output value is deterministically determined by the input.
Then, $W \circ f$ is a 
stochastic mapping
$\mathcal{V}$ to $\mathcal{Y}$, and we have
\begin{align}
(W \circ f)(y|v)=W(y|f(v))
\Label{12-18-4}
\end{align}
for $v\in \mathcal{V}$ and $y\in \mathcal{Y}$.
Given a probability transition matrix 
$W' \in \mathcal{W}(\mathcal{U}$, $\mathcal{V})$,
we define $f \circ W' \in  \mathcal{W}(\mathcal{U}$, $\mathcal{X})$ by
\begin{align}
(f \circ W') (x|u):= \sum_{v\in f^{-1}(x) }W'(v|u)
\Label{12-18-5}
\end{align}
for $x\in \mathcal{X}$ and $u\in \mathcal{U}$.
As a special case,
given a distribution $Q$ on $\mathcal{V}$,
$f \circ Q$ is defined as a distribution on $\mathcal{X}$
in the following way.
\begin{align}
(f \circ Q) (x):= \sum_{v\in f^{-1}(x) }Q(v).
\Label{12-18-6}
\end{align}

Remember that
$W_x$ denotes the output distribution on the output alphabet 
$\mathcal{Y}$ with input $x$. 
Then, $W_X$ is the random variable taking its values on the output distributions on $\mathcal{Y}$.
Given a real valued function $g$ of distributions on $\mathcal{Y}$,
we regard $g(W_X)$ as a random variable 
taking the value 
$g(W_x)$ with the probability $P_X(x)$.
Hence, we obtain
\begin{align*}
\rE_{X} g(W_X)= \sum_x P_X(x) g(W_x),
\end{align*}
where $\rE_{X}$ denotes the expectation concerning $X$.

Given two random variables $X$ and $Y$,
for a real valued function $h$ on ${\cal X} \times {\cal Y}$,
we regard $\rE_{X|Y}h(X,Y)$ as 
a random variable taking the value $\rE_{X|Y=y}h(X,y)$ 
with the probability $P_Y(y)$.
In order to identify an information quantity, e.g., 
mutual information $I(X;Y)$ and the Shannon entropy $H(X)$, 
we sometimes need to specify the distribution $P$ of interest.
In such a case, we use the notations $I(X;Y)[P]$ and $H(X)[P]$
for identifying what distribution is considered.

Further, in this paper, 
we discuss our codes and their performances in the single-shot setting\cite{
single}
when their descriptions do not require their asymptotic discussions.
However, in several parts, we need to treat 
$n$-fold memoryless extensions when we discuss their asymptotic performances.
Hence, we need to prepare the notations for 
$n$-fold independent and identical distributions
and $n$-fold memoryless extensions of given channels.
For a given probability distributions $Q$
and $P_X$ of the random variable $X$ on ${\cal X}$, 
we denote their $n$-fold independent and identical distributions
by $Q^n$ and $P_X^n$.

When we consider the random variables on $\mathcal{X}^n$, 
even if they do not obey the independent and identical distributions, 
we denote the random variables by $X^n$ and 
denote their distributions by $P_{X^n}$.
However, when we consider a general sequence of random variables
those take values 
not in the product sets $\mathcal{X}^n$
but in general sets $\mathcal{X}_n$,
we denote the random variables by $X_n$ 
and denote their distributions by $P_{X_n}$.
Similarly, for a given probability transition matrices $W$
and $P_{Y|X}$ from ${\cal X}$ to ${\cal Y}$,
we denote their $n$-fold memoryless extensions by $W^n$
and $P_{Y|X}^n$.

We also 
denote the set of positive
real numbers by $\mathbf{R}^+$,
and denote the set of non-negative
real numbers by $\mathbf{R}_{\ge 0}$.

\section{Information Quantities}\Label{s1-2}
In this paper, to evaluate the secrecy and the decoding error probabilities,
we employ several information quantities.
For distributions 
$P_A$ on ${\cal A}$ and $P_{AB}$ on ${\cal A}\times {\cal B}$,
we define 
R\'{e}nyi entropy and conditional R\'{e}nyi entropy 
\begin{align}
H_{1+\rho}(A) & := -\frac{1}{\rho}\log \sum_a P_A(a)^{1+\rho} \nonumber\\
H_{1+\rho}(A|B) &:= -\frac{1}{\rho}\log \sum_{a,b} P_B(b) P_{A|B=b}(a)^{1+\rho}.\nonumber
\end{align}
$H_{1}(A)$ and $H_{1}(A|B)$ are defined to be $H(A)$ and $H(A|B)$.
Then, we have several important properties for R\'{e}nyi entropy and conditional R\'{e}nyi entropy. 
Since $\rho \mapsto \rho H_{1+\rho}(A)$, $\rho \mapsto \rho H_{1+\rho}(A|B)$ are concave
and $\lim_{\rho\to 0}\rho H_{1+\rho}(A)=
\lim_{\rho\to 0}\rho H_{1+\rho}(A|B)=0$, we have
\begin{align}
H_{1+\rho'}(A)  \le H_{1+\rho}(A) , \quad
H_{1+\rho'}(A|B) \le H_{1+\rho}(A|B) \Label{eq-1-13}
\end{align}
for $0 \le \rho \le \rho'$.

Similarly, as is shown in \cite{hayashi11},
we have the following proposition for the function 
\begin{align}
\psi(\rho|Q\|P) & :=\log \sum_{a}Q(a)^{1+\rho}P(a)^{-\rho} .
\end{align}

\begin{proposition}\cite{hayashi11}\Label{3-22-4L}
The function $\psi(\rho| Q \| P)$ satisfies the following properties:
\begin{description}
\item[(1)] 
$\rho \mapsto \psi(\rho| Q \| P)$ 
is convex.
\item[(2)] 
$\psi(0| Q \| P)=0$.
\item[(3)] 
$\frac{d}{d\rho}\psi(\rho| Q \| P)|_{\rho=0}=
D( Q \| P)$.
\item[(4)] 
The relations
\begin{align}
D(Q\|P):= \sum_{a} P(a) \log \frac{P(a)}{Q(a)}
=&
\lim_{\rho\to +0} \frac{\psi(\rho|Q\|P)}{\rho} 
\nonumber \\
\le &
\frac{\psi(\rho|Q\|P)}{\rho}
\Label{12-18-1}
\end{align}
hold for $0 < \rho$\footnote{Item (4) was not directly given in \cite{hayashi11}. However, it can be shown by the combination of other items.}.
\end{description}
\end{proposition}

For a given channel $W$ from ${\cal X}$ to ${\cal Y}$, 
we define the function \cite{hayashi11}:
\begin{align}
\psi(\rho|W,P_X) &:=
\log \sum_{x}P_X(x) e^{\psi(\rho|W_x\|W\circ P_X)}.
\Label{5-30-1}
\end{align}
When the channel is written as $P_{Z|L}$,
$\psi(\rho|W,P)$
can be rewritten as follows.
\begin{align}
\psi(\rho| P_{Z|L}, P_L) 
=
\log \sum_z \sum_\ell P_L(\ell) P_{Z|L}(z|\ell)^{1+\rho} P_Z(z)^{-\rho}.
\Label{eq:psid}
\end{align}
This quantity is extended as
\begin{align}
& \psi(\rho| P_{Z|V}, P_{V|U},P_U) \nonumber \\
:=& 
\log \sum_u P_U(u) \sum_v P_{V|U}(v|u) \sum_z P_{Z|V}(z|v)^{1+\rho} P_{Z|U}(z|u)^{-\rho}.
\Label{eq:psid-1}
\end{align}
for conditional distributions $P_{Z|V}$, $P_{V|U}$ 
and a distribution $P_U$.
Also, we introduce the following functions as in \cite{hayashi11}.
\begin{align}
&E_0(\rho| P_{Z|L},P_L) \nonumber \\
:=& \log \sum_z\left(
\sum_{\ell} P_{L}(\ell) (P_{Z|L}(z|\ell)^{1/(1-\rho)})\right)^{1-\rho},
\Label{phid} \\
&E_0(\rho| P_{Z|V}, P_{V|U},P_U) \nonumber \\
:=& \log \sum_u P_U(u) \sum_z\left(
\sum_{v} P_{V|U}(v|u) (P_{Z|V}(z|v)^{1/(1-\rho)})\right)^{1-\rho}.\Label{phid-2}
\end{align}
Observe that $E_0$ is essentially Gallager's function $E_0$
\cite{gallager68}.
As can be easily shown, these quantities satisfy the additivity as follows\cite{hayashi11,gallager68}.
\begin{align}
\psi(\rho| P_{Z|L}^n, P_L^n) 
&=n \psi(\rho| P_{Z|L}, P_L) \Label{3-24-4eq}\\
\psi(\rho| P_{Z|V}^n, P_{V|U}^n,P_U^n)
&=n \psi(\rho| P_{Z|V}, P_{V|U},P_U) \Label{3-24-3eq}\\
E_0(\rho| P_{Z|L}^n,P_L^n)
&= n E_0(\rho| P_{Z|L},P_L) \Label{3-24-2eq}\\
E_0(\rho| P_{Z|V}^n, P_{V|U}^n,P_U^n)
&=n E_0(\rho| P_{Z|V}, P_{V|U},P_U) \Label{3-24-1eq}
\end{align}
Then, we have the following proposition.
\begin{proposition}\cite{gallager68,hayashi11}\Label{lem12-4-1}
We have the following five items for fixed $0<\rho< 1$ and fixed conditional distribution $P_{Z|L}$.
\begin{description}
\item[(1)] 
The function $\rho\mapsto E_0(\rho| P_{Z|L}, P_L)$ is convex for a given distribution $P_L$\cite{gallager68}.

\item[(2)] 
$\exp(E_0(\rho| P_{Z|L}, P_L))$ is concave with respect to $P_L$\cite[Lemma 1]{hayashi11}.

\item[(3)] 
The relation $
\psi(\rho| P_{Z|L}, P_L) \le E_0(\rho| P_{Z|L}, P_L)$, i.e.,  
\begin{align}
\exp(\psi(\rho| P_{Z|L}, P_L))
& \le \exp(E_0(\rho| P_{Z|L}, P_L))
\Label{psileqphi} 
\end{align}
holds for any distribution $P_L$ of $L$\cite[(16)]{hayashi11}.

\item[(4)] 
The relation 
\begin{align}
\lim_{\rho \to 0}\frac{\psi(\rho| P_{Z|L}, P_L)}{\rho}
=
\lim_{\rho \to 0}\frac{E_0(\rho| P_{Z|L}, P_L)}{\rho}
=I(Z;L)\Label{3-24-20eq}
\end{align}
holds for a distribution $P_L$\cite[Section III]{hayashi11}\cite{gallager68}.
\end{description}
\end{proposition}

\begin{lemma}\Label{3-24-4L}
When two distributions $Q_L$ and $P_L$
of $L$ satisfy
$P_L(\ell) \le C_1 Q_L(\ell)$ for any $\ell$ with given constants $C_1 \ge 1$ and $0< \rho <1$,
we have
\begin{align}
\exp(E_0(\rho| P_{Z|L}, P_L))
& \le 
C_1 \exp(E_0(\rho| P_{Z|L}, Q_L)) \Label{Haya-2} .
\end{align}
\end{lemma}

\begin{proof}
\eqref{Haya-2} can be shown as follows.
\begin{align*}
& \exp(E_0(\rho| P_{Z|L}, P_L))
=
\sum_z\left(
\sum_{\ell} P_{L}(\ell) (P_{Z|L}(z|\ell)^{1/(1-\rho)})\right)^{1-\rho} \\
\le &
\sum_z\left(
\sum_{\ell} C_1 Q_{L}(\ell) (P_{Z|L}(z|\ell)^{1/(1-\rho)})\right)^{1-\rho} \\
\le &
C_1^{1-\rho}\sum_z\left(
\sum_{\ell} Q_{L}(\ell) (P_{Z|L}(z|\ell)^{1/(1-\rho)})\right)^{1-\rho} \\
=&
C_1^{1-\rho} \exp(E_0(\rho| P_{Z|L}, Q_L))
\le
C_1 \exp(E_0(\rho| P_{Z|L}, Q_L)).
\end{align*}
\end{proof}

As a generalization of Item (4) of Proposition \ref{lem12-4-1},
we have the following lemma.
\begin{lemma}\Label{3-24-3L}
The relation 
\begin{align}
\lim_{\rho \to 0}\frac{\psi(\rho| P_{Z|V}, P_{V|U},P_U)}{\rho}
=&
\lim_{\rho \to 0}\frac{E_0(\rho| P_{Z|V}, P_{V|U},P_U)}{\rho} \nonumber \\
=&I(Z;V|U)\Label{3-24-21eq}
\end{align}
holds for a distribution $P_U$,
and conditional distributions $P_{Z|V}$ and $P_{V|U}$.
\end{lemma}

\begin{proof}
Due to \eqref{3-24-20eq}, we have
\begin{align*}
e^{\psi(\rho| P_{Z|V}, P_{V|U},P_U)}
=&
\sum_{u}P_U(u) 1+ \rho I(Z;V|U=u)+ o(\rho)\nonumber \\
=&
1+ \rho I(Z;V|U)+ o(\rho).
\end{align*}
Taking the logarithm, we obtain 
$\lim_{\rho \to 0}\frac{\psi(\rho| P_{Z|V}, P_{V|U},P_U)}{\rho}
=I(Z;V|U)$.
Similarly, we can show 
$\lim_{\rho \to 0}\frac{E_0(\rho| P_{Z|V}, P_{V|U},P_U)}{\rho}
=I(Z;V|U)$.
\end{proof}

Considering the Legendre transforms, we define 
\begin{align}
\tilde{E}^{\psi}(R,P_{Z,V,U})
&:=
\max_{0 \le \rho \le 1} \rho R - \psi(\rho|P_{Z|V},P_{V|U},P_U), \Label{1-31-3b}\\
\tilde{E}^{E_0}(R,P_{Z,V,U})
& :=
\max_{0 \le \rho \le 1} \rho R -E_0(\rho|P_{Z|V},P_{V|U},P_U).\Label{1-31-3}
\end{align}

Taking the maximum, we define
\begin{align}
E_{0,\max}(\rho|P_{Z|V}):= &
\max_{P_V} E_0(\rho| P_{Z|V}, P_V) \nonumber \\
=&
\log
\max_{P_V}
\sum_{z}
(\sum_{v}P_V(v)
P_{Z|V}(z|v )^{\frac{1}{1-\rho}}  )^{1-\rho} \nonumber \\
=& \max_{P_{VU}}E_0(\rho| P_{Z|V} ,P_{V|U},P_U ) . 
\Label{eq10001}
\end{align}

\begin{lemma}\Label{3-22-1L}
The function $\rho \mapsto 
E_{0,\max}(\rho|P_{Z|V})$ is convex.
\end{lemma}

\begin{proof}
Given convex functions $x \mapsto f_i(x)$, 
the function $x \mapsto \max_{i}f_i(x)$ is also convex.
Hence, the item (1) of Proposition \ref{lem12-4-1} yields the desired argument.
\end{proof}

Next, 
for $\overline{W}^Z \in \mathcal{W}(\mathcal{V}$, $\mathcal{Z})$,
we consider a different information quantity $\tilde{E}^l$:
\begin{align}
&\tilde{E}^l
(R, \overline{W}^Z\times Q_{VU}) \nonumber \\
:=& 
\min_{{W}^Z \in \mathcal{W}(\mathcal{U}\times \mathcal{V},\mathcal{Z})}
\biggl( D({W}^Z\| \overline{W}^Z| Q_{VU})  \nonumber \\
&\hspace{10ex}+
[R-I(V;Z|U)[W^Z\times Q_{VU}] ]_+ \biggr).
\Label{1-31-1-k}
\end{align}
Due to Item (3) of Proposition \ref{lem12-4-1}, we have
\begin{align}
\tilde{E}^{\psi} (R, \overline{W}^Z\times Q_{VU})
\ge 
\tilde{E}^{E_0}(R, \overline{W}^Z\times Q_{VU}) .
\Label{3-22-3}
\end{align}
In this paper, we will derive the following relations:
\begin{align}
\tilde{E}^l (R, \overline{W}^Z\times Q_{VU}) 
\ge & 
\tilde{E}^{E_0}(R, \overline{W}^Z\times Q_{VU}) 
\Label{3-22-1}
\end{align}
and
\begin{align}
&\min_{Q_V}\tilde{E}^l (R, \overline{W}^Z\times Q_{V}) 
=\min_{Q_V}\tilde{E}^{E_0} (R, \overline{W}^Z\times Q_{V})
\nonumber \\
=&\max_{\rho \in [0,1]} \rho R-E_0(\rho|\overline{W}^Z) \Label{3-22-2}
\end{align}
as Theorems \ref{t-12-20-2} and \ref{t-12-20-1} in Section \ref{s12}, respectively.

Similar to $\tilde{E}^l$,
we introduce the following quantities
for $W^Y \in \mathcal{W}(\mathcal{V},\mathcal{Y})$
and $W^Z \in \mathcal{W}(\mathcal{V},\mathcal{Z})$
\begin{align}
&\hat{E}^b(R_{\mathrm{p}},R_{\mathrm{c}}, \tilde{W}^Y\times Q_{VU})\nonumber \\
:=&
\min \biggl( 
[I(VU;Y)[\tilde{W}^Y\times Q_{U,V}] - R_{\mathrm{p}}-R_{\mathrm{c}}]_+,\nonumber \\
&\hspace{10ex} 
[I(V;Y|U)[\tilde{W}^Y\times Q_{U,V}] - R_{\mathrm{p}}]_+\biggr), \\
&\tilde{E}^b(R_{\mathrm{p}},R_{\mathrm{c}}, W^Y\times Q_{VU}) \nonumber \\
:=& \min_{\tilde{W}^Y \in \mathcal{W}(\mathcal{U}\times \mathcal{V},\mathcal{Y})}
D(\tilde{W}^Y\| W^Y| Q_{VU})
+\hat{E}^b(R_{\mathrm{p}},R_{\mathrm{c}}, \tilde{W}^Y\times Q_{VU}),
\Label{1-31-1}
\\
& \tilde{E}^e(R_{\mathrm{c}}, W^Z\times Q_{U})  \nonumber \\
:=&\min_{\tilde{W}^Z \in \mathcal{W}(\mathcal{U}\times \mathcal{V},\mathcal{Z})}
D(\tilde{W}^Z\| W^Z| Q_{VU})+
[I(U;Z)[\tilde{W}^Z\times Q_{VU}] - R_{\mathrm{c}}]_+,
\Label{1-31-2}
\end{align}
where $D(\tilde{W}~Y\|W^Y| Q_{VU})$ is defined for 
$\tilde{W}^Y, W^Y \in \mathcal{W}(\mathcal{V},\mathcal{Y})$ as
\begin{align}
D(\tilde{W}^Y\|W^Y| Q_{VU}):=\sum_{u,v}Q_{VU}(u,v) D(\tilde{W}^{Y}_{u,v}\|W^Y_{v}). \end{align}
In the above definition, 
$W^Y$ and $W^Z$ are treated as elements of
$\mathcal{W}(\mathcal{U}\times \mathcal{V},\mathcal{Y})$
and $\mathcal{W}(\mathcal{U}\times \mathcal{V},\mathcal{Z})$,
respectively.

\section{Broadcast Channels with Confidential Messages}\Label{s2}
\subsection{Review of Existing Results}\Label{s2-1}
First, we give a formulation 
of broadcast channels with confidential messages
with single shot setting\cite{single}.
Let Alice, Bob, and Eve be as defined in Section \ref{sec1}.
$\mathcal{X}$ denotes the channel input alphabet and
$\mathcal{Y}$ (resp.\ $\mathcal{Z}$) denotes
the channel output alphabet to Bob (resp.\ Eve).
We assume that $\mathcal{X}$, $\mathcal{Y}$, and
$\mathcal{Z}$ are finite unless otherwise stated.

We denote the conditional probability of the channel
to Bob (resp.\ Eve) by $P_{Y|X}$ (resp.\ $P_{Z|X}$).
The purpose of broadcast channels with confidential messages is the following.
(1) Alice reliably sends the common message $E$ to Bob and Eve.
(2) Alice confidentially and reliably sends the secret message $S$ to Bob.
Here, we denote the sets of the common messages and the secret messages by $\mathcal{E}$ and $\mathcal{S}$.
Our code is given by
Alice's stochastic encoder $\varphi_{a}$ from $\mathcal{S} \times \mathcal{E}$ to $\mathcal{X}$,
Bob's deterministic decoder $\varphi_{b}: \mathcal{Y} \rightarrow \mathcal{S} \times \mathcal{E}$ and
Eve's deterministic decoder $\varphi_{e}: \mathcal{Z} \rightarrow \mathcal{E}$.
The triple $\varphi= (\varphi_a,\varphi_b,\varphi_e)$ is called a code for broadcast channels with confidential messages.
Then, 
when the common message $E$ and the secret message $S$ 
obey the distribution $P_{S,E}$,
the performance is evaluated by the following quantities.
(1) The sizes of the sets of the common messages and the secret messages, i.e., $|\mathcal{E}|$ and $|\mathcal{S}|$.
(2) Bob's decoding error probability
$P_b[P_{Y|X},\varphi,P_{S,E}]$, which is
the probability $\mathrm{Pr}\{(S, E) \neq \varphi_{b}(Y) \}$
under the distribution 
$(P_{Y|X}\circ \varphi_{a}) \times P_{S,E}$.
(3) Eve's decoding error probability
$P_e[P_{Y|X},\varphi,P_{S,E}]$,
which is the probability 
$\mathrm{Pr}\{ E \neq \varphi_{e}(Z) \}$ 
under the distribution 
$(P_{Z|X}\circ \varphi_{a}) \times P_{S,E}$.
(4) Eve's uncertainty $H(S|Z)[P_{Z|X},\varphi_a,P_{S,E}]$,
which is the conditional entropy $H(S|Z)$ under the distribution 
$(P_{Z|X}\circ \varphi_{a}) \times P_{S,E}$.
Since these quantities are functions of the channel and the code,
such dependencies are denoted by the symbol $[P_{Y|X},\varphi,P_{S,E}]$ in the above notation.
Instead of $H(S|Z)[P_{Z|X},\varphi_a,P_{S,E}]$, we sometimes treat 
(5) leaked information $I(S;Z)[P_{Z|X},\varphi_a,P_{S,E}]$,
which is the mutual information $I(S;Z)$ under the distribution 
$(P_{Z|X}\circ \varphi_{a}) \times P_{S,E}$.

We sometimes need to evaluate the error probability 
when $S$ and/or $E$ is fixed.
In such a case, we denote it by 
$P_b[P_{Y|X},\varphi,P_{E|S=s}]$,
$P_b[P_{Y|X},\varphi,S=s,E=e]$, and
$P_e[P_{Y|X},\varphi,P_{S|E=e}]$.

Now, we review the asymptotic formulation 
of broadcast channels with confidential messages
with 
the $n$-fold discrete memoryless extension
when both of the common messages and the secret messages are 
subject to uniform distributions.
The set $\mathcal{S}_n$ denotes the set of the confidential message
and $\mathcal{E}_n$ does the set of the common message
when the block coding of length $n$ is used.
We shall define the achievability of a rate triple
$(R_1$, $R_e$, $R_0)$,
where 
$R_0$ and $R_1$ are the rates of the common and confidential messages,
and
$R_e$ is the entropy rate conditioned with Eve's random variable
for the confidential message.
For the notational convenience, we fix the base of logarithm,
including one used in entropy and mutual information, to the
base of natural logarithm.

\begin{definition}\cite{csiszar78}
The rate triple $(R_1$, $R_e$, $R_0)$ is said
to be \emph{achievable} for the 
information leakage rate criterion
if the following condition holds.
The size of the sets of the common and confidential messages
are $|\mathcal{E}_n| =e^{n R_0}$ and $|\mathcal{S}_n| =e^{n R_1}$.
The common and confidential messages are subject to 
the uniform and independent distribution on $\mathcal{S}_n$
and $\mathcal{E}_n$.
There exists a sequence of 
the codes $\varphi_n=(\varphi_{a,n},\varphi_{b,n},\varphi_{e,n})$, i.e., 
Alice's stochastic encoder $\varphi_{a,n}$ from 
$\mathcal{S}_n \times \mathcal{E}_n$ to $\mathcal{X}^n$,
Bob's deterministic decoder $\varphi_{b,n}: \mathcal{Y}^n
\rightarrow \mathcal{S}_n \times \mathcal{E}_n$ and
Eve's deterministic decoder $\varphi_{e,n}: \mathcal{Z}^n
\rightarrow \mathcal{E}_n$ such that
\begin{align*}
\lim_{n\rightarrow \infty} 
P_b[P_{Y|X}^n,\varphi_{n},P_{\mix,\mathcal{S}_n,\mathcal{E}_n}]&=0 \\
\lim_{n\rightarrow \infty} 
P_e[P_{Z|X}^n,\varphi_{n},P_{\mix,\mathcal{S}_n,\mathcal{E}_n}]&=0 \\
\liminf_{n\rightarrow \infty} 
\frac{H(S_n|Z^n)[P_{Y|X}^n,\varphi_{a,n},P_{\mix,\mathcal{S}_n,\mathcal{E}_n}]}{n} &\geq
R_e.
\end{align*}
The capacity region with the information leakage rate criterion of the BCC is the closure of 
the achievable rate triples for the information leakage rate criterion.
\end{definition}

\begin{theorem}\label{th1}\cite{csiszar78}
The capacity region with the information leakage rate criterion of the BCC 
is given by
the set of $R_0$, $R_1$ and $R_e$ such that
there exists a Markov chain $U\rightarrow V \rightarrow X \rightarrow 
YZ$ and
\begin{eqnarray*}
R_1 + R_0 &\leq& I(V;Y|U)+\min[I(U;Y),I(U;Z)],\\
R_0 &\leq& \min[I(U;Y),I(U;Z)],\\
R_e & \leq & I(V;Y|U)-I(V;Z|U),\\
R_e & \leq &R_1.
\end{eqnarray*}
\end{theorem}
As described in \cite{liang09},
$U$ can be regarded as the common message,
$V$ the combination of the common and the confidential messages,
and $X$ the transmitted signal.

In this paper, we treat the source-channel universal coding for BCC,
in which, we guarantee the security independently of the choice of 
the source distribution.
While the lower bound of the above conditional entropy 
$H(S_n|Z^n)[P_{Y|X}^n,\varphi_{a,n},P_{S_n,E_n}]$ depends on the 
the source distribution $P_{S_n,E_n}$,
we can find an upper bound of mutual information that 
does not depend on the source distribution,
as is shown in Section \ref{s10}.
As a preparation for the above source-channel universal coding for BCC,
we propose another type of capacity region
for the uniform and independent distributed case
while the non-uniform and dependent case will be treated latter.

\begin{definition}\Label{region-l}
The rate triple $(R_1$, $R_l$, $R_0)$ is said
to be \emph{achievable} for the leaked information criterion
if the following conditions hold.
In this notation, $R_1$, $R_l$, and $R_0$ denote the rates of 
the confidential message, the leaked information, and the common message, respectively.
The size of the sets of the common and confidential messages
are $|\mathcal{E}_n| =e^{n R_0}$ and $|\mathcal{S}_n| =e^{n R_1}$,
and the common and confidential messages are subject to 
the uniform and independent distribution on $\mathcal{S}_n$
and $\mathcal{E}_n$.
There exists a sequence of 
the codes $\varphi_n=(\varphi_{a,n},\varphi_{b,n},\varphi_{e,n})$, i.e., 
Alice's stochastic encoder $\varphi_{a,n}$ from 
$\mathcal{S}_n \times \mathcal{E}_n$ to $\mathcal{X}^n$,
Bob's deterministic decoder $\varphi_{b,n}: \mathcal{Y}^n
\rightarrow \mathcal{S}_n \times \mathcal{E}_n$ and
Eve's deterministic decoder $\varphi_{e,n}: \mathcal{Z}^n
\rightarrow \mathcal{E}_n$ such that
\begin{align*}
\lim_{n\rightarrow \infty} 
P_b[P_{Y|X}^n,\varphi_{n},P_{\mix,\mathcal{S}_n,\mathcal{E}_n}]&=0 \\
\lim_{n\rightarrow \infty} 
P_e[P_{Z|X}^n,\varphi_{n},P_{\mix,\mathcal{S}_n,\mathcal{E}_n}]&=0 \\
\limsup_{n\rightarrow \infty} 
\frac{I(S_n;Z^n)[P_{Y|X}^n,\varphi_{a,n},P_{\mix,\mathcal{S}_n,\mathcal{E}_n}]}{n} & \leq  R_l.
\end{align*}
The capacity region with the leaked information criterion of the BCC is the closure of 
the achievable rate triples.
\end{definition}

The capacity region with the leaked information criterion of the BCC is characterized 
as a corollary of Theorem \ref{th1}.
\begin{corollary}\Label{cor2}
The capacity region with the leaked information criterion of the BCC 
is given by the set of $R_0$, $R_1$ and $R_l$, 
such that
there exists a Markov chain $U\rightarrow V \rightarrow X \rightarrow 
YZ$ and
\begin{eqnarray*}
R_1 + R_0 &\leq& I(V;Y|U)+\min[I(U;Y),I(U;Z)],\\
R_0 &\leq& \min[I(U;Y),I(U;Z)],\\
R_l & \geq & R_1-[I(V;Y|U)-I(V;Z|U)]_+,
\end{eqnarray*}
where $[x]_+:=\max (x,0)$.
That is,
when
$R_1 + R_0 < I(V;Y|U)+\min[I(U;Y),I(U;Z)]$
and
$R_0 < \min[I(U;Y),I(U;Z)]$,
there exists a sequence of 
the codes $\varphi_n=(\varphi_{a,n},\varphi_{b,n},\varphi_{e,n})$, i.e., 
Alice's stochastic encoder $\varphi_{a,n}$ from 
$\mathcal{S}_n \times \mathcal{E}_n$ to $\mathcal{X}^n$,
Bob's deterministic decoder $\varphi_{b,n}: \mathcal{Y}^n
\rightarrow \mathcal{S}_n \times \mathcal{E}_n$ and
Eve's deterministic decoder $\varphi_{e,n}: \mathcal{Z}^n
\rightarrow \mathcal{E}_n$ such that
\begin{align*}
\lim_{n\rightarrow \infty} 
P_b[P_{Y|X}^n,\varphi_{n},P_{\mix,\mathcal{S}_n,\mathcal{E}_n}]&=0 \\
\lim_{n\rightarrow \infty} 
P_e[P_{Z|X}^n,\varphi_{n},P_{\mix,\mathcal{S}_n,\mathcal{E}_n}]&=0 
\end{align*}
and
\begin{align*}
&\limsup_{n\rightarrow \infty} 
\frac{I(S_n;Z^n)[P_{Y|X}^n,\varphi_{a,n},P_{\mix,\mathcal{S}_n,\mathcal{E}_n}]}{n} \\
 \leq & R_1-I[(V;Y|U)-I(V;Z|U)]_+.
\end{align*}
\end{corollary}

\subsection{Our Approach to BCC}\Label{s2-2}
Next, we consider the BCC 
with the single-shot setting when 
the common and confidential messages do not obey 
the uniform and independent distributions on $\mathcal{S}$
and $\mathcal{E}$, i.e., 
the confidential message $S$ may have a correlation with the common messages $E$.
When the confidential message $S$ is independent of the common messages $E$,
\begin{align*}
& I(S;Z)\le I(S;Z E) 
=  I(S;Z |E)+ I(S;E)
=  I(S;Z |E),\\
& I(S;Z) 
= H(S)-H(S|Z) \ge H(S|E)-H(S|Z) \\
= & H(S|E)-(H(S|ZE)+I(S;E|Z)) 
= I(S;Z|E)-I(S;E|Z) \\
\ge & I(S;Z|E)-H(E|Z) 
\ge I(S;Z|E)-H(E|\varphi_e(Z)).
\end{align*}
When the error probability goes to zero,
Fano's inequality guarantees that $H(E|Z)$ goes to zero.
Hence, 
$I(S;Z)$ and $I(S;Z|E)$ have the same asymptotic behaviors.
So, even if we replace $I(S;Z)$ by $I(S;Z|E)$ in Definition \ref{region-l},
we obtain the same capacity region.
However, when the confidential message $S$ 
is dependent on the common messages $E$,
$I(S;Z)$ and $I(S;Z|E)$ have the different asymptotic behavior
as follows.
Since
\begin{align*}
& I(S;Z)= I(S;ZE)-I(S;E|Z) \\
\ge & I(S;E)-H(E|Z) 
\ge  I(S;E)-H(E|\varphi_e(Z)),
\end{align*}
$I(S;Z)$ is asymptotically lower bounded by $I(S;E)$
when the error probability goes to zero.
That is, when the mutual information $I(S;E)$ is positive,
the mutual information $I(S;Z)$ cannot go to zero
because Eve can infer the secret message from the common message.
Thus, it is not suitable to treat the mutual information $I(S;Z)$ as leaked information from $Z$.
Hence, we adopt the conditional mutual information $I(S;Z|E)$ as leaked information from $Z$.

\begin{remark}
Csisz\'ar and K\"orner \cite{csiszar78} treated BCC with non-uniform information source.
However, their formulation was different from our formulation in the following point.
In their formulation, they fixed a correlated non-uniform 
distribution $P_{S,E}$
on ${\cal S}\times {\cal E}$
and 
assumed that 
the information source $S_n$ and $E_n$ obey its $n$-fold 
independent and identical distribution $P_{S,E}^n $.
In addition to this, their code depends on the distribution $P_{S,E}$.
However, in our formulation, we do not assume 
the independent and identical distributed condition for the distribution 
$P_{S_n,E_n}$ of the information source $S_n$ and $E_n$.
This is because
information source is not given as an independent and identical distribution or known, in general.
Hence, we study a universal code independent of the distribution 
$P_{S_n,E_n}$ of sources in Section \ref{s10}.
Thus, our code is useful for a realistic case.
\end{remark}

\section{Broadcast Channels with Degraded Message Sets}\label{sec:bcd}
\subsection{Capacity Region}\Label{s4-2}
Next, we review the broadcast channel with degraded message sets (abbreviated as
BCD) considered by K\"orner and Marton \cite{korner77}
in the single-shot setting.
If we set $R_e= 0$ in the BCC,
the secrecy requirement is removed from BCC, and
the coding problem is equivalent to BCD. 
In this problem, we treat 
the private message $S_{\mathrm{p}}$ taking values in ${\cal S}_{\mathrm{p}}$
and the common message $S_{\mathrm{c}}$ taking values in ${\cal S}_{\mathrm{c}}$.
\begin{corollary}\label{cor:bcd}\cite{korner77}
The capacity region of the BCD is given by
the pair of the rate $R_\mathrm{c}$ of common message 
and the rate $R_{\mathrm{p}}$ of private message such that
there exists a Markov chain $U\rightarrow V = X \rightarrow YZ$ and
\begin{eqnarray*}
R_{\mathrm{c}} &\leq& \min[I(U;Y),I(U;Z)],\\
R_{\mathrm{c}}+R_{\mathrm{p}} & \leq & I(V;Y|U)+\min[I(U;Y),I(U;Z)].
\end{eqnarray*}
\end{corollary}

Note that the statement of our Corollary \ref{cor:bcd} is the same as \cite[Corollary 5]{csiszar78} 
and different from \cite{korner77}. 
However, as is stated in \cite[Remark 5]{csiszar78}, 
the equivalence between the two statements
can be easily shown by some algebra.

Here, we only consider a sequence of codes that achieves 
the rate pair $(R_{\mathrm{c}},R_{\mathrm{p}})$ satisfying
\begin{eqnarray}
R_{\mathrm{c}}  <  \min[I(U;Y),I(U;Z)],~
R_{\mathrm{p}}  <  I(V;Y|U). \Label{11-9-1}
\end{eqnarray}
For a given Markov chain $U\rightarrow V = X \rightarrow YZ$,
we construct an ensemble of codes by the following random coding
with the single-shot setting,
which is mathematically equivalent to 
the construction by Kaspi and Merhav \cite{kaspi11}.

\begin{ensemble}[Kaspi and Merhav \protect{\cite[Section II]{kaspi11}}]\footnote{%
A code ensemble and a code construction play a distinguished role in this paper
because they give a procedure to make our codes.
Hence, we give them serial numbers that are separate from other environments, Theorems, Lemmas, and Remarks.
Although both of a code ensemble and a code construction 
give a procedure for our code,
the procedure by a code ensemble is less practical,
and that by a code construction is more practical.
To clarify this difference, we assigned 
one of two environments to them dependently of their properties.
Code constructions will be given in Section \ref{s8} after 
code ensembles are presented in the previous sections.}\Label{con0}
For an arbitrary element $s_{\mathrm{c}} \in {\cal S}_{\mathrm{c}}$,
$\Phi_{\mathrm{c}}(s_{\mathrm{c}})$ is the random variable taking values in ${\cal U}$ and is subject to the distribution $P_U$,
and is independent of $\Phi_{\mathrm{c}}(s_{\mathrm{c}}')$ with  $s_{\mathrm{c}}' \neq s_{\mathrm{c}} \in {\cal S}_{\mathrm{c}}$.
For an arbitrary element $s_{\mathrm{p}} \in {\cal S}_{\mathrm{p}}$,
$\Phi_{\mathrm{p}}(s_{\mathrm{c}},s_{\mathrm{p}})$ is the random variable taking values in ${\cal V}$,
is independent of $\Phi_{\mathrm{p}}(s_{\mathrm{c}}',s_{\mathrm{p}}')$ with  $s_{\mathrm{c}}' \neq s_{\mathrm{c}}$,
and depends on the random variable $\Phi_{\mathrm{c}}(s_{\mathrm{c}})$.
Under the condition $\Phi_{\mathrm{c}}(s_{\mathrm{c}})=u$,
the random variable 
$\Phi_{\mathrm{p}}(s_{\mathrm{c}},s_{\mathrm{p}})$
is subject to the distribution $P_{V|U=u}$
and is conditionally independent of
$\Phi_{\mathrm{p}}(s_{\mathrm{c}},s_{\mathrm{p}'})$
with  $s_{\mathrm{p}}' \neq s_{\mathrm{p}}$.
Bob's decoder $\Phi_{b}$ and Eve's decoder $\Phi_{e}$ are defined as
the maximum likelihood decoders.
The quartet $(\Phi_{\mathrm{p}},\Phi_{\mathrm{c}},\Phi_{b},\Phi_{e})$
is abbreviated as $\Phi$.

Here, 
the all values of 
the random variables
$\{\Phi_{\mathrm{c}}(s_{\mathrm{c}})\}_{s_{\mathrm{c}}}$
and
$\{\Phi_{\mathrm{p}}(s_{\mathrm{c}},s_{\mathrm{p}})\}_{s_{\mathrm{c}},s_{\mathrm{p}}}$
are disclosed to all players
prior to the real communication
because these random variables decides our code.
\end{ensemble}

\begin{lemma}\cite[Theorem 1 and Section IV]{kaspi11}\Label{lem0}
The above ensemble of codes $\Phi$ satisfies the following inequalities.
\begin{align}
\rE_{\Phi} P_b[P_{Y|V},\Phi]
\le &
|{\cal S}_{\mathrm{p}}|^{\rho} e^{E_0(-\rho| P_{Y|V}, P_{V|U},P_U)}
\nonumber \\
&+
(|{\cal S}_{\mathrm{c}}||{\cal S}_{\mathrm{p}}|)^{\rho} e^{E_0(-\rho| P_{Y|U,V},P_{U,V})} \Label{ineq-31}\\
\rE_{\Phi} P_e[P_{Z|V},\Phi]
\le &
|{\cal S}_{\mathrm{c}}|^{\rho} e^{E_0(-\rho| P_{Z|U},P_U)},\Label{ineq-32}
\end{align}
where $E_0(-\rho| P_{Z|U},P_U)$ and $E_0(-\rho| P_{Y|V}, P_{V|U},P_U)$
are defined in \eqref{phid} and \eqref{phid-2}.
\end{lemma}

Here, we should remark that Inequalities \eqref{ineq-31} and \eqref{ineq-32}
hold for any distribution over the messages
because 
the proof by \cite{kaspi11} does not make any assumption for the distribution over the messages.

Due to Lemma \ref{lem0},
Markov inequality guarantees that
\begin{align*}
&\rm{Pr} \Omega_1 < \frac{1}{2}, \quad
\rm{Pr} \Omega_2 < \frac{1}{2}\\
&\Omega_1 :=\Biggl\{ \!\!\!
\begin{array}{ll}
P_b[P_{Y|V},\Phi,P_{\mix,\mathcal{S}_{\mathrm{p}},\mathcal{S}_{\mathrm{c}}}]
> \!\! & \!\!
2|{\cal S}_{\mathrm{p}}|^{\rho} e^{E_0(-\rho| P_{Y|V}, P_{V|U},P_U)} \\
&\!\!+
2(|{\cal S}_{\mathrm{c}}||{\cal S}_{\mathrm{p}}|)^{\rho} e^{E_0(-\rho| P_{Y|U,V},P_{U,V})} 
\end{array}\!\!\!
\Biggr\} \\
&\Omega_2 :=
\{
P_e[P_{Z|V},\Phi,P_{\mix,\mathcal{S}_{\mathrm{p}},\mathcal{S}_{\mathrm{c}}}]
> 
2|{\cal S}_{\mathrm{c}}|^{\rho} e^{E_0(-\rho| P_{Z|U},P_U)} \}.
\end{align*}
Since 
$\rm{Pr} (\Omega_1 \cup \Omega_2)<1$,
we have
$\rm{Pr} (\Omega_1^c \cap \Omega_2^c)>0$.
That is,
for an arbitrary distribution $P_{\mathcal{S}_{\mathrm{p}},\mathcal{S}_{\mathrm{c}}}$ over the messages,
there exists a code $\varphi$ such that
\begin{align}
 P_b[P_{Y|V},\varphi,
P_{\mathcal{S}_{\mathrm{p}},\mathcal{S}_{\mathrm{c}}}]
\le &
2|{\cal S}_{\mathrm{p}}|^{\rho} e^{E_0(-\rho| P_{Y|V}, P_{V|U},P_U)} \nonumber \\
&+
2(|{\cal S}_{\mathrm{c}}||{\cal S}_{\mathrm{p}}|)^{\rho} e^{E_0(-\rho| P_{Y|U,V},P_{U,V})} \Label{12-28-3}
\\
P_e[P_{Z|V},\varphi,P_{\mathcal{S}_{\mathrm{p}},\mathcal{S}_{\mathrm{c}}}]
\le &
2|{\cal S}_{\mathrm{c}}|^{\rho} e^{E_0(-\rho| P_{Z|U},P_U)}.
\Label{12-28-4}
\end{align}
Now, we apply 
the above inequalities to the $n$-fold discrete memoryless extension.
Then, 
for an arbitrary distribution $P_{\mathcal{S}_{p,n},\mathcal{S}_{c,n}}$ 
over the messages,
there exists a sequence of codes $\varphi_n$
with the rate of common message $R_{\mathrm{c}}$ and the rate of private message $R_{\mathrm{p}}$ 
of length $n$
such that
\begin{align}
 P_b[P_{Y|V}^n,\varphi_n,
P_{\mathcal{S}_{p,n},\mathcal{S}_{c,n}}
]
\le &
2 e^{n (\rho R_{\mathrm{p}}+ E_0(-\rho| P_{Y|V}, P_{V|U},P_U))}\nonumber \\
&+
2 e^{n (\rho (R_{\mathrm{p}}+R_{\mathrm{c}}) + E_0(-\rho| P_{Y|U,V},P_{U,V}))}
\Label{ineq-6-b--}
\\
 P_e[P_{Z|V}^n,\varphi_n,P_{\mathcal{S}_{p,n},\mathcal{S}_{c,n}}]
\le &
2 e^{n (\rho R_{\mathrm{c}}+ E_0(-\rho| P_{Z|U},P_U))}.
\Label{ineq-6-c--}
\end{align}
The above values go to zero 
under the condition (\ref{11-9-1}),
because the condition (\ref{11-9-1}) guarantees that both exponents are positive
with sufficiently small $\rho>0$.

Indeed, Kaspi and Merhav \cite{kaspi11}
derived a better bound than (\ref{ineq-32}) by employing four parameters
even in the single-shot setting.
The bound (\ref{ineq-32}) can be seen as a special case of 
Kaspi and Merhav \cite{kaspi11}'s bound.
Since the bound (\ref{ineq-32}) can derive the capacity region of SMC,
we only use the bound (\ref{ineq-32}) for simplicity.

\subsection{Universal Code for BCD}\Label{s4-3}
K\"orner and Sgarro \cite{korner80} provided the code 
that attains the above rate region universally for source and channel in the following sense.

\begin{theorem}\cite{korner80}\Label{lem-11-25-3-b}
For an arbitrary real number $\epsilon>0$, there exists an integer $N$ satisfying the following.
For an arbitrary integer $n \ge N$, a given joint type $Q_{VU}$ 
of length $n$ on the sets $\mathcal{V}\times \mathcal{U}$, 
and rates $R_{\mathrm{p}}$ and $R_{\mathrm{c}}$,
there exists a code $\varphi_n$ with the rates $R_{\mathrm{p}}$ and $R_{\mathrm{c}}$
such that
\begin{align}
 &P_b[W^n,\varphi_n,{S}_{p,n}=s_{p,n},{S}_{c,n}=s_{c,n}] \nonumber \\
 \leq & \exp(-n[ \tilde{E}^{b}(R_{\mathrm{p}}, R_{\mathrm{c}}, W^Y \times Q_{U,V})-\epsilon ]),
\Label{Haya-51-d}\\
 &P_e[W^n,\varphi_n,{S}_{p,n}=s_{p,n},{S}_{c,n}=s_{c,n}] \nonumber \\
 \leq & \exp(-n[ \tilde{E}^{e}(R_{\mathrm{c}},W^Z  \times Q_{U,V})-\epsilon ])
\Label{Haya-52-d}
\end{align}
for 
any $s_{p,n}\in \mathcal{S}_{p,n}$, $s_{c,n}\in \mathcal{S}_{c,n}$
and any $W \in  \mathcal{W}(\mathcal{V}$, $\mathcal{Y}\times$ 
$\mathcal{Z})$,
where the exponents
$\tilde{E}^{b}(R_{\mathrm{p}}, R_{\mathrm{c}}, W^Y \times Q_{U,V})$
and
$\tilde{E}^{e}(R_{\mathrm{c}},W^Z  \times Q_{U,V})$ 
are defined in \eqref{1-31-1} and \eqref{1-31-2}, respectively.
\end{theorem}

\section{General Channel Resolvability}\Label{s4}
In the wire-tap channel model,
when the dummy message obeys the uniform distribution,
channel resolvability \cite{han93} can be used for guaranteeing the security \cite{hay-wire}.
In this paper, we consider the security of SMC with non-uniform and dependent secret messages.
For the analysis of this case,
we have to consider 
the secrecy when 
the dummy message does not necessarily obey the uniform distribution.
Hence,
the security evaluation \cite{hay-wire} based on the original channel resolvability 
cannot be extended to 
the security of SMC with non-uniform and dependent secret messages.
Thus, we need a generalization of 
channel resolvability.
In this section, we propose a generalization of channel resolvability
in the single-shot setting.

First, we fix a channel $W$ from the alphabet ${\cal X}$ 
to the alphabet ${\cal Y}$.
For a fixed distribution $P_X$ on ${\cal X}$,
we focus on an encoder $\Lambda$ from the message set 
${\cal A}$ to the alphabet ${\cal X}$.
The purpose of the encoder $\Lambda$
is approximation of the average output distribution $W \circ P_X$
by the output distribution with input $\Lambda(A)$.
The original channel resolvability \cite{han93} treats
the minimum asymptotic rate of 
$|{\cal A}|$
such that the output distribution $W \circ \Lambda \circ P_{\mix,\mathcal{A}}$
can approximate the average output distribution $W \circ P_X$
with a suitable choice of $\Lambda$
in the sense that the variational distance goes to zero.
In the single-shot setting,
the problem can be converted to the following way:
How well the given average output distribution 
$W \circ P_X$ can be approximated 
by the output distribution $W \circ \Lambda \circ P_{\mix,\mathcal{A}}$
when the cardinality $|\mathcal{A}|$ is less than a given amount.
In this paper, we consider this approximation problem
when the message $A$ does not obey the uniform distribution $P_{\mix,\mathcal{A}}$.
Since our problem can be regarded as a generalization of channel resolvability, 
it is called general channel resolvability,
which is essential for the secure multiplex coding with common messages with dependent and non-uniform secret messages.

Now, we apply the random coding on the alphabet $A$ with the probability distribution $P_A$.
For an arbitrary $a \in {\cal A}$, $\Lambda(a)$ is the random variable subject to the distribution $P_X$ on ${\cal X}$.
For $a \neq a' \in {\cal A}$, 
$\Lambda(a)$ is independent of $\Lambda(a')$.
Then, 
the random encoder $\Lambda:=\{\Lambda(a)\}_{a \in {\cal A}}$ 
gives the map from ${\cal A}$ to ${\cal X}$ as $a\mapsto \Lambda(a)$.

Then, we have the following theorem:
\begin{theorem}[General channel resolvability]\Label{lem-01}
For $\rho \in (0,1]$, we have
\begin{align*}
& \rE_{\Lambda} e^{\rho D( W \circ \Lambda \circ P_A \| W \circ P_X )}
\le
\rE_{\Lambda} e^{\psi(\rho| W \circ \Lambda \circ P_A \| W\circ P_X )} \\
\le &
1+ e^{-\rho H_{1+\rho}(A)}e^{\psi(\rho|W,P_X)}.
\end{align*}
\end{theorem}   
By applying Jensen inequality to the function $x \mapsto e^x$,
Theorem \ref{lem-01} yields 
\begin{align*}
& \rE_{\Lambda} D( W \circ \Lambda \circ P_A \| W\circ P_X )
\le 
\frac{1}{\rho}\log 
\rE_{\Lambda} e^{\rho D( W \circ \Lambda \circ P_A \| W\circ P_X)}\nonumber \\
\le &
\frac{1}{\rho}\log (1+ e^{-\rho H_{1+\rho}(A)}e^{\psi(\rho|W,P_X)}),
\end{align*}
which is non-uniform generalization of \cite[Lemma 2]{hay-wire}.
This theorem will be used for the proof of Theorem \ref{lem1}.

\begin{IEEEproof}
Due to (\ref{12-18-1}), we have
\begin{align*}
\rho D( W \circ \Lambda \circ P_A \| W\circ P_X)
\le
\psi(\rho| W \circ \Lambda \circ P_A \| W\circ P_X).
\end{align*}
The average of $e^{\psi(\rho| W \circ \Lambda \circ P_A \| W\circ P_X)}$ is evaluated as
\begin{align}
& \rE_{\Lambda} 
e^{\psi(\rho| W \circ \Lambda \circ P_A \| W\circ P_X)}\nonumber \\
=&
\rE_{\Lambda} 
\sum_{y}
\Big(\sum_{a} P_A(a) W_{\Lambda(a)}(y)\Big)^{1+\rho} (W\circ P_X)(y)^{-\rho} \nonumber \\
=&
\rE_{\Lambda} 
\sum_{y}
\Big(\sum_{a} P_A(a) W_{\Lambda(a)}(y)\Big)
\Big(\sum_{a'} P_A(a') W_{\Lambda(a')}(y)\Big)^\rho
(W\circ P_X)(y)^{-\rho} \nonumber \\
=&
\sum_{y}
\sum_{a} 
\Bigl( \rE_{\Lambda(a)} 
P_A(a) W_{\Lambda(a)}(y) 
\rE_{\Lambda|\Lambda(a)} 
\Big(P_A(a) W_{\Lambda(a)}(y)\nonumber \\
& \quad +
\sum_{a'\neq a } P_A(a') W_{\Lambda(a')}(y)\Big)^\rho
(W\circ P_X)(y)^{-\rho} \Bigr) \nonumber \\
\le &
\sum_{y}
\sum_{a} 
\Bigl( \rE_{\Lambda(a)} 
P_A(a) W_{\Lambda(a)}(y)
\Big(P_A(a) W_{\Lambda(a)}(y) \nonumber \\
&\quad +
\rE_{\Lambda|\Lambda(a)} 
\sum_{a'\neq a } P_A(a') W_{\Lambda(a')}(y)\Big)^\rho
(W\circ P_X)(y)^{-\rho} \Bigr) \Label{1}\\
= &
\sum_{y}
\sum_{a} 
\Bigl( \rE_{\Lambda(a)} 
P_A(a) W_{\Lambda(a)}(y)
\Big(P_A(a) W_{\Lambda(a)}(y) \nonumber \\
&\quad +
\sum_{a'\neq a } P_A(a') (W\circ P_X)(y)\Big)^\rho
(W\circ P_X)(y)^{-\rho} \Bigr) \nonumber \\
\le &
\sum_{y}
\sum_{a} 
\Bigl( \rE_{\Lambda(a)} 
P_A(a) W_{\Lambda(a)}(y)
\bigl(P_A(a) W_{\Lambda(a)}(y)+(W\circ P_X)(y)\bigr)^\rho \nonumber \\
&\quad \cdot (W\circ P_X)(y)^{-\rho} \Bigr) 
\Label{2} \\
\le &
\sum_{y}
\sum_{a} 
\rE_{\Lambda(a)} 
P_A(a) W_{\Lambda(a)}(y)
\nonumber \\
& \quad
(P_A(a)^\rho W_{\Lambda(a)}(y)^\rho 
+(W\circ P_X)(y)^\rho)
(W\circ P_X)(y)^{-\rho} \Label{3-b} \\
= &
\sum_{y}
\sum_{a} 
\rE_{\Lambda(a)} 
P_A(a) W_{\Lambda(a)}(y)
(1+ P_A(a)^\rho W_{\Lambda(a)}(y)^\rho (W\circ P_X)(y)^{-\rho} )\nonumber \\
= &
1+ \sum_{y}
\sum_{a} 
\rE_{\Lambda(a)} 
P_A(a)^{1+\rho} W_{\Lambda(a)}(y)^{1+\rho} (W\circ P_X)(y)^{-\rho} \nonumber \\
= &
1+ 
\sum_{a} 
P_A(a)^{1+\rho} 
\sum_{y}
\sum_{x}
P_X(x) W_{x}(y)^{1+\rho} (W\circ P_X)(y)^{-\rho}  \nonumber \\
= &
1+ 
(\sum_{a} P_A(a)^{1+\rho} )e^{\psi(\rho|W,P_X)}.\nonumber
\end{align}
In the above derivation, 
(\ref{1}) follows from the concavity of $x \mapsto x^\rho$,
(\ref{2}) follows from $\sum_{a'\neq a } P_A(a') \le 1$,
(\ref{3-b}) follows from the inequality $(x+y)^\rho \le x^\rho+y^\rho$.
\end{IEEEproof}

Next, in order to reduce the complexity of encoding,
we consider the case when ${\cal X}$ and ${\cal A}$ are
Abelian groups.
We introduce the following condition for the ensemble for injective homomorphisms
$F$ from ${\cal A}$ to ${\cal X}$.
\begin{condition}\Label{C2-b}
Let $F$ be a random variable that takes its values on
injective\footnote{The condition of injectivity is not 
necessarily for Theorem \ref{Lee3}.
However, the injectivity for $F$ will needed in the discussion in Subsection \ref{s8-2}.
Hence, to avoid to make so many conditions, we assume the injectivity, here.} homomorphisms from ${\cal A}$ to ${\cal X}$.
For arbitrary elements 
$x \neq 0 \in {\cal X}$ and $a \neq 0 \in {\cal A}$,
the relation $F(a)=x$ holds with probability at most 
$\frac{1}{|{\cal X}|-1}$.
\end{condition}

When ${\cal X}$ and ${\cal A}$ are vector spaces over a finite field $\bF_q$,
the set of all injective homomorphisms from ${\cal A}$ to ${\cal X}$ satisfies Condition \ref{C2-b}.

\begin{remark}\Label{3-25R}
When ${\cal X}$ and ${\cal A}$ have the same Abelian group structure as
the vector space over a finite field $\bF_2$ with the the same dimension $k$,
these can be regarded as the finite filed $\bF_{2^k}$.
For $y \in \bF_{2^k}$,
the homomorphism $f_y$ from ${\cal A}$ to ${\cal X}$
from ${\cal A}$ to ${\cal X}$
is defined by the multiplication as $f_y:x\to xy $.
Then, as mentioned in \cite[Remark 9]{MH11}, 
when the random variable $Y$ chosen in $\bF_{2^k}$ subject to the uniform distribution,
the function-valued random variable $f_Y$ satisfies Condition \ref{C2-b}.
To realize the function-valued random variable $f_Y$,
we need to choose a finite filed $\bF_{2^k}$ with efficient multiplication.
Constructions of such a finite filed $\bF_{2^k}$ are given in 
\cite[Appendix D]{HT13},
\cite[Section 7.3.1]{Assche}.
\end{remark}

We choose another random variable $G$ in ${\cal X}$ 
that obeys the uniform distribution on ${\cal X}$ 
and is independent of the choice of $F$.
Then, we define a map $\Lambda_{F,G}(a):=F(a)+G$
and have the following theorem:

\begin{theorem}[Algebraic channel resolvability]\Label{Lee3}
Under the above choice,
we obtain
\begin{align}
& 
\rE_{F,G} e^{\rho {D}( W \circ \Lambda_{F,G} \circ P_A  \| W \circ P_{\mix,\mathcal{X}})} 
\le 
\rE_{F,G} e^{{\psi}(\rho| W \circ \Lambda_{F,G} \circ P_A  \| W \circ P_{\mix,\mathcal{X}})} \nonumber \\
\le &
1+ e^{-\rho H_{1+\rho}(A)}e^{{\psi}(\rho|W, P_{\mix,\mathcal{X}})}.
\Label{9-19-2-b}
\end{align}
\end{theorem}   
This theorem will be used for the proof of Lemma \ref{lem2-1},
which is essential for the proof of Theorem \ref{lem2}.

\begin{IEEEproof}
We introduce the random variable $Z_a:=\Lambda_{F,G}(a) =F(a)+G$.
The random variable $Z_a$ is independent of the choice of $F$.
For $a'\in {\cal A}$, $\Lambda_{F,G}(a')= F(a'-a)+Z_a$.
Since 
$(|{\cal X}|-1)\rE_{F|Z_a} W_{\Lambda_{F,G}(a)}(y)
=(|{\cal X}|-1)\rE_{F} W_{F(a'-a)+Z_a}(y)
\le \sum_x W_x (y)
=|{\cal X}| W \circ P_{\mix,\mathcal{X}} (y)$ for $a\in \mathcal{A}$ and $y \in {\cal Y}$,
we obtain $\rE_{F|Z_a} W_{\Lambda_{F,G}(a)}(y)
\le \frac{|{\cal X}|}{|{\cal X}|-1} W \circ P_{\mix,\mathcal{X}}(y)$ 
for $a\in {\cal A}$ and $y \in {\cal Y}$.
Further, 
since $F$ is injective, we have $|{\cal A}|\le |{\cal X}| $,
which implies
$\sum_{a} P_A(a)^2 \ge \frac{1}{|{\cal A}|}
\ge \frac{1}{|{\cal X}|}$.
Hence, since $x \mapsto x^\rho$ is concave,
we obtain
\begin{align}
\sum_{a} P_A(a) (\frac{1-P_A(a)}{1-1/|{\cal X}|})^\rho
\le
(\frac{1-\sum_{a} P_A(a)^2 }{1-1/|{\cal X}|})^\rho
\le
(\frac{1-1/|{\cal X}| }{1-1/|{\cal X}|})^\rho
=1 .\Label{ineq-202}
\end{align}
Our proof of Theorem \ref{lem-01} can be applied to our proof of 
Theorem \ref{Lee3}
by replacing $\Lambda(a)$, $\Lambda|\Lambda(a)$, and $P_X$ 
by $Z_a$, $F|Z_a$ and $P_{\mix,\mathcal{X}}$.
Then, we obtain
\begin{align}
& \rE_{F,G} e^{\psi(\rho| W \circ \Lambda_{F,G} \circ P_A \| W \circ P_{\mix,\mathcal{X}})} \nonumber \\
\le &
\sum_{y}
\sum_{a} 
\Bigl(
\rE_{Z_a} 
P_A(a) W_{\Lambda_{F,G}(a)}(y)
\Big(P_A(a) W_{\Lambda_{F,G}(a)}(y)\nonumber \\
& +
\rE_{F|Z_a} 
\sum_{a'\neq a } P_A(a') W_{\Lambda_{F,G}(a')}(y)\Big)^\rho
W \circ P_{\mix,\mathcal{X}}(y) ^{-\rho} \Bigr) \Label{11-20-1eq}\\
\le &
\sum_{y}
\sum_{a} 
\Bigl(
\rE_{Z_a} 
P_A(a) W_{Z_a}(y)
\Bigl(P_A(a) W_{Z_a}(y)\nonumber \\
& +
\frac{|{\cal X}|}{|{\cal X}|-1} \sum_{a'\neq a } P_A(a') 
W \circ P_{\mix,\mathcal{X}}(y) \Bigr)^\rho
W \circ P_{\mix,\mathcal{X}}(y)^{-\rho} \Bigr)\Label{11-22-1b} \\
\le &
\sum_{y}
\sum_{a} 
\Bigl(
\rE_{Z_a} 
P_A(a) W_{Z_a}(y)
\Bigl(P_A(a) W_{Z_a}(y)\nonumber\\
&+
\frac{1-P_A(a)}{1-1/|{\cal X}|}W \circ P_{\mix,\mathcal{X}} (y)
\Bigr)^\rho
W \circ P_{\mix,\mathcal{X}}(y)^{-\rho} \Bigr)\Label{11-22-1}\\
\le &
\sum_{y}
\sum_{a}
\Bigl( 
\rE_{Z_a} 
P_A(a) W_{Z_a}(y)
\Bigl(P_A(a)^\rho W_{Z_a}(y)^\rho\nonumber\\
&+
(\frac{1-P_A(a)}{1-1/|{\cal X}|})^\rho W \circ P_{\mix,\mathcal{X}}(y)^\rho \Bigr)
W \circ P_{\mix,\mathcal{X}}(y)^{-\rho} \Bigr)\Label{3} \\
= &
\sum_{y}
\sum_{a} 
\Bigl( 
\rE_{Z_a} 
P_A(a) W_{Z_a}(y)
((\frac{1-P_A(a)}{1-1/|{\cal X}|})^\rho \nonumber\\
&+ P_A(a)^\rho W_{Z_a}(y)^\rho W \circ P_{\mix,\mathcal{X}}(y)^{-\rho}) \Bigr)\nonumber\\
= &
\sum_{a} P_A(a) (\frac{1-P_A(a)}{1-1/|{\cal X}|})^\rho \nonumber\\
&+ \sum_{y}
\sum_{a} 
\rE_{Z_a} 
P_A(a)^{1+\rho} W_{Z_a}(y)^{1+\rho} W \circ P_{\mix,\mathcal{X}}(y)^{-\rho})\nonumber  \\
= &
\sum_{a} P_A(a) (\frac{1-P_A(a)}{1-1/|{\cal X}|})^\rho \nonumber\\
&+ 
\sum_{a} 
P_A(a)^{1+\rho} 
\sum_{y}
\sum_{x}
P_X(x) W_{x}(y)^{1+\rho} W \circ P_{\mix,\mathcal{X}}(y)^{-\rho}) \nonumber\\
\le &
1
+ 
(\sum_{a} P_A(a)^{1+\rho} )e^{\psi(\rho|W,P_{\mix,\mathcal{X}})}.
\end{align}
In the above derivation, 
\eqref{11-20-1eq} follows in the same way as \eqref{1},
(\ref{11-22-1b}) follows from Condition \ref{C2-b}, 
(\ref{11-22-1}) follows from $\sum_{a'\neq a } P_A(a') \le 1$,
(\ref{3}) follows from the inequality $(x+y)^\rho \le x^\rho+y^\rho$.
The final inequality follows from (\ref{ineq-202}).
\end{IEEEproof}


\begin{figure*}
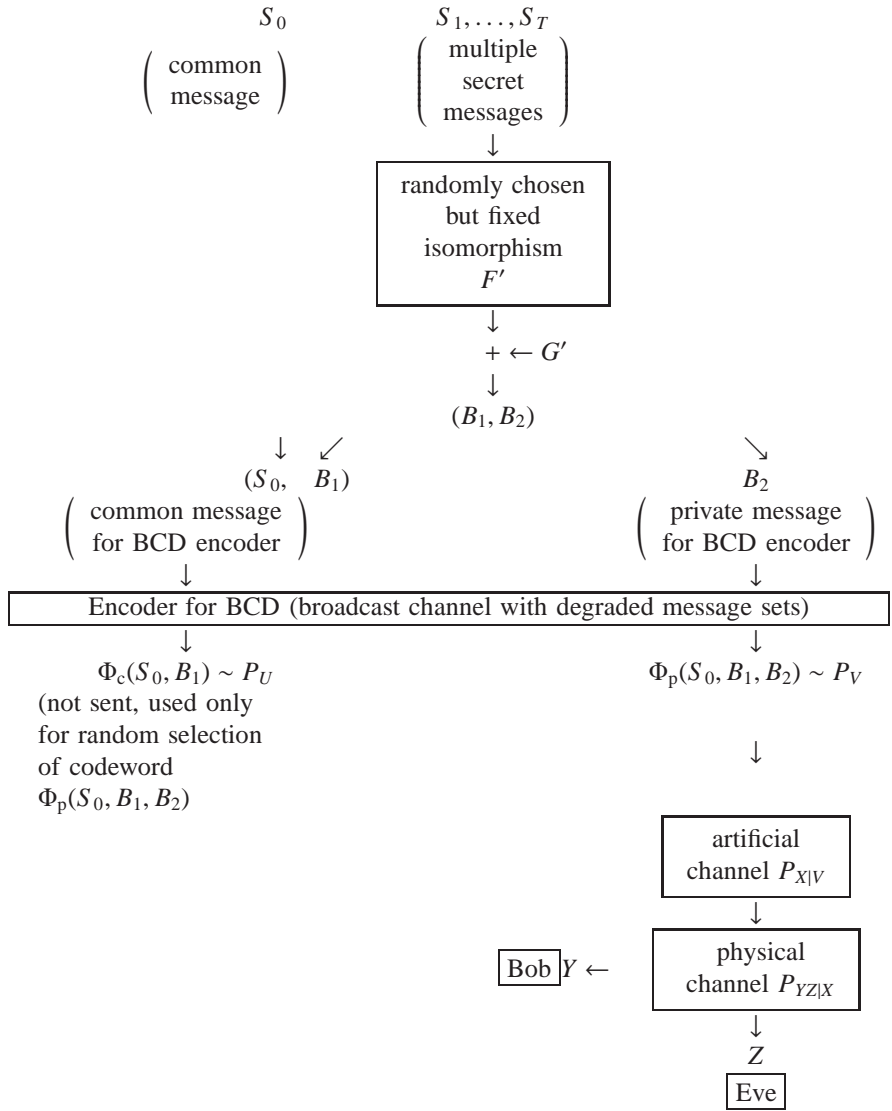

\[
\begin{array}{cccc}
\multicolumn{1}{r}{S_0}&&S_1,\ldots,S_T&\\
\multicolumn{1}{r}{\left(
\begin{tabular}{c}common\\
message
\end{tabular}\right)}&&\left(
\begin{tabular}{c}multiple\\
secret\\
messages
\end{tabular}\right)&
\\
&&\downarrow&\\
&&\fbox{\begin{tabular}{c}randomly chosen\\
but fixed\\
isomorphism\\
$F'$
\end{tabular}}&\\
&&\downarrow&\\
&&\hspace*{2.5eM}+\leftarrow  G'   &\\
&&\downarrow&\\
&&(B_1,B_2)&\\
\multicolumn{1}{r}{\downarrow}&\swarrow&&\searrow\\
\multicolumn{1}{r}{(S_0,}&\multicolumn{1}{l}{B_1)}&&B_2\\
\multicolumn{2}{c}{%
\left(
\begin{tabular}{c}common message\\
for BCD encoder
\end{tabular}\right)}&&\left(
\begin{tabular}{c}private message\\
for BCD encoder
\end{tabular}\right)\\
\multicolumn{2}{c}{\downarrow} &&\downarrow\\\hline
\multicolumn{4}{|c|}{\mbox{Encoder for BCD (broadcast channel with degraded
message sets)}}\\\hline
\multicolumn{2}{c}{\downarrow}&&\downarrow\\
\multicolumn{2}{c}{\Phi_{\mathrm{c}}(S_0,B_1)\sim P_U}&&
\Phi_{\mathrm{p}} (S_0,B_1,B_2)\sim P_V\\
\multicolumn{1}{l}{\hbox{\begin{tabular}{l}(not sent, used only\\
for random selection\\
of codeword \\
$\Phi_{\mathrm{p}}(S_0,B_1,B_2)$\end{tabular}
}}&&&\downarrow\\
&&&\fbox{\begin{tabular}{c}artificial\\
channel $P_{X|V}$\end{tabular}}\\
&&&\downarrow\\
&&\multicolumn{1}{r}{\fbox{Bob}Y\leftarrow}&\fbox{
\begin{tabular}{c}physical\\
channel $P_{YZ|X}$\end{tabular}}\\
&&&\begin{array}{c}\downarrow\\
Z\\
\fbox{Eve}
\end{array}
\end{array}
\]
\caption{Communication structure
used in Sections \ref{s5}--\ref{s9}}\label{fig:encoder}
\end{figure*}

In the following, we assume that
the input alphabet $\cX$ is an Abelian group,
and 
an action of $\cX$ on the output alphabet $\cY$
is given as $x\cdot y $ for $x \in \cX$ and $y \in \cY$.
A channel $W$ from $\cX$ to $\cY$ is regular
in the sense of Delsarte-Piret \cite{delsarte82},
if there is a probability distribution $P_Y$ such that
\begin{align*}
W_x(y)=P_Y(x\cdot y).
\end{align*}
Since a regular channel $W$ satisfies 
\begin{align*}
D( W \circ \Lambda_{F,g} \circ P_A \| W \circ P_{\mix,\mathcal{X}})
=
D( W \circ \Lambda_{F,g'}\circ P_A \| W \circ P_{\mix,\mathcal{X}})
\end{align*}
for any $g,g' \in \cX$,
we obtain the following corollary.
This corollary implies that
we do not need the additional random variable $G$
in the regular channel case.
\begin{corollary}\Label{cor1}
When the channel $W$ is a regular channel given by a distribution 
$P_Y$
on $\cY$,
we obtain
\begin{align}
&
\rE_{F} e^{\rho {D}( W\circ \Lambda_{F,g} \circ P_A \| W \circ P_{\mix,\mathcal{X}})} \le 
\rE_{F} e^{{\psi}(\rho| W\circ \Lambda_{F,g} \circ P_A \| W \circ P_{\mix,\mathcal{X}})} \nonumber \\
\le &
1+ e^{-\rho H_{1+\rho}(A)}e^{{\psi}(\rho|W,P_{\mix,\mathcal{X}})}
= 
1+ e^{-\rho H_{1+\rho}(A)}e^{{\psi}(\rho|P_Y \| \overline{P}_Y)}
\Label{9-19-2-c}
\end{align}
for any $g \in \cX$,
where $\overline{P}_Y(y):=\sum_{x} P_{\mix,\mathcal{X}}(x) P_Y(x\cdot y) $.
\end{corollary}   

\begin{proof}
Due to Theorem \ref{lem-01},
it is enough to show 
${\psi}(\rho|W,P_{\mix,\mathcal{X}})=
{\psi}(\rho|P_Y \| \overline{P}_Y)$.
Since $\overline{P}_Y(y)=
W \circ P_{\mix,\mathcal{X}}(y)
=
W \circ P_{\mix,\mathcal{X}}(x\cdot y)$,
we have
\begin{align*}
&e^{{\psi}(\rho|W,P_{\mix,\mathcal{X}})}
=
\sum_{x} 
P_{\mix,\mathcal{X}}(x)
\sum_{y} 
P_Y(x\cdot y )^{1+\rho}
\overline{P}_Y(y)^{-\rho} \\
=&
\sum_{x} 
P_{\mix,\mathcal{X}}(x)
\sum_{y} 
P_Y( y )^{1+\rho}
\overline{P}_Y(x^{-1}\cdot y)^{-\rho} \\
=&
\sum_{x} 
P_{\mix,\mathcal{X}}(x)
\sum_{y} 
P_Y( y )^{1+\rho}
\overline{P}_Y( y)^{-\rho} \\
=&
\sum_{y} 
P_Y( y )^{1+\rho}
\overline{P}_Y( y)^{-\rho} 
=
e^{{\psi}(\rho|P_Y \| \overline{P}_Y)}.
\end{align*}
\end{proof}

\section{Secure Multiplex Coding with Common Messages:
Single-Shot Setting}\Label{s5}
In this section,
we give the formulation of 
the secure multiplex coding with common messages.
After the formulation,
we give 
two kinds of random construction of codes for the secure multiplex coding with common messages
and evaluate their performance in the single-shot setting.

\subsection{Formulation and Preparation}\Label{s5-1}
In the secure multiplex coding with common messages, 
Alice sends the common message $S_0$ to Bob and Eve, and $T$ secret messages $S_1,\ldots, S_T$ to Bob.
We do not necessarily assume the uniformity 
nor independence
for the distributions of messages $S_0,S_1,\ldots, S_T$.
Hence, there might exist statistical correlations among messages $S_0,S_1,\ldots, S_T$.
Even in this scenario, 
Alice and Bob can use $S_1,\ldots,S_{i-1}$, $S_{i+1},\ldots,S_{T}$ as
random bits making $S_i$ ambiguous to Eve.
When we focus on $S_{\mathcal{I}}:= (S_i ;i\in \mathcal{I})$ for 
a non-empty proper subset
$\mathcal{I}(\neq \emptyset) \subsetneq \{1,\ldots, T\}$, 
the remaining information $S_{\mathcal{I}^c}$ 
serves as random bits making $S_{\mathcal{I}}$ ambiguous to Eve.
The messages $S_0,S_1,\ldots, S_T$ are assumed to belong to the sets 
$\mathcal{S}_0,\mathcal{S}_1,\ldots,\mathcal{S}_T$.
The set $\mathcal{S}_1 \times \ldots\times \mathcal{S}_T$ of all secret messages 
is denoted by $\mathcal{S}$.
In order to explain the SMC model without $\mathcal{S}_0$, 
we consider the following example.
Consider the case when 
$S_1,\ldots, S_T$ are personal information for $T$ persons.
That is, $S_i$ corresponds to the personal information of the $i$-th person. 
Assume that it is required only to keep 
the secrecy of the respective personal information
$S_1,\ldots, S_T$ from the third party.
The secrecy of the relation among respective personal informations
is not required.
For example,
when $S_1,\ldots,S_{T}$ are the uniform random bits with the same size,
the secrecy of the sum $S_1\oplus\ldots\oplus S_{T}$ is not required,
where $\oplus$ is exclusive OR.
In order to treat this secrecy problem,
we give a formulation of the SMC model as follows.

The purpose of the coding in the SMC model is
to reliably send the messages $S_0,S_1,\ldots, S_T$ to Bob,
and
to make $S_{\mathcal{I}}$ ambiguous to Eve by
using the remaining information $S_{\mathcal{I}^c}$
for several non-empty proper subsets $\mathcal{I} \subsetneq  \{
1, \ldots, T\}$.
Our code is given by
Alice's stochastic encoder $\varphi_{a}$ from $\mathcal{S} \times \mathcal{S}_0$ to $\mathcal{X}$,
Bob's deterministic decoder $\varphi_{b}: \mathcal{Y} \rightarrow \mathcal{S} \times \mathcal{S}_0$ and
Eve's deterministic decoder $\varphi_{e}: \mathcal{Z} \rightarrow \mathcal{S}_0$.
The triple $\varphi= (\varphi_a,\varphi_b,\varphi_e)$ is called a code 
for the secure multiplex coding with common messages.
Then, the performance is evaluated by the following quantities:
(1) The sizes of the sets of the common messages and all of the secret messages, i.e., 
$|\mathcal{S}_0|,|\mathcal{S}_1|,\ldots,|\mathcal{S}_T|$.
(2) Bob's decoding error probability
$P_b[P_{Y|X},\varphi,P_{S_{\cal T}}]$, which is the probability
$\mathrm{Pr}\{(S_0,S_1,\ldots, S_T) \neq \varphi_{b}(Y) \}$
under the distribution 
$(P_{Y|X} \circ \varphi_a) \times P_{S_{\cal T}} $
with ${\cal T}:=\{0,\ldots, T\}$.
(3) Eve's decoding error probability
$P_e[P_{Z|X},\varphi,P_{S_{\cal T}}]$,
which is the probability
$\mathrm{Pr}\{ S_0 \neq \varphi_{e}(Z) \}$
under the distribution 
$(P_{Z|X} \circ \varphi_a )\times P_{S_{\cal T}} $.
(4) Leaked information $I(S_{\mathcal{I}};Z|S_0)[P_{Z|X},\varphi_a,P_{S_{\cal T}}]$
for non-empty proper subset $\mathcal{I} \subsetneq  \{
1, \ldots, T\}$,
which is the mutual information
$I(S_{\mathcal{I}};Z|S_0)$ 
under the distribution 
$(P_{Z|X} \circ \varphi_a )\times P_{S_{\cal T}} $.
Instead of $I(S_{\mathcal{I}};Z|S_0)[P_{Z|X},\varphi_a, P_{S_{\cal T}}]$, 
other researchers sometimes treat 
(5) Eve's uncertainty $H(S_{\mathcal{I}}|Z,S_0)[P_{Z|X},\varphi_a,P_{S_{\cal T}}]$,
which is the conditional entropy
$H(S_{\mathcal{I}}|Z,S_0)$
under the distribution 
$(P_{Z|X} \circ \varphi_a )\times P_{S_{\cal T}} $.
However, when we treat the universality of our code, 
leaked information $I(S_{\mathcal{I}};Z|S_0)[P_{Z|X},\varphi_a,P_{S_{\cal T}}]$ 
is used as criterion for performance of our code.
That is, we adopt leaked information $I(S_{\mathcal{I}};Z|S_0)[P_{Z|X},\varphi_a,P_{S_{\cal T}}]$ rather than Eve's uncertainty $H(S_{\mathcal{I}}|Z,S_0)[P_{Z|X},\varphi_a,P_{S_{\cal T}}]$.

In the above formulation, we treat the leaked information $I(S_{\mathcal{I}};Z|S_0)[P_{Z|X},\varphi_a, P_{S_{\cal T}}]$
for several non-empty proper subsets $\mathcal{I}\subsetneq \{1,\ldots, T\}$.
Depending on the situation, 
we decide which non-empty proper subset $\mathcal{I}$ is considered.
Hence, 
in that case,
we can fix a family ${\bf J}$ of non-empty proper subsets $\mathcal{I}$ of $\{1, \ldots, T\}$
for which we discuss the leaked information $I(S_{\mathcal{I}};Z|S_0)[P_{Z|X},\varphi_a, P_{S_{\cal T}}]$. 
For example, in the case of the above personal information,
we consider the subsets $\{1\},\{2\}, \ldots, \{T\}$.
Hence, we choose ${\bf J}$ as ${\bf J}:= \{ \{1\},\{2\}, \ldots, \{T\} \}$.
When we do not specify the family ${\bf J}$, we treat the leaked information $I(S_{\mathcal{I}};Z|S_0)[P_{Z|X},\varphi_a, P_{S_{\cal T}}]$ 
for all non-empty proper subsets $\mathcal{I} $ of $\{1,\ldots, T\}$.


This model can be regarded as a generalization of the wire-tap model 
in the following way. 
When there is no common messages and $T=2$,
there exist only two messages $\mathcal{S}_1$ and $\mathcal{S}_2$
in the secure multiplex coding.
In the wire-tap channel model,
${S}_1$ corresponds to the message to be secretly sent to Bob,
and ${S}_2$ does to the dummy message making $S_1$ ambiguous to Eve.
As a special case of our code,
a wire-tap code is given by
Alice's stochastic encoder $\varphi_{a}$ from $\mathcal{S}_1 \times \mathcal{S}_2$ to $\mathcal{X}$ and
Bob's deterministic decoder $\varphi_{b}: \mathcal{Y} \rightarrow \mathcal{S}_1$.
Then, the performance is evaluated by the following quantities.
(1) The size of the secret message $|\mathcal{S}_1|$.
(2) Bob's decoding error probability
$P_b[P_{Y|X},\varphi,P_{S_{1,2}}]$.
(4) Leaked information $I(S_{1};Z)[P_{Z|X},\varphi_a,P_{S_{1,2}}]$.

In order to guarantee that the leaked information is small,
we employ the method of generalized channel resolvability given in Section \ref{s4}.
In order to employ this method,
we have to use the random coding method to construct a code $\varphi$.
In this section, we propose two kinds of random construction for our code.
For a simple application of Theorem \ref{lem-01}, 
which is a simple generalization of channel resolvability,
we propose the first construction in Subsection \ref{s5-3}.
When there is no common message, 
this construction achieves the capacity region, as is mentioned in Remark \ref{R-11-29}.
However, it cannot fully achieve the capacity region that will be defined in Section \ref{s6-2}
when there exists a common message $S_0$.

To resolve this defect, in Subsection \ref{s5-2},
we propose the second construction,
which attains the capacity region.
This construction has two steps.
In the first step, 
similar to the BCD encoder,
we use the superposition random coding.
In the second step, as illustrated in Fig. \ref{fig:encoder},
we split the confidential message into the private message
$B_2$ and a part $B_1$ of the common message encoded by the BCD encoder.
The coding scheme for BCC in \cite{csiszar78} uses this kind
of message splitting. 
The average leaked information under this kind of construction 
is evaluated by Theorem \ref{Lee3}, which is an algebraic version of
channel resolvability.
However,
when there is no common message, 
the first construction  
realizes a better exponential decreasing rate 
for leaked information than 
the second construction.  


When we fix a code $\varphi$, we obtain the following observations.
Any distribution $\tilde{P}_Z$ on $\mathcal{Z}$ 
and any non-empty proper subset
$\mathcal{I}\subsetneq \{1,\ldots, T\}$
satisfy
\begin{align}
&\rho 
I(S_{\mathcal{I}}; Z |S_{0}) [P_{Z|V}, \varphi,P_{S_{\cal T}}]\nonumber \\
=&
\rho \sum_{s_{0}} P_{S_{0}}(s_{0}) I(S_{\mathcal{I}}; Z|S_{0}=s_{0}) [P_{Z|V}, \varphi,P_{S_{\cal T}}] \nonumber \\
= &
\rho \sum_{s_{0}} P_{S_{0}}(s_{0}) 
D(P_{Z,S_{\mathcal{I}}|S_{0}=s_{0}, \varphi }\|
P_{Z|S_{0}=s_{0}, \varphi }\times
P_{S_{\mathcal{I}}|S_{0}=s_{0}, \varphi }
) \nonumber \\
\le &
\rho \sum_{s_{0}} P_{S_{0}}(s_{0}) 
D(P_{Z,S_{\mathcal{I}}|S_{0}=s_{0}, \varphi }\|
\tilde{P}_Z  \times
P_{S_{\mathcal{I}}|S_{0}=s_{0}, \varphi }
) \Label{12-18-2} \\
= &
\sum_{s_{0}} P_{S_{0}}(s_{0}) 
\sum_{s_{\mathcal{I}}} P_{S_{\mathcal{I}}|S_{0} }(s_{\mathcal{I}}|s_{0})
\rho 
D(P_{Z|S_{\mathcal{I}}=s_{\mathcal{I}},S_{0}=s_{0} ,\varphi} \| \tilde{P}_Z ) ,
\end{align}
where
(\ref{12-18-2}) follows from the following general inequality
\begin{align}
D(P_{X,Y} \| P_{X} \times P_{Y})
\le
D(P_{X,Y} \| Q_{X} \times P_{Y})\Label{2-26-1}
\end{align}
for any distribution $Q_X$ over $\mathcal{X}$. 
Due to (\ref{12-18-1}), we have
\begin{align}
\rho 
D(P_{Z|S_{\mathcal{I}}=s_{\mathcal{I}},S_{0}=s_{0} ,\varphi} \| \tilde{P}_Z ) 
\le 
\psi (\rho |P_{Z|S_{\mathcal{I}}=s_{\mathcal{I}},S_{0}=s_{0}, \varphi }\| \tilde{P}_Z ) .
\end{align}
Thus, combining Jensen inequality and the above observations, 
we obtain the following lemma.
\begin{lemma}
Any distribution $\tilde{P}_Z$ on $\mathcal{Z}$ and any
non-empty proper subset
$\mathcal{I}\subsetneq \{1,\ldots, T\}$
satisfy
\begin{align}
& e^{\rho I(S_{\mathcal{I}}; Z |S_{0})[P_{Z|V}, \varphi,P_{S_{\cal T}}]}
\le
e^{\sum_{s_{0}} P_{S_{0}}(s_{0}) 
\sum_{s_{\mathcal{I}}} P_{S_{\mathcal{I}}|S_{0} }(s_{\mathcal{I}}|s_{0})
\rho D(P_{Z|S_{\mathcal{I}}=s_{\mathcal{I}},S_{0}=s_{0},\varphi }\| \tilde{P}_Z ) } \nonumber \\
\le &
\sum_{s_{0}} P_{S_{0}}(s_{0}) 
\sum_{s_{\mathcal{I}}} P_{S_{\mathcal{I}}|S_{0} }(s_{\mathcal{I}}|s_{0})
e^{\rho D(P_{Z|S_{\mathcal{I}}=s_{\mathcal{I}},S_{0}=s_{0},\varphi }\| \tilde{P}_Z ) } \Label{1-13-1} \\
\le &
\sum_{s_{0}} P_{S_{0}}(s_{0}) 
\sum_{s_{\mathcal{I}}} P_{S_{\mathcal{I}}|S_{0} }(s_{\mathcal{I}}|s_{0})
e^{\psi (\rho |P_{Z|S_{\mathcal{I}}=s_{\mathcal{I}},S_{0}=s_{0},\varphi }\| \tilde{P}_Z ) } .
\Label{ineq-12}
\end{align}
\end{lemma}

\subsection{First Construction}\Label{s5-3}
Now, we introduce the first kind of random coding for SMC.
\begin{ensemble}\Label{con2}
For a given Markov chain $U\rightarrow V \rightarrow X \rightarrow YZ$,
we give the random coding 
$\Phi_{\mathrm{c}}$ and $\Phi_{\mathrm{p}}$
in the same way as Code Ensemble \ref{con0} 
with 
$\mathcal{S}_{\mathrm{c}}= \mathcal{S}_0$
and
$\mathcal{S}_{\mathrm{p}}= \mathcal{S}_1 \times \cdots \times \mathcal{S}_T$.
Similar to the case of BCD,
Bob's decoder $\Phi_{b}$ and Eve's decoder $\Phi_{e}$ are defined as
the maximum likelihood decoders.
Hence, our code is written by the quartet$(\Phi_{\mathrm{c}},\Phi_{\mathrm{p}},\Phi_{b},\Phi_{e})$. 
\end{ensemble}

As a special case of Code Ensemble \ref{con2},
a wire-tap code is given as the case when 
$T=2$ and we do not have the random variables $S_{0}$.
The averaged performance of 
the above code is evaluated by the following theorem.
Indeed, we cannot derive the capacity region from the following theorem.
However, the following theorem has an advantage when 
the conditional mutual information goes to zero.
As is explained in Section \ref{s7},
the following theorem yields a better bound for the exponential decreasing rate
of the conditional mutual information
than Theorem \ref{lem2} in a specific case.

\begin{theorem}\Label{lem1}
The above ensemble of codes 
$\Phi=(\Phi_{\mathrm{c}},\Phi_{\mathrm{p}},\Phi_{b},\Phi_{e})$ satisfies the following inequalities.
\begin{align}
& \rE_{\Phi}
\exp (\rho I(S_{\mathcal{I}}; Z |S_{0}) [P_{Z|V}, \Phi,P_{S_{\cal T}}]  ) \nonumber \\
\le &
1+ e^{- \rho H_{1+\rho}(S_{\mathcal{I}^c}|S_{\mathcal{I}},S_{0})+\psi(\rho| P_{Z|V}, P_{V|U},P_U)} ,\Label{ineq-1} \\
& \rE_{\Phi}
P_b [P_{Y|V}, \Phi,P_{S_{\cal T}}]\nonumber \\
\le &
|{\cal S}|^{\rho} e^{E_0(-\rho| P_{Y|V}, P_{V|U},P_U)}
+
(|{\cal S}_0||{\cal S}|)^{\rho} e^{E_0(-\rho| P_{Y|U,V},P_{U,V})}, \Label{ineq-31-a}\\
& \rE_{\Phi}
P_e [P_{Z|V}, \Phi,P_{S_{\cal T}}]
\le 
|{\cal S}_0|^{\rho} e^{E_0(-\rho| P_{Z|U},P_U)}.\Label{ineq-32-a}
\end{align}
\end{theorem}

Theorem \ref{lem1} yields the following observation.
Applying Jensen's inequality to the convex function $x \mapsto e^x$, we obtain
\begin{align}
& 
\rE_{\Phi}
\rho I(S_{\mathcal{I}}; Z |S_{0}) [P_{Z|V}, \Phi,P_{S_{\cal T}}]  \nonumber \\
\le &
\log (1+ e^{- \rho H_{1+\rho}(S_{\mathcal{I}^c}|S_{\mathcal{I}},S_{0})+ \psi(\rho| P_{Z|V}, P_{V|U},P_U)}) \nonumber \\
\le &
e^{- \rho H_{1+\rho}(S_{\mathcal{I}^c}|S_{\mathcal{I}},S_{0})+\psi(\rho| P_{Z|V}, P_{V|U},P_U)}
. \Label{ineq-2-k}
\end{align}
The number of non-empty proper subsets $\mathcal{I}\subsetneq \{1,\ldots, T\}$
is $2^{T}-2$.
Similar to (\ref{12-28-3}) and (\ref{12-28-4}),
since $2(2^{T}-2)+2=2^{T+1}-2 <2^{T+1}$, 
Markov inequality guarantees that
there exists a code $\varphi$ such that
\begin{align}
& \exp (\rho I(S_{\mathcal{I}}; Z |S_{0}) [P_{Z|V}, \varphi,P_{S_{\cal T}}]  )\nonumber \\
\le & 
2^{T+1}(1+ e^{- \rho H_{1+\rho}(S_{\mathcal{I}^c}|S_{\mathcal{I}},S_{0})+\psi(\rho| P_{Z|V}, P_{V|U},P_U)}) \nonumber \\
\le &
2^{T+2} e^{[- \rho H_{1+\rho}(S_{\mathcal{I}^c}|S_{\mathcal{I}},S_{0})+\psi(\rho| P_{Z|V}, P_{V|U},P_U)]_+ } ,
\Label{ineq-6}\\
& \rho 
I(S_{\mathcal{I}}; Z |S_{0}) [P_{Z|V}, \varphi,P_{S_{\cal T}}]  \nonumber \\
\le &
2^{T+1} e^{- \rho H_{1+\rho}(S_{\mathcal{I}^c}|S_{\mathcal{I}},S_{0})+\psi(\rho| P_{Z|V}, P_{V|U},P_U)},
\Label{ineq-6-a-}\\
& P_b [P_{Y|V}, \varphi,P_{S_{\cal T}}]  \nonumber \\
\le &
2^{T+1}|{\cal S}|^{\rho} e^{E_0(-\rho| P_{Y|V}, P_{V|U},P_U)}
+
2^{T+1}|{\cal S}_0|^{\rho}e^{E_0(-\rho| P_{Y|U},P_U)} ,
\Label{ineq-6-b-}
\\
& P_e [P_{Z|V}, \varphi,P_{S_{\cal T}}] \nonumber \\ 
\le &
2^{T+1}|{\cal S}_0|^{\rho} e^{E_0(-\rho| P_{Z|U},P_U)}.
\Label{ineq-6-c-}
\end{align}
Taking the logarithm in (\ref{ineq-6}),
we obtain
\begin{align}
& 
I(S_{\mathcal{I}}; Z |S_{0}) [P_{Z|V}, \Phi,P_{S_{\cal T}}]  \nonumber \\
\le &
(T+2)\frac{\log 2 }{\rho}
+[\frac{1}{\rho}\psi(\rho| P_{Z|V}, P_{V|U},P_U)
-H_{1+\rho}(S_{\mathcal{I}^c}|S_{\mathcal{I}},S_{0})]_+
.\Label{ineq-2}
\end{align}

\quad {\it Proof of Theorem \ref{lem1}:}\quad

Inequalities (\ref{ineq-31-a}) and (\ref{ineq-32-a}) can be shown by Lemma \ref{lem0}.
The remaining inequality (\ref{ineq-1}) can be shown as follows.
\begin{align*}
& \rE_{\Phi} e^{\rho I(S_{\mathcal{I}}; Z |S_{0},\Phi)} \\
\stackrel{(a)}{\le} &
\rE_{\Phi} \sum_{s_{0}} P_{S_{0}}(s_{0}) 
\sum_{s_{\mathcal{I}}} P_{S_{\mathcal{I}}|S_{0} }(s_{\mathcal{I}}|s_{0})
e^{\psi (\rho |P_{Z|S_{\mathcal{I}}=s_{\mathcal{I}},S_{0}=s_{0},\Phi }\| 
P_{Z|U=\Phi_{\mathrm{c}}(s_{0})} ) } \\
=&
\sum_{s_{0}} P_{S_{0}}(s_{0}) 
\sum_{s_{\mathcal{I}}} P_{S_{\mathcal{I}}|S_{0} }(s_{\mathcal{I}}|s_{0})
\nonumber \\
&\cdot \rE_{\Phi_{\mathrm{c}} } 
\rE_{\Phi_{\mathrm{p}}|\Phi_{\mathrm{c}} } 
e^{\psi (\rho |P_{Z|S_{\mathcal{I}}=s_{\mathcal{I}},S_{0}=s_{0},\Phi }\| 
P_{Z|U=\Phi_{\mathrm{c}}(s_{0})} ) } \\
\stackrel{(b)}{\le} & 
\sum_{s_{0}} P_{S_{0}}(s_{0}) 
\sum_{s_{\mathcal{I}}} P_{S_{\mathcal{I}}|S_{0} }(s_{\mathcal{I}}|s_{0})
\nonumber \\
&\cdot \rE_{\Phi_{\mathrm{c}}} 
(1+
e^{- \rho H_{1+\rho}(S_{\mathcal{I}^c}|S_{\mathcal{I}}=s_{\mathcal{I}},S_{0}=s_{0})}
e^{\psi(\rho| P_{Z|V}, P_{V|U=\Phi_{\mathrm{c}}(s_{0})})}) \\
= &
\sum_{s_{0}} P_{S_{0}}(s_{0}) 
\sum_{s_{\mathcal{I}}} P_{S_{\mathcal{I}}|S_{0} }(s_{\mathcal{I}}|s_{0})
\nonumber \\
&\cdot (1+
e^{- \rho H_{1+\rho}(S_{\mathcal{I}^c}|S_{\mathcal{I}}=s_{\mathcal{I}},S_{0}=s_{0})}
e^{\psi(\rho| P_{Z|V}, P_{V|U},P_U)} )\\
= &
1+
e^{- \rho H_{1+\rho}(S_{\mathcal{I}^c}|S_{\mathcal{I}},S_{0})}
e^{\psi(\rho| P_{Z|V}, P_{V|U},P_U)} ,
\end{align*}
$(a)$ follows from application of 
(\ref{ineq-12}) to the case with $\tilde{P}_Z=P_{Z|U=\Phi_{\mathrm{c}}(s_{0})} $,
and 
$(b)$ follows from Theorem \ref{lem-01}.
\endIEEEproof

\subsection{Second Construction}\Label{s5-2}
Next, we give the second kind of random coding for SMC as follows.
\begin{ensemble}\Label{con1}
{\em First Step:}\quad
For a given Markov chain 
$U\rightarrow V \rightarrow X \rightarrow YZ$,
we introduce two random variables $B_1$ and $B_2$ that 
take values in Abelian groups ${\cal B}_1$ and ${\cal B}_2$
and are subject to the uniform distributions.
The pair of random variables $(B_1,B_2)$ is 
used for sending the all of secret messages in
${\cal S}_{1}\times \cdots \times {\cal S}_{T}$.
Assuming that 
${\cal S}_1\times \ldots\times {\cal S}_T$ has an Abelian group structure,
we give the random coding 
$\Phi_{\mathrm{c}}$ and $\Phi_{\mathrm{p}}$
in the same way as Code Ensemble \ref{con0} 
with 
$\mathcal{S}_{\mathrm{c}}= \mathcal{S}_0 \times \mathcal{B}_1$
and
$\mathcal{S}_{\mathrm{p}}= \mathcal{B}_2$.



{\em Second Step:}\quad
We choose an ensemble satisfying Condition \ref{C2-b} of isomorphisms
$F'$ from ${\cal S}_{1}\times \cdots \times {\cal S}_{T}$ to ${\cal B}_1\times {\cal B}_2$ as Abelian groups.
We choose the random variable $G'\in {\cal B}_1\times {\cal B}_2$ that obeys the uniform distribution on ${\cal B}_1\times {\cal B}_2$ 
and is independent of the choice of $F'$ and anything else.
Then, we define a map $\Lambda_{F',G'}(s):=F'(s)+G'$.
Combining the above codes, we construct the code 
$\Phi_{a}=\Phi_{\mathrm{p}} \circ \Lambda_{F',G'}:
{\cal S}_0 \times {\cal S}_1\times \cdots \times {\cal S}_{T}
\to 
 {\cal V}$ as $(s_0,s_1,\ldots,s_T) \mapsto \Phi_{\mathrm{p}} (s_0,\Lambda_{F',G'}(s_1,\ldots,s_T))$.
Similar to the case of BCD,
Bob's decoder $\Phi_{b}$ and Eve's decoder $\Phi_{e}$ are defined as
the maximum likelihood decoders.
Hence, our code is written by the triple
$(\Phi_a,\Phi_b,\Phi_e)$. 
The structure of encoder is illustrated in Fig.\ \ref{fig:encoder}.
\end{ensemble}

As a special case of Code Ensemble \ref{con1},
a wire-tap code is given as the case when 
$T=2$ and we do not have the random variables $S_{0}$.
For a fixed code $\varphi_{\mathrm{p}}$, 
$P_{Z|S_{0}=s_{0},\Phi_{\mathrm{p}}=\varphi_{\mathrm{p}}}$
denotes the average output distribution of the channel of
the transmitted codeword $\varphi_{\mathrm{p}}(s_0,B_1,B_2)$
averaged over $B_1,B_2$.
In order to evaluate the averaged performance of 
the above code $(\Phi_a,\Phi_b,\Phi_e)$,
we prepare the following lemma.

\begin{lemma}\Label{lem2-1}
When the code $\Phi_{\mathrm{p}}$ is fixed to $\varphi_{\mathrm{p}}$ in the BCD part,
we have the following average performance.
\begin{align}
& 
\rE_{F',G'}
\exp (\rho I(S_{\mathcal{I}}; Z |S_{0})[P_{Z|V},
\varphi_{\mathrm{p}} \circ \Lambda_{F',G'},P_{S_{\cal T}}])\nonumber  \\
\le &
\rE_{F',G'}
\sum_{s_{0}} P_{S_{0}}(s_{0}) 
\sum_{s_{\mathcal{I}}} P_{S_{\mathcal{I}}|S_{0} }(s_{\mathcal{I}}|s_{0})
\nonumber  \\
&\qquad \qquad \qquad \qquad \cdot e^{\rho D( P_{Z|S_{\mathcal{I}}=s_{\mathcal{I}},S_0=s_0,\Phi_{\mathrm{p}}=\varphi_{\mathrm{p}}}\| 
P_{Z|S_{0}=s_{0},\Phi_{\mathrm{p}}=\varphi_{\mathrm{p}}} )} 
\nonumber \\
\le &
1+
\sum_{s_{0}} P_{S_{0}}(s_{0}) 
\sum_{s_{\mathcal{I}}} P_{S_{\mathcal{I}}|S_{0} }(s_{\mathcal{I}}|s_{0})
e^{- \rho H_{1+\rho}(S_{\mathcal{I}^c}|S_{\mathcal{I}}=s_{\mathcal{I}},S_{0}=s_{0})} \nonumber \\
&\qquad \qquad \qquad \qquad
 \cdot e^{\psi ( \rho| P_{Z|B_1,B_2,S_0=s_0,\Phi_{\mathrm{p}}=\varphi_{\mathrm{p}}}, P_{\mix, \mathcal{B}_1,\mathcal{B}_2} )} 
\Label{ineq-11-1} .
\end{align}
Further, 
when $P_{Z|V}$ is a regular channel
and the map $\varphi_{\mathrm{p}}|_{S_0=s_0}:
(b_1,b_2) \mapsto \varphi_{\mathrm{p}}(b_1,b_2 ,s_0)$ 
is a homomorphism from an Abelian group $\mathcal{B}_1\times \mathcal{B}_2 $ to an Abelian group $\mathcal{V}$ for any $s_0\in \mathcal{S}_0$,
the inequalities (\ref{ineq-11-1}) hold 
even when $G'$ is a constant $g'$.
\end{lemma}

Lemma \ref{lem2-1} will be applied 
for the evaluation of the performance of Code Ensemble \ref{con1}.
However, it will be also used 
for the evaluation of the performance of another type of codes 
without common messages based on 
a specific error correcting code in Section \ref{s8}.
Hence, Lemma \ref{lem2-1} addresses the case when 
the map $\varphi_{\mathrm{p}}|_{S_0=s_0}$ is a homomorphism.

Lemma \ref{lem2-1} yields the following observation.
Applying Jensen's inequality for the convex function $x \mapsto e^x$
and the inequality $\log(1+x) \le x$, we obtain
\begin{align}
& 
\rE_{F',G'}
\rho I(S_{\mathcal{I}}; Z |S_{0})[P_{Z|V},\varphi_{\mathrm{p}} \circ \Lambda_{F',G'},P_{S_{\cal T}}] \nonumber \\
\le &
\log \Bigl(
1+
\sum_{s_{0}} P_{S_{0}}(s_{0}) 
\sum_{s_{\mathcal{I}}} P_{S_{\mathcal{I}}|S_{0} }(s_{\mathcal{I}}|s_{0})
e^{- \rho H_{1+\rho}(S_{\mathcal{I}^c}|S_{\mathcal{I}}=s_{\mathcal{I}},S_{0}=s_{0})}
\nonumber \\
&\qquad \qquad \qquad \qquad
 \cdot 
e^{\psi ( \rho| P_{Z|B_1,B_2,S_0=s_0,\Phi_{\mathrm{p}}=\varphi_{\mathrm{p}}}, P_{\mix, \mathcal{B}_1,\mathcal{B}_2} )} 
\Bigr)
\nonumber \\
\le &
\sum_{s_{0}} P_{S_{0}}(s_{0}) 
\sum_{s_{\mathcal{I}}} P_{S_{\mathcal{I}}|S_{0} }(s_{\mathcal{I}}|s_{0})
e^{- \rho H_{1+\rho}(S_{\mathcal{I}^c}|S_{\mathcal{I}}=s_{\mathcal{I}},S_{0}=s_{0})}
\nonumber \\
&\qquad \qquad \qquad \qquad
 \cdot 
e^{\psi ( \rho| P_{Z|B_1,B_2,S_0=s_0,\Phi_{\mathrm{p}}=\varphi_{\mathrm{p}}}, P_{\mix, \mathcal{B}_1,\mathcal{B}_2} )} 
. \Label{ineq-2-c}
\end{align}

\begin{IEEEproof}
Applying (\ref{1-13-1}) and \eqref{ineq-12} to the case when 
$\tilde{P}_Z=\tilde{P}_{Z|S_{0}=s_{0},\Phi_{\mathrm{p}}=\varphi_{\mathrm{p}}}  $,
we obtain
\begin{align}
& \rE_{F',G'}
 e^{\rho I(S_{\mathcal{I}}; Z |S_{0}) [P_{Z|V},\varphi_{\mathrm{p}} \circ \Lambda_{F',G'},P_{S_{\cal T}}] } \nonumber\\
\le &
\rE_{F',G'}
\sum_{s_{0}} P_{S_{0}}(s_{0}) 
\sum_{s_{\mathcal{I}}} P_{S_{\mathcal{I}}|S_{0} }(s_{\mathcal{I}}|s_{0})
\nonumber \\
&\qquad \qquad \qquad \qquad
 \cdot 
e^{\rho D(P_{Z|S_{\mathcal{I}}=s_{\mathcal{I}},S_{0}=s_{0},\Phi_{\mathrm{p}}=\varphi_{\mathrm{p}} }\| {P}_{Z|S_{0}=s_{0},\Phi_{\mathrm{p}}=\varphi_{\mathrm{p}}} ) } 
\nonumber \\
\le &
\rE_{F',G'|\Phi_{\mathrm{p}}=\varphi_{\mathrm{p}}} 
\sum_{s_{0}} P_{S_{0}}(s_{0}) 
\sum_{s_{\mathcal{I}}} P_{S_{\mathcal{I}}|S_{0} }(s_{\mathcal{I}}|s_{0})
\nonumber \\
&\qquad \qquad \qquad \qquad
 \cdot 
e^{\psi (\rho |P_{Z|S_{\mathcal{I}}=s_{\mathcal{I}},S_{0}=s_{0},\Phi_{\mathrm{p}}=\varphi_{\mathrm{p}} }\| {P}_{Z|S_{0}=s_{0},\Phi_{\mathrm{p}}=\varphi_{\mathrm{p}}} ) } 
\Label{ineq-21}.
\end{align}
For a fixed $s_{\mathcal{I}}$,
we apply Theorem \ref{Lee3} to the case when 
${\cal A}$ is $ 
{\cal S}_{\mathcal{I}^c} $,
${\cal X}$ is ${\cal B}_1\times {\cal B}_2$, 
$G$ is $
G'+F'(s_{\mathcal{I}},0)$, which is independent of $F'$, and
$F$ is the map $s_{\mathcal{I}^c} \mapsto F'(0,s_{\mathcal{I}^c})$
that satisfies Condition \ref{C2-b}.
Then, 
$\Lambda_{F',G'}(s_{\mathcal{I}},s_{\mathcal{I}^c})
=F'(s_{\mathcal{I}},s_{\mathcal{I}^c})+G'=F'(0,s_{\mathcal{I}^c})+ Z_{s_{\mathcal{I}}}$.
Thus, we obtain
\begin{align}
& \rE_{F',G'}
e^{\psi (\rho |P_{Z|S_{\mathcal{I}}=s_{\mathcal{I}},S_{0}=s_{0},\Phi_{\mathrm{p}}=\varphi_{\mathrm{p}} }\| \tilde{P}_{Z|S_{0}=s_{0},\Phi_{\mathrm{p}}=\varphi_{\mathrm{p}}} ) } 
\nonumber \\
\le &
1+
e^{- \rho H_{1+\rho}(S_{\mathcal{I}^c}|S_{\mathcal{I}}=s_{\mathcal{I}},S_{0}=s_{0})}
e^{\psi ( \rho| P_{Z|B_1,B_2,S_0,\Phi_{\mathrm{p}}=\varphi_{\mathrm{p}}}, P_{\mix, \mathcal{B}_1,\mathcal{B}_2} )} 
\Label{ineq-22}  .
\end{align}
Thus, we obtain (\ref{ineq-11-1}).

Further, 
when $P_{Z|V}$ is a regular channel
and the map $\varphi_{\mathrm{p}}|_{S_0=s_0} :(b_1,b_2) \mapsto \varphi_{\mathrm{p}}(b_1,b_2 ,s_0)$ 
is a homomorphism from an Abelian group $\mathcal{B}_1 \times \mathcal{B}_2 $ to an Abelian group $\mathcal{V}$ for any $s_0\in \mathcal{S}_0$,
the channel $P_{Z|V} \circ \varphi_{\mathrm{p}}|_{S_0=s_0}$ is a regular channel
from $\mathcal{B}_1 \times \mathcal{B}_2 $ to $\mathcal{V}$.
Hence, due to Corollary \ref{cor1},
the inequalities (\ref{ineq-11-1}) hold
even when $G'$ is a constant $g'$.
\end{IEEEproof}

Using the above lemma, we obtain the following theorem,
which gives the averaged performance of 
the above code $(\Phi_a,\Phi_b,\Phi_e)$.
By using this theorem, we will give the capacity region in 
Subsection \ref{s6-2}.

\begin{theorem}\Label{lem2}
Assume that the code $\Phi=(\Phi_a,\Phi_b,\Phi_e)$ is
the ensemble given in Code Ensemble \ref{con1}.
Then,
the inequalities
\begin{align}
& \rE_{\Phi_a} 
\exp (\rho I(S_{\mathcal{I}}; Z |S_{0})[P_{Z|V},\Phi_a,P_{S_{\cal T}}] )\nonumber \\
\le &
\rE_{\Phi_a} 
\sum_{s_{0}} P_{S_{0}}(s_{0}) 
\sum_{s_{\mathcal{I}}} P_{S_{\mathcal{I}}|S_{0} }(s_{\mathcal{I}}|s_{0})
e^{\rho D( P_{Z|S_{\mathcal{I}}=s_{\mathcal{I}},S_0=s_0,\Phi_a}\| 
P_{Z|S_{0}=s_{0},\Phi_{\mathrm{p}}} )} \nonumber \\
\le &
1+ |{\cal B}_1|^{\rho} e^{- \rho H_{1+\rho}(S_{\mathcal{I}^c}|S_{\mathcal{I}},S_{0})+E_0(\rho| P_{Z|V}, P_{V|U},P_U)},
\Label{ineq-11} 
\end{align}
and
\begin{align}
\rE_{\Phi}
P_b [P_{Y|V},\Phi,P_{S_{\cal T}}] 
\le &
|{\cal B}_2|^{\rho} e^{E_0(-\rho| P_{Y|V}, P_{V|U},P_U)}\nonumber \\
&+
(|{\cal S}_0| |{\cal S}|)^{\rho}
 e^{E_0(-\rho| P_{Y|U,V},P_{U,V})} 
\Label{ineq-33} \\
\rE_{\Phi}
P_e [P_{Z|V},\Phi,P_{S_{\cal T}}] 
\le &
|{\cal S}_0|^{\rho} e^{E_0(-\rho| P_{Z|U},P_U)}.
\Label{ineq-34} 
\end{align}
hold. 
\end{theorem}

Theorem \ref{lem2} yields the following observation.
Applying Jensen's inequality to the convex function $x \mapsto e^x$, we obtain
\begin{align}
& 
\rE_{\Phi_a} 
\rho I(S_{\mathcal{I}}; Z |S_{0}) [P_{Z|V},\Phi_a,P_{S_{\cal T}}] \nonumber\\
\le &
\log (1+ |{\cal B}_1|^{\rho} e^{- \rho H_{1+\rho}(S_{\mathcal{I}^c}|S_{\mathcal{I}},S_{0})+ E_0(\rho| P_{Z|V}, P_{V|U},P_U)}) \nonumber\\
\le &
|{\cal B}_1|^{\rho} e^{- \rho H_{1+\rho}(S_{\mathcal{I}^c}|S_{\mathcal{I}},S_{0})+E_0(\rho| P_{Z|V}, P_{V|U},P_U)}
. \Label{ineq-2-}
\end{align}
Here, 
we choose $\rho_0$ as
\begin{align}
\rho_0:= 
\argmin_{\rho \in [0,1]}
\Bigl [ &
\log |{\cal B}_1|+\frac{1}{\rho}E_0(\rho| P_{Z|V}, P_{V|U},P_U)
\nonumber\\ 
& \qquad -H_{1+\rho}(S_{\mathcal{I}^c}|S_{\mathcal{I}},S_{0}) \Bigr]_+
+
(T+2)\frac{\log 2 }{\rho}. 
\end{align}
Then,
Similar to (\ref{12-28-3}) and (\ref{12-28-4}),
since $2(2^{T}-2)+2=2^{T+1}-2 <2^{T+1}$, 
Markov inequality guarantees that
there exists a code $\varphi=(\varphi_a,\varphi_b,\varphi_e)$ such that
\begin{align}
& \exp (\rho_0 I(S_{\mathcal{I}}; Z |S_{0}) [P_{Z|V},\varphi_a,P_{S_{\cal T}}] )
\nonumber\\
\le &
2^{T+1}
(1+ |{\cal B}_1|^{\rho_0} e^{- \rho_0 H_{1+\rho}(S_{\mathcal{I}^c}|S_{\mathcal{I}},S_{0})+E_0(\rho_0| P_{Z|V}, P_{V|U},P_U)}) \nonumber\\
\le &
2^{T+2} e^{[\rho_0 \log |{\cal B}_1|- \rho_0 H_{1+\rho_0}(S_{\mathcal{I}^c}|S_{\mathcal{I}},S_{0})+E_0(\rho_0| P_{Z|V}, P_{V|U},P_U),P_{S_{\cal T}}]_+ } ,
\Label{ineq-6-}\\
&  I(S_{\mathcal{I}}; Z |S_{0}) [P_{Z|V},\varphi_a,P_{S_{\cal T}}] 
\nonumber\\
\le &
\min_{0 \le \rho\le 1}
\frac{2^{T+1}}{\rho} |{\cal B}_1|^{\rho} e^{- \rho H_{1+\rho}(S_{\mathcal{I}^c}|S_{\mathcal{I}},S_{0})+E_0(\rho| P_{Z|V}, P_{V|U},P_U)},
\Label{ineq-6-a}\\
& 
P_b [P_{Y|V},\varphi,P_{S_{\cal T}}] \nonumber\\
\le &
2^{T+1} 
\min_{0 \le \rho\le 1}
(|{\cal B}_2|^{\rho} e^{E_0(-\rho| P_{Y|V}, P_{V|U},P_U)}
+
(|{\cal S}_0||{\cal S}|)^{\rho}e^{E_0(-\rho| P_{Y|UV},P_{UV})} ),
\Label{ineq-6-b}\\
& 
P_e [P_{Z|V},\varphi,P_{S_{\cal T}}] \nonumber\\
\le &
2^{T+1} 
\min_{0 \le \rho\le 1}
|{\cal S}_0|^{\rho} e^{E_0(-\rho| P_{Z|U},P_U)}
\Label{ineq-6-c}
\end{align}
for any non-empty proper subset $\mathcal{I}\subsetneq \{1,\ldots, T\}$.
Taking the logarithm in (\ref{ineq-6-}),
we obtain
\begin{align}
& 
I(S_{\mathcal{I}}; Z |S_{0}) [P_{Z|V},\Phi_a,P_{S_{\cal T}}] 
\nonumber\\
\le &
\Bigl [\log |{\cal B}_1|+\frac{1}{\rho_0}E_0(\rho_0| P_{Z|V}, P_{V|U},P_U)
-H_{1+\rho_0}(S_{\mathcal{I}^c}|S_{\mathcal{I}},S_{0}) \Bigr]_+
\nonumber\\
& +
(T+2)\frac{\log 2 }{\rho_0}  \nonumber \\
=& \min_{\rho \in [0,1]}
\Bigl [\log |{\cal B}_1|+\frac{1}{\rho}E_0(\rho| P_{Z|V}, P_{V|U},P_U)
-H_{1+\rho}(S_{\mathcal{I}^c}|S_{\mathcal{I}},S_{0}) \Bigr]_+
\nonumber\\
& +
(T+2)\frac{\log 2 }{\rho}.\Label{ineq-2-e}
\end{align}

\quad {\it Proof of Theorem \ref{lem2}:}\quad
We show (\ref{ineq-11}). Using (\ref{psileqphi}), we obtain
\begin{align}
& 
\rE_{\Phi_{\mathrm{p}},\Phi_{\mathrm{c}} } 
e^{\psi ( \rho| P_{Z|B_1,B_2,S_0=s_0,\Phi_{\mathrm{p}}}, 
P_{\mix, \mathcal{B}_1,\mathcal{B}_2} )}  \nonumber\\
\le &
\rE_{\Phi_{\mathrm{p}},\Phi_{\mathrm{c}} } 
e^{E_0 ( \rho| P_{Z|B_1,B_2,S_0=s_0,\Phi_{\mathrm{p}}}, 
P_{\mix, \mathcal{B}_1,\mathcal{B}_2} )}  \Label{ineq-43}\\
= &
\rE_{\Phi_{\mathrm{p}},\Phi_{\mathrm{c}} } 
\sum_{z}(\sum_{b_1,b_2}P_{B_1,B_2}(b_1,b_2)P_{Z|B_1,B_2,S_0=s_0,\Phi_{\mathrm{p}}}(z|b_1,b_2)^{\frac{1}{1-\rho}}  )^{1-\rho} \nonumber\\
= &
\rE_{\Phi_{\mathrm{p}},\Phi_{\mathrm{c}} } 
\sum_{z}
(\sum_{b_1,b_2}
\frac{1}{|{\cal B}_1||{\cal B}_2|}
P_{Z|V}(z|\Phi_{\mathrm{p}}(s_0,b_1,b_2) )^{\frac{1}{1-\rho}}  )^{1-\rho} \nonumber\\
\le &
\rE_{\Phi_{\mathrm{p}},\Phi_{\mathrm{c}} } 
\sum_{z}
\sum_{b_1}(\sum_{b_2}
\frac{1}{|{\cal B}_1||{\cal B}_2|}
P_{Z|V}(z|\Phi_{\mathrm{p}}(s_0,b_1,b_2) )^{\frac{1}{1-\rho}}  )^{1-\rho} 
\Label{ineq-26}\\
= &
\rE_{\Phi_{\mathrm{p}},\Phi_{\mathrm{c}} } 
\sum_{z}
\sum_{b_1}
\frac{|{\cal B}_1|^\rho}{|{\cal B}_1|}
(\sum_{b_2}
\frac{1}{|{\cal B}_2|}
P_{Z|V}(z|\Phi_{\mathrm{p}}(s_0,b_1,b_2) )^{\frac{1}{1-\rho}}  )^{1-\rho} 
\Label{12-26-5}\\
\le &
\rE_{\Phi_{\mathrm{c}} } 
\sum_{z}
\sum_{b_1}
\frac{|{\cal B}_1|^\rho}{|{\cal B}_1|}
(
\sum_{b_2}
\frac{1}{|{\cal B}_2|}
\rE_{\Phi_{\mathrm{p}}|\Phi_{\mathrm{c}} } 
P_{Z|V}(z|\Phi_{\mathrm{p}}(s_0,b_1,b_2) )^{\frac{1}{1-\rho}}  )^{1-\rho} 
\Label{ineq-44}\\
= &
\sum_{z}
\sum_{b_1}
\frac{|{\cal B}_1|^\rho}{|{\cal B}_1|}
\rE_{\Phi_{\mathrm{c}} } 
(
\sum_{b_2}
\frac{1}{|{\cal B}_2|}
\sum_{v}
P_{V|U}(v|\Phi_{\mathrm{c}}(s_0,b_1))
P_{Z|V}(z|v )^{\frac{1}{1-\rho}}  )^{1-\rho} 
\Label{5-30-2}\\
= &
\sum_{z}
\sum_{b_1}
\frac{|{\cal B}_1|^\rho}{|{\cal B}_1|}
\rE_{\Phi_{\mathrm{c}} } 
(
\sum_{v}
P_{V|U}(v|\Phi_{\mathrm{c}}(s_0,b_1))
P_{Z|V}(z|v )^{\frac{1}{1-\rho}}  )^{1-\rho} 
\nonumber\\
= &
\sum_{z}
\sum_{b_1}
\frac{|{\cal B}_1|^\rho}{|{\cal B}_1|}
\sum_u P_U(u)
(
\sum_{v}
P_{V|U}(v|u)
P_{Z|V}(z|v )^{\frac{1}{1-\rho}}  )^{1-\rho} 
\nonumber\\
= &
\sum_{z}
|{\cal B}_1|^\rho
\sum_u P_U(u)
(
\sum_{v}
P_{V|U}(v|u)
P_{Z|V}(z|v )^{\frac{1}{1-\rho}}  )^{1-\rho} 
\nonumber\\
=&
|{\cal B}_1|^{\rho}
e^{E_0(\rho| P_{Z|V}, P_{V|U},P_U)},\Label{ineq-13}
\end{align}
where
(\ref{ineq-43}),  
(\ref{ineq-26}) 
(\ref{ineq-44}), 
and
(\ref{5-30-2})
follow from 
(\ref{psileqphi}), 
the inequality $(x+y)^{1-\rho} \le x^{1-\rho}+y^{1-\rho}$, 
the concavity of $x \mapsto x^{1-\rho}$,
and
the definition of the ensemble of the code $\Phi_{\mathrm{p}}$,
respectively.

Summarizing the above discussion,
we obtain
\begin{align}
& \rE_{\Phi_a} e^{\rho I(S_{\mathcal{I}}; Z |S_{0}) [P_{Z|V},\Phi_a,P_{S_{\cal T}}] } \nonumber\\
\le &
\rE_{\Phi_a} 
\sum_{s_{0}} P_{S_{0}}(s_{0}) 
\sum_{s_{\mathcal{I}}} P_{S_{\mathcal{I}}|S_{0} }(s_{\mathcal{I}}|s_{0})
e^{\rho D( P_{Z|B_1,B_2,S_0=s_0,\Phi_{\mathrm{p}}}\| \tilde{P}_{Z|S_{0}=s_{0},\Phi_{\mathrm{p}}} )} \Label{1-13-20} \\
= &
\rE_{\Phi_{\mathrm{p}} } \rE_{F',G'|\Phi_{\mathrm{p}}} 
\sum_{s_{0}} P_{S_{0}}(s_{0}) 
\sum_{s_{\mathcal{I}}} P_{S_{\mathcal{I}}|S_{0} }(s_{\mathcal{I}}|s_{0})
\nonumber \\
& \qquad \qquad \qquad \cdot e^{\rho D( P_{Z|B_1,B_2,S_0=s_0,\Phi_{\mathrm{p}}}\| \tilde{P}_{Z|S_{0}=s_{0},\Phi_{\mathrm{p}}} )} \nonumber \\
\le &
\sum_{s_{0}} P_{S_{0}}(s_{0}) 
\sum_{s_{\mathcal{I}}} P_{S_{\mathcal{I}}|S_{0} }(s_{\mathcal{I}}|s_{0})
\nonumber \\
& \cdot
\rE_{\Phi_{\mathrm{p}} } 
(1+
e^{- \rho H_{1+\rho}(S_{\mathcal{I}^c}|S_{\mathcal{I}}=s_{\mathcal{I}},S_{0}=s_{0})}
e^{\psi ( \rho| P_{Z|B_1,B_2,S_0,\Phi_{\mathrm{p}}}, P_{B_1,B_2} )} )
\Label{ineq-22-a} \\ 
\le &
\sum_{s_{0}} P_{S_{0}}(s_{0}) 
\sum_{s_{\mathcal{I}}} P_{S_{\mathcal{I}}|S_{0} }(s_{\mathcal{I}}|s_{0})
\nonumber \\
& \cdot
(1+
e^{- \rho H_{1+\rho}(S_{\mathcal{I}^c}|S_{\mathcal{I}}=s_{\mathcal{I}},S_{0}=s_{0})}
|{\cal B}_1|^{\rho} 
e^{E_0(\rho| P_{Z|V}, P_{V|U},P_U)}
) \Label{ineq-23}\\
= &
1+
e^{- \rho H_{1+\rho}(S_{\mathcal{I}^c}|S_{\mathcal{I}},S_{0})}
|{\cal B}_1|^{\rho} 
e^{E_0 (\rho| P_{Z|V}, P_{V|U},P_U)} ,
\nonumber
\end{align}
where 
(\ref{1-13-20}), (\ref{ineq-22-a}), and (\ref{ineq-23}) 
follow from
(\ref{1-13-1}),
the second inequality in Lemma \ref{lem2-1}, and (\ref{ineq-13}), respectively.
Then, we obtain (\ref{ineq-11}).

Further, (\ref{ineq-33}) and (\ref{ineq-34}) 
follow from Lemma \ref{lem0}.
\endIEEEproof

\subsection{Group Symmetry}\Label{s5-4}
Next, when the channel has a nice property with respect to group action,
we treat the upper bound of the leaked information with a fixed BCD code $\varphi_{\mathrm{p}}$.
That is, we discuss the upper bound given in Lemma \ref{lem2-1}
under an assumption for group action,
which will be given latter.
The following analysis is required for evaluation of universal coding
in Sections \ref{s9} and \ref{s10}
and a practical code construction in Subsection \ref{s8-1-2}.

For simplicity, we first discuss the case with no common message, i.e., 
$|\mathcal{S}_0|=1$ and $|\mathcal{B}_1|=1$.
Assume that a group $\mathcal{G}$ acts on $\mathcal{V}$ and $\mathcal{Z}$.
The action of $g \in \mathcal{G}$ is written as 
$g \cdot v$ and $g \cdot z$ for $v \in \mathcal{V}$ and $z \in \mathcal{Z}$.
Then, due to Eqs.~(\ref{12-18-4}), (\ref{12-18-5}), and (\ref{12-18-6}),
we have
\begin{align*}
(g^{-1} \circ P_{Z|V}\circ g)  (z| v)&=P_{Z|V}(g\cdot z|g\cdot v) \\
(g^{-1} \circ P_{V}) (v)     &=P_{V}(g\cdot v) .
\end{align*}
Then, the set $\mathcal{V}$ can be divided to orbits $\{\mathcal{V}_o \}_{o\in O}$ by the action of $\mathcal{G}$.
The set $O$ of indexes of the orbits is called the orbit space.
Given a code $\varphi_{\mathrm{p}}$ as an injective map from $\mathcal{B}_2$ to $\mathcal{V}$, 
Recall that we denote the uniform distribution on the image $\im \varphi_{\mathrm{p}}$
by $P_{\mix, \im \varphi_{\mathrm{p}}}$,
and
we define the distribution $P_{\varphi_{\mathrm{p}}}(o)
:= |\im \varphi_{\mathrm{p}}\cap \mathcal{V}_o |/|\im \varphi_{\mathrm{p}}|$
on the orbit space $O$
and the distribution 
$\overline{P}_{\varphi_{\mathrm{p}}}$
on $\mathcal{V}$ 
by $\overline{P}_{\varphi_{\mathrm{p}}}(v)
:=\frac{P_{\varphi_{\mathrm{p}}}(o)}{|\mathcal{V}_o|}$ 
when the element $v$ belongs to the subset $\mathcal{V}_o$.  
Then, we obtain the following lemma.
\begin{lemma}\Label{l12-3-1}
When the relation 
$g^{-1} \circ P_{Z|V}\circ g=P_{Z|V}$ holds for any $g \in \mathcal{G}$, $v \in \mathcal{Z}$, and $v \in \mathcal{V}$,
\begin{align}
&\psi ( \rho| P_{Z|B_2,\Phi_{\mathrm{p}}=\varphi_{\mathrm{p}}}, 
P_{\mix, \mathcal{B}_2} 
) 
=
\psi ( \rho| P_{Z|V}, P_{\mix, \im \varphi_{\mathrm{p}}} ) \nonumber \\
\le &
E_0 ( \rho| P_{Z|V}, P_{\mix, \im \varphi_{\mathrm{p}}} )
\le
E_0(\rho|P_{Z|V},\overline{P}_{\varphi_{\mathrm{p}}} )
\Label{12-2-1} .
\end{align}
In particular,
when the image $\im \varphi_{\mathrm{p}}$ is included in one orbit $\mathcal{V}_o$,
$\overline{P}_{\varphi_{\mathrm{p}}}$ is the uniform distribution on the orbit $\mathcal{V}_o$.
\end{lemma}

\begin{proof}
Since 
$e^{E_0 ( \rho| g^{-1} \circ P_{Z|V}\circ g, g^{-1} \circ P_{\mix, \varphi_{\mathrm{p}}} )}
=
e^{E_0 ( \rho| P_{Z|V}, g^{-1} \circ P_{\mix, \varphi_{\mathrm{p}}} )}$, 
we have
\begin{align}
& 
e^{\psi ( \rho| P_{Z|V}, P_{\mix, \im \varphi_{\mathrm{p}}} )}
\le
e^{E_0 ( \rho| P_{Z|V}, P_{\mix, \im \varphi_{\mathrm{p}}} )} \nonumber \\
=&
\sum_{g\in \mathcal{G}} 
\frac{1}{|\mathcal{G}|}
e^{E_0 ( \rho| g^{-1} \circ P_{Z|V}\circ g, g^{-1} \circ P_{\mix, \im \varphi_{\mathrm{p}}} )} \nonumber \\
=&
\sum_{g\in \mathcal{G}} 
\frac{1}{|\mathcal{G}|}
e^{E_0 ( \rho| P_{Z|V}, g^{-1} \circ P_{\mix, \im \varphi_{\mathrm{p}}} )} \nonumber \\
\le &
e^{E_0 ( \rho| P_{Z|V}, \sum_{g\in \mathcal{G}} 
\frac{1}{|\mathcal{G}|}
g^{-1} \circ P_{\mix, \im \varphi_{\mathrm{p}}} )} 
= 
e^{E_0 ( \rho| P_{Z|V}, \overline{P}_{\varphi_{\mathrm{p}}})} 
\Label{ineq-11-1d} .
\end{align}
\end{proof}

Next, we consider the general case.
Assume that a group $\mathcal{G}$ acts on $\mathcal{U}$, $\mathcal{V}$, and $\mathcal{Z}$.
The code pair code $(\varphi_{\mathrm{c}},\varphi_{\mathrm{p}})$ is a map 
from $\mathcal{S}_0 \times \mathcal{B}_1\times \mathcal{B}_2$
to $\mathcal{U} \times \mathcal{V}$.
For a given $s_0 \in \mathcal{S}_0$, we define 
the maps 
$\varphi_{\mathrm{c}}|_{S_0=s_0} $
and
$(\varphi_{\mathrm{c}},\varphi_{\mathrm{p}})|_{S_0=s_0}$
by 
\begin{align*}
\varphi_{\mathrm{c}}|_{S_0=s_0}(b_1)
&:=
\varphi_{\mathrm{c}}(s_0,b_1)
\in \mathcal{U}
\\
(\varphi_{\mathrm{c}},\varphi_{\mathrm{p}})|_{S_0=s_0}(b_1,b_2)
&:=
(\varphi_{\mathrm{c}}(s_0,b_1),\varphi_{\mathrm{p}}(s_0,b_1,b_2))
\in \mathcal{U} \times \mathcal{V}.
\end{align*}
For simplicity, we assume that
the image of $(\varphi_{\mathrm{c}},\varphi_{\mathrm{p}})|_{S_0=s_0}$ 
is included in one orbit in $\mathcal{U} \times \mathcal{V}$, 
which is denoted by $(\mathcal{V}\times \mathcal{U})_o$.
Hence,
the image of $\varphi_{\mathrm{c}}|_{S_0=s_0}$ is included in one orbit in $\mathcal{U}$, which is denoted by $\mathcal{U}_o$.
\begin{lemma}\Label{l12-3-2}
Assume that
the image of $(\varphi_{\mathrm{c}},\varphi_{\mathrm{p}})|_{S_0=s_0}$ 
is included in a orbit $(\mathcal{V}\times \mathcal{U})_o$
in $\mathcal{U} \times \mathcal{V}$.
When the relation 
$g^{-1} \circ P_{Z|V}\circ g=P_{Z|V}$ holds for any $g \in \mathcal{G}$, 
the relation
\begin{align}
&e^{\psi ( \rho| P_{Z|B_1,B_2,S_0=s_0,\Phi_{\mathrm{p}}=\varphi_{\mathrm{p}}}, P_{\mix, \mathcal{B}_1,\mathcal{B}_2} ) }\nonumber \\
\le &
|{\cal B}_1|^\rho
e^{E_0(\rho|P_{Z|V},P_{V|U, \mix, (\mathcal{V}\times \mathcal{U})_o}, P_{\mix, \mathcal{U}_o} )}
\Label{12-2-2} 
\end{align}
holds for any $s_0\in  \mathcal{S}_0$.
\end{lemma}

\begin{proof}
For a given $u \in \mathcal{U}_o$, 
we define the stabilizer of $u$ by
$\mathcal{H}_u:=
\{ g \in \mathcal{G}| g \cdot u=u\}$,
which is a subgroup of $\mathcal{G}$.
For arbitrary $u \in \mathcal{U}_o$,
we define the two subsets 
$\mathcal{V}_u', \mathcal{V}_u \subset \mathcal{V}$ by 
$\{u\} \times\mathcal{V}_u'
=\im (\varphi_{\mathrm{c}},\varphi_{\mathrm{p}})|_{S_0=s_0}
\cap (\{u\} \times \mathcal{V})$
and
$\{u\} \times\mathcal{V}_u
=(\mathcal{V}\times \mathcal{U})_o
\cap (\{u\} \times \mathcal{V})$.
Then, we obtain the relations
\begin{align}
P_{V|U=u, \mix, \im (\varphi_{\mathrm{c}},\varphi_{\mathrm{p}})|_{S_0=s_0}}
&=P_{V|\mix, \mathcal{V}_u'}\Label{12-3-15b}\\
P_{V|U=u, \mix, (\mathcal{V}\times \mathcal{U})_o}
&=P_{V|\mix, \mathcal{V}_u} .
\Label{12-3-15}
\end{align}
For the definitions of the left hand  sides,
see (\ref{12-19-10}).
We can also show that
\begin{align*}
\cup_{g\in \mathcal{H}_u}
\{g \cdot v |v \in  \mathcal{V}_u'\}
=\mathcal{V}_u.
\end{align*}
Since 
$g^{-1} \circ P_{V|U=g \cdot u, \mix, (\mathcal{V}\times \mathcal{U})_o}
=P_{V|U=u, \mix, (\mathcal{V}\times \mathcal{U})_o}$,
the condition $g^{-1} \circ P_{Z|V}\circ g=P_{Z|V}$  implies that 
\begin{align}
&
e^{E_0 ( \rho| g^{-1} \circ P_{Z|V} \circ g,g^{-1} \circ P_{V|U=g \cdot u, \mix, (\mathcal{V}\times \mathcal{U})_o}
)} \nonumber \\
=&
e^{E_0 ( \rho| P_{Z|V} ,P_{V|U=u, \mix, (\mathcal{V}\times \mathcal{U})_o}
)}.\Label{12-3-16}
\end{align}
We obtain the following relations.
In the following derivation,
(\ref{12-26-5a})
and
(\ref{12-3-14}) 
follow from 
(\ref{12-26-5}) and
(\ref{12-3-16}), respectively.
Applying Lemma \ref{l12-3-1} to the case of $\mathcal{G}=\mathcal{H}_u$,
we obtain the inequality (\ref{12-3-12}) from (\ref{12-3-15b}) and (\ref{12-3-15}).
\begin{align}
& e^{\psi ( \rho| P_{Z|B_1,B_2,S_0=s_0,\Phi_{\mathrm{p}}=\varphi_{\mathrm{p}}}, P_{\mix, \mathcal{B}_1,\mathcal{B}_2} )} \nonumber \\
\le &
\sum_{z}
\sum_{b_1}
\frac{|{\cal B}_1|^\rho}{|{\cal B}_1|}
(\sum_{b_2}
\frac{1}{|{\cal B}_2|}
P_{Z|V}(z|\varphi_{\mathrm{p}}(s_0,b_1,b_2) )^{\frac{1}{1-\rho}}  )^{1-\rho} 
\Label{12-26-5a}\\
= &
|{\cal B}_1|^\rho
\sum_{z}
\sum_{u}P_{U, \mix, \im \varphi_{\mathrm{c}}|_{S_0=s_0}
} (u)
\nonumber \\
& \hspace{13ex}
\cdot\Bigl[
\sum_{v}
P_{V|U=u, \mix, 
\im (\varphi_{\mathrm{c}},\varphi_{\mathrm{p}})|_{S_0=s_0}} (v)
P_{Z|V}(z|v )^{\frac{1}{1-\rho}}  \Bigr]^{1-\rho} 
\nonumber \\
= & 
|{\cal B}_1|^\rho
\sum_{u}P_{U, \mix, 
\im \varphi_{\mathrm{c}}|_{S_0=s_0}
} (u)
e^{E_0 ( \rho| P_{Z|V}, P_{V|U=u, \mix, 
\im (\varphi_{\mathrm{c}},\varphi_{\mathrm{p}})|_{S_0=s_0}}
)} \nonumber \\
\le & 
|{\cal B}_1|^\rho
\sum_{u}P_{U, \mix, 
\im \varphi_{\mathrm{c}}|_{S_0=s_0}
} (u)
e^{E_0 ( \rho| P_{Z|V}, P_{V|U=u, \mix, (\mathcal{V}\times \mathcal{U})_o}
)} \Label{12-3-12}\\
= & 
|{\cal B}_1|^\rho
\sum_{g\in \mathcal{G}} 
\frac{1}{|\mathcal{G}|}
\sum_{u}P_{U, \mix, \im \varphi_{\mathrm{c}}|_{S_0=s_0}} (g \cdot u)
e^{E_0 ( \rho| P_{Z|V} ,P_{V|U= u, \mix, (\mathcal{V}\times \mathcal{U})_o}
)} \Label{12-3-14}\\
= & 
|{\cal B}_1|^\rho
\sum_{u}P_{U, \mix, (\mathcal{V}\times \mathcal{U})_o} (u)
e^{E_0 ( \rho| P_{Z|V} ,P_{V|U= u, \mix, (\mathcal{V}\times \mathcal{U})_o}
)} \nonumber \\
= & 
|{\cal B}_1|^\rho
e^{E_0 ( \rho| P_{Z|V} ,P_{V|U= u, \mix, (\mathcal{V}\times \mathcal{U})_o},P_{U, \mix, \mathcal{U}_o} )}.
\nonumber 
\end{align}
\end{proof}

\begin{remark}
\Label{s5-5}
Section \ref{s5} deals with the security when
a channel $P_{Z|V}$ from $\mathcal{V}$
to $\mathcal{Z}$ is given.
The discussion of Section \ref{s5} can be extended to 
the case with a channel $P_{Z|VU}$
from $\mathcal{V}\times \mathcal{U}$
to $\mathcal{Z}$.
In this case, 
$\psi(\rho|P_{Z|V},P_{V|U},P_U)$
and 
$E_0(\rho|P_{Z|V},P_{V|U},P_U)$
are modified to 
\begin{align}
&\psi(\rho|P_{Z|V,U},P_{V|U},P_U) \nonumber \\
:=& 
\log \sum_u P_U(u) \sum_v P_{V|U}(v|u) \sum_z P_{Z|V,U}(z|v,u)^{1+\rho} P_{Z|U}(z|u)^{-\rho}\nonumber \\
&E_0(\rho|P_{Z|V,U},P_{V|U},P_U) \nonumber \\
:=&
\log \sum_u P_U(u) \sum_z\left(
\sum_{v} P_{V|U}(v|u) P_{Z|V,U}(z|v,u)^{1/(1-\rho )} \right)^{1-\rho}.\nonumber 
\end{align}
All of the discussions in this section are still valid 
even if we replace $P_{Z|V}(z|v)$ by $P_{Z|V,U}(z|v,u)$ with the above modification.
These extensions to the channel
$P_{Z|VU}$ will be used in Section \ref{s9}
as a mathematical tool for our proof.
\end{remark}

\section{Asymptotic Conditional Uniformity}\Label{s6-2-}
\subsection{Three Kinds of Asymptotic Conditional Uniformity Conditions}\Label{s6-2-2}
In SMC, we use the message $S_{\mathcal{I}^c}$ as a dummy message.
The secrecy of the message $S_{\mathcal{I}}$
depends on the conditional entropy of the dummy message $S_{\mathcal{I}^c}$
given $S_{\mathcal{I}}$.
Then, it is not easy to treat the asymptotic performance 
without fixing the conditional entropy rate of the dummy message $S_{\mathcal{I}^c}$.
Hence, we need to characterize the randomness of the dummy message $S_{\mathcal{I}^c}$
under the condition with respect to $S_{\mathcal{I}}$
in the asymptotic setting.
In order to treat the capacity region and the strong security,
we introduce several kinds of asymptotic conditional uniformity conditions
for a general sequence of
source distributions $P_{S_{{\cal T},n}}$
on the message sets ${\cal S}_{i,n}$ for $i=0,1,\ldots,T$
satisfying
the relations
$|{\cal S}_{i,n}|:=e^{n R_{i}}$ for $i=0,1,\ldots,T$.

\begin{definition}
The sequence of distributions $P_{S_{{\cal T},n}}$
of the dummy message $S_{\mathcal{I}^c,n}$ 
is called {\it weak asymptotically conditionally uniform} (WACU)
for a non-empty proper subset $\mathcal{I}(\neq \emptyset) \subsetneq \{1,\ldots, T\}$
when 
\begin{align}
\lim_{n\to \infty}
\frac{1}{n}H(S_{\mathcal{I}^c,n}|S_{\mathcal{I},n},S_{0,n})
=
\sum_{i\in \mathcal{I}^c} R_i \Label{11-9-3}.
\end{align}
\end{definition}

\begin{definition}
The sequence of distributions $P_{S_{{\cal T},n}}$
of the dummy message $S_{\mathcal{I}^c,n}$ 
is called {\it semi-weak asymptotically conditionally uniform} (SWACU)
for a non-empty proper subset $\mathcal{I}(\neq \emptyset) \subsetneq \{1,\ldots, T\}$
when
the relation
\begin{align}
\lim_{n\to \infty}
\frac{1}{n}H_{1+\frac{\delta}{n}}(S_{\mathcal{I}^c,n}|S_{\mathcal{I},n},S_{0,n})
=
\sum_{i\in \mathcal{I}^c} R_i \Label{11-9-3-d}
\end{align}
holds for any $\delta>0$.
\end{definition}

\begin{definition}\Label{D12-18-2}
Fix an arbitrary fixed real number $\epsilon \ge 0$.
The sequence of distributions $P_{S_{{\cal T},n}}$
of the dummy message $S_{\mathcal{I}^c,n}$
is called $\epsilon$-{\it strong asymptotically conditionally uniform} 
($\epsilon$-SACU) for 
for a non-empty proper subset $\mathcal{I}(\neq \emptyset) \subsetneq \{1,\ldots, T\}$
when the relation
\begin{align}
\underline{H}_{\log}(\mathcal{I}^c)
\ge
\sum_{i\in \mathcal{I}^c} (R_i -\epsilon),
\Label{11-9-3-b}
\end{align}
where
\begin{align}
\underline{H}_{\log}(\mathcal{I}^c):=
\lim_{\delta \to \infty} \liminf_{n\to \infty}
\frac{1}{n}H_{1+\frac{\delta \log n}{n}}(S_{\mathcal{I}^c,n}|S_{\mathcal{I},n},S_{0,n}).
\Label{11-9-3-bb}
\end{align}
Since 
$\rho-1$ behaves as $\delta\frac{\log n}{n}$ 
in (\ref{11-9-3-bb}),
we use the subscript $\log$ in (\ref{11-9-3-bb}).
In the case of $\epsilon=0$, 
it is simply called {\it strong asymptotically conditionally uniform}
 (SACU)
for a non-empty proper subset $\mathcal{I}(\neq \emptyset) \subsetneq \{1,\ldots, T\}$.
In this case, the condition (\ref{11-9-3-b}) is equivalent with
\begin{align}
\underline{H}_{\log}(\mathcal{I}^c)
=
\sum_{i\in \mathcal{I}^c} R_i 
\Label{11-9-3-b2}
\end{align}
because the opposite inequality holds due to the cardinalities of respective 
message sets.
\end{definition}

In particular, 
when 
the sequence of distributions $P_{S_{{\cal T},n}}$
of the dummy message $S_{\mathcal{I}^c,n}$ 
is WACU
for any non-empty proper subset $\mathcal{I}\subsetneq \{1,\ldots, T\}$,
it is simply called WACU.
We sometimes fix a family ${\bf J}$ of non-empty proper subsets $\mathcal{I}$ of $\{1, \ldots, T\}$,
and treat only non-empty proper subsets $\mathcal{I} \in {\bf J}$.
In this case, we call 
the sequence of distributions $P_{S_{{\cal T},n}}$
WACU for a family ${\bf J}$
when it is WACU
for any non-empty proper subset $\mathcal{I} \in {\bf J}$.
We also apply these conventions 
to SWACU, SACU, and $\epsilon$-SACU.
The relations among the above conditions are summarized as follows.

\begin{theorem}\Label{th-12-26-1}
The following relations hold.
\[
\begin{array}{ccccc}
 \hbox{SACU} & \Rightarrow & \hbox{SWACU} & \Leftrightarrow & \hbox{WACU} \\
 \Downarrow   &     & &&\\
 \epsilon\hbox{-SACU} & & &&
\end{array}
\]
\end{theorem}

\begin{proof}
The equivalence between SWACU and WACU will be shown as Lemma \ref{lem22} in Appendix \ref{as1}.
Other relations are trivial from their definitions.
\end{proof}

In fact, as is shown in Subsection  \ref{as2},
even if the original information does not satisfy
the WACU condition (\ref{11-9-3}) or the SACU condition (\ref{11-9-3-b2}) with $\epsilon=0$,
if we apply Slepian-Wolf data compression \cite{Slepian} to the original sources
so that the total compressed rate of the whole data attains the entropy rate of the whole sources,
the compressed data satisfies the WACU condition (\ref{11-9-3}) and/or the SACU condition (\ref{11-9-3-b2}).
Similarly,
as is shown in Subsection \ref{as2},
even if the original information does not satisfy
the $\epsilon$-SACU condition (\ref{11-9-3-b}),
if we apply Slepian-Wolf data compression \cite{Slepian} to the original sources
so that the error probability goes to zero exponentially
and the difference between the entropy rate of the whole system and 
the total compressed rate is less than $\epsilon$,
the compressed data satisfies 
the $\epsilon$-SACU condition (\ref{11-9-3-b}).


\subsection{Asymptotic Conditional Uniformity Conditions and Slepian-Wolf Data Compression}\Label{as2}
In Subsection \ref{s7-2}, 
we have introduced several 
asymptotic conditional uniformity conditions.
In this subsection, 
we clarify which kind of data compressed by Slepian-Wolf compression satisfies
asymptotic conditional uniformity conditions.
For this purpose, we assume that
the random variables $S_{\cal T}^n=(S_0^n, S_1^n,\ldots S_T^n)$ are subject to the 
$n$-fold stationary ergodic joint distribution $P_{S_{\cal T}}^n$ 
over $\mathcal{S}_{0}^n\times \mathcal{S}_{1}^n\times \cdots \times \mathcal{S}_{T}^n$.
The symbols $H(S_0,\ldots, S_T)$, $H(S_{\cal I})$, and 
$H(S_0,S_{\cal I})$ describe the entropy rates of the respective random variables for any 
non-empty proper subset $\mathcal{I}\subsetneq \{1,\ldots, T\}$.
The following theorem treats the WACU condition for the compressed data.

\begin{theorem}\Label{th-12-23-3}
We choose the asymptotic compression rates $R_0,\ldots, R_T$
such that $\sum_{i=0}^T R_i = H(S_0,\ldots, S_T)$
and 
$\sum_{i \in {\cal I}} R_i \le H(S_{\cal I})$, 
$R_0+\sum_{i \in {\cal I}} R_i \le H(S_0,S_{\cal I})$ 
for any 
non-empty proper subset
$\mathcal{I}\subsetneq \{1,\ldots, T\}$.
Choose a sequence $m_n$ such that $\frac{m_n}{n} \to 1$. 

Let 
$\varphi_i^n: \mathcal{S}_{i}^{m_n} \to \{1,\ldots, \lceil e^{n R_i} \rceil\}$
be Slepian-Wolf encoders
and
$\hat{\varphi}^n: \{1,\ldots, \lceil e^{n R_0} \rceil\}\times \cdots \times \{1,\ldots, \lceil e^{n R_T} \rceil\} 
\to \mathcal{S}_{0}^{m_n} \times \cdots \times \mathcal{S}_{T}^{m_n}$ be its Slepian-Wolf decoder 
for any positive integer $n$
such that
\begin{align}
\varepsilon(\varphi^n,\hat{\varphi}^n)
:=\mathrm{Pr} \{  
(S_0^{m_n}, \ldots S_T^{m_n})
\neq
\hat{\varphi}^n (\varphi_0^n(S_0^{m_n}),\ldots, \varphi_T^n(S_T^{m_n}))
\} 
&\to 0,\Label{12-23-4}
\end{align}
where $\varphi^n=(\varphi_0^n,\ldots, \varphi_T^n) $.
Then, we have
\begin{align}
\lim_{n \to \infty}
\frac{1}{n}H(
(\varphi_i^n(S_i^{m_n}))_{i \in {\cal I}^c}|
(\varphi_i^n(S_i^{m_n}))_{i \in {\cal I}},\varphi_0^n(S_0^{m_n}))
&= \sum_{i \in {\cal I}^c} R_i \Label{12-23-3}
\end{align} 
for any non-empty proper subset
$\mathcal{I}\subsetneq \{1,\ldots, T\}$.
That is, the compressed data satisfies the WACU condition (\ref{11-9-3}).
\end{theorem}

\begin{remark}\Label{rem1}
Theorem \ref{th-12-23-3} gives only a sufficient condition (\ref{12-23-4})
for the compressed data satisfying the WACU condition. 
For construction of the compressed data satisfying the WACU condition,
it is needed to clarify the existence of a code whose 
the compressed data satisfying the condition (\ref{12-23-4}).

In the single terminal 
Markovian case, 
under the condition $\frac{m_n}{n} \to 1$, 
the second order asymptotic analysis in \cite[Section VII]{hayashi08} guarantees that there exists 
sequence of the pairs of an encoder and a decoder satisfying (\ref{12-23-4}) if and only if $\frac{n-m_n}{\sqrt{n}} \to \infty$.
The extension to the Slepian-Wolf coding has been done with the i.i.d. case \cite{TK12}.
For the boundary of the attainable rate region of Slepian-Wolf data compression in the stationary ergodic case \cite{Cover},
we can show the existence of 
the pair of an encoder and a decoder satisfying (\ref{12-23-4})
with a suitable choice of the sequence $m_n$ under the condition $\frac{m_n}{n} \to 1$ in the following way\footnote{The following discussion does not require any property for source distribution.
That is, 
it can be extended to 
Slepian-Wolf data compression for the general information source 
\cite{M-K} in the sense of Han-Verd\'{u}\cite{han93}.}.

Choose the rates $R_i+\delta$ for any $\delta>0$.
Let 
$\varphi_{i,\delta}^n: \mathcal{S}_{i}^{n} \to \{1,\ldots, \lceil e^{n R_i(1+\delta)} \rceil\}$
be Slepian-Wolf encoders
and
$\hat{\varphi}_{\delta}^n: \{1,\ldots, \lceil e^{n R_0(1+\delta)} \rceil\}\times \cdots \times \{1,\ldots, \lceil e^{n R_T(1+\delta)} \rceil\} 
\to \mathcal{S}_{0}^{n} \times \cdots \times \mathcal{S}_{T}^{n}$
be its Slepian-Wolf decoder 
such that $\varepsilon(\varphi_{\delta}^n,\hat{\varphi}_{\delta}^n) \to 0$
with $\varphi_{\delta}^n:=(\varphi_{0,\delta}^n,\ldots,
\varphi_{T,\delta}^n) $.
For an arbitrary integer $l$, we choose an integer $n_l$ such that 
the inequality $\varepsilon(\varphi_{1/l}^{n},\hat{\varphi}_{1/l}^{n}) \le \frac{1}{l}$ holds for any $n \ge n_l$. 
We define $m_n$ to be $m_n := \lfloor \frac{n}{1+1/l} \rfloor$, where we choose $l$ such that $n_l \le n < n_{l+1} $.
Here, we can choose the integer $l$ for any positive integer $n$.
The construction guarantees that $R_i (1+1/l) (m_n+1) \ge  R_i n \ge R_i (1+1/l) m_n$.
We define the pair of an encoder and a decoder $(\varphi^{n},\hat{\varphi}^{n})$
to be $(\varphi_{1/l}^{m_n},\hat{\varphi}_{1/l}^{m_n})$.
That is, $\varphi^{n}_i$ is chosen to be $\varphi_{i,1/l}^{m_n}$.
Our choices guarantee that
$\frac{m_n}{n} \cong \frac{1}{1+1/l} \to 1$,
and 
$\varepsilon(\varphi^{n},\hat{\varphi}^{n})
=\varepsilon(\varphi_{1/l}^{m_n},\hat{\varphi}_{1/l}^{m_n})
\le 1/l \to 0$.
In this construction,
the encoder $\varphi^{n}_i$ is a map from
$\mathcal{S}_{i}^{m_n}$
to
$\{1,\ldots, \lceil e^{m_n R_i(1+1/l)} \rceil\}
\subset \{1,\ldots, \lceil e^{n R_i} \rceil\}$
because $R_i n \ge  m_n R_i(1+1/l)$.
Hence, 
the pair of an encoder and a decoder $(\varphi^{n},\hat{\varphi}^{n})$
satisfies the assumption of Theorem \ref{th-12-23-3}.
\end{remark}

\begin{proofof}{Theorem \ref{th-12-23-3}}
Assume that
the code $\varphi^n=(\varphi_0^n,\ldots, \varphi_T^n)$ 
satisfies (\ref{12-23-4}).
Since the stationary ergodic source satisfies the strong converse property 
for the data compression, 
due to 
folklore source coding theorem
\cite[Theorem 3.1]{han-folklore},
the code $\varphi^n$ satisfies 
\begin{align*}
\lim_{n \to \infty}
\frac{1}{n}H(\varphi_0^n(S_0^{m_n}),\ldots, \varphi_T^n(S_T^{m_n}))
=\sum_{i=0}^T R_i.
\end{align*} 
Since 
$\frac{1}{n}H(
(\varphi_i^n(S_i^{m_n}))_{i \in {\cal I}^c}|
(\varphi_i^n(S_i^{m_n}))_{i \in {\cal I}},\varphi_0^n(S_0^{m_n}))
\le \sum_{i \in {\cal I}^c} R_i$
and
$\frac{1}{n}H(( \varphi_i^n(S_i^{m_n} ))_{i \in {\cal I}},\varphi_0^n (S_0^{m_n}))
\le R_0+\sum_{i \in {\cal I}} R_i$,
we obtain (\ref{12-23-3}).
\end{proofof}


In Subsection \ref{s7-2}, 
we have introduced 
the $\epsilon$-strong asymptotic conditional uniformity (\ref{11-9-3-b})
as another kind of asymptotic conditional uniformity.
The following theorem shows the $\epsilon$-strong asymptotic conditional uniformity for the compressed data.

\begin{theorem}\Label{th-12-23-1}
We fix a sequence $m_n$ such that $\frac{m_n}{n} \to 1$. 
We also fix an arbitrary $\epsilon \ge 0$
and an arbitrary non-empty proper subset
$\mathcal{I}\subsetneq \{1,\ldots, T\}$.
Then,
we choose the asymptotic compression rates $R_0,\ldots, R_T$ such that 
$\sum_{i=0}^T R_i = H(S_0,\ldots, S_T)+\epsilon$
and 
\begin{align}
\sum_{i \in {\cal I}} R_i \le H(S_{\cal I}), \quad
R_0+\sum_{i \in {\cal I}} R_i \le H(S_0,S_{\cal I}).
\Label{1-22-1}
\end{align}
We choose 
a Slepian-Wolf encoder
$\varphi^n=(\varphi^n_0, \ldots, \varphi^n_T)$ and a Slepian-Wolf decoder $\hat{\varphi}^n$ 
as a map 
$\varphi_i^n: \mathcal{S}_{i}^{m_n} \to \{1,\ldots, \lceil e^{n R_i} \rceil\}$
and a map
$\hat{\varphi}^n: \{1,\ldots, \lceil e^{n R_0} \rceil\}\times \cdots \times \{1,\ldots, \lceil e^{n R_T} \rceil\} 
\to \mathcal{S}_{0}^{m_n} \times \cdots \times \mathcal{S}_{T}^{m_n}$.
When the decoding error probability $\varepsilon(\varphi^n,\hat{\varphi}^n)$ 
satisfies that
\begin{align}
\varepsilon(\varphi^n,\hat{\varphi}^n) p(n)\to 0
\Label{12-26-1}
\end{align}
for any polynomial $p(n)$,
the relation
\begin{align}
&\liminf_{n \to \infty}
\frac{1}{n}H_{1+\rho_n}
((\varphi_i^n (S_i^n ) )_{i \in {\cal I}^c}|
(\varphi_i^n (S_i^n ) )_{i \in {\cal I}},\varphi_0^n (S_0^n ) ) \nonumber\\
\ge & (\sum_{i \in {\cal I}^c} R_i) -\epsilon \ge \sum_{i \in {\cal I}^c} (R_i-\epsilon) 
\Label{12-23-1}
\end{align} 
holds with
$\rho_n=\frac{\delta \log n}{n}$ for any $\delta>0$.
That is,
the compressed data $(\varphi_0^n(S_0^n),\ldots, \varphi_T^n(S_T^n))$
satisfies the $\epsilon$-SACU condition (\ref{11-9-3-b})
for the non-empty proper subset $\mathcal{I}\subsetneq \{1,\ldots, T\}$.
In particular, in the case of $\epsilon=0$,
the compressed data $(\varphi_0^n(S_0^n),\ldots, \varphi_T^n(S_T^n))$
satisfies the SACU condition
for the non-empty proper subset $\mathcal{I}\subsetneq \{1,\ldots, T\}$.

Hence, if the relation (\ref{1-22-1}) holds for any non-empty proper subset $\mathcal{I}\subsetneq \{1,\ldots, T\}$,
the compressed data $(\varphi_0^n(S_0^n),\ldots, \varphi_T^n(S_T^n))$
satisfies the $\epsilon$-SACU condition (\ref{11-9-3-b}).
\end{theorem}

\begin{remark}\Label{rem2}
Theorem \ref{th-12-23-1} gives only a sufficient condition 
(\ref{12-26-1}) for 
the compressed data satisfying the $\epsilon$-SACU condition (\ref{11-9-3-b}).
Hence, it is necessary to clarify the existence of a code whose compressed data satisfying the condition (\ref{12-26-1}).

In the i.i.d. case, 
for an arbitrary $\epsilon>0$ and 
an arbitrary sequence $m_n$ satisfying $\lim_{n\to \infty}\frac{m_n}{n}=1$,
there exists 
a sequence of Slepian-Wolf codes
$(\varphi^n,\hat{\varphi}^n)$
with any rate tuples given in Theorem \ref{th-12-23-1}
such that the decoding error probability $\varepsilon(\varphi^n,\hat{\varphi}^n)$ goes to zero exponentially with respect to $n$\cite{e4}.
That is, there exists a Slepian-Wolf code satisfying the condition 
(\ref{12-26-1}) in Theorem \ref{th-12-23-1}.
However, it is not so easy to give a required code in the case of $\epsilon=0$.
In Appendix \ref{as5}, we give such a code
when $m_n:=\frac{n}{1+\frac{c}{n^t}}$ with $t>1/2$ and $\infty> c>0$.
\end{remark}

\subsection{Proof of Theorem \ref{th-12-23-1}}
For the proof of Theorem \ref{th-12-23-1}, we prepare the following lemma for treating the relation between 
the conditional R\'{e}nyi entropy of the compressed data and the decoding error probability. 
The following lemma treats 
the single terminal data compression for a random variable $S$ 
on a set $\mathcal{S}$ in the single-shot setting.
\begin{lemma}\Label{lem-11-24-1}
Any encoder $\varphi: \mathcal{S} \to \{1, \ldots, M\}$ 
and any decoder $\hat{\varphi}:
\{1, \ldots, M\} \to \mathcal{S}$ for a random variable $S$
satisfy 
\begin{align}
e^{-\rho H_{1+\rho}(S)}
\le 
e^{-\rho H_{1+\rho}(\varphi(S) )}
\le
2^{\rho} e^{-\rho H_{1+\rho}(S)}
+2^{\rho}\varepsilon(\varphi,\hat{\varphi})^{1+\rho},\Label{11-24-5}
\end{align}
where 
$\varepsilon(\varphi,\hat{\varphi})$ is the
decoding error probability 
$\mathrm{Pr} \{ S \neq \hat{\varphi} (\varphi(S))\}$.
\end{lemma}

\begin{proof} 
First, we show the first inequality.
Using the inequality $x^{1+\rho}+y^{1+\rho} \le (x+y)^{1+\rho}$ for $x,y\ge 0$,
we obtain
\begin{align*}
\Bigl(\sum_{s \in \varphi^{-1}(i)} P_S(s) \Bigr)^{1+\rho} 
\ge
\sum_{s \in \varphi^{-1}(i)} P_S(s)^{1+\rho} 
\end{align*}
for any $i=1, \ldots, M$.
Hence,
\begin{align*}
& e^{-\rho H_{1+\rho}(\varphi(S) )}
=
\sum_{i=1}^M \Bigl(\sum_{s \in \varphi^{-1}(i)} P_S(s) \Bigr)^{1+\rho}\nonumber \\
\ge &
\sum_{i=1}^M \sum_{s \in \varphi^{-1}(i)} P_S(s)^{1+\rho} 
= 
\sum_{s } P_S(s)^{1+\rho} 
= e^{-\rho H_{1+\rho}(S)},
\end{align*}
which implies the first inequality of (\ref{11-24-5}).

Next, we show the second inequality of (\ref{11-24-5}).
Given an arbitrary element $i$ in the codebook,
we have two cases:
(1) 
The element $s_i:= \hat{\varphi}(i)$ belongs to $\varphi^{-1}(i)$, i.e., 
there 
exists exact one element $s_i \in \varphi^{-1}(i)$ such that
$\hat{\varphi}(\varphi(s_i))=s_i $.
(2) 
There exists no element $s_i \in \varphi^{-1}(i)$ such that
$\hat{\varphi}(\varphi(s_i))=s_i $.
In case (1),
\begin{align*}
& \Bigl(\sum_{s \in \varphi^{-1}(i)} P_S(s)\Bigr)^{1+\rho}
=\Bigl(P_S(s_i)+ \sum_{s \in \varphi^{-1}(i): \hat{\varphi}(\varphi(s)) \neq s} P_S(s)\Bigr)^{1+\rho}\\
=&2^{1+\rho}
\Bigl(\frac{1}{2}P_S(s_i)+ \frac{1}{2}\sum_{s \in \varphi^{-1}(i): \hat{\varphi}(\varphi(s)) \neq s} P_S(s) \Bigr)^{1+\rho} \\
\le &
2^{1+\rho}
\Bigl(\frac{1}{2} P_S(s_i)^{1+\rho} 
+ \frac{1}{2}\Bigl(\sum_{s \in \varphi^{-1}(i): \hat{\varphi}(\varphi(s)) \neq s} P_S(s)\Bigr)^{1+\rho} \Bigr)\\
= &
2^{\rho} P_S(s_i)^{1+\rho} + 2^{\rho}\Bigl(\sum_{s \in \varphi^{-1}(i): \hat{\varphi}(\varphi(s)) \neq s} P_S(s)\Bigr)^{1+\rho}.
\end{align*}
In case (2),
\begin{align*}
\Bigl(\sum_{s \in \varphi^{-1}(i)} P_S(s)\Bigr)^{1+\rho}
=\Bigl(\sum_{s \in \varphi^{-1}(i): \hat{\varphi}(\varphi(s)) \neq s} P_S(s)
\Bigr)^{1+\rho}.
\end{align*}
Hence, we obtain
\begin{align}
& e^{-\rho H_{1+\rho}(\varphi(S) )}
=
\sum_{i} \Bigl(\sum_{s \in \varphi^{-1}(i)} P_S(s)\Bigr)^{1+\rho}\nonumber \\
\le &
2^{\rho} \sum_{i} P_S(s_i)^{1+\rho} 
+ 2^{\rho}\sum_{i}
\Bigl(\sum_{s \in \varphi^{-1}(i): \hat{\varphi}(\varphi(s)) \neq s} P_S(s)
\Bigr)^{1+\rho}\nonumber \\
\le &
2^{\rho} \sum_{s} P_S(s)^{1+\rho} 
+ 2^{\rho}
\Bigl(\sum_{i}\sum_{s \in \varphi^{-1}(i): \hat{\varphi}(\varphi(s)) \neq s} P_S(s)\Bigr)^{1+\rho}\Label{11-24-1} \\
= &
2^{\rho} \sum_{s} P_S(s)^{1+\rho} 
+ 2^{\rho}
\Bigl(\sum_{s : \hat{\varphi}(\varphi(s)) \neq s} P_S(s)\Bigr)^{1+\rho} \nonumber \\
= &
2^{\rho} e^{-\rho H_{1+\rho}(S)}
+2^{\rho}\varepsilon(\varphi,\hat{\varphi})^{1+\rho},\nonumber 
\end{align}
where (\ref{11-24-1}) follow from the inequality
$x^{1+\rho}+y^{1+\rho} \le (x+y)^{1+\rho}$ for $x,y\ge 0$.
Hence, we obtain the second inequality.
\end{proof}

Then, we obtain the following corollary of Lemma \ref{lem-11-24-1}.
The following corollary treats 
the single terminal data compression for 
a general sequence of random variables $S_n$. 
\begin{corollary}\Label{lem-11-23-1}
Let $\varphi^n$ be an encoder and $\hat{\varphi}^n$ be a decoder
for a general sequence of random variables $S_n$.
When the decoding error probabilities 
$\varepsilon(\varphi^n,\hat{\varphi}^n)$
and the sequence $\{\rho_n\}$ of positive real numbers
satisfy
\begin{align}
\lim_{n \to \infty}
\varepsilon(\varphi^n,\hat{\varphi}^n)^{1+\rho_n}
e^{\rho_n H_{1+\rho_n}(S_n)}
=0,\Label{12-19-1}
\end{align}
we have
\begin{align}
\lim_{n \to \infty} \frac{1}{n}H_{1+\rho_n}(\varphi^n( S_n) )
=
\lim_{n \to \infty} \frac{1}{n}H_{1+\rho_n}(S_n).
\end{align}
\end{corollary}

\begin{proofof}{Corollary \ref{lem-11-23-1}}
The inequality $
\lim_{n \to \infty} \frac{1}{n}H_{1+\rho_n}(\varphi^n( S_n) )
\le
\lim_{n \to \infty} \frac{1}{n}H_{1+\rho_n}(S_n)$ 
follows from the first inequality (\ref{11-24-5}).
We show only the inequality 
$\lim_{n \to \infty} \frac{1}{n}H_{1+\rho_n}(\varphi^n( S_n) )
\ge
\lim_{n \to \infty} \frac{1}{n}H_{1+\rho_n}(S_n)$. 
Using the second inequality in (\ref{11-24-5}), we have
\begin{align}
& \lim_{n \to \infty} \frac{1}{n} H_{1+\rho_n}(\varphi^n( S_n) ) 
= \lim_{n \to \infty} \frac{-1}{n\rho_n}
\log e^{- \rho_n H_{1+\rho_n}(\varphi^n( S_n) ) }
\nonumber\\
\ge &
\lim_{n \to \infty} \frac{-1}{n\rho_n}
\log (2^{\rho_n} e^{-\rho_n H_{1+\rho_n}(S_n)}
+2^{\rho_n}\varepsilon (\varphi^n,\hat{\varphi}^n)^{1+\rho_n}) \nonumber \\
= &
\lim_{n \to \infty} \frac{-1}{n\rho_n}
\log (2^{\rho_n} e^{-\rho_n H_{1+\rho_n}(S_n)}) \Label{11-24-3} \\
=&
\lim_{n \to \infty} 
\frac{1}{n} (H_{1+\rho_n}(S_n) -\log 2 )
=
\lim_{n \to \infty} \frac{1}{n}H_{1+\rho_n}(S_n) ,\nonumber 
\end{align}
where (\ref{11-24-3}) follows from the assumption (\ref{12-19-1}).
\end{proofof}

Now, we show Theorem \ref{th-12-23-1}.

\begin{proofof}{Theorem \ref{th-12-23-1}}
For the proof of Theorem \ref{th-12-23-1},
we choose $\rho_n'$ so that $\rho_n'(1-\rho_n')=\rho_n$. 
Since $\lim_{n \to \infty} \frac{m_n}{n}=1$
and
$\rho \ge \rho_n'$ for all $n$,
we have
\begin{align*}
& H_{1+\rho}(S_0,\ldots, S_T)
\le
\liminf_{n \to \infty} \frac{1}{n}H_{1+\rho_n'}
(S_0^{m_n},\ldots, S_T^{m_n}) \\
\le &
\limsup_{n \to \infty} \frac{1}{n}H_{1+\rho_n'}
(S_0^{m_n},\ldots, S_T^{m_n})
\le
H(S_0,\ldots, S_T).
\end{align*}
Since $\rho_n' \to 0$
and 
$\lim_{\rho \to +0}H_{1+\rho}(S_0,\ldots, S_T)=H(S_0,\ldots, S_T)$,
\begin{align}
\lim_{n \to \infty} \frac{1}{n}H_{1+\rho_n'}
(S_0^{m_n},\ldots, S_T^{m_n})
=
H(S_0,\ldots, S_T).
\Label{2-25-1}
\end{align}
Since $\rho_n'$ behaves as $\frac{\delta \log n}{n}$,
due to the relation (\ref{2-25-1}),
the quantity $e^{\rho_n' H_{1+\rho_n'}(S_0^{m_n}, \ldots S_T^{m_n})}$ behaves as
$e^{\delta (\log n)  H(S_0,\ldots, S_T)} =n^{\delta H(S_0,\ldots, S_T)}$.
Since
$\varepsilon(\varphi^n,\hat{\varphi}^n)^{1+\rho_n'} \le \varepsilon(\varphi^n,\hat{\varphi}^n)$,
the condition (\ref{12-26-1})
guarantees the condition (\ref{12-19-1}).
Hence, Corollary \ref{lem-11-23-1} guarantees that
\begin{align*}
\lim_{n \to \infty}
\frac{1}{n}H_{1+\rho_n'}(\varphi_0^n(S_0^{m_n}),\ldots, \varphi_T^n(S_T^{m_n}))
=(\sum_{i=0}^T R_i)-\epsilon.
\end{align*} 
Since
$\log |\varphi_0^n (\mathcal{S}_0^{m_n}) 
\times \prod_{i \in {\cal I}} \varphi_i^n (\mathcal{S}_i^{m_n}) |=
n (R_0+\sum_{i \in {\cal I}} R_i)$,
Corollary \ref{lem-11-23-2} in Appendix \ref{as4} implies 
(\ref{12-23-1}).
\end{proofof}

\section{Secure Multiplex Coding with Common Messages:
Asymptotic Performance}\Label{s6}
In this section,
we treat the asymptotic performance
for the secure multiplex coding with common messages
when the channel is given as the $n$-fold 
discrete memoryless channel of a given broadcast channel $P_{YZ|X}$.
First, 
we treat what performance can be achieved 
by using Code Ensemble \ref{con1} and Theorem \ref{lem2} in Subsection \ref{s5-2}
without any assumption for the distribution of sources.
In the next step,
we define the capacity region under the asymptotic uniformity of information sources.
In SMC, this restriction for the sources is essential for our definition of the capacity region.
After this definition, we concretely give the capacity region.

\subsection{General Sequence of Information Sources}\Label{s6-1}
First, we treat the secure multiplex coding with common messages
with general sequence of information sources.
For a given set of rates $(R_i)_{i=0}^T$,
we give a general sequence of
source distributions $P_{S_{{\cal T},n}}$
on the message sets ${\cal S}_{i,n}$ for $i=0,1,\ldots,T$
satisfying
the relations
$|{\cal S}_{i,n}|:=e^{n R_{i}}$ for $i=0,1,\ldots,T$.
For a given Markov chains $U \to V\to X\to Y Z$,
we give an asymptotic code construction in the following way.



\begin{const}\Label{as-con2}
Let $\varphi_n$ be a code given in Code Ensemble \ref{con2} in Subsection \ref{s5-3}
satisfying
(\ref{ineq-2}),
(\ref{ineq-6-a-}),
(\ref{ineq-6-b-}), and
(\ref{ineq-6-c-}) 
of length $n$ with 
$|{\cal S}_{i,n}|:=e^{n R_i}$ for $i=0,1,\ldots,T$
and a given Markov chain $U \to V \to X$.
\end{const}
The performance of the code $\varphi_n$ of Code Construction \ref{as-con2}
is characterized as follows.
The conditions (\ref{ineq-6-b-}) and (\ref{ineq-6-c-})
guarantee (\ref{Haya-51-w}) and (\ref{Haya-52-w}) given as follows.
\begin{align}
& \liminf_{n\to \infty}\frac{-1}{n}\log 
P_b [P_{Y|V}^n,\Phi_n,P_{S_{{\cal T},n}}] \nonumber \\
\ge &
-\rho \sum_{i=1}^T R_i- 
\max[ 
E_0 (-\rho| P_{Y|V}, P_{V|U},P_U)  , 
E_0 (-\rho| P_{Y|U,V},P_{V,U})] 
,\Label{Haya-51-w}\\
& \liminf_{n\to \infty}\frac{-1}{n}\log 
P_e [P_{Z|V}^n,\Phi_n,P_{S_{{\cal T},n}}]
\ge 
-\rho R_0- E_0(-\rho| P_{Z|U},P_U) 
\Label{Haya-52-w} 
\end{align}
with any $\rho\in (0,1]$.
Further, due to (\ref{ineq-2}), the leaked information for $S_{\mathcal{I},n}$ can be evaluated as
\begin{align}
& \frac{1}{n}
I(S_{\mathcal{I},n}; Z^n| S_{0,n})[P_{Z|V}^n ,\varphi_{a,n},P_{S_{{\cal T},n}}]
\nonumber 
\\
\le &
\Bigl[ \frac{1}{\rho}\psi(\rho| P_{Z|V}, P_{V|U},P_U)
-\frac{1}{n}H_{1+\rho}(S_{\mathcal{I}^c,n}|S_{\mathcal{I},n},S_{0,n})
\Bigr]_+ 
\nonumber 
\\
&+
(T+2)\frac{\log 2 }{n \rho}.
\nonumber 
\end{align}
We substitute $\rho=a/n$ with an arbitrary real $a>0$
and take the limits $n\to \infty$.
Then, \eqref{3-24-21eq} of Lemma \ref{3-24-3L} leads the inequality
\begin{align}
& 
\limsup_{n\to \infty} \frac{1}{n}I(S_{\mathcal{I},n}; Z^n| S_{0,n})[P_{Z|V}^n ,\varphi_{a,n},P_{S_{{\cal T},n}}] \nonumber \\
\le &
\Bigl[ I(V;Z|U) - 
\liminf_{n\to \infty}
\frac{1}{n}H_{1+a/n}(S_{\mathcal{I}^c,n}|S_{\mathcal{I},n},S_{0,n})
\Bigr]_+ \!
+\! (T+2)\frac{\log 2 }{a}.\nonumber
\end{align}
Taking the limits $a\to \infty$,
we obtain
\begin{align}
& 
\limsup_{n\to \infty} \frac{1}{n}I(S_{\mathcal{I},n}; Z^n| S_{0,n})[P_{Z|V}^n ,\varphi_{a,n},P_{S_{{\cal T},n}}] \nonumber\\
\le &
\Bigl[ I(V;Z|U) 
- 
\lim_{a \to \infty}
\liminf_{n\to \infty}
\frac{1}{n}H_{1+a/n}(S_{\mathcal{I}^c,n}|S_{\mathcal{I},n},S_{0,n})
\Bigr]_+ .\Label{11-27-1b}
\end{align}
So, the asymptotic performance of our code 
given in Code Construction \ref{as-con2}
is characterized in 
\eqref{Haya-51-w},
\eqref{Haya-52-w}, and \eqref{11-27-1b}.

In Code Construction \ref{as-con2}, the parameter $R_0$ is chosen to be 
$R_{\mathrm{c}}$ in BCD.
However, to realize the capacity region of SMC, 
we need to choose the parameter $R_0$ to be
a smaller value than $R_{\mathrm{c}}$ in BCD in general. 
To realize such a choice, we introduce another 
code construction by using Code Ensemble \ref{con1} in Subsection \ref{s5-2}.
As is explained in Remark \ref{R-11-29}, 
such a construction is crucial for achieving the capacity region in general
although Code Construction \ref{as-con2} achieves the capacity region with no common message.

\begin{const}\Label{as-con}
For a given set of rates $(R_i)_{i=0}^T$,
we introduce other parameters $R_{\mathrm{p}}$ and $R_{\mathrm{c}}$ satisfying
\begin{align}
R_{\mathrm{c}} +R_{\mathrm{p}}= \sum_{i=0}^T R_i, \quad
R_{\mathrm{c}} \ge R_0.
\end{align}
In the following, we denote the set of 
$((R_i)_{i=0}^T,R_{\mathrm{p}},R_{\mathrm{c}})$ satisfying the above condition
by ${\cal R}_T$.
In order to apply Code Ensemble \ref{con1} in Subsection \ref{s5-2},
we fix Abelian groups
${\cal B}_{1,n}$ and ${\cal B}_{2,n}$ satisfying 
$|{\cal B}_{1,n}| := e^{n (R_{\mathrm{c}}-R_0)}$
and
$|{\cal B}_{2,n}| := e^{n R_{\mathrm{p}}}$.
Applying Code Ensemble \ref{con1} and Theorem \ref{lem2}
to the $n$-fold discrete memoryless extension 
$U^n \to V^n\to X^n\to Y^n Z^n$ of the above Markov chain
and the Abelian groups ${\cal B}_{1,n}$ and ${\cal B}_{2,n}$,
we find the code $\varphi_n=(\varphi_{a,n},\varphi_{b,n},\varphi_{e,n})$ 
with the message sets ${\cal S}_{i,n}$ for $i=0,1,\ldots,T$
satisfying (\ref{ineq-6-}), (\ref{ineq-6-a}), (\ref{ineq-6-b}), and (\ref{ineq-6-c}). 
\end{const}

The performance of the code $\varphi_n$ of Code Construction \ref{as-con}
is characterized as follows.
The relations (\ref{ineq-6-b}) and (\ref{ineq-6-c}) guarantee that
\begin{align}
&\liminf_{n\to \infty}\frac{-1}{n}\log P_b [P_{Y|V}^n ,\varphi_n,P_{S_{{\cal T},n}}]\nonumber \\
 \ge & 
\min \Bigl[ 
-\rho R_{\mathrm{p}} -E_0 (-\rho| P_{Y|V}, P_{V|U},P_U)  , 
\nonumber \\
&\qquad 
-\rho (R_{\mathrm{p}}+R_{\mathrm{c}}) - E_0 (-\rho| P_{Y|U,V},P_{V,U})\Bigr] 
,\Label{Haya-51-v}\\
&\liminf_{n\to \infty}\frac{-1}{n}\log P_e [P_{Z|V}^n ,\varphi_n,P_{S_{{\cal T},n}}]
 \ge 
-\rho R_{\mathrm{c}}- E_0(-\rho| P_{Z|U},P_U) 
 \Label{Haya-52-v} 
\end{align}
for any $\rho \in (0,1]$.
Hence, due to \eqref{3-24-20eq} and \eqref{3-24-21eq}, 
above both exponents 
(\ref{Haya-51-v}) and (\ref{Haya-52-v})
are positive, i.e., 
both error probabilities go to zero exponentially
when
\begin{align*}
&R_{\mathrm{p}} < I(Y;V|U), \quad
R_{\mathrm{p}}+R_{\mathrm{c}} < I(Y;VU)=I(Y;U)+I(Y;V|U), 
\\
&R_{\mathrm{c}} < I(Z;U) ,
\end{align*}
which are satisfied when
\begin{align}
R_{\mathrm{c}} < \min[I(Y;U),I(Z;U) ], \quad R_{\mathrm{p}} < I(Y;V|U)
\Label{11-9-8}.
\end{align}
Further, due to (\ref{ineq-2-e}), the leaked information for $S_{\mathcal{I},n}$ can be evaluated as
\begin{align}
& \frac{1}{n}
I(S_{\mathcal{I},n}; Z^n| S_{0,n})[P_{Z|V}^n ,\varphi_{a,n},P_{S_{{\cal T},n}}]
\nonumber 
\\
\le &
\Bigl[ [R_{\mathrm{c}}-R_0]_+ 
+\frac{1}{\rho}E_0(\rho| P_{Z|V}, P_{V|U},P_U)
-\frac{1}{n}H_{1+\rho}(S_{\mathcal{I}^c,n}|S_{\mathcal{I},n},S_{0,n})
\Bigr]_+ 
\nonumber \\
& + (T+2)\frac{\log 2 }{n \rho}.
\nonumber 
\end{align}
Similar to \eqref{11-27-1b},
we obtain
\begin{align}
& 
\limsup_{n\to \infty} \frac{1}{n}I(S_{\mathcal{I},n}; Z^n| S_{0,n})[P_{Z|V}^n ,\varphi_{a,n},P_{S_{{\cal T},n}}] \nonumber\\
\le &
\Bigl[ (R_{\mathrm{c}}-R_0) +I(V;Z|U) \nonumber \\
&- 
\lim_{a \to \infty}
\liminf_{n\to \infty}
\frac{1}{n}H_{1+a/n}(S_{\mathcal{I}^c,n}|S_{\mathcal{I},n},S_{0,n})
\Bigr]_+ .\Label{11-27-1}
\end{align}
So, the asymptotic performance of our code in 
Code Construction \ref{as-con}
is characterized in (\ref{Haya-51-v}), (\ref{Haya-52-v}), and (\ref{11-27-1}).

\subsection{Capacity Region}\Label{s6-2}
Next, in order to characterize the limit of the asymptotic performance of
the secure multiplex coding with common messages, 
we define the capacity region 
based on the WACU condition (\ref{11-9-3}).
For this purpose, 
we treat the transmission rate tuple 
$(R_i)_{i=0, \ldots,T}=
(R_0$, $R_1$, \ldots, $R_T)$
and the information leakage rate tuple $( R_{l,\mathcal{I}})_{ \emptyset \neq \mathcal{I} \subsetneq \{1, \ldots, T\}}$,
where $\mathcal{I}$ takes every non-empty proper subset of $\{1, \ldots, T\}$.
The latter describes the rates of the leaked information for the message $S_{\mathcal{I},n}$.
Combining both tuples, we call  
$((R_i)_{i=0, \ldots,T}, ( R_{l,\mathcal{I}})_{ \emptyset \neq \mathcal{I} \subsetneq \{1, \ldots, T\}})$ 
the rate tuple.

\begin{definition}\Label{D12-18-1}
The rate tuple
$((R_i)_{i=0, \ldots,T}, ( R_{l,\mathcal{I}})_{ \emptyset \neq \mathcal{I} \subsetneq \{1, \ldots, T\}})$ 
is said to be \emph{achievable} for the secure multiplex coding with $T$ secret messages for the channel $P_{YZ|X}$
if 
there exist a sequence of
codes $\varphi_n=(\varphi_{a,n},\varphi_{b,n},\varphi_{e,n}) $, i.e., 
Alice's stochastic encoder $\varphi_{a,n}$ from 
$\mathcal{S}_{0,n} \times \mathcal{S}_{1,n} \times \cdots \times \mathcal{S}_{T,n}$ to $\mathcal{X}^n$,
Bob's deterministic decoder $\varphi_{b,n}: \mathcal{Y}^n
\rightarrow \mathcal{S}_{0,n}\times \mathcal{S}_{1,n} \times \cdots \times \mathcal{S}_{T,n}$ 
and
Eve's deterministic decoder $\varphi_{e,n}: \mathcal{Z}^n \rightarrow \mathcal{S}_{0,n}$ satisfying the following conditions:
(1) The $i$-th secret message set $\mathcal{S}_{i,n}$ 
has cardinality $e^{n R_i}$ for $i=1$, \ldots, $T$,
and the common message set $\mathcal{S}_{0,n}$ 
has cardinality $e^{n R_0}$. 
(2)
When a sequence of joint distributions 
$P_{S_{{\cal T},n}}$
on the message sets $\mathcal{S}_{i,n}$ for $T=0,1, \ldots,T$
satisfies 
the WACU condition (\ref{11-9-3})
for a non-empty proper subset $\mathcal{I}(\neq \emptyset) \subsetneq \{1,\ldots, T\}$,
the relations
\begin{align}
\lim_{n\rightarrow \infty} P_b [P_{Y|X}^n,\varphi_n,P_{S_{{\cal T},n}}] &=0
\Label{11-9-4-a}
\\
\lim_{n\rightarrow \infty} P_e [P_{Z|X}^n,\varphi_n,P_{S_{{\cal T},n}}] &=0
\Label{11-9-5-a}
\\
\limsup_{n\rightarrow \infty} I(S_{\mathcal{I},n};Z^n|S_0) [P_{Z|X}^n,\varphi_{a,n},P_{S_{{\cal T},n}}] 
& \le R_{l,\mathcal{I}}
\end{align}
hold.
The capacity region 
$\mathcal{C}$
of the secure multiplex coding is the closure of the achievable
 rate tuples
$((R_i)_{i=0, \ldots,T}, ( R_{l,\mathcal{I}})_{ \emptyset \neq \mathcal{I} \subsetneq \{1, \ldots, T\}})$. 
\end{definition}

\begin{theorem}\Label{th:smc}
The capacity region of the secure multiplex coding with
common messages is given by
the set of rate tuples
$((R_i)_{i=0, \ldots,T}, ( R_{l,\mathcal{I}})_{ \emptyset \neq \mathcal{I} \subsetneq \{1, \ldots, T\}})$
such that
there exist a Markov chain $U\rightarrow V \rightarrow X \rightarrow YZ$ and
\begin{eqnarray}
R_0 &\leq & \min[I(U;Y),I(U;Z)],\nonumber \\
\sum_{i=0}^T R_i &\leq & I(V;Y|U)+\min[I(U;Y),I(U;Z)]\nonumber \\
R_{l,\mathcal{I}} &\geq & 
 \sum_{i\in \mathcal{I}} R_i -[I(V;Y|U)- I(V;Z|U)  ]_+ 
\Label{11-14-1}
\end{eqnarray}
for any
non-empty proper subset
$\mathcal{I}\subsetneq \{1,\ldots, T\}$.
\end{theorem}

Now, we define the capacity region $\mathcal{C}_{\rm nc}$ of the secure multiplex coding with
no common messages as
the set of rate tuples
$((R_i)_{i=1, \ldots,T}, ( R_{l,\mathcal{I}})_{ \emptyset \neq \mathcal{I} \subsetneq \{1, \ldots, T\}})$ satisfying 
$(0,(R_i)_{i=1, \ldots,T}, ( R_{l,\mathcal{I}})_{ \emptyset \neq \mathcal{I} \subsetneq \{1, \ldots, T\}}) \in \mathcal{C}$.
As a corollary, the case with no common message 
is characterized as follows.

\begin{corollary}\Label{C-11-29}
$\mathcal{C}_{\rm nc}$ is given as
the set of rate tuples
$((R_i)_{i=1, \ldots,T}, ( R_{l,\mathcal{I}})_{ \emptyset \neq \mathcal{I} \subsetneq \{1, \ldots, T\}})$
such that
there exist a Markov chain $V \rightarrow X \rightarrow YZ$ and
\begin{eqnarray}
\sum_{i=1}^T R_i &\leq & I(V;Y) \nonumber \\
R_{l,\mathcal{I}} &\geq & 
 \sum_{i\in \mathcal{I}} R_i -[I(V;Y)- I(V;Z)  ]_+ 
\Label{11-14-1b}
\end{eqnarray}
for any
non-empty proper subset
$\mathcal{I}\subsetneq \{1,\ldots, T\}$.
\end{corollary}

\begin{proofof}{Theorem \ref{th:smc}}
The converse part of this coding theorem follows from
that for Corollary \ref{cor2} with the uniform distribution on the whole message sets.
The direct part can be shown by Lemma \ref{L-11-29}.
That is, 
for a rate tuple 
$((R_i)_{i=1, \ldots,T}, ( R_{l,\mathcal{I}})_{ \emptyset \neq \mathcal{I} \subsetneq \{1, \ldots, T\}})$
given in \eqref{11-14-1} and an arbitrary small real number $\varepsilon>0$,
the rate tuple 
$((R_i-\frac{\epsilon}{T})_{i=1, \ldots,T}, ( R_{l,\mathcal{I}})_{ \emptyset \neq \mathcal{I} \subsetneq \{1, \ldots, T\}})$
can be achieved by Lemma \ref{L-11-29}
when the $T+1$-th message $S_{T+1}$ is used as the dummy message subject to the uniform distribution
and its rate $R_{T+1}$ is chosen to be $\max (I(V;Y|U)-\sum_{i=0}^{T} R_i-\frac{\epsilon}{T}, 0)$.
\end{proofof}

\begin{remark}\Label{R-11-29}
As is mentioned in Proof of Theorem \ref{th:smc},
to derive the capacity region, we employ Lemma \ref{L-11-29}, 
which is based on Code Construction \ref{as-con} instead of Code Construction \ref{as-con2}
because the case $\sum_{i=1}^T R_i  > I(V;Y|U)$ requires Code Construction \ref{as-con}.
This is the reason why we introduce 
Code Construction \ref{as-con} as well as
Code Construction \ref{as-con2}.
When $\sum_{i=1}^T R_i  \le I(V;Y|U)$,
the rate tuple 
$((R_i)_{i=1, \ldots,T}, ( R_{l,\mathcal{I}})_{ \emptyset \neq \mathcal{I} \subsetneq \{1, \ldots, T\}})$
given in \eqref{11-14-1} 
can be approximately achieved by Lemma \ref{L-11-29b},
which is based on Code Construction \ref{as-con2}.
That is,
the rate tuple 
$((R_i-\frac{\epsilon}{T} )_{i=1, \ldots,T}, ( R_{l,\mathcal{I}})_{ \emptyset \neq \mathcal{I} \subsetneq \{1, \ldots, T\}})$
can be achieved by Lemma \ref{L-11-29b}
when the $T+1$-th message $S_{T+1}$ is used as the dummy message subject to the uniform distribution
and its rate $R_{T+1}$ is chosen to be 
$\max (I(V;Y|U)-\sum_{i=0}^{T} (R_i-\frac{\epsilon}{T} )-\epsilon, 0)$.
Then, Code Construction \ref{as-con2} gives only the special rate tuple in the capacity region.


When there is no common message, 
it is enough to attain the region given in Corollary \ref{C-11-29}.
Hence, it is sufficient to consider the case with $R_0=0$, which 
implies that $\sum_{i=1}^T R_i \le I(V;Y|U)$. 
That is, if we need to show only Corollary \ref{C-11-29}, 
it is enough to use Lemma \ref{L-11-29b}, which is based on Code Construction \ref{as-con2} instead of Code Construction \ref{as-con}.
\end{remark}

\begin{lemma}\Label{L-11-29b}
Choose a sufficiently small real number $\epsilon >0$ and $(R_i)_{i=0}^{T+1}$ for $i=0,1,\ldots,T,T+1$ satisfying 
\begin{eqnarray}
R_0 &< & \min[I(U;Y),I(U;Z)],  \Label{11-9-6b}\\
\sum_{i=1}^{T+1} R_i  &< & I(V;Y|U) \le (\sum_{i=1}^{T+1} R_i) +\epsilon. \Label{11-9-7b}
\end{eqnarray}
Then, 
the code $\varphi_n$ given by Code Construction \ref{as-con2} 
satisfies
\begin{align}
\lim_{n\to \infty} P_b [P_{Y|V}^n,\varphi_n,P_{S_{{\cal T},n}}
\times P_{S_{T+1,n}}
] &= 0 \Label{11-9-4}\\
\lim_{n\to \infty} P_e [P_{Z|V}^n,\varphi_n,P_{S_{{\cal T},n}}
\times P_{S_{T+1,n}}
] &= 0 \Label{11-9-5} 
\end{align}
and
\begin{align}
& \limsup_{n\to \infty} \frac{1}{n}I(S_{\mathcal{I},n}; Z_n| S_{0,n})[P_{Z|V}^n,\varphi_n,P_{S_{{\cal T},n}}\times P_{S_{T+1,n}}
] \nonumber \\
\le & \sum_{i\in \mathcal{I}} R_i - [ I(V;Y|U) - I(V;Z|U) ]_+ +\epsilon
\Label{11-9-9}
\end{align}
when the sequence of the joint distributions
$P_{S_{{\cal T},n}}$ of information source satisfies 
the WACU condition (\ref{11-9-3}) for any non-empty proper subset $\mathcal{I}\subsetneq \{1,\ldots, T\}$
and
$P_{S_{T+1,n}}$ is the uniform distribution.
\end{lemma}

\begin{lemma}\Label{L-11-29}
Choose a sufficiently small real number $\epsilon >0$ and $(R_i)_{i=0}^{T+1}$ 
for $i=0,1,\ldots,T,T+1$ satisfying 
\begin{align}
R_0 < & \min[I(U;Y),I(U;Z)], \Label{11-9-6}\\
I(V;Y|U) \le (\sum_{i=0}^{T+1} R_i) +\epsilon < & I(V;Y|U)+\min[I(U;Y),I(U;Z)]. \Label{11-9-7}
\end{align}
Then, 
the code $\varphi_n$ given by Code Construction \ref{as-con} with
the choices 
\begin{align}
R_{\mathrm{p}}:= I(V;Y|U)
-\epsilon \hbox{ and }
R_{\mathrm{c}}:= \sum_{i=0}^{T+1} R_i-R_{\mathrm{p}}
\Label{3-22-A}
\end{align}
satisfies
\eqref{11-9-4}, 
\eqref{11-9-5}, 
and
\eqref{11-9-9}
when the sequence of the joint distributions
$P_{S_{{\cal T},n}}$ of information source satisfies 
the WACU condition (\ref{11-9-3}) for any non-empty proper subset $\mathcal{I}\subsetneq \{1,\ldots, T\}$
and
$P_{S_{T+1,n}}$ is the uniform distribution.
\end{lemma}

\begin{proofof}{Lemma \ref{L-11-29b}}
Since the conditions (\ref{11-9-6b}) and (\ref{11-9-7b}) guarantee the conditions (\ref{11-9-8}), 
we obtain (\ref{11-9-4}) and (\ref{11-9-5}).
We need to show only (\ref{11-9-9}).
Assume that $I(V;Y|U) \le I(V;Z|U)$.
Since $|{\cal S}_{\mathcal{I},n}| =e^{n \sum_{i\in \mathcal{I}} R_i}$,
we obtain
$\frac{1}{n}I(S_{\mathcal{I},n}; Z_n| S_{0,n})[P_{Z|V}^n,\varphi_n,
P_{S_{{\cal T},n}} \times P_{S_{T+1,n}}
] 
\le \sum_{i\in \mathcal{I}} R_i $, which implies 
(\ref{11-9-9}).
Hence, it is enough to consider the case
$I(V;Y|U) > I(V;Z|U)$.
Since, as is shown in Lemma \ref{lem22} in Appendix \ref{as1},
the equivalence between 
the SWACU condition (\ref{11-9-3-d}) and 
the WACU condition (\ref{11-9-3}) holds,
we obtain
\begin{align}
\lim_{a \to \infty}
\lim_{n\to \infty}
\frac{1}{n}H_{1+a/n}(S_{\mathcal{I}^c,n}|S_{\mathcal{I},n},S_{0,n})
=\sum_{i\in \mathcal{I}^c} R_i \Label{11-27-2}.
\end{align}
The relations (\ref{11-27-1b}) and (\ref{11-27-2}) yield
\begin{align}
& 
\limsup_{n\to \infty} \frac{1}{n}I(S_{\mathcal{I},n}; Z_n| S_{0,n})[P_{Z|V}^n ,\varphi_{a,n},P_{S_{{\cal T},n}}
\times P_{S_{T+1,n}}
] \nonumber \\
\le &
I(V;Z|U) - \sum_{i\in \mathcal{I}^c} R_i \nonumber \\
= &
  - \sum_{i=1}^{T+1} R_i + I(V;Z|U) + \sum_{i\in \mathcal{I}} R_i \nonumber \\
\le &
\epsilon
 - I(V;Y|U) + I(V;Z|U) + \sum_{i\in \mathcal{I}} R_i ,\Label{11-9-10b}
\end{align}
which implies 
(\ref{11-9-9})
\end{proofof}

\begin{proofof}{Lemma \ref{L-11-29}}
Since the conditions (\ref{11-9-6}), (\ref{11-9-7}), and \eqref{3-22-A}
guarantee the conditions (\ref{11-9-8}), 
we obtain (\ref{11-9-4}) and (\ref{11-9-5}).
We need to show only (\ref{11-9-9}).
When $I(V;Y|U) \le I(V;Z|U)$, we can show (\ref{11-9-9}) by the same way as Lemma \ref{L-11-29b}.
Hence, it is enough to consider the case
$I(V;Y|U) > I(V;Z|U)$.
By the same way as Lemma \ref{L-11-29b},
the relations (\ref{11-27-1}) and (\ref{11-27-2}) yield
\begin{align}
& 
\limsup_{n\to \infty} \frac{1}{n}I(S_{\mathcal{I},n}; Z_n| S_{0,n})[P_{Z|V}^n ,\varphi_{a,n},P_{S_{{\cal T},n}}
\times P_{S_{T+1,n}}
] \nonumber \\
\le &
 (R_{\mathrm{c}}-R_0) +I(V;Z|U) - \sum_{i\in \mathcal{I}^c} R_i \nonumber \\
= &
 R_{\mathrm{c}} - \sum_{i=0}^{T+1} R_i + I(V;Z|U) + \sum_{i\in \mathcal{I}} R_i \nonumber \\
= &
 - R_{\mathrm{p}} + I(V;Z|U) + \sum_{i\in \mathcal{I}} R_i .\Label{11-9-10}
\end{align}
Therefore,
since $R_{\mathrm{p}}= 
I(V;Y|U)-\epsilon$, 
(\ref{11-9-10}) implies (\ref{11-9-9})
when $I(V;Y|U)>I(V;Z|U)$.
\end{proofof}

\section{Secure Multiplex Coding with Common Messages:
Strong Security}\Label{s7}
\subsection{Strong Security}\Label{s7-2}
In this section, we treat the strong security.
A sequence of codes $\varphi_n$ is called {\it strongly secure} 
for a subset $\mathcal{I}\subsetneq \{1,\ldots, T\}$
and a sequence of distributions $P_{S_{{\cal T},n}}$
when the relation 
\begin{align}
\lim_{n\to \infty} I(S_{\mathcal{I},n}; Z_n| S_{0,n})[P_{Z|X}^n,\varphi_n,P_{S_{{\cal T},n}}] =0
\Label{11-13-2}
\end{align}
holds.
Now, we fix a family ${\bf J}$ of non-empty proper subsets $\mathcal{I}$ of $\{1, \ldots, T\}$,
and consider only the security of the messages $S_{\mathcal{I},n}$ for all $\mathcal{I} \in {\bf J}$.


\begin{theorem}\Label{thm-11-22-b}
Assume that the transmission rate tuple $(R_i)_{i=0, \ldots,T}=(R_0,R_1, \ldots, R_T)$
belongs to the inner of the capacity region with $R_{l,\mathcal{I}}=0$
for any subset $\mathcal{I}\in {\bf J}$, 
i.e., there exist an information leakage rate tuple
$( R_{l,\mathcal{I}})_{ \emptyset \neq \mathcal{I} \in {\bf J}^c }$
such that
\begin{align}
((R_i)_{i=0, \ldots,T}, 
( 0 )_{ \mathcal{I} \in {\bf J} },
( R_{l,\mathcal{I}})_{ \emptyset \neq \mathcal{I} \in {\bf J}^c }
)
\in \inn (\mathcal{C}),
\end{align}
where $\inn(\mathcal{C})$ denotes the inner of the set $\mathcal{C}$.
Then, there exists a Markov chain $U \to V \to X$ such that
\begin{align}
\epsilon :=& \min_{\mathcal{I}\in {\bf J}} \frac{I(V;Y|U) - I(V;Z|U)- \sum_{i \in \mathcal{I}} R_i}{|\mathcal{I}^c| } >0, \\
R_0 < & \min[I(U;Y),I(U;Z)],\nonumber \\
\sum_{i=0}^T R_i < & I(V;Y|U)+\min[I(U;Y),I(U;Z)] .\nonumber 
\end{align}
Next, we choose $R_{T+1}:= \max (I(V;Y|U)-\sum_{i=0}^{T} R_i, 0)$
and a small real $\epsilon'>0$ 
such that
$\epsilon' < \frac{\epsilon}{2}$,
$\epsilon' <  I(V;Y|U)+\min[I(U;Y),I(U;Z)]-\sum_{i=0}^{T+1} R_i $.
The code $\varphi_n$ given by Code Construction \ref{as-con} with
the choices $R_{\mathrm{p}}:= I(V;Y|U)
-\epsilon'$ and $R_{\mathrm{c}}:= \sum_{i=0}^{T+1} R_i-R_{\mathrm{p}}$
satisfies
(\ref{11-9-4}), (\ref{11-9-5}), and
the strong security
\begin{align}
\lim_{n\to \infty} I(S_{\mathcal{I},n}; Z_n| S_{0,n})[P_{Z|V}^n,\varphi_n,P_{S_{{\cal T},n}}] =0
\Label{11-13-2b}
\end{align}
for any subset $\mathcal{I}\in {\bf J}$
when the sequence of distributions $P_{S_{{\cal T},n}}$
satisfies
the $(\epsilon-2\epsilon')$-SACU condition (\ref{11-9-3-b}) for the subset $\mathcal{I}$.
\end{theorem}

Thanks to Theorem \ref{thm-11-22-b},
the strong security holds at all inner points of the capacity region $\mathcal{C}$
with $R_{l,\mathcal{I}}=0$
for any subset $\mathcal{I}\in {\bf J}$
under the $\epsilon$-SACU condition (\ref{11-9-3-b})
for any subset $\mathcal{I}\in {\bf J}$.

Here, we address the relation with the paper \cite{yamamoto05}.
When there is no common message,
the paper \cite{yamamoto05} defined the region ${\cal R}_{\rm sto}^I$ as follows.
\begin{definition}\Label{D12-18-b}
The region ${\cal R}_{\rm sto}^I$ is the closure of the set of the rate tuples
$(R_i)_{i=1, \ldots,T}$ satisfying the following.
There exist a sequence of
codes $\varphi_n=(\varphi_{a,n},\varphi_{b,n},\varphi_{e,n}) $, i.e., 
Alice's stochastic encoder $\varphi_{a,n}$ from 
$\mathcal{S}_{1,n} \times \cdots \times \mathcal{S}_{T,n}$ to $\mathcal{X}^n$,
Bob's deterministic decoder $\varphi_{b,n}: \mathcal{Y}^n
\rightarrow \mathcal{S}_{1,n}\times \mathcal{S}_{1,n} \times \cdots \times \mathcal{S}_{T,n}$ 
satisfying the following conditions:
(1) The $i$-th secret message set $\mathcal{S}_{i,n}$ 
has cardinality $e^{n R_i}$ for $i=1$, \ldots, $T$,
(2)
When the message obeys the uniform distribution,
the relations \eqref{11-9-4-a} and
\begin{align}
\limsup_{n\rightarrow \infty} I(S_{t,n};Z^n|S_0) [P_{Z|X}^n,\varphi_{a,n},P_{S_{{\cal T},n}}
\times P_{S_{T+1},n}
] =0
\end{align}
hold for $t=1, \ldots, T$.
\end{definition}

On the other hand, we define the region 
$\tilde{\cal R}_{\rm sto}^I$ as the set of rate tuples
$(R_i)_{i=1, \ldots,T}$
such that
there exists a Markov chain $V \rightarrow X \rightarrow YZ$ and
\begin{eqnarray}
\sum_{i=1}^T R_i \leq  I(V;Y), \quad
R_t  \le  [I(V;Y)- I(V;Z)  ]_+ 
\Label{11-14-1c}
\end{eqnarray}
for $t=1, \ldots, T$.
Then, Theorem \ref{thm-11-22-b} and Corollary \ref{C-11-29}
guarantee the relation
\begin{align}
{\cal R}_{\rm sto}^I
=\tilde{\cal R}_{\rm sto}^I,
\Label{11-14-1d}
\end{align}
which is the same as the result by the paper \cite[(138)]{yamamoto05}.
Here,
Corollary \ref{C-11-29} implies 
${\cal R}_{\rm sto}^I \subset \tilde{\cal R}_{\rm sto}^I$ and 
Theorem \ref{thm-11-22-b} does 
${\cal R}_{\rm sto}^I \supset \inn(\tilde{\cal R}_{\rm sto}^I)$.
Since ${\cal R}_{\rm sto}^I$ and $\tilde{\cal R}_{\rm sto}^I$
are the closed sets, we obtain \eqref{11-14-1d}.

In order to show Theorem \ref{thm-11-22-b}, we prepare the following lemma.

\begin{lemma}\Label{thm-11-22}
We fix a subset $\mathcal{I}\subsetneq \{1,\ldots, T\}$.
Assume that 
the transmission rate tuple $(R_i)_{i=0, \ldots,T}$, the sequence of distributions $P_{S_{{\cal T},n}}$,
and a Markov chain $U \to V \to X$
satisfy that
\begin{align}
\delta':= &\frac{1}{2}
\Bigl(
\underline{H}_{\log}(\mathcal{I}^c)
\nonumber \\
&\qquad -(\sum_{i=1}^T R_i - I(V;Y|U) + I(V;Z|U))\Bigr) >0,
\Label{11-13-1} \\
R_0 < & \min[I(U;Y),I(U;Z)],\nonumber \\
\sum_{i=0}^T R_i < & I(V;Y|U)+\min[I(U;Y),I(U;Z)].\nonumber 
\end{align}
When we choose 
$R_{T+1}:= \max (I(V;Y|U)-\sum_{i=0}^{T} R_i, 0)$
and a small real $\epsilon'>0$ such that
$\epsilon' \le \delta'$ and 
$\epsilon' <  I(V;Y|U)+\min[I(U;Y),I(U;Z)]-\sum_{i=0}^{T+1} R_i $,
the code $\varphi_n$ given by Code Construction \ref{as-con} with
the choices $R_{\mathrm{p}}:= I(V;Y|U)
-\epsilon'$ and $R_{\mathrm{c}}:= \sum_{i=0}^{T+1} R_i-R_{\mathrm{p}}$
satisfies
(\ref{11-9-4}), (\ref{11-9-5}), and
the strong security
\begin{align}
\lim_{n\to \infty} I(S_{\mathcal{I},n}; Z_n| S_{0,n})[P_{Z|V}^n,\varphi_n,P_{S_{{\cal T},n}}
\times P_{S_{T+1},n}
] =0
\Label{11-13-2c}.
\end{align}
\end{lemma}

\begin{proofof}{Theorem \ref{thm-11-22-b}}
First, we fix an arbitrary subset $\mathcal{I}\in {\bf J}$.
Hence, 
\begin{align*}
& \sum_{i\in \mathcal{I}^c} (R_i -(\epsilon-2\epsilon'))
-(\sum_{i=1}^{T+1} R_i - I(V;Y|U) + I(V;Z|U)) \\
\ge &
(\sum_{i\in \mathcal{I}^c} R_i) -|\mathcal{I}^c|(\epsilon-2\epsilon')
-(\sum_{i=1}^{T+1} R_i - I(V;Y|U) + I(V;Z|U)) \\
= &
I(V;Y|U) - I(V;Z|U) -\sum_{i\in \mathcal{I}} R_i
-|\mathcal{I}^c|(\epsilon-2\epsilon') \\
\ge &
|\mathcal{I}^c|\epsilon
-|\mathcal{I}^c|(\epsilon-2\epsilon') 
=2|\mathcal{I}^c| \epsilon'\ge 2 \epsilon'.
\end{align*}
Thus, since the sequence of distributions $P_{S_{{\cal T},n}}$
satisfies
the $\epsilon-2\epsilon'$-SACU condition (\ref{11-9-3-b}) for the subset $\mathcal{I}$,
\begin{align}
\delta':=
&\frac{1}{2}
\Bigl(
\underline{H}_{\log}(\mathcal{I}^c)
\nonumber \\
&\qquad -(\sum_{i=1}^{T+1} R_i - I(V;Y|U) + I(V;Z|U))\Bigr) \nonumber \\
\ge &
\frac{1}{2}
\Bigl(
\sum_{i\in \mathcal{I}^c} (R_i -(\epsilon-2\epsilon') )
-(\sum_{i=1}^{T+1} R_i - I(V;Y|U) + I(V;Z|U)) \Bigr)
\nonumber \\
\ge & \epsilon'. \nonumber 
\end{align}
Hence, any real number $\epsilon'>0$ given in Theorem \ref{thm-11-22-b}
satisfies the condition for $\epsilon'>0$ in Lemma \ref{thm-11-22}.
Thus, applying Lemma \ref{thm-11-22}, 
we obtain (\ref{11-13-2b}) for the subset $\mathcal{I}$.
Since the subset $\mathcal{I}$ is an arbitrary element of ${\bf J}$,
we obtain Theorem \ref{thm-11-22-b}.
\end{proofof}

\begin{proofof}{Lemma \ref{thm-11-22}}
Since $\epsilon'>0$, we have the second condition of (\ref{11-9-8}).
Due to the choice of $\epsilon'>0$,
\begin{align*}
0=& I(V;Y|U) - \epsilon' -R_{\mathrm{p}}\\
> &I(V;Y|U)-\Bigl( I(V;Y|U)+\min[I(U;Y),I(U;Z)]-\sum_{i=0}^{T+1} R_i
\Bigr)\\
& -R_{\mathrm{p}}\\
= &\sum_{i=0}^{T+1} R_i- \min[I(U;Y),I(U;Z)] -R_{\mathrm{p}}\\
= &R_{\mathrm{c}}- \min[I(U;Y),I(U;Z)] ,
\end{align*}
which implies the first condition of (\ref{11-9-8}).
Hence, we obtain (\ref{11-9-4}) and (\ref{11-9-5}).

Next, 
we define 
\begin{align*}
\rho_n:=& \frac{2 \log n}{n\delta'}, \\
C_n:=&
\Bigl(-\rho_n n (R_{\mathrm{c}}-R_0) +\rho_n H_{1+\rho_n}(S_{\mathcal{I}^c,n}|S_{\mathcal{I},n},S_{0,n}) \\
&\hspace{17ex} -n E_0(\rho_n| P_{Z|V}, P_{V|U},P_U) \Bigr).
\end{align*}

The condition (\ref{11-13-1}) and $\epsilon' \le \delta'$ 
imply that
\begin{align}
& \liminf_{n \to \infty}
\frac{C_n}{n \rho_n} \nonumber \\
=&
\liminf_{n\to \infty}
\frac{1}{n}H_{1+\rho_n}(S_{\mathcal{I}^c,n}|S_{\mathcal{I},n},S_{0,n})
-\sum_{i=1}^{T+1} R_i +R_{\mathrm{p}} - I(V;Z|U) \nonumber \\
\ge &
\underline{H}_{\log}(\mathcal{I}^c)
-\sum_{i=1}^{T+1} R_i +I(V;Y|U)-\delta' - I(V;Z|U) \nonumber \\
= &
\frac{1}{2}
\Bigl( 
\underline{H}_{\log}(\mathcal{I}^c)
-\sum_{i=1}^{T+1} R_i +I(V;Y|U) - I(V;Z|U)
\Bigr) \nonumber \\
=& \delta'> 0.
\end{align}
That is, we can choose a sufficiently large integer $N$ such that
\begin{align}
\frac{C_n}{n \rho_n} \ge \frac{\delta'}{2}
\Label{11-23-1}
\end{align}
for $n \ge N$.
Due to (\ref{ineq-6-a}), the leaked information for $S_{\mathcal{I},n}$ can be evaluated as
\begin{align}
I(S_{\mathcal{I},n}; Z_n| S_{0,n})[P_{Z|V}^n,\varphi_n ,P_{S_{{\cal T},n}}] 
\le 
\frac{2^{T+2}}{\rho_n} e^{-C_n}.\nonumber 
\end{align}
Since (\ref{11-23-1}) implies that
\begin{align*}
&-\log (\frac{2^{T+2}}{\rho_n} e^{-C_n})
=-(T+2)\log 2 + C_n +\log \rho_n \\
\ge &
-(T+2)\log 2 + \frac{\delta'}{2} n \rho_n +\log \rho_n \\
=&
-(T+2)\log 2 + \log \log n - \log \frac{\delta'}{2}
\to \infty,
\end{align*}
we obtain (\ref{11-13-2c}).
\end{proofof}

\subsection{Exponential Decreasing Rate}\Label{s7-3}
In this subsection,
we treat the exponential decreasing rate of leaked information.
In this subsection, 
we assume that the $T+1$-th message $S_{T+1,n}$ is subject to the uniform distribution.
We simplify $P_{S_{{\cal T},n}} \times P_{S_{T+1,n}}$ by $P_{S_{{\cal T},n}}$. 
For a subset $\mathcal{I}\subsetneq \{1,\ldots, T\}$, 
we denote the complementary set in $\{1,\ldots, T\}$ by $\mathcal{I}^c$
and simplify the set $\mathcal{I}^c\cup \{T+1\}$ to $\mathcal{I}^{c,*}$.
Unfortunately, 
the $\epsilon$-SACU condition (\ref{11-9-3-b}) 
is not sufficient for deriving a good exponential decreasing rate of leaked information.
Hence, in this subsection, 
given a sequence of distributions $P_{S_{{\cal T},n}}$,
we introduce the following quantity
\begin{align}
\underline{H}_{1+\rho}(\mathcal{I}^{c,*})
:=
\liminf_{n \to \infty}
\frac{1}{n}
H_{1+\rho}(S_{\mathcal{I}^{c,*},n}|S_{\mathcal{I},n},S_{0,n}) 
\end{align}
for any subset $\mathcal{I} \subset \{1, \ldots, T\}$
and any $\rho \in (0,1]$.

\begin{theorem}\Label{thm-1-15-1}
For given $(R_i)_{i=0}^T$,
we choose $R_{\mathrm{p}}$ and $R_{\mathrm{c}}$
as follows.
\begin{align*}
R_{\mathrm{c}} \ge R_0, \quad
R_{\mathrm{c}} +R_{\mathrm{p}} = \sum_{i=0}^{T+1} R_i.
\end{align*}
We fix a real number $\epsilon > 0$.
We choose a code $\varphi_n$ given by Code Construction \ref{as-con} with
the above choices $R_{\mathrm{p}}$ and $R_{\mathrm{c}}$ and a given 
Markov chain $U \to V \to X$.
When the sequence of distributions $P_{S_{{\cal T},n}}$
satisfies
the $\epsilon$-SACU condition (\ref{11-9-3-b}) for a non-empty proper subset 
$\mathcal{I}(\neq \emptyset) \subsetneq \{1,\ldots, T\}$,
the sequence of codes $\varphi_n$ satisfies (\ref{Haya-51-v}), (\ref{Haya-52-v}), and 
\begin{align}
& \liminf_{n\to \infty}\frac{-1}{n}\log I(S_{\mathcal{I},n}; Z_n| S_{0,n})[P_{Z|V}^n,\Phi_n,P_{S_{{\cal T},n}}] \nonumber \\
\ge &
\sup_{0 < \rho < 1} 
\rho (\underline{H}_{1+\rho}(\mathcal{I}^{c,*})- R_{\mathrm{c}}+R_0 )
 -E_0(\rho|P_{Z|V},P_{V|U},P_U).
\Label{bound-2}
\end{align}
In particular, when the distribution $P_{S_{{\cal T},n}}$ is uniform,
we obtain
\begin{align}
& \liminf_{n\to \infty}\frac{-1}{n}\log I(S_{\mathcal{I},n}; Z_n| S_{0,n})
[P_{Z|V}^n,\Phi_n,P_{S_{{\cal T+1},n}}] \nonumber \\
\ge &
\tilde{E}^{E_0}(R_{\mathrm{p}} -\sum_{i\in \mathcal{I}}R_i ,P_{Z,V,U}),
\Label{bound-2b}
\end{align}
where $\tilde{E}^{E_0}(R ,P_{Z,V,U})$ is defined in \eqref{1-31-3}.
\end{theorem}

Theorem \ref{thm-1-15-1} yields the following observation.
When $R_{\mathrm{p}} - \epsilon- \sum_{i\in \mathcal{I}}R_i >  I(V;Z|U)$
and $\underline{H}_{1+\rho}(\mathcal{I}^c) \ge (\sum_{i\in \mathcal{I}^c}R_i)-\epsilon$ holds with a small $\rho>0$,
the exponent (\ref{bound-2}) is positive, i.e., 
the leaked information goes to zero exponentially.
In particular, when 
\begin{align}
\sum_{i=1}^{T+1} R_i < I(V;Y|U),~
R_0 < \min[I(U;Y),I(U;Z)],
\Label{eq-223}
\end{align}
we can choose $R_{\mathrm{p}}$ and $R_{\mathrm{c}}$ by 
\begin{align}
R_{\mathrm{p}}:= \sum_{i=1}^{T+1} R_i, \quad
R_{\mathrm{c}}:= R_0. \Label{1-18-1}
\end{align}
Then, the inequalities 
(\ref{Haya-51-v}) and (\ref{Haya-52-v})
can be simplified to 
(\ref{Haya-51-w}) and (\ref{Haya-52-w}).
Then, the both decoding error probabilities goes zero exponentially.
Further, the inequality
(\ref{bound-2}) can be simplified to 
\begin{align}
& \liminf_{n\to \infty}\frac{-1}{n}\log I(S_{\mathcal{I},n}; Z_n| S_{0,n})[P_{Z|V}^n,\Phi_n,P_{S_{{\cal T+1},n}}]\nonumber \\
\ge &
\sup_{0 < \rho < 1} 
\rho \underline{H}_{1+\rho}(\mathcal{I}^{c,*})
 -E_0(\rho|P_{Z|V},P_{V|U},P_U).
\Label{bound-b}
\end{align}
Further, in the case of (\ref{eq-223}) and (\ref{1-18-1}),
when the WACU condition holds for $\mathcal{I}$,
the inequality (\ref{11-27-1}) can be simplified to 
\begin{align}
& \limsup_{n\to \infty}\frac{1}{n} I(S_{\mathcal{I},n}; Z_n| S_{0,n})[P_{Z|V}^n,\Phi_n,P_{S_{{\cal T+1},n}}]\nonumber \\
\le &
R_{\mathrm{c}}-R_0+I(V;Z|U) - \sum_{i \in \mathcal{I}^{c,*}} R_i
=
I(V;Z|U) - \sum_{i \in \mathcal{I}^{c,*}}R_i.\Label{12-28-2}
\end{align}

\begin{proofof}{Theorem \ref{thm-1-15-1}}
In Subsection \ref{s6-1},
we have already shown (\ref{Haya-51-v}) and (\ref{Haya-52-v}).
Hence, we need to only show (\ref{bound-2}).
Due to (\ref{ineq-6-a}), the leaked information for $S_{\mathcal{I},n}$ can be evaluated as
\begin{align}
& I(S_{\mathcal{I},n}; Z_n| S_{0,n})[P_{Z|V}^n,\varphi_n ,P_{S_{{\cal T+1},n}}] \nonumber\\ 
\le & 
\frac{2^{T+2}}{\rho} e^{
\rho n (R_{\mathrm{c}}-R_0) 
-\rho H_{1+\rho}(S_{\mathcal{I}^{c,*},n}|S_{\mathcal{I},n},S_{0,n}) 
+n E_0(\rho| P_{Z|V}, P_{V|U},P_U)
}.\nonumber 
\end{align}
Hence,
\begin{align}
& \liminf_{n\to \infty}\frac{-1}{n}\log I(S_{\mathcal{I},n}; Z_n| S_{0,n})
\nonumber \\
\ge &
\rho 
\liminf_{n \to \infty}
\frac{1}{n}H_{1+\rho}(S_{\mathcal{I}^{c,*},n}|S_{\mathcal{I},n},S_{0,n})
\nonumber \\
& 
-\rho (R_{\mathrm{c}}-R_0) -E_0 (\rho| P_{Z|V}, P_{V|U},P_U) \nonumber\\
\ge & 
\rho (\underline{H}_{1+\rho}(\mathcal{I}^{c,*})- R_{\mathrm{c}}+R_0 )
 -E_0(\rho|P_{Z|V},P_{V|U},P_U).
\nonumber
\end{align}
Taking the supremum for $\rho \in [0,1]$,
we obtain (\ref{bound-2}).
\end{proofof}

When the condition (\ref{eq-223}) holds,
the exponent (\ref{bound-b}) can be improved by using 
Theorem \ref{lem1} with Code Construction \ref{as-con2} in the following way.

\begin{theorem}\Label{thm-1-15-1-2}
We fix a real number $\epsilon \ge 0$.
Let $\varphi_n$ be a code given in Code Construction \ref{as-con2}
in Subsection \ref{s6-1}. 
The sequence of codes $\varphi_n$ satisfies
(\ref{Haya-51-w}), (\ref{Haya-52-w}), (\ref{12-28-2}), and
\begin{align}
& \liminf_{n\to \infty}\frac{-1}{n}\log I(S_{\mathcal{I},n}; Z_n| S_{0,n})[P_{Z|V}^n,\Phi_n,P_{S_{{\cal T+1},n}}] \nonumber \\
\ge &
\max_{0 \le \rho \le 1} 
\rho \underline{H}_{1+\rho}(\mathcal{I}^{c,*}) - \psi(\rho|P_{Z|V},P_{V|U},P_U).
\Label{bound-1}
\end{align}
In particular, when the distribution $P_{S_{{\cal T},n}}$ is uniform,
we obtain
\begin{align*}
& \liminf_{n\to \infty}\frac{-1}{n}\log I(S_{\mathcal{I},n}; Z_n| S_{0,n})[P_{Z|V}^n,\Phi_n,P_{S_{{\cal T+1},n}}]
\nonumber \\
\ge & \tilde{E}^{\psi}(\sum_{i\in \mathcal{I}^{c,*}}R_i ,P_{Z,V,U}),
\end{align*}
where $\tilde{E}^{\psi}(R ,P_{Z,V,U})$ is defined in \eqref{1-31-3b}.
\end{theorem}

Now, we compare Theorems \ref{thm-1-15-1} and \ref{thm-1-15-1-2}.
Since the RHS of (\ref{bound-1}) is larger than the RHS of (\ref{bound-b}) due to (\ref{psileqphi}),  
Theorem \ref{thm-1-15-1-2} is better than Theorem \ref{thm-1-15-1}
when the relation (\ref{eq-223}) holds.
Otherwise, the error exponent of (\ref{Haya-51-w}) and/or (\ref{Haya-52-w})
is not positive.
That is,
Theorem \ref{thm-1-15-1-2} cannot yield a reliable communication.
In summary, 
Theorem \ref{thm-1-15-1} has a wider applicability than 
Theorem \ref{thm-1-15-1-2}.
In the special case (\ref{eq-223}),
Theorem \ref{thm-1-15-1-2} is better than Theorem \ref{thm-1-15-1}.

\begin{proof}
Relations (\ref{Haya-51-w}) and (\ref{Haya-52-w}) have been shown in Subsection \ref{s6-1}.
Due to the $\epsilon$-SACU condition,
(\ref{11-27-1b}) guarantees (\ref{12-28-2}).
Using (\ref{ineq-6-a-}) and the $\epsilon$-SACU condition,
we obtain
\begin{align*}
&I(S_{\mathcal{I},n}; Z_n| S_{0,n})[P_{Z|V}^n,\Phi_n,P_{S_{{\cal T+1},n}}] \\
\le &
\frac{2^{T+2}}{\rho} e^{
-\rho 
H_{1+\rho}(S_{\mathcal{I}^{c,*},n}|S_{\mathcal{I},n},S_{0,n})
+n \psi (\rho| P_{Z|V}, P_{V|U},P_U)}.
\end{align*}
Then,
\begin{align}
& \liminf_{n\to \infty}\frac{-1}{n}\log I(S_{\mathcal{I},n}; Z_n| S_{0,n})[P_{Z|V}^n,\Phi_n,P_{S_{{\cal T},n}}] \nonumber \\
\ge &
\rho \underline{H}_{1+\rho}(\mathcal{I}^{c,*}) - \psi(\rho|P_{Z|V},P_{V|U},P_U).
\end{align}
Hence, we obtain (\ref{bound-1}).
\end{proof}

When the above discussion is applied to the wire-tap channel model,
we obtain an extension of existing results to the case of the asymptotic uniform dummy message.
That is, we consider the case with no common messages and $T=2$ when
${S}_1$ corresponds to the message to be secretly sent to Bob,
and ${S}_2$ does to the dummy message making $S_1$ ambiguous to Eve.
For a given rate $R_1$ of secret message
and a given rate $R_2$ of dummy message,
the RHS of (\ref{Haya-51-w}) coincides with the Gallager exponents,
the RHS of (\ref{bound-b}) coincides with the RHS of (59) in
\cite{hay-wire},
and 
the RHS of (\ref{bound-1}) coincides with the exponents of the RHS of (15) in
\cite{hayashi11}.

\section{Practical Code Construction}\Label{s8}
In Section \ref{s8},
we consider how we can construct practically usable 
encoder and decoder for the secure multiplex coding.
When the channel has additive structure, 
the paper \cite[Section V]{hayashi11} constructed a code for wire-tap channel code
from an ordinary linear error correcting code,
and the paper \cite[Section VI]{yamamoto05} did a secure multiple code without common message from an ordinary linear error correcting code.
Here, we construct a secure multiple code with/without common message
when the channel does not necessarily have additive structure
and the message does not necessarily obey the uniform distribution.
We shall show how to convert an ordinary
error correcting code without secrecy consideration
to a code for the secure multiplex coding.
In this section, we treat practical code construction in the single-shot setting
unless otherwise stated.

It is a common practice to assume the uniform distribution
of messages when one evaluates the decoding error probability,
and decoding error probabilities with non-uniform message
distributions are rarely considered in practice.
Thus, we always assume the uniform message distribution
because this assumption is necessary
for the analysis of the decoding error probability.
However, this assumption is unnecessary for
that of the leaked information to Eve.
The analysis of this section holds for general channels with finite alphabets except for Lemma \ref{l-12-21-2}.
Only Lemma \ref{l-12-21-2} assumes the regularity of the channel.


\subsection{First Practical Code Construction: First Type Evaluation}\Label{s8-1}
We construct a code for the secure multiplex coding based on a
given code $\varphi_{\mathrm{p}}$ for BCD
with the common message in $\mathcal{S}_{\mathrm{c}}$ 
and the private message in $\mathcal{S}_{\mathrm{p}}$.
We assume that encoding and decoding of $\varphi_{\mathrm{p}}$
can be efficiently executed.
We shall attach $F'$ and $G'$ in the second step of Code Ensemble \ref{con1}
to $\varphi_{\mathrm{p}}$ so that the resulting code for SMC
enables efficient encoding and decoding.
This type of construction is much more practical than 
Code Ensemble \ref{con1}
because Code Ensemble \ref{con1} uses the random coding for 
the error correcting code $\varphi_{\mathrm{p}}$,
which does not enable efficient encoding nor decoding.
To use the code with $F'$ and $G'$ attached,
we have to evaluate decoding error probability and the amount of information leaked to Eve.
The former is less than or equal to that of the underlying error correcting 
code $\varphi_{\mathrm{p}}$, and
the average of the latter over the ensemble of $F'$ and $G'$  can be evaluated by Lemma \ref{lem2-1} 
with a fixed error correcting code $\varphi_{\mathrm{p}}$.
In our code, we employ a dummy message to realize 
the secrecy of message when the leaked information is very close to the mutual information with the normal receiver and the number of $T$ is fixed.
Now, we present a code construction.


\begin{const}\Label{con3}
First, in order to apply Lemma \ref{lem2-1},
we divide the common message set $\mathcal{S}_{\mathrm{c}}$ of the BCD code $\varphi_{\mathrm{p}}$ 
to $\mathcal{S}_0 \times \mathcal{B}_1$,
and denote the private message set $\mathcal{S}_{\mathrm{p}}$ 
of $\varphi_{\mathrm{p}}$ by $\mathcal{B}_2$.
That is, the code $\varphi_{\mathrm{p}}$
is regarded as a map from $\mathcal{S}_0 \times \mathcal{B}_1\times \mathcal{B}_2$ 
to $\mathcal{X}$.
Then, based on the code $\varphi_{\mathrm{p}}$,
assuming the Abelian group structures in ${\cal B}_1$ and $\mathcal{B}_2$,
we choose an ensemble of isomorphisms\footnote{Remark \ref{3-25R} discusses an efficient realization of an ensemble of isomorphisms $F$ satisfying Condition \ref{C2-b}.}
$F'$ from ${\cal S}_{1}\times \cdots \times {\cal S}_{T+1}$ to ${\cal B}_1\times \mathcal{B}_2$ as Abelian groups
satisfying Condition \ref{C2-b}
while we do not assume any algebraic assumption for the code $\varphi_{\mathrm{p}}$.
In this scenario, 
$S_0$ is common message, 
$S_1$, \ldots, $S_{T}$ are secret messages,
and $S_{T+1}$ is the dummy randomness whose secrecy is not
required. 
We choose the random variable $G'\in {\cal B}_1\times {\cal B}_2$ that obeys the uniform distribution on ${\cal B}_1\times {\cal B}_2$ 
and is independent of the choice of $F'$ and anything else.
Then, by defining a map $\Lambda_{F',G'}(s):=F'(s)+G'$,
we obtain our encoder 
$\varphi_{\mathrm{p}} \circ \Lambda_{F',G'}(s_0,s_1,\ldots,s_{T+1})=
\varphi_{\mathrm{p}} (s_0,\Lambda_{F',G'}(s_1,\ldots,s_{T+1}))$.
The decoder is constructed by applying the inverse 
$\Lambda_{F',G'}^{-1}(b_1,b_2)= {F'}^{-1}((b_1,b_2)-G')$
to the decoded message of the code $\varphi_{\mathrm{p}}$.
\end{const}

The average of the leaked information of the above constructed code
is evaluated as follows.
\begin{lemma}\label{lem:practical1}
For a subset $\mathcal{I}\subsetneq \{1,\ldots, T\}$, 
the quantity $E_{0,\max}(\rho|P_{Z|V})$ defined in \eqref{eq10001} satisfies 
\begin{align}
& 
\rE_{F',G'}
I(S_{\mathcal{I}}; Z |S_{0})[P_{Z|V},
\varphi_{\mathrm{p}} \circ \Lambda_{F',G'},P_{S_{\cal T}}] \nonumber  \\
\le &
\frac{e^{E_{0,\max}(\rho|P_{Z|V})- \rho H_{1+\rho}(S_{\mathcal{I}^{c,*}}|S_{\mathcal{I}},S_{0})}}{\rho}
\Label{eq10000}.
\end{align}
\end{lemma}

\begin{proof}
Applying Lemma \ref{lem2-1}, we obtain
\begin{align}
& 
\rE_{F',G'}
\exp (\rho I(S_{\mathcal{I}}; Z |S_{0})[P_{Z|V},
\varphi_{\mathrm{p}} \circ \Lambda_{F',G'},P_{S_{\cal T}}])\nonumber  \\
\le &
1+
\sum_{s_{0}} P_{S_{0}}(s_{0}) 
\sum_{s_{\mathcal{I}}} P_{S_{\mathcal{I}}|S_{0} }(s_{\mathcal{I}}|s_{0})
e^{- \rho H_{1+\rho}(S_{\mathcal{I}^{c,*}}|S_{\mathcal{I}}=s_{\mathcal{I}},S_{0}=s_{0})} \nonumber \\
&\qquad \qquad \qquad \qquad
 \cdot e^{\psi ( \rho| P_{Z|B_1,B_2,S_0=s_0,\varphi_{\mathrm{p}}}, P_{\mix, \mathcal{B}_1,\mathcal{B}_2} )} 
\Label{ineq-11-1-b} .
\end{align}
Since
\begin{align*}
& e^{\psi ( \rho| P_{Z|B_1,B_2,\varphi_{\mathrm{p}}, S_0}, P_{\mix, \mathcal{B}_1,\mathcal{B}_2} )} 
\le
e^{E_0 ( \rho| P_{Z|B_1,B_2,\varphi_{\mathrm{p}}, S_0}, P_{\mix, \mathcal{B}_1,\mathcal{B}_2} )}  \\
= &
\sum_{z}
(\sum_{b_1,b_2}
\frac{1}{|{\cal B}_1||{\cal B}_2|}
P_{Z|V}(z|\varphi_{\mathrm{p}}(s_0,b_1,b_2) )^{\frac{1}{1-\rho}}  )^{1-\rho},  \\
&\sum_{s_{\mathcal{I}}} P_{S_{\mathcal{I}}|S_{0} }(s_{\mathcal{I}}|s_{0})
e^{- \rho H_{1+\rho}(S_{\mathcal{I}^{c,*}}|S_{\mathcal{I}}=s_{\mathcal{I}},S_{0}=s_{0})}
=
e^{- \rho H_{1+\rho}(S_{\mathcal{I}^{c,*}}|S_{\mathcal{I}},S_{0}=s_{0})} ,
\end{align*}
we obtain
\begin{align}
& 
\rE_{F',G'}
\exp (\rho I(S_{\mathcal{I}}; Z |S_{0})[P_{Z|V},
\varphi_{\mathrm{p}} \circ \Lambda_{F',G'},P_{S_{\cal T}}])\nonumber  \\
\le &
1+
\sum_{s_{0}} P_{S_{0}}(s_{0}) 
e^{- \rho H_{1+\rho}(S_{\mathcal{I}^{c,*}}|S_{\mathcal{I}},S_{0}=s_{0})} \nonumber \\
&\qquad 
 \cdot 
\sum_{z}
(\sum_{b_1,b_2}
\frac{1}{|{\cal B}_1||{\cal B}_2|}
P_{Z|V}(z|\varphi_{\mathrm{p}}(s_0,b_1,b_2) )^{\frac{1}{1-\rho}}  )^{1-\rho} 
\Label{ineq-11-1-d} .
\end{align}
It can be simplified as follows.
\begin{align}
&
\sum_{z}
(\sum_{b_1,b_2}
\frac{1}{|{\cal B}_1||{\cal B}_2|}
P_{Z|V}(z|\varphi_{\mathrm{p}}(s_0,b_1,b_2) )^{\frac{1}{1-\rho}}  )^{1-\rho} 
\nonumber \\
\le &
\max_{P_V}
\sum_{z}
(\sum_{v}P_V(v)
P_{Z|V}(z|v )^{\frac{1}{1-\rho}}  )^{1-\rho} \nonumber \\
=& \max_{P_V} e^{E_{0}(\rho|P_{Z|V},P_V)}
=e^{E_{0,\max}(\rho|P_{Z|V})}.\nonumber 
\end{align}
That is, using the relation
$\sum_{s_{0}} P_{S_{0}}(s_{0}) 
e^{- \rho H_{1+\rho}(S_{\mathcal{I}^{c,*}}|S_{\mathcal{I}},S_{0}=s_{0})} =
e^{- \rho H_{1+\rho}(S_{\mathcal{I}^{c,*}}|S_{\mathcal{I}},S_{0})}$,
we have
\begin{align}
& 
\rE_{F',G'}
\exp (\rho I(S_{\mathcal{I}}; Z |S_{0})[P_{Z|V},
\varphi_{\mathrm{p}} \circ \Lambda_{F',G'},P_{S_{\cal T}}])\nonumber  \\
\le &
1+
e^{- \rho H_{1+\rho}(S_{\mathcal{I}^{c,*}}|S_{\mathcal{I}},S_{0})} 
e^{E_{0,\max}(\rho|P_{Z|V})}.
\Label{ineq-11-1-e} 
\end{align}
Combining the Jensen inequality for $x \mapsto e^x$,
we obtain the desired upper bound (\ref{eq10000}).
\end{proof}

The logarithm of the RHS of (\ref{eq10000}) has the following property.
\begin{lemma}\Label{l-12-21-1}
The functions 
$\rho \mapsto 
E_{0}(\rho|P_{Z|V})- \rho H_{1+\rho}(S_{\mathcal{I}^{c,*}}|S_{\mathcal{I}},S_{0})
-\log \rho$
and
$\rho \mapsto 
E_{0,\max}(\rho|P_{Z|V})- \rho H_{1+\rho}(S_{\mathcal{I}^{c,*}}|S_{\mathcal{I}},S_{0})
-\log \rho$
are convex.
\end{lemma}

\begin{proof}
The function $\rho \mapsto E_0(\rho| \overline{W}^Z ,Q_{V})$
is convex \cite{gallager68}.
Also the function $\rho \mapsto \rho H_{1+\rho}(S_{\mathcal{I}^{c,*}}|S_{\mathcal{I}},S_{0})$
is concave.
Hence, 
$E_{0}(\rho|P_{Z|V},Q_{V})- \rho H_{1+\rho}(S_{\mathcal{I}^{c,*}}|S_{\mathcal{I}},S_{0})
-\log \rho$
is convex.
Similarly, due to Lemma \ref{3-22-1L}, 
the function $\rho \mapsto 
E_{0,\max}(\rho|P_{Z|V})- \rho H_{1+\rho}(S_{\mathcal{I}^{c,*}}|S_{\mathcal{I}},S_{0})
-\log \rho$
is convex.
\end{proof}

As is explained latter,
the bound $e^{E_{0,\max}(\rho|P_{Z|V})}$
is computable in the discrete memoryless case.
On the other hand, the error probabilities 
can be upper bounded by the average
error probabilities of the code $\varphi_{\mathrm{p}}$.

Next, we determine the necessary amount of
dummy randomness so that the amounts of leaked information is
below specified levels.
Suppose that we are given arbitrary error-correcting code $\varphi_{\mathrm{p}}$
for the broadcast channel $P_{YZ|V}$.
The code $\varphi_{\mathrm{p}}$ can be, for example, an LDPC code
\cite{ldpcbook} or a Turbo code \cite{turbobook}
when there is no common message.
Then, we assume that $S_{T+1}$ obeys the uniform distribution on
its alphabet $\mathcal{S}_{T+1}$ and is statistically independent of
all other random variables.
As a corollary to Lemma \ref{lem:practical1},
we have:
\begin{lemma}\label{lem:practical2}
For $\mathcal{I} \subset \{1, \ldots, T\}$, we have
\begin{align}
& \rE_{F',G'}
I(S_{\mathcal{I}}; Z |S_{0})[P_{Z|V},
\varphi_{\mathrm{p}} \circ \Lambda_{F',G'},P_{S_{\cal T}}] \nonumber  \\
\le &
\frac{e^{E_{0,\max}(\rho|P_{Z|V})- \rho (\log |\mathcal{S}_{T+1}|
+H_{1+\rho}(S_{\mathcal{I}^{c}}|S_{\mathcal{I}},S_{0}))}}{\rho}
\Label{eq10002}.
\end{align}
\end{lemma}

By using Eq.\ (\ref{eq10002}), from $\varphi_{\mathrm{p}}$ we can construct
a code for the secure multiplex coding as follows.
For each proper nonempty set $\mathcal{I} \subsetneq
\{1$, \ldots, $T\}$, $\epsilon_{\mathcal{I}}$
denotes the maximum acceptable information leakage for $I(S_{\mathcal{I}}; Z)$.
Denote by $\epsilon_2$
the maximum acceptable probability  for
a chosen $F',G'$ not making $I(S_{\mathcal{I}}; Z |S_{0})$ below
$\epsilon_{\mathcal{I}}$
for some $\mathcal{I}$.

Adjust the size $|\mathcal{S}_{T+1}|$ 
of the dummy randomness so that
\[
\epsilon_{\mathcal{I}} := \frac{2^{T} }{\epsilon_2}
\left(
\inf_{\rho\in (0,1)}\frac{e^{E_{0,\max}(\rho|P_{Z|V})- \rho (
\log|\mathcal{S}_{T+1}|+H_{1+\rho}(S_{\mathcal{I}^c}|S_{\mathcal{I}}, S_{0}))}
}{\rho  }
\right) .
\]
Then, due to (\ref{eq10002}), 
we obtain 
\[
\rE_{F',G'} 
 I(S_{\mathcal{I}}; Z |S_{0})[P_{Z|V},\varphi_{\mathrm{p}} \circ \Lambda_{F',G'},P_{S_{{\cal T}}}]
\leq \epsilon_2 \epsilon_{\mathcal{I}} / 2^{T}
\]

Then,
by the Markov inequality
the probability of choosing $F'$ and $G'$
making $I(S_{\mathcal{I}}; Z |S_{0}) \leq \epsilon_{\mathcal{I}}$
simultaneously for all $\mathcal{I}
\subsetneq  \{1$, \ldots, $T\}$
is $\geq 1 - \epsilon_2$.

When the channel is a regular channel
in the sense of Delsarte-Piret \cite{delsarte82},
the value $E_{0,\max}(\rho|P_{Z|V})$
can be calculated as follows:
\begin{lemma}\Label{l-12-21-2}
When the channel $P_{Z|V}$ is regular 
in the sense of Delsarte-Piret \cite{delsarte82},
\begin{align}
E_{0,\max} ( \rho| P_{Z|V})=E_0 ( \rho| P_{Z|V},P_{\mix,\mathcal{V}}).
\Label{eq-233}
\end{align}
Further, 
when the code $\varphi_{\mathrm{p}}$ is a homomorphism as Abelian group, 
the inequality
\begin{align}
& 
\rE_{F'|G'=g'}
I(S_{\mathcal{I}}; Z |S_{0})[P_{Z|V},
\varphi_{\mathrm{p}} \circ \Lambda_{F',g'},P_{S_{\cal T}}] \nonumber  \\
\le &
\frac{e^{E_{0}(\rho|P_{Z|V},P_{\mix,\mathcal{V}})- 
\rho (\log|\mathcal{S}_{T+1}|+H_{1+\rho}(S_{\mathcal{I}^c}|S_{\mathcal{I}},S_{0}))}}{\rho}
\Label{eq10000-b}
\end{align}
holds for any $g'\in G'$.
\end{lemma}

Thanks to Lemma \ref{l-12-21-2}, 
in the regular case,
when the code $\varphi_{\mathrm{p}}$ is a homomorphism as Abelian group, 
the above procedure for the construction of our code (Code Construction \ref{con3})
can be simplified 
to the following way.
It is enough to choose $F'$ and to fix $G'$ to be $0$,
and  we can replace 
$E_{0,\max}(\rho|P_{Z|V})$ by $E_{0}(\rho|P_{Z|V},P_{\mix,\mathcal{V}})$.
That is, it is enough to calculate 
$\inf_{\rho \in (0,1)}
E_{0}(\rho|P_{Z|V},P_{\mix,\mathcal{V}})
- \rho (\log |\mathcal{S}_{T+1}| + H_{1+\rho}(S_{\mathcal{I}^{c}}|S_{\mathcal{I}},S_{0}))
-\log \rho$.
Due to Lemma \ref{l-12-21-1},
$E_{0}(\rho|P_{Z|V},P_{\mix,\mathcal{V}})
- \rho (\log |\mathcal{S}_{T+1}| + H_{1+\rho}(S_{\mathcal{I}^{c}}|S_{\mathcal{I}},S_{0}))
-\log \rho$ is convex with respect to $\rho$, 
and the infimum is computable 
by the bisection method \cite[Algorithm 4.1]{Boyd}.

\begin{proofof}{Lemma \ref{l-12-21-2}}
First, we choose $P_{V}'$ such that
\begin{align}
E_{0,\max} ( \rho| P_{Z|V})=E_0 ( \rho| P_{Z|V},P_{V}').
\Label{3-22-7eq}
\end{align}
Define $P_{V,v_0}'$ for $v_0 \in {\cal V}$ by 
\begin{align*}
P_{V,v_0}'(v)=P_{V}'(v+v_0).
\end{align*}
Then,
\begin{align}
e^{E_0 ( \rho| P_{Z|V},P_{V}')}
= e^{E_0 ( \rho| P_{Z|V},P_{V,v_0}')}.
\Label{3-22-8eq}
\end{align}
Hence, we obtain
\begin{align*}
& e^{E_{0,\max} ( \rho| P_{Z|V})}
\stackrel{(a)}{=}
e^{E_0 ( \rho| P_{Z|V},P_{V}')}
\stackrel{(b)}{=}
\sum_{v_0\in {\cal V}} \frac{1}{|{\cal V}|}e^{E_0 ( \rho| P_{Z|V},P_{V,v_0}')} \\
\stackrel{(c)}{\le} &
e^{E_0 ( \rho| P_{Z|V},\sum_{v_0\in {\cal V}} \frac{1}{|{\cal V}|} P_{V,v_0}')}
=
e^{E_0 ( \rho| P_{Z|V},P_{\mix,\mathcal{V}})}
\stackrel{(d)}{\le}  e^{E_{0,\max} ( \rho| P_{Z|V})},
\end{align*}
where 
$(a)$, $(b)$, $(c)$, and $(d)$ follow from
\eqref{3-22-7eq},
\eqref{3-22-8eq},
the concavity of $P_{V}\mapsto e^{E_0 ( \rho| P_{Z|V},P_{V})}$ 
(Item (2) of Proposition \ref{lem12-4-1}), and
the definition \eqref{eq10001} of $E_{0,\max} ( \rho| P_{Z|V})$,
respectively.
Thus, we have (\ref{eq-233}).

Next, we show (\ref{eq10000-b}).
When the code $\varphi_{\mathrm{p}}$ is a homomorphism as Abelian group, 
as is mentioned in Lemma \ref{lem2-1},
we have
$\rE_{F'|G'=g'}
I(S_{\mathcal{I}}; Z |S_{0})[P_{Z|V},
\varphi_{\mathrm{p}} \circ \Lambda_{F',g'},P_{S_{\cal T}}] 
=
\rE_{F',G'}
I(S_{\mathcal{I}}; Z |S_{0})[P_{Z|V},
\varphi_{\mathrm{p}} \circ \Lambda_{F',g'},P_{S_{\cal T}}]$.
Hence, 
combining (\ref{eq10002}), we obtain (\ref{eq10000-b}).
\end{proofof}

When the channel 
is given as the 
$n$-fold discrete memoryless extension $P_{Z|V}^n$ of
$P_{Z|V}$,
$E_{0,\max}(\rho|P_{Z|V}^n)$ 
has the following characterization.
Using \cite{arimoto73}, we obtain
\begin{align}
\max_{P_{V^n}}
\sum_{z^n}
(\sum_{v^n}P_{V^n}(v^n)
P_{Z^n|V^n}(z^n|v^n )^{\frac{1}{1-\rho}}  )^{1-\rho} 
=
e^{nE_{0,\max}(\rho|P_{Z|V})}.\nonumber 
\end{align}
Thus, 
we can apply the above discussion to 
the $n$-fold memoryless case
by replacing 
$E_{0,\max}(\rho|P_{Z|V})$ and
$P_{Z|V}$ 
by 
$n E_{0,\max}(\rho|P_{Z|V})$
and $P_{Z|V}^n$.
That is, it is enough to calculate 
$\inf_{\rho \in (0,1)}
n E_{0,\max}(\rho|P_{Z|V})
- \rho (\log |\mathcal{S}_{T+1}| + H_{1+\rho}(S_{\mathcal{I}^{c}}|S_{\mathcal{I}},S_{0}))
-\log \rho$.
Since, 
as is mentioned in Proposition \ref{lem12-4-1},
$Q_V \mapsto e^{E_0(\rho| \overline{W}^Z ,Q_{V})}$ is concave
and $x \mapsto \log x$ is monotone increasing and concave,
$Q_V \mapsto E_0(\rho| \overline{W}^Z ,Q_{V})$ is concave.
Hence, $E_{0,\max}(\rho|P_{Z|V},Q_V)=\max_{Q_V} E_{0}(\rho|P_{Z|V},Q_V)$ can be easily computed. 
Due to Lemma \ref{l-12-21-1},
$n E_{0,\max}(\rho|P_{Z|V})
- \rho (\log |\mathcal{S}_{T+1}| + H_{1+\rho}(S_{\mathcal{I}^{c}}|S_{\mathcal{I}},S_{0}))
-\log \rho$ is convex concerning
with respect to
$\rho$, the infimum is computable 
by the bisection method \cite[Algorithm 4.1]{Boyd}.
Therefore, we can calculate
the minimum size $|\mathcal{S}_{T+1}|$
satisfying that
$n E_{0,\max}(\rho|P_{Z|V})
- \rho (\log |\mathcal{S}_{T+1}| + H_{1+\rho}(S_{\mathcal{I}^{c}}|S_{\mathcal{I}},S_{0}))
-\log \rho$ is smaller than a specified level for all of $\mathcal{I} \subsetneq \{1, \ldots, T\}$.

\subsection{First Practical Construction: Second Type Evaluation}\Label{s8-1-2}
In the above discussion, 
we have to consider the maximum value $E_{0,\max}(\rho|P_{Z|V})$.
However, when there is no common message and the channel $P_{Z|V}$ is not regular,  
one can
improve the bound (\ref{eq10000}) in the $n$-fold memoryless case 
under the same code construction (Code Construction \ref{con3})
as the following way.
In the following, we treat the $n$-fold memoryless extension $P_{Z|V}^n$.
Given an encoder $\varphi_{\mathrm{p}}: {\cal B}_2 \to {\cal V}^n$,
we define the weight distribution $P_{\varphi_{\mathrm{p}}}$
over the set $T_n({\cal V})$ 
of types of length $n$ of the set $\mathcal{V}$ by 
\begin{align}
P_{\varphi_{\mathrm{p}}} (Q_V):=
\frac{
|\{ v^n \in \im  \varphi_{\mathrm{p}} |  \hbox{the type of }v^n \hbox{ is } Q_V. \}|
}{|\im  \varphi_{\mathrm{p}}|}
\end{align}
for $Q_V \in T_n({\cal V})$.
Using the above weight distribution $P_{\varphi_{\mathrm{p}}} $,
we define the distribution
\begin{align*}
\overline{P}_{\varphi_{\mathrm{p}}}(v^n)
:= \frac{P_{\varphi_{\mathrm{p}}} (Q_V)}{|T_n(Q_V)|}
\end{align*}
for $v^n \in \mathcal{V}^n$, where $Q_V$ is the type of $v^n$
and 
\begin{align*}
T_n(Q_V):= 
\{v^n \in \mathcal{U}^n |
\hbox{the type of }v^n
\hbox{ is }Q_V. \} .
\end{align*}

We construct our code by the same way as Subsection \ref{s8-1}.
We apply Lemma \ref{l12-3-1} to the case when 
$\mathcal{G}$ is the $n$-th permutation group,
$\mathcal{V}$ is $\mathcal{V}^n$,
and $P_{Z|V}$ is $P_{Z|V}^n$.
Then,
\begin{align*}
& e^{\psi ( \rho| P_{Z^n|B_1}, P_{\mix, \mathcal{B}_2} )} 
\le
e^{E_0 ( \rho| P_{Z|V}^n,\overline{P}_{\varphi_{\mathrm{p}}})}.
\end{align*}
Hence, combining (\ref{ineq-11-1-b}), we obtain
\begin{align*}
& \rE_{F',G'}
\exp (\rho I(S_{\mathcal{I}}; Z )[P_{Z|V}^n,
\varphi_{\mathrm{p}} \circ \Lambda_{F',G'},P_{S_{\cal T}}])\nonumber  \\
\le &
1+
e^{E_0 ( \rho| P_{Z|V}^n,\overline{P}_{\varphi_{\mathrm{p}}})- 
\rho 
(\log |\mathcal{S}_{T+1}|+
H_{1+\rho}(S_{\mathcal{I}^c}|S_{\mathcal{I}}))} .
\end{align*}
Since $e^x$ is convex, we obtain
\begin{align*}
& \rE_{F',G'}
I(S_{\mathcal{I}}; Z )[P_{Z|V}^n,
\varphi_{\mathrm{p}} \circ \Lambda_{F',G'},P_{S_{\cal T}}]\nonumber  \\
\le &
\frac{
e^{E_0 ( \rho| P_{Z|V}^n,\overline{P}_{\varphi_{\mathrm{p}}})- 
\rho (\log |\mathcal{S}_{T+1}|+
H_{1+\rho}(S_{\mathcal{I}^c}|S_{\mathcal{I}}))}}{\rho} .
\end{align*}

However, it is not easy to calculate the weight distribution $P_{\varphi_{\mathrm{p}}}$
for a given code $\varphi_{\mathrm{p}}$, 
but it is possible to give an upper bound 
for each $P_{\varphi_{\mathrm{p}}} (Q_V)$ in some special cases.
For example, the upper bound in the case of binary BCH codes is discussed in \cite{Kasami}.
We assume that another distribution 
$Q_{\varphi_{\mathrm{p}}}$
over the set $T_n({\cal V})$ 
and a constant $C_1$ satisfy
\begin{align*}
C_1 Q_{\varphi_{\mathrm{p}}} (Q_V)
\ge P_{\varphi_{\mathrm{p}}} (Q_V)
\end{align*}
for any $Q_V \in T_n({\cal V})$.
Similar to $\overline{P}_{\varphi_{\mathrm{p}}}$,
we define the distribution $\overline{Q}_{\varphi_{\mathrm{p}}}$ by
\begin{align*}
\overline{Q}_{\varphi_{\mathrm{p}}}(v^n)
:= \frac{Q_{\varphi_{\mathrm{p}}} (Q_V)}{|T_n(Q_V)|}
\end{align*}
for $v^n \in \mathcal{V}^n$, where $Q_V$ is the type of $v^n$.
Hence, Proposition \ref{lem12-4-1} yields 
\begin{align*}
e^{E_0 ( \rho| P_{Z|V}^n,\overline{P}_{\varphi_{\mathrm{p}}})}
\le
C_1 e^{E_0 ( \rho| P_{Z|V}^n,\overline{Q}_{\varphi_{\mathrm{p}}})}.
\end{align*}
Therefore, we obtain
\begin{align}
& \rE_{F',G'}
I(S_{\mathcal{I}}; Z )[P_{Z|V}^n,
\varphi_{\mathrm{p}} \circ \Lambda_{F',G'},P_{S_{\cal T}}]\nonumber  \\
\le &
C_1 \frac{
e^{E_0 ( \rho| P_{Z|V}^n,\overline{Q}_{\varphi_{\mathrm{p}}})- 
\rho (\log |\mathcal{S}_{T+1}|+H_{1+\rho}(S_{\mathcal{I}^c}|S_{\mathcal{I}}))}}{\rho} .
\Label{12-21-10}
\end{align}
When $C_1$ is sufficiently small and
$\overline{Q}_{\varphi_{\mathrm{p}}}$ does not give the maximum $E_{0,\max} ( \rho| P_{Z|V}^n)$, 
the RHS of (\ref{12-21-10}) is smaller than the RHS of (\ref{eq10000}).
Similar to the regular case of Subsection \ref{s8-1}, 
we can calculate $
\inf_{\rho \in (0,1)}
E_0 ( \rho| P_{Z|V}^n,\overline{Q}_{\varphi_{\mathrm{p}}}) - \rho (\log |\mathcal{S}_{T+1}| + H_{1+\rho}(S_{\mathcal{I}^{c}}|S_{\mathcal{I}},S_{0}))
-\log \rho + \log C_1$
by the bisection method \cite[Algorithm 4.1]{Boyd}.
Therefore, 
in the above case, 
the method in this subsection improves that in Subsection \ref{s8-1}.


\subsection{Second Practical Construction}\Label{s8-2}
In the previous construction, 
when the channel is not a regular channel,
we have to use an upper bound (\ref{eq10000}), 
which is larger than 
$\frac{e^{E_{0}(\rho|P_{Z|V},P_{\mix,\mathcal{V}})- \rho H_{1+\rho}(S_{\mathcal{I}^{c,*}}|S_{\mathcal{I}},S_{0})}}{\rho}$.
In order to use a smaller upper bound
$\frac{e^{E_{0}(\rho|P_{Z|V},P_{\mix,\mathcal{V}})- \rho H_{1+\rho}(S_{\mathcal{I}^{c,*}}|S_{\mathcal{I}},S_{0})}}{\rho}$ even for a non-regular channel,
we introduce another practical construction
when there is no common message.

Assume that ${\cal V}$ has an Abelian group structure.
Now, we give a code ensemble from 
an arbitrary Abelian group ${\cal B}$
and 
an arbitrary encoder $\varphi:{\cal B}_2 \to {\cal V}$
satisfying that
the map $\varphi$ is an injective homomorphism.
In particular, 
when ${\cal B}_2$ and ${\cal V}$ are vector spaces over the finite field $\bF_2$,
the map $\varphi$ can be given as a linear code, such as an LDPC code 
\cite{ldpcbook} or a Turbo code \cite{turbobook}.
However, we do not necessarily need to assume any algebraic structure in the channel $P_{Z,Y|V}$, for now.
We stress that in Code Ensemble \ref{con4}
we use single encoder $\varphi$,
while in Code Construction \ref{const3}
we use multiple encoders with the same code length and
different information rates.

\begin{ensemble}\Label{con4}
We modify the random code given in Lemma \ref{lem2-1} as follows.
We choose an ensemble of isomorphisms
$F'$ from ${\cal S}_{1}\times \cdots \times {\cal S}_{T+1}$ to ${\cal B}_2$
satisfying Condition \ref{C2-b}.
We choose the random variable $G''\in {\cal V}$ 
that obeys the uniform distribution on ${\cal V}$ 
statistically independent of the choice of $F'$.
Then, we define the encoder  
$\tilde{\Lambda}_{F',G''}(s):=(\varphi \circ F')(s)+G''$. 
The decoder is given by 
$\hat{\tilde{\Lambda}}_{F',G''}(v)= {F'}^{-1}(\hat{\varphi}(v-G''))$
by using the decoder $\hat{\varphi}$ of $\varphi$.

This code ensemble can be understood in the following way.
We define the random variable 
$H$ in the quotient group ${\cal V}/\varphi ({\cal B}_2)$ 
that obeys the uniform distribution.
Let $\{y_{h}\}$ be the set of coset representatives.
Let $G'$ be the random variable subject to the uniform distribution on ${\cal B}_2$.
Then, $G''$ is given as $\varphi(G')+y_{H}$.
That is,
the encoder and the decoder can be given as follows.
$\tilde{\Lambda}_{F',G',H}(s):=(\varphi \circ F')(s)+G'+y_{H}$
and
$\hat{\tilde{\Lambda}}_{F',G',H}(v):= {F'}^{-1}(\hat{\varphi}(v-G'-y_{H}))$.
\end{ensemble}

In Code Ensemble \ref{con4},
the random variable 
$H$ corresponds to the choice of the codebook
for error correction.
Let
$\varepsilon_H$
be the decoding error probability 
when we use $H$ as the codebook
and the message obeys the uniform distribution.
Hence, we consider that
$\varepsilon_H$ expresses 
the decoding error probability 
when we use $H$ as the codebook in the following code construction.

For Code Ensemble \ref{con4}, we have the following lemma:
\begin{lemma}
The inequality
\begin{align}
& \rE_{F',G',H} 
e^{\rho I(S_{\mathcal{I}}; Z )[P_{Z|V},
\tilde{\Lambda}_{F',G',H},P_{S_{{\cal T}}}] } \nonumber \\
\le &
1+e^{- \rho H_{1+\rho}(S_{\mathcal{I}^{c,*}}|S_{\mathcal{I}})}
e^{E_0 ( \rho| P_{Z|V},P_{\mix,\mathcal{V}})}\Label{11-24-30}
\end{align}
\end{lemma}
holds for each subset $\mathcal{I} \subsetneq 
\{1$, \ldots, $T\}$.
Thus, applying Jensen inequality to $x\mapsto e^x$,
we have
\begin{align}
& \rE_{F',G',H} 
I(S_{\mathcal{I}}; Z )[P_{Z|V},\tilde{\Lambda}_{F',G',H},P_{S_{{\cal T}}}]  \nonumber \\
\le & 
\frac{e^{E_{0}(\rho|P_{Z|V},P_{\mix,\mathcal{V}})- \rho H_{1+\rho}(S_{\mathcal{I}^{c,*}}|S_{\mathcal{I}})}}{\rho}.
\Label{eq10003}
\end{align}

\begin{IEEEproof}
We apply (\ref{ineq-11-1-d}) to the case 
when $|\mathcal{S}_0|=1$, 
$\mathcal{S}_0=\{s_0\}$,
$|\mathcal{B}_1|=1$,
$\mathcal{B}_1=\{b_1\}$,
and the map $\varphi_{\mathrm{p}}$ is given as
$\varphi_{\mathrm{p}}(s_0,b_1,b_2)=\varphi(b_2)+y_h$
for any $b_2 \in \mathcal{B}_2$. 
Then, we obtain
\begin{align*}
& 
\rE_{F',G'} 
e^{\rho I(S_{\mathcal{I}}; Z)  [P_{Z|V},\tilde{\Lambda}_{F',G',h},P_{S_{{\cal T}}}] }
\nonumber \\
\le &
1+
e^{- \rho H_{1+\rho}(S_{\mathcal{I}^{c,*}}|S_{\mathcal{I}})}
\sum_{z}
(\sum_{b_2}
\frac{1}{|{\cal B}_2|}
P_{Z|V}(z|\varphi(b_2)+y_h )^{\frac{1}{1-\rho}}  )^{1-\rho} .
\end{align*}
Hence, we obtain
\begin{align*}
& \rE_{F',G',H} 
e^{\rho I(S_{\mathcal{I}}; Z)  [P_{Z|V},\tilde{\Lambda}_{F',G',H},P_{S_{{\cal T}}}] }
\nonumber \\
=& \rE_{H} 
\rE_{F',G'|H} 
e^{\rho I(S_{\mathcal{I}}; Z)  [P_{Z|V},\tilde{\Lambda}_{F',G',H},P_{S_{{\cal T}}}] }
\nonumber \\
\le &
1+
e^{- \rho H_{1+\rho}(S_{\mathcal{I}^{c,*}}|S_{\mathcal{I}})}
\rE_{H} 
\sum_{z}
(\sum_{b_2}
\frac{1}{|{\cal B}_2|}
P_{Z|V}(z|\varphi(b_2)+y_H )^{\frac{1}{1-\rho}}  )^{1-\rho} \\
\le &
1+e^{- \rho H_{1+\rho}(S_{\mathcal{I}^{c,*}}|S_{\mathcal{I}})}
\sum_{z}
(
\rE_{H} 
\sum_{b_2}
\frac{1}{|{\cal B}_2|}
P_{Z|V}(z|\varphi(b_2)+y_H )^{\frac{1}{1-\rho}}  )^{1-\rho} \\
=&
1+e^{- \rho H_{1+\rho}(S_{\mathcal{I}^{c,*}}|S_{\mathcal{I}})}
e^{E_0 ( \rho| P_{Z|V},P_{\mix,\mathcal{V}})},
\end{align*}
which implies (\ref{11-24-30}).
\end{IEEEproof}



In order to construct 
a code for the secure multiplex coding (with no common message),
we define the notations as follows.
Let $\epsilon_{\mathcal{I}}$ be the 
maximum acceptable information
leakage for $I(S_{\mathcal{I}}; Z)$ for 
each $\mathcal{I} \subsetneq
\{1$, \ldots, $T\}$.
Let $\epsilon_{b}$ be the maximum acceptable error probability.
Let $\epsilon_2$ be the 
the maximum acceptable probability 
a chosen $F',G''$ 
not making $I(S_{\mathcal{I}}; Z )$ below $\epsilon_{\mathcal{I}}$.
These parameters $\epsilon_{b}$, $\epsilon_{\mathcal{I}}$, and  
$\epsilon_2$ 
are the requirements for our code construction.

\begin{const}\label{const3}
In this construction,
in contrast to Subsections \ref{s8-1} and \ref{s8-1-2}
we assume that we are given multiple error-correcting
codes with the same code length $n$ and different information rates.
Using (\ref{eq10003}),
we construct a code for the secure multiplex coding (with no common message)
as follows:
\begin{enumerate}
\item
We choose 
a suitable Abelian group $\mathcal{B}_2$,
a suitable code $\varphi$,
a suitable sacrifice bit length (the size of $T$-th message), 
and a suitable real value $\epsilon_1 \in (0,1)$
satisfying that
\begin{align}
&\epsilon_{b} \ge \frac{\rE_H \varepsilon_H}{\epsilon_1 } 
\Label{1-16-1} \\
&\epsilon_{\mathcal{I}} \ge 2^{T} \min_{\rho\in (0,1)}\frac{e^{E_{0}(\rho|P_{Z|V},P_{\mix,\mathcal{V}})- \rho H_{1+\rho}(S_{\mathcal{I}^{c,*}}|S_{\mathcal{I}})}}{\rho \epsilon_2 (1-\epsilon_1)} .
\Label{1-21-11} 
\end{align}

\item
We choose $H$ randomly.
Then, we check that $\varepsilon_H$ is less than $\epsilon_{b}$.
If not, we choose another $H$.
We repeat this process until it is successful.
We denote the final choice of $H$ by $H'$.
Thanks to Markov inequality and (\ref{1-16-1}),
the successful probability for one trial is 
at least  $1-\epsilon_1$.

\item
We choose $F'$ and $G'$ randomly.
Then, 
we obtain the pair of 
the encoder 
$\tilde{\Lambda}_{F',G',H'}(s):=(\varphi \circ F')(s)+G'+y_{H'}$
and
the decoder
$\hat{\tilde{\Lambda}}_{F',G',H'}(v):= {F'}^{-1}(\hat{\varphi}(v-G'-y_{H'}))$.
\end{enumerate}
\end{const}

\begin{theorem}
Under the above construction,
the inequality
\begin{align}
 I(S_{\mathcal{I}}; Z )[P_{Z|V},\tilde{\Lambda}_{F',G',H'},
P_{S_{{\cal T}}}] 
\le  \epsilon_{\mathcal{I}} \Label{1-17-1}
\end{align}
holds for all subsets $\mathcal{I} \subsetneq \{1, \ldots, T\}$
with at least with probability $1-\epsilon_2$.
\end{theorem}

\begin{proof}
Markov inequality guarantees that
${\rm Pr} \{ \varepsilon_H \le \epsilon_{b} \}
\ge 1-\epsilon_1$.
Hence, we obtain
\begin{align*}
&\rE_{F',G',H'} 
 I(S_{\mathcal{I}}; Z )[P_{Z|V},\tilde{\Lambda}_{F',G',H},
P_{S_{{\cal T}}}] \\
= &
\rE_{F',G',H| \varepsilon_H \le \epsilon_{b} } 
 I(S_{\mathcal{I}}; Z )[P_{Z|V},\tilde{\Lambda}_{F',G',H},
P_{S_{{\cal T}}}] \\
\le &
\frac{{\rm Pr} \{ \varepsilon_H \le \epsilon_{b} \} }
{{\rm Pr} \{ \varepsilon_H \le \epsilon_{b} \} }
\rE_{F',G',H| \varepsilon_H \le \epsilon_{b} } 
 I(S_{\mathcal{I}}; Z )[P_{Z|V},\tilde{\Lambda}_{F',G',H},
P_{S_{{\cal T}}}] 
\nonumber \\
&+
\frac{{\rm Pr} \{ \varepsilon_H > \epsilon_{b} \} }
{{\rm Pr} \{ \varepsilon_H \le \epsilon_{b} \} }
\rE_{F',G',H| \varepsilon_H > \epsilon_{b} } 
 I(S_{\mathcal{I}}; Z )[P_{Z|V},\tilde{\Lambda}_{F',G',H},
P_{S_{{\cal T}}}] 
\\
=& 
\frac{1}{{\rm Pr} \{ \varepsilon_H \le \epsilon_{b} \} }
\rE_{F',G',H} 
 I(S_{\mathcal{I}}; Z )[P_{Z|V},\tilde{\Lambda}_{F',G',H},
P_{S_{{\cal T}}}] \\
\le &
\frac{1}{1-\epsilon_1}
\rE_{F',G',H} 
 I(S_{\mathcal{I}}; Z )[P_{Z|V},\tilde{\Lambda}_{F',G',H},
P_{S_{{\cal T}}}] \\
\le & \epsilon_2 \epsilon_{\mathcal{I}} / 2^{T} 
\end{align*}
for every $\mathcal{I}$,
where
$\rE_{F',G',H| \varepsilon_H \le \epsilon_{b} } $ denotes
the expectation under the condition $\varepsilon_H \le \epsilon_{b}$.
The final inequality follows from (\ref{eq10003}).
Since the above choice of $F'$, $G'$ and $H'$ is restricted to 
the set $\{(f',g',h') |\varepsilon_h \le \epsilon_{b}\}$,
due to Markov inequality,
the probability of choosing $F'$, $G'$ and $H'$ making 
(\ref{1-17-1}) simultaneously for all
 $\mathcal{I} \subsetneq \{1, \ldots, T\}$
is not less than $1-\epsilon_2$.
\end{proof}

Further, when the channel 
is given as the 
$n$-fold discrete memoryless extension $P_{Z|V}^n$ of
$P_{Z|V}$,
the quantity $E_{0}(\rho|P_{Z|V}^n,P_{\mix,{\cal V}^n})$ 
is simplified to
$n E_{0}(\rho|P_{Z|V},P_{\mix,{\cal V}})$. 
Hence, similar to the regular case of Subsection \ref{s8-1}, 
we can calculate the 
right hand side of (\ref{1-21-11})
by the bisection method \cite[Algorithm 4.1]{Boyd}.

\section{Channel-Universal Coding for Secure Multiplex Coding with Common Messages}\Label{s9}
In order to treat universal coding for the multiplex coding with common messages,
we introduce the universally attainable exponents 
of the multiplex coding with common messages
in the $n$-fold discrete memoryless setting
by adjusting the original definition for the BCD given by
K\"orner and Sgarro \cite{korner80}.
Similar to Subsection \ref{s7-3},
in this section, we employ $T+1$-th message $S_{T+1}$ as a dummy message subject to the 
uniform distribution, and
assume that the $T+1$-th message $S_{T+1,n}$ is subject to the uniform distribution.
We simplify $P_{S_{{\cal T},n}} \times P_{S_{T+1,n}}$ by $P_{S_{{\cal T},n}}$. 
For a subset $\mathcal{I}\subsetneq \{1,\ldots, T\}$, 
we denote the complementary set in $\{1,\ldots, T\}$ by $\mathcal{I}^c$
and simplify the set $\mathcal{I}^c\cup \{T+1\}$ to $\mathcal{I}^{c,*}$.

In order to treat universal coding for secure multiplex coding with common messages,
we focus on $2^{T+1}-2$ functions to express 
the evaluations of the exponential decreasing rates of decoding error probabilities 
and the asymptotic evaluations of leaked information.
For describing bounds of the exponential decreasing rates of both decoding error probabilities, we need two functions.
For treating the asymptotic evaluations of leaked information,
we need $2^{T+1}-4$ functions
because the number of non-empty proper subsets $\mathcal{I}(\neq \emptyset)\subsetneq \{1,\ldots, T\}$ is $2^T-2$
and we treat the exponential decreasing rates and the information leakage rates of leaked information for respective 
non-empty proper subsets $\mathcal{I}(\neq \emptyset)\subsetneq \{1,\ldots, T\}$.
Then, we need to treat $2^{T+1}-2$ functions.
Since we do not assume the uniformity, we cannot describe 
our bounds of the exponential decreasing rate and the information leakage rate of leaked information
as functions of the rate tuples $(R_{\mathrm{p}}$, $R_{\mathrm{c}}$, $(R_i)_{i=0,1,\ldots,T,T+1} )$.
In the following discussion, we treat our bound of the exponential decreasing rate of leaked information for 
a non-empty proper subset $\mathcal{I}(\neq \emptyset) \subsetneq \{1,\ldots, T\}$
as a function of $\underline{H}_{2}(\mathcal{I}^{c,*})$, $R_{\mathrm{c}}$, and $R_0$ as well as the channel $W$.
Similarly, we treat our bound of the information leakage rate of leaked information for 
a non-empty proper subset $\mathcal{I}(\neq \emptyset)  \subsetneq \{1,\ldots, T\}$
as a function of $\underline{H}_{\log}(\mathcal{I}^{c,*})$, $R_{\mathrm{c}}$, and $R_0$ as well as the channel $W$.
Our bounds of the exponential decreasing rates of both decoding error probabilities
are described as functions of $R_{\mathrm{p}}$, $R_{\mathrm{c}}$, and the channel $W$.
Hence, the outcomes of the above $2^{T+1}-2$ functions are decided by $2^{T+1}-1$ real numbers 
$R_{\mathrm{p}}$, $R_{\mathrm{c}}$, $R_0$, 
and $(\underline{H}_{2}(\mathcal{I}^{c,*}),\underline{H}_{\log}(\mathcal{I}^{c,*}) )_{\mathcal{I}(\neq \emptyset) \subsetneq \{1,\ldots, T\}}$
as well as the channel $W$.

\begin{definition}\Label{def:univexp}
A set of functions $({E}^{b}$, ${E}^{e}$,
$({E}_+^{\mathcal{I}}$, 
${E}_-^{\mathcal{I}})_{\mathcal{I}\subsetneq \{1,\ldots,T\} } )$ 
from 
$\mathbf{R}_{\ge 0}^{2^{T+1}-1} \times 
\mathcal{W}(\mathcal{X}$, $\mathcal{Y}\times \mathcal{Z})$
to $\mathbf{R}_{\ge 0}^{2^{T+1}-2}$ is said to be
a universally attainable set of exponents and information leakage rate
for the
family $\mathcal{W}(\mathcal{X}$, $\mathcal{Y}\times \mathcal{Z})$
if for any $\epsilon>0$ and any rate tuples
$(R_{\mathrm{p}}$, $R_{\mathrm{c}}$, $(R_i)_{i=0,1,\ldots,T} )$,
there exist 
a sufficiently large integer $N$
and
a sequence of codes $\varphi_n$ of length $n$ satisfying the following conditions:
(1)
The $i$-th secret message set $\mathcal{S}_{i,n}$ of the code $\varphi_n$
has cardinality $e^{n R_i}$ for $i=1$, \ldots, $T$,
and the common message sets $\mathcal{S}_{0,n}$ 
has cardinality $e^{n R_0}$. 
(2)
Any sequence of joint distributions $P_{S_{{\cal T},n}}$ for all of the $i$-th secret $S_{i,n}$ on $\mathcal{S}_{i,n}$ and the common message $S_{0,n}$ on $\mathcal{S}_{0,n}$ satisfies 
the inequalities
\begin{align}
P_b[W^n,\varphi_n,P_{S_{{\cal T+1},n}}]
 \leq & \exp(-n
[ 
{E}^{b}
(R_{\mathrm{p}}, R_{\mathrm{c}},R_0,
W)
-\epsilon]),\Label{Haya-51-f}\\
P_e[W^n,\varphi_n,P_{S_{{\cal T+1},n}}]
\leq & \exp(-n
[ {E}^{e}
(R_{\mathrm{p}}, R_{\mathrm{c}},R_0,
W)
-\epsilon]),\Label{Haya-52-f}
\end{align}
and
\begin{align}
&
\liminf_{n\to \infty} \frac{-1}{n} \log I(S_{\mathcal{I},n}; Z^n| S_{0,n})[W^n,\varphi_n,P_{S_{{\cal T+1},n}}] \nonumber \\
\ge &
{E}_+^{\mathcal{I}}
(R_{\mathrm{p}}, R_{\mathrm{c}},R_0,
(\underline{H}_{2}(\mathcal{I'}^{c,*}),\underline{H}_{\log}(\mathcal{I'}^{c,*}) )_{\mathcal{I'}(\neq \emptyset) \subsetneq \{1,\ldots, T\}},W)
,
\Label{eq:logexp} \\
&
\limsup_{n\to \infty} \frac{1}{n} I(S_{\mathcal{I},n}; Z^n| S_{0,n})[W^n,\varphi_n,P_{S_{{\cal T+1},n}}] \nonumber \\
\leq & 
{E}_-^{\mathcal{I}}
(R_{\mathrm{p}}, R_{\mathrm{c}},R_0,
(\underline{H}_{2}(\mathcal{I'}^{c,*}),\underline{H}_{\log}(\mathcal{I'}^{c,*}) )_{\mathcal{I'}(\neq \emptyset) \subsetneq \{1,\ldots, T\}},W)
,\Label{eq:logexp2}
\end{align}
hold for 
any channel 
$W \in \mathcal{W}(\mathcal{X}$, $\mathcal{Y}\times \mathcal{Z})$,
any non-empty proper subset $\mathcal{I}(\neq \emptyset)\subsetneq \{1,\ldots, T\}$,
and any $n \ge N$.
Here, 
${E}^{b}
(R_{\mathrm{p}}, R_{\mathrm{c}},R_0,
(\underline{H}_{2}(\mathcal{I'}^{c,*}),\underline{H}_{\log}(\mathcal{I'}^{c,*}) )_{\mathcal{I'}(\neq \emptyset) \subsetneq \{1,\ldots, T\}},W)
$
and
${E}^{e}
(R_{\mathrm{p}}, R_{\mathrm{c}},R_0,
(\underline{H}_{2}(\mathcal{I'}^{c,*}),\underline{H}_{\log}(\mathcal{I'}^{c,*}) )_{\mathcal{I'}(\neq \emptyset) \subsetneq \{1,\ldots, T\}},W)
$
are abbreviated to
${E}^{b}
(R_{\mathrm{p}}, R_{\mathrm{c}},R_0,
,W)$
and
${E}^{e}
(R_{\mathrm{p}}, R_{\mathrm{c}},R_0,
W)$
because
they do not depend on\par
\noindent $(\underline{H}_{2}(\mathcal{I'}^{c,*}),\underline{H}_{\log}(\mathcal{I'}^{c,*}) )_{\mathcal{I'}(\neq \emptyset) \subsetneq \{1,\ldots, T\}}$.
\end{definition}

For the reason why we employ the limiting forms in (\ref{eq:logexp}) and (\ref{eq:logexp2}), see Remark \ref{R1-2-1}.
Note that we do not consider here the universality for source while 
K\"orner and Sgarro \cite{korner80} show the universality for source as well as that for channel, as reviewed in Theorem \ref{lem-11-25-3-b} of this paper.
In order to guarantee the secrecy for $\mathcal{S}_{\mathcal{I},n}$, we need sufficient randomness of $\mathcal{S}_{\mathcal{I}^c,n}$.
That is, the secrecy of $\mathcal{S}_{\mathcal{I},n}$ depends on 
$\underline{H}_{2}(\mathcal{I}^c)$ and $\underline{H}_{\log}(\mathcal{I}^c)$,
which depends on the source distribution.
Hence, it is impossible to show the universality for source in SMC.

We fix a distribution $Q_{VU}$ on $\mathcal{U}\times \mathcal{V}$ and
a channel $\Xi: \mathcal{V}\rightarrow \mathcal{X}$. 
Then, we present a universally attainable set of exponents and leaked information rate
in terms of $Q_{VU}$ and $\Xi$ in the following way.
Given 
a broadcast $W:\mathcal{X}\rightarrow \mathcal{Y}\times \mathcal{Z}$
and  
the real numbers
$(R_{\mathrm{p}}, R_{\mathrm{c}},R_0,
(\underline{H}_{2}(\mathcal{I'}^{c,*}),\underline{H}_{\log}(\mathcal{I'}^{c,*}) )_{\mathcal{I'}(\neq \emptyset) \subsetneq \{1,\ldots, T\}})$,
the tuple of exponents and information leakage rate
are given as
\begin{align}
E^{b} =& {E}^{b}
(R_{\mathrm{p}}, R_{\mathrm{c}},R_0,
W)
\nonumber \\
:= &
\tilde{E}^{b} (R_{\mathrm{p}}, R_{\mathrm{c}}, (W^Y\circ \Xi) \times Q_{VU}),
\Label{eq:univs1}\\
E^{e} =& {E}^{e}
(R_{\mathrm{p}}, R_{\mathrm{c}},R_0,
W)
\nonumber \\
:= &
\tilde{E}^{e} ( R_{\mathrm{c}}, (W^Z\circ \Xi) \times Q_{VU}), 
\Label{eq:univs2}\\
E_+^{\mathcal{I}}=& 
{E}_+^{\mathcal{I}}
(R_{\mathrm{p}}, R_{\mathrm{c}},R_0,
(\underline{H}_{2}(\mathcal{I'}^{c,*}),\underline{H}_{\log}(\mathcal{I'}^{c,*}) )_{\mathcal{I'}(\neq \emptyset) \subsetneq \{1,\ldots, T\}},W)
\nonumber \\
:=& \tilde{E}^{l}(\underline{H}_{2}(\mathcal{I}^{c,*})-R_{\mathrm{c}}+R_0 ,(W^Z\circ \Xi ) \times Q_{VU} ) ,
\Label{eq:univs3} \\
E_-^{\mathcal{I}}=& 
{E}_-^{\mathcal{I}}
(R_{\mathrm{p}}, R_{\mathrm{c}},R_0,
(\underline{H}_{2}(\mathcal{I'}^{c,*}),\underline{H}_{\log}(\mathcal{I'}^{c,*}) )_{\mathcal{I'}(\neq \emptyset) \subsetneq \{1,\ldots, T\}},W)
\nonumber \\
:=& 
I(V;Z|U)[(W^Z\circ \Xi) \times Q_{VU}]
-\underline{H}_{\log}(\mathcal{I}^{c,*}) +R_{\mathrm{c}}-R_0
\Label{eq:univs4}
\end{align}
for a non-empty proper subset $\mathcal{I}(\neq \emptyset) \subsetneq \{1,\ldots, T\}$,
where $\tilde{E}^{b}$, $\tilde{E}^{e}$, $\tilde{E}^{E_0}$,
and $\tilde{E}^l$ are given by
(\ref{1-31-1}), (\ref{1-31-2}), (\ref{1-31-3}), and \eqref{1-31-1-k}, respectively.

Hence, our quadruple of exponents and information leakage rate
depends on $Q_{VU}$ and $\Xi$.

\begin{theorem}[Extension of {\cite[Theorem 1, part (a)]{korner80}}]
\Label{univ:thm}
Eqs.\ (\ref{eq:univs1})--(\ref{eq:univs4}) are 
universally attainable rates of exponents and information leakage rate
in the sense of Definition \ref{def:univexp}.
\end{theorem}


\begin{proof}
In the proof, since we treat the channel $W^{Z} \circ \Xi:{\cal V}\to {\cal Z}$,
we abbreviate it as $\overline{W}^{Z}$.
First, we give the outline of our proof.
We shall modify the constant composition code used by K\"orner and Sgarro \cite{korner80}.
We do not evaluate the decoding error probability,
because that of our code is not larger than that given in \cite{korner80}.
Observe that our exponents in
Eqs.\ (\ref{eq:univs1}) and (\ref{eq:univs2}) are the same
as \cite{korner80} with the channel $\overline{W}^{Z}=W^Z \circ \Xi$.
We shall evaluate only the mutual information.
For this purpose, we prepare general notations and properties of type and conditional type in Step (1).
Next, in Steps (2) and (3),
we prepare several notations and properties 
of type and conditional type that are specific to our proof.
In Step (4), we apply the random coding and evaluate the leaked information when the channel is given by the conditional types.
Then, we choose a code whose leaked information is evaluated for all conditional types
and whose error is evaluated for all discrete memoryless channels.
In Step (5), we evaluate the leaked information under the above chosen code for all discrete memoryless channels.

\noindent{\it Step (1): Preparation of general notations and properties of type and conditional type:\quad}

For the following construction of our code,
we prepare general notations for types.
These notations will be used also in the next section.
For a given type $Q_U$ of length $n$ on a set $\mathcal{U}$,
we define the set $T_n(Q_{U})$ as
\begin{align*}
T_n(Q_{U}) :=&
\{u^n \in \mathcal{U}^n |
\hbox{the type of }u^n
\hbox{ is }Q_U \} .
\end{align*}
Hence, for a given type $Q_{VU}$ of length $n$ on a set $\mathcal{V}\times \mathcal{U}$,
the set $T_n(Q_{VU})$ is written as
\begin{align*}
T_n(Q_{VU}) =&
\{(u^n,v^n) \in \mathcal{V}^n\times \mathcal{U}^n |
\hbox{the type of }(v^n,u^n)
\hbox{ is }Q_{VU} \} .
\end{align*}
The marginal distribution $Q_U$ over $\mathcal{U}$ of the type $Q_{VU}$ of length $n$ 
on the set $\mathcal{V}\times \mathcal{U}$
is a type of length $n$ on the set $\mathcal{U}$.
Given a type $Q_V$ of length $n$ on the set $\mathcal{U}$, 
we define the set of conditional types on the set $\mathcal{V}$
with respect to $Q_V$ as
\begin{align*}
&{\cal T}_{n,\mathcal{V}}(Q_{U}) \\
:=&
\{ \hbox{probability transition matrix } W \hbox{ from } \mathcal{U} \hbox{ to } \mathcal{V} \\
& \hspace{3ex} |W \times Q_{U} \hbox{ is a type of length } n
\hbox{ on a set } \mathcal{V}\times \mathcal{U} \} .
\end{align*}
The cardinality 
$|{\cal T}_{n,\mathcal{V}}(Q_{U}) |$ is upper bounded as \cite{csiszarbook}
\begin{align}
|{\cal T}_{n,\mathcal{V}}(Q_{U}) |
\le
(n+1)^{|\mathcal{V}\times \mathcal{U}|}. \Label{12-24-1}
\end{align}
In particular, given a type $Q_{VU}$ of length $n$ 
on the set $\mathcal{V}\times \mathcal{U}$,
we define the conditional type $Q_{V|U}$ such that $Q_{VU}= Q_{V|U}\times Q_U$.
We also define the set $T_n(Q_{V|U})_{U^n=u^n}$ as
\begin{align*}
T_n(Q_{V|U})_{U^n=u^n} :=&
\{v^n \in \mathcal{V}^n |
\hbox{the type of }(v^n,u^n)
\hbox{ is }Q_{VU} \} .
\end{align*}

We denote the uniform distribution $P_{\mix, T_n(Q_{U})}$ on 
$T_n(Q_{U})$ by $\Upsilon_n(Q_{U})$.
Then, for a given type $Q_{VU}$ of length $n$ on a set 
$\mathcal{V}\times \mathcal{U}$,
$\Upsilon_n(Q_{VU})$ represents the uniform distribution 
$P_{\mix, T_n(Q_{VU})} $
on $T_n(Q_{VU})$.
Further, for an arbitrary $W \in {\cal T}_{n,\mathcal{V}}(Q_{U})$,
$\Upsilon_n(W \times Q_{U})$ represents the uniform distribution on 
$T_n(W \times Q_{U})$.
Then, we define the probability transition matrix $\Upsilon_n(W)$ from $\mathcal{V}^n $ to $\mathcal{U}^n$
such that $\Upsilon_n(W)\times \Upsilon_n(Q_{U})
=\Upsilon_n(W \times Q_{U})$.

When $P_{V^n U^n}$ is a distribution over $\mathcal{V}^n\times \mathcal{U}^n$ and invariant under the permutation of the indices,
the distribution $P_{V^n U^n}$ can be written as
\begin{align}
P_{V^n U^n}= \sum_{Q_{VU}} \lambda_{P_{V^n U^n}} (Q_{VU})\Upsilon_n(Q_{VU})
\end{align}
with non-negative constants $\lambda (Q_{VU})$.
In particular, 
the independent and identical distribution 
$P_{V}^n$ of $P_V$
can be written as
\begin{align}
P_{V}^n= \sum_{Q_{V}} \lambda_{P_{V}^n} (Q_{V})\Upsilon_n(Q_{V})
\end{align}
with
\begin{align}
\lambda_{P_{V}^n} (Q_{V}) = P_{V}^n( T_n(Q_V))
\le e^{-nD(Q_V\|P_V)}.\Label{12-20-3}
\end{align}
When the marginal distribution over $\mathcal{U}^n$ 
of $P_{V^n U^n}$ 
can be written as $P_{\mix, T_n(Q_{U})}=\Upsilon_n(Q_{U})$ with a type $Q_U$ on the set $\mathcal{U}$,
we have
\begin{align}
P_{V^n U^n}
&= \sum_{Q_{V|U} \in {\cal T}_{n,\mathcal{V}}(Q_{U}) } 
\lambda_{P_{V^n U^n}} (Q_{V|U}\times Q_U)\Upsilon_n(Q_{V|U}\times Q_U) \nonumber \\
&= \sum_{Q_{V|U} \in {\cal T}_{n,\mathcal{V}}(Q_{U}) } 
\lambda_{P_{V^n U^n}} (Q_{V|U}\times Q_U) (\Upsilon_n(Q_{V|U})\times \Upsilon_n(Q_U)) \nonumber \\
&= \Bigl(\sum_{Q_{V|U} \in {\cal T}_{n,\mathcal{V}}(Q_{U}) } 
\lambda_{P_{V^n U^n}} (Q_{V|U}\times Q_U)\Upsilon_n(Q_{V|U})
\Bigr)\times \Upsilon_n(Q_U) .\Label{12-20-1}
\end{align}
We define the channel $P_{V^n| U^n}$ by 
$P_{V^n U^n}=P_{V^n| U^n}\times \Upsilon_n(Q_{U})$
and 
the real number $\lambda_{P_{V^n| U^n}} (Q_{V|U}):=\lambda_{P_{V^n U^n}} (Q_{V|U}\times Q_U)$
for $Q_{V|U} \in {\cal T}_{n,\mathcal{V}}(Q_{U})$.
Then,
we obtain
\begin{align}
P_{V^n |U^n}
= \sum_{Q_{V|U} \in {\cal T}_{n,\mathcal{V}}(Q_{U}) } 
\lambda_{P_{V^n| U^n}} (Q_{V|U})\Upsilon_n(Q_{V|U}). 
\Label{12-20-2}
\end{align}
Now, we consider the $n$-fold discrete memoryless channel $P_{V|U}^n$.
For a given type $Q_U$ on the set $\mathcal{U}$,
we apply the relation (\ref{12-20-1}) to 
the joint distribution $P_{V|U}^n|_{T_n(Q_{U})}\times \Upsilon_n(Q_U) $.
Then, (\ref{12-20-2}) implies that 
\begin{align}
P_{V|U}^n|_{T_n(Q_{U})}
= \sum_{Q_{V|U} \in {\cal T}_{n,\mathcal{V}}(Q_{U}) } 
\lambda_{P_{V|U}^n} (Q_{V|U}) \Upsilon_n (Q_{V|U}).\Label{12-20-6}
\end{align}
Choosing $u^n \in T_n(Q_{U})$, we have 
\begin{align}
\Upsilon_n (Q_{V|U}')(T_n(Q_{V|U})_{U^n=u^n} |U^n=u^n)
= 
\left\{
\begin{array}{ll}
1 & \hbox{ if } Q_{V|U}'= Q_{V|U}\\
0 & \hbox{ otherwise. } 
\end{array}
\right. \Label{12-20-7}
\end{align}
Combining (\ref{12-20-6}) and (\ref{12-20-7}), we obtain
\begin{align}
& \lambda_{P_{V|U}^n} (Q_{V|U}) \nonumber \\
=&
P_{V|U}^n|_{T_n(Q_{U})}(T_n(Q_{V|U})_{U^n=u^n}  |U^n=u^n) \nonumber \\
=&
\prod_{u \in \mathcal{U}} (P_{V|U=u})^{n Q_U(u)}(T_{n_u}(Q_{V|U=u}) )\nonumber \\
\le &
e^{-\sum_{u \in \mathcal{U}} n Q_U(u) D(Q_{V|U=u}\| P_{V|U=u})} \Label{12-20-5} \\
=&
e^{-n D(Q_{V|U}\| P_{V|U}| Q_U)},\Label{12-20-4}
\end{align}
where (\ref{12-20-5}) follows from (\ref{12-20-3}).

\noindent{\it Step (2): Preparation of notations and properties of conditional types based on a joint type on $\mathcal{U}\times \mathcal{V}$:\quad}

In this step, we prepare several important properties 
based on a type of length $n$ on the set 
$\mathcal{U} \times \mathcal{V} \times \mathcal{Z}$.
Now, we focus on a conditional type $W^Z \in 
{\cal T}_{n,\mathcal{Z}}(Q_{VU}) $, which gives
a type $W^Z \times Q_{VU}$ of length $n$ on the set 
$\mathcal{U} \times \mathcal{V} \times \mathcal{Z}$.
Note that
in order to make a type of length $n$ on the set 
$\mathcal{U} \times \mathcal{V} \times \mathcal{Z}$,
we need to choose $W^Z$ not from 
${\cal T}_{n,\mathcal{Z}}(Q_{V}) $
but from ${\cal T}_{n,\mathcal{Z}}(Q_{VU}) $.
Now, we treat 
the channel $\overline{W}^{Z}$ 
as a channel from $\mathcal{V}\times \mathcal{U} $
to $\mathcal{Z}$
while the output distribution of the channel $\overline{W}^{Z}$ 
does not depend on the choice of $u \in \mathcal{U}$. 
In our code $\varphi_{a,n}$,
the random variable $V^n U^n$ takes values in 
the subset $T_n(Q_{VU})$.
Hence, 
it is sufficient to treat the channel whose input alphabet is the subset $T_n(Q_{VU})$ of $\mathcal{V}^n\times\mathcal{U}^n$.
Based on (\ref{12-20-6}),
we make a convex decomposition
\begin{align}
\overline{W}^{Z,n}|_{T_n(Q_{VU})}
=&
\sum_{{W}^Z \in \mathcal{T}_{n,\mathcal{Z}}(Q_{VU})}
\lambda_{n,T} ({W}^Z) 
\Upsilon_n({W}^Z),
\Label{sc1-12-1-b}
\end{align}
with non-negative constants $\lambda_{n,T} ({W}^Z)$.
Then, due to (\ref{12-20-4}), we have 
\begin{align}
\lambda_{n,T} ({W}^Z)&\le e^{-nD({W}^Z\|\overline{W}^Z| Q_{VU})} .\Label{11-29-2}
\end{align}


For an arbitrary code $\varphi_{a,n}$,
the joint convexity of the conditional relative entropy yields that
\begin{align}
& I(S_{\mathcal{I},n};Z^n|S_{0,n})[\overline{W}^{Z,n},\varphi_{a,n},P_{S_{{\cal T+1},n}}]
\nonumber \\
\le &
\sum_{{W}^Z \in \mathcal{T}_{n,\mathcal{Z}}(Q_{VU})}
\lambda_{n,T} ({W}^Z) 
I(S_{\mathcal{I},n};Z^n|S_{0,n})
[\Upsilon_n(W^Z),\varphi_{a,n},P_{S_{{\cal T+1},n}}]
.\Label{Haya-21}
\end{align}

Next, in order to treat each channel $\Upsilon_n({W}^Z)$, 
we fix a conditional type $W^Z \in 
{\cal T}_{n,\mathcal{Z}}(Q_{VU}) $
and study the properties of the channel $\Upsilon_n({W}^Z)$.
Under the joint type $Q_{ZVU}:=W^Z \times Q_{VU}$,
we define the numbers
\begin{align*}
N(U)&:=|T_n(Q_{U})|, \quad 
N(UZ):=|T_n( (W^Z \circ Q_{V|U} )\times Q_U )|,\\
N(VU)&:=|T_n(Q_{VU})|, \quad
N(VUZ):=|T_n(W^Z \times Q_{VU})|,
\end{align*}
and
\begin{align*}
N(Z|U)&:=N(UZ)/ N(U) , \quad
N(V|UZ):=N(VUZ)/ N(UZ) , \\
N(V|U)&:=N(VU)/N(U), \quad
N(Z|VU):= N(VUZ)/N(VU).
\end{align*}
Then, due to \cite{csiszarbook}, we have
\begin{align}
|{\cal T}_{n,\mathcal{Z}}(Q_{U})|^{-1} e^{nH(Z|U)[W^Z\times Q_{VU}]} 
&\le N(Z|U) \le e^{nH(Z|U)[W^Z\times Q_{VU}]} \Label{12-3-4}\\
|{\cal T}_{n,\mathcal{Z}}(Q_{VU})|^{-1} e^{nH(Z|VU)[W^Z\times Q_{VU}]} 
&\le N(Z|VU) \le e^{nH(Z|VU)[W^Z\times Q_{VU}]} .
\Label{12-3-5}
\end{align}
Then, we obtain the following lemma.

\begin{lemma}
Any conditional type $W^Z \in {\cal T}_{n,\mathcal{Z}}(Q_{VU})$
satisfies 
\begin{align}
&E_0(\rho| \Upsilon_n(W^Z), P_{V^n|U^n, \mix, T_n(Q_{VU})},P_{\mix, T_n(Q_{U})}) 
\nonumber \\
=&\rho \log \frac{N(Z|U)}{ N(Z|VU) } \Label{12-3-1}\\
=&\rho I(V;Z|U)[\Upsilon_n(W^Z)\times P_{\mix, T_n(Q_{VU})}] \Label{12-3-2}\\
\le & n \rho I(V;Z|U)[W^Z\times Q_{VU}]+ \rho \log |{\cal T}_{n,\mathcal{Z}}(Q_{VU})|\Label{12-3-3}
\end{align}
for any $\rho \in (0,1)$.
Here $P_{V^n|U^n, \mix, T_n(Q_{VU})}$ is defined as a special case of Eq.(\ref{12-19-10}).
\end{lemma}

\begin{proof}
Under the joint type $Q_{ZVU}:=W^Z\times Q_{VU}$,
since $\Upsilon_n(W^Z)= P_{Z^n|V^n U^n, \mix, T_n(Q_{ZVU})}$,
we obtain
\begin{align*}
& e^{E_0(\rho| \Upsilon_n(W^Z), P_{V^n|U^n, \mix, T_n(Q_{VU})},P_{\mix, T_n(Q_{U})})} \\
=& e^{E_0(\rho| P_{Z^n|V^n U^n, \mix, T_n(Q_{ZVU})}, P_{V^n|U^n, \mix, T_n(Q_{VU})},P_{\mix, T_n(Q_{U})})} \\
=&
\sum_{u^n \in T_n(Q_{U})}
\frac{1}{N(U)}
\sum_{z^n \in T_n(Q_{Z|U})_{U^n=u^n})}
\Biggl(
\\
& \sum_{v \in T_n(Q_{V|ZU})_{Z^n U^n=(z^n,u^n)})}
P_{V^n|U^n, \mix, T_n(Q_{VU})}(v^n|u^n)
\\
&\hspace{12ex} \cdot \big(
P_{Z^n|V^n U^n, \mix, T_n(Q_{ZVU})} (z^n|v^n,u^n)
\big)^{\frac{1}{1-\rho}}\Biggr)^{1-\rho} \\
=&
\sum_{u^n \in T_n(Q_{U})}
\frac{1}{N(U)}
\sum_{z^n \in T_n(Q_{Z|U})_{U^n=u^n})}
\biggl( \\
&\sum_{v \in T_n(Q_{V|ZU})_{Z^n U^n=(z^n,u^n)})}
\frac{1}{N(V|U)}
(\frac{1}{N(Z|VU)})^{\frac{1}{1-\rho}}\biggr)^{1-\rho} \\
=&
N(U)
\frac{1}{N(U)}
N(Z|U)
( N(V|UZ)
\frac{1}{N(V|U)}
(\frac{1}{N(Z|VU)})^{\frac{1}{1-\rho}})^{1-\rho} \\
=&
\frac{N(ZU)^{\rho}  N(VU)^{\rho}}{ N(VUZ)^{\rho} N(U)^{\rho} } 
= \frac{N(Z|U)^{\rho}}{ N(Z|VU)^{\rho} } ,
\end{align*}
which implies (\ref{12-3-1}).
Since
\begin{align*}
& \log N(Z|U)- \log N(Z|VU) \\
=& H(Z|U)[\Upsilon_n(W^Z)\times P_{\mix, T_n(Q_{VU})}] \\
&-H(Z|VU)[\Upsilon_n(W^Z)\times P_{\mix, T_n(Q_{VU})}] \\
= &I(V;Z|U)[\Upsilon_n(W^Z)\times P_{\mix, T_n(Q_{VU})}] ,
\end{align*}
we obtain (\ref{12-3-2}).
Combining (\ref{12-3-4}) and (\ref{12-3-5}),
we obtain (\ref{12-3-3}).
\end{proof}

\noindent{\it Step (3): Preparation of notations and properties concerning conditional types based on a type on $\mathcal{V}$:\quad}

In this step, we focus only on a convex decomposition different from (\ref{sc1-12-1-b}).
For a given type $Q_{V}$ of length $n$ on a set $\mathcal{V}$,
we focus on the set
\begin{align*}
{\cal W}_{n,\mathcal{Z}}(Q_{V}) :=&
\{ \Upsilon_n(W^Z) |W^Z \in {\cal T}_{n,\mathcal{Z}}(Q_{V})\}.
\end{align*}
In our code $\varphi_{a,n}$,
the random variable $V^n$ takes values in 
the subset $T_n(Q_V)$.
Hence, 
if we focus on the set $\mathcal{V}^n$ as inputs, 
it is sufficient to treat the channel whose input alphabet is the subset $T_n(Q_V)$ of ${\cal V}^n$.
Then, 
due to (\ref{12-20-6}), we have another type of convex combination:
\begin{align}
\overline{W}^{Z,n}|_{T_n(Q_V)}
=\sum_{\Theta_n\in {\cal W}_{n,\mathcal{Z}}(Q_V)}
\lambda_{n,W} (\Theta_n) \Theta_n,
\Label{1-12-1}
\end{align}
where
$\lambda_{n,W} (\Theta_n)$ is a non-negative constant.
Then, for an arbitrary code $\varphi_{a,n}$,
the joint convexity of the conditional relative entropy yields that
\begin{align}
& I(S_{\mathcal{I},n};Z^n|S_{0,n})[\overline{W}^{Z,n},\varphi_{a,n},P_{S_{{\cal T+1},n}}]
\nonumber \\
\le &
\sum_{\Theta_n\in {\cal W}_{n,\mathcal{Z}}(Q_V)}
\lambda_{n,W} (\Theta_n) 
I(S_{\mathcal{I},n};Z^n|S_{0,n})[\Theta_n,\varphi_{a,n},P_{S_{{\cal T+1},n}}]
.\Label{Haya-21-b}
\end{align}

Next, we introduce the quantity
\begin{align}
&\varepsilon_{n,\rho,\mathcal{I}}(W^{Z^n}, Q_{V^n,U^n})\nonumber \\
:=&
\exp \biggl(n \rho (R_{\mathrm{c}}-R_{0}) 
-\rho H_{1+\rho}(S_{\mathcal{I}^{c,*},n}|S_{\mathcal{I},n},S_{0,n})
\nonumber \\
&\hspace{20ex}+E_0 (\rho|W^{Z^n},Q_{V^n|U^n},Q_{U^n}) \biggr)
\end{align}
for any channel $W^{Z^n}$ from ${\cal V}^n$ to ${\cal Z}^n$
and any distribution $Q_{V^n U^n}$ on ${\cal V}^n \times {\cal U}^n$.

Then, we have the following lemma.
\begin{lemma}\Label{l-12-26-2}
Any joint type $Q_{VU}$ of length $n$ on a set $\mathcal{V}\times \mathcal{U}$ 
and any channel $\Theta_n\in {\cal W}_{n,\mathcal{Z}}(Q_V)$
satisfy
\begin{align}
& \exp (E_0 (\rho| \overline{W}^{Z,n},
P_{V^n|U^n, \mix, T_n(Q_{VU})},P_{\mix, T_n(Q_{U})})) \nonumber \\
\le &
(n+1)^{|\mathcal{U}|^2|\mathcal{V}|}
 \exp (E_0 (\rho| \overline{W}^{Z,n},Q_{V|U}^n,Q_{U}^n)) ,
\Label{Haya-30} \\
& \lambda_{n,W} (\Theta_n) 
\varepsilon_{n,\rho}(\Theta_n, P_{\mix, T_n(Q_{VU})}) \nonumber \\
\le & 
(n+1)^{|\mathcal{U}|^2|\mathcal{V}|}
\varepsilon_{n,\rho,\mathcal{I}}(\overline{W}^{Z,n},Q_{V,U}) .
\Label{Haya-22-c} 
\end{align}
We have
\begin{align}
& \limsup_{n \to \infty}
\frac{1}{n\rho_n}
\log \varepsilon_{n,\rho_n,\mathcal{I}}(\overline{W}^{Z,n},Q_{V,U}^n)\nonumber \\
\le &
I(V;Z|U)[\overline{W}^Z \times Q_{VU}]
-\underline{H}_{\log}(\mathcal{I}^{c,*}) +R_{\mathrm{c}}-R_0
=E_-^{\mathcal{I}}.
\Label{11-25-5}
\end{align}
with $\rho_n=\frac{\delta \log n}{n}$ for any $\delta>0$.
Further, 
when $S_{\mathcal{I}^{c,*},n}$ is the uniform random number and independent of 
$S_{\mathcal{I},n}$ and $S_{0,n}$, 
we have
\begin{align}
\varepsilon_{n,\rho,\mathcal{I}}(\overline{W}^{Z,n},Q_{V,U}^n)
=\varepsilon_{1,\rho,\mathcal{I}}(\overline{W}^{Z},Q_{V,U})^n
\Label{12-22-1}
\end{align}
and
\begin{align}
\lim_{\rho \to 0}
\frac{[\log \varepsilon_{1,\rho,\mathcal{I}}(\overline{W}^{Z},Q_{V,U})]_+ }{\rho} 
&=
I(V;Z|U) -R_{\mathrm{p}}+\sum_{i \in \mathcal{I}} R_i \Label{11-25-5b}.
\end{align}
The convergence in (\ref{11-25-5b}) is uniform.
\end{lemma}

\begin{proof}
First, we show (\ref{Haya-30}).
For arbitrary $u\in \mathcal{U}$ and $v\in \mathcal{V}$, 
the distribution 
$P_{\mix, T_n(Q_{VU})}$
satisfies 
\begin{equation}
P_{V^n|U^n, \mix, T_n(Q_{VU})}(v|u)
\leq (n+1)^{|\mathcal{U}\times \mathcal{V}|} 
Q^n_{V|U}(v|u) \Label{eq101}
\end{equation}
by \cite[Lemma 2.5, Chapter 1]{csiszarbook}, and
\begin{equation}
P_{\mix, T_n(Q_{U})}(u) \leq (n+1)^{|\mathcal{U}|} Q^n_U(u), \Label{eq102}
\end{equation}
by \cite[Lemma 2.3, Chapter 1]{csiszarbook}.
Then, due to the relation (\ref{eq101}), and (\ref{eq102}), 
Lemma \ref{3-24-4L}
with $C_1= (n+1)^{|\mathcal{U}|^2|\mathcal{V}|}$
yields the relation (\ref{Haya-30}).

Next, we show (\ref{Haya-22-c}).
We can also show that 
\begin{align}
& \lambda_{n,W} (\Theta_n) 
e^{E_0(\rho|\Theta_n, P_{V^n|U^n, \mix, T_n(Q_{VU})}, P_{\mix, T_n(Q_{U})} )} \nonumber\\
=&\sum_{u}P_{\mix, T_n(Q_{U})}(u) \sum_{z}
\Biggl(\sum_{v}P_{V^n|U^n, \mix, T_n(Q_{VU})}(v|u)
\nonumber \\
& \hspace{26ex} \cdot \biggl(
\lambda_{n,W} (\Theta_n) 
\Theta_n(z|v)
\biggr)^{\frac{1}{1-\rho}}\Biggr)^{1-\rho} \nonumber\\
\le &
\sum_{u}P_{\mix, T_n(Q_{U})}(u) \sum_{z}
\Biggl(\sum_{v}P_{V^n|U^n, \mix, T_n(Q_{VU})}(v|u) \nonumber\\
& \hspace{18ex} \cdot \biggl(\sum_{\Theta_n'\in {\cal W}_{n,\mathcal{Z}}(Q_V)}
\lambda_{n,W} (\Theta_n') 
\Theta_n'(z|v) \biggr)^{\frac{1}{1-\rho}} \Biggr)^{1-\rho} \nonumber\\
=& e^{E_0(\rho|\overline{W}^{Z,n},
P_{V^n|U^n, \mix, T_n(Q_{VU})},P_{\mix, T_n(Q_{U})})}.
\Label{Haya-22}
\end{align}
Combining (\ref{Haya-30}) and (\ref{Haya-22}), 
we obtain
\begin{align}
& (n+1)^{|\mathcal{U}|^2|\mathcal{V}|}
e^{E_0(\rho|\overline{W}^{Z,n},Q_{V|U}^n,Q_{U}^n )}
\nonumber \\
\ge &
\lambda_{n,W} (\Theta_n) 
e^{E_0(\rho|\Theta_n, P_{V^n|U^n, \mix, T_n(Q_{VU})}, P_{\mix, T_n(Q_{U})} )}.
\Label{Haya-22-b}
\end{align}
Due to the definition of $\varepsilon_{n,\rho}(W^{Z^n}, Q_{V^n,U^n})$,
the relation (\ref{Haya-22-b}) is equivalent with the relation (\ref{Haya-22-c}).

By using \eqref{3-24-1eq}, the relation (\ref{11-25-5}) can be shown as follows.
\begin{align}
& \limsup_{n \to \infty}
\frac{1}{n\rho_n}
\log \varepsilon_{n,\rho_n,\mathcal{I}}
(\overline{W}^{Z,n},Q_{V,U}^n)\nonumber \\
= 
& \limsup_{n \to \infty}
\biggl[
(R_{\mathrm{c}}-R_{0}) 
-\frac{1}{n} H_{1+\frac{\delta \log n}{n}}(S_{\mathcal{I}^{c,*},n}|S_{\mathcal{I},n},S_{0,n})
\nonumber\\
&\hspace{25ex} +\frac{1}{\rho_n}E_0 (\rho_n|\overline{W}^{Z},Q_{V|U},Q_{U})  \biggr] \nonumber\\
\le &
R_{\mathrm{c}}-R_{0}
-\underline{H}_{\log}(\mathcal{I}^{c,*}) 
+I(V;Z|U) 
=E_-^{\mathcal{I}}.\nonumber
\end{align}
The relations (\ref{12-22-1}) and (\ref{11-25-5b})
are trivial.
\end{proof}

\noindent{\it Step (4): Evaluation of the leaked information when the channel is given by the uniform distribution on a fixed conditional type:\quad}

Recall the fixed code $\varphi_{\mathrm{p},n}$ for BCD given in Theorem \ref{lem-11-25-3-b}.
The message sets of the code $\varphi_{\mathrm{p},n}$ are
$\mathcal{S}_{0,n}\times \mathcal{B}_{1,n}$ and $\mathcal{B}_{2,n}$
with 
$|\mathcal{B}_{1,n}|=e^{n(R_{\mathrm{c}}-R_0)}$ and
$|\mathcal{B}_{2,n}|=e^{n R_{\mathrm{p}}}$.
We attach 
the other random coding $\Lambda_{F,G,n}$ for message $S_{1,n},\ldots,S_{T,n}$ given as Second Step of Code Ensemble \ref{con1} in Subsection \ref{s5-2}
to the code $\varphi_{\mathrm{p},n}$.
That is, the encoder is given by $\Phi_{a,n}=(\varphi_{\mathrm{p},n},\Lambda_{F,G,n})$.
In the following,
Bob's decoder $\Phi_{b,n}$
and
Eve's decoder $\Phi_{e,n}$
are given as the maximum mutual information decoder.
We treat the ensemble of 
codes $\Phi_n:=(\Phi_{a,n},\Phi_{b,n},\Phi_{e,n})$.

First, related to the decomposition (\ref{sc1-12-1-b}), we focus on a fixed arbitrary element $W^Z \in {\cal T}_{n,\mathcal{Z}}(Q_{VU})$,
We recall the discussion in Subsection \ref{s5-4}.
As is mentioned in Remark \ref{s5-5},
the discussion in Section \ref{s5} can be applied the channel $W^Z$, whose 
output distribution depends on the element of $\mathcal{U}$ as well as the element of $\mathcal{V}$.
Then, we apply Lemma \ref{l12-3-2} to the case when 
$P_{Z|V}=W^Z$,
$\mathcal{G}$ is the $n$-th permutation group,
$(\mathcal{U}\times \mathcal{V})_o$ is $T_n(Q_{UV})$, and
$P_{V|U}$ is $\Upsilon_n(W^Z)$.
Note that the $n$-th permutation group acts on $T_n(Q_{UV})$ transitively.
We obtain
\begin{align*}
&e^{\psi ( \rho| P_{Z^n|B_1,B_2,S_0=s_0}, P_{\mix, \mathcal{B}_1,\mathcal{B}_2} )} \nonumber \\
=&e^{\psi ( \rho| \Upsilon_n(W^Z), P_{V^n|U^n, \mix, \im \varphi_{\mathrm{p}}}, P_{U, \mix, \im \varphi_{\mathrm{p}}} ) }\\
\le &
e^{n\rho (R_{\mathrm{c}}-R_0)
+E_0(\rho| \Upsilon_n(W^Z), P_{V^n|U^n, \mix, T_n(Q_{VU})},P_{\mix, T_n(Q_{U})})}.
\end{align*}
Combining Lemma \ref{lem2-1} and the above inequality, we obtain
\begin{align}
& \rE_{\Phi_{a,n}}
\exp (\rho I(S_{\mathcal{I},n};Z^n|S_{0,n})[\Upsilon_n(W^Z),\Phi_{a,n},P_{S_{{\cal T+1},n}}] ) 
\nonumber\\
\leq &
1+ 
e^{n \rho (R_{\mathrm{c}}-R_{0}) -\rho H_{1+\rho}(S_{\mathcal{I}^{c,*},n}|S_{\mathcal{I},n},S_{0,n})}
e^{E_0(\rho| \Upsilon_n(W^Z), P_{V^n|U^n, \mix, T_n(Q_{VU})},P_{\mix, T_n(Q_{U})})} 
\Label{eq100-bc}.
\end{align}
Hence, we obtain the following relations.
In the following derivation, 
the first inequality follows from the convexity of $x \mapsto e^x$.
The third inequality follows from (\ref{12-3-3}).
\begin{align}
&\exp (\rho \rE_{\Phi_{a,n}} 
I(S_{\mathcal{I},n};Z^n|S_{0,n})[\Upsilon_n(W^Z),\Phi_{a,n} ,P_{S_{{\cal T+1},n}}] ) 
\nonumber \\
\leq &
\rE_{\Phi_{a,n}}
\exp (\rho I(S_{\mathcal{I},n};Z^n|S_{0,n})[\Upsilon_n(W^Z),\Phi_{a,n},P_{S_{{\cal T+1},n}}] ) 
\nonumber\\
\leq &
1+ 
e^{n \rho (R_{\mathrm{c}}-R_{0}) -\rho H_{1+\rho}(S_{\mathcal{I}^{c,*},n}|S_{\mathcal{I},n},S_{0,n})}
e^{E_0(\rho| \Upsilon_n(W^Z), P_{V^n|U^n, \mix, T_n(Q_{VU})},P_{\mix, T_n(Q_{U})})} \nonumber \\
\leq &
1+ 
|{\cal T}_{n,\mathcal{Z}}(Q_{VU})|^{\rho}
e^{n \rho (R_{\mathrm{c}}-R_{0})- \rho H_{1+\rho}(S_{\mathcal{I}^{c,*},n}|S_{\mathcal{I},n},S_{0,n})}
e^{n \rho I(V;Z|U)[W^Z\times Q_{VU}]}   
\nonumber 
\end{align}
for any $\rho \in (0,1)$.
Taking the limit $\rho \to 1-0$, we have
\begin{align}
&\exp ( \rE_{\Phi_{a,n}} 
I(S_{\mathcal{I},n};Z^n|S_{0,n})[\Upsilon_n(W^Z),\Phi_{a,n} ,P_{S_{{\cal T+1},n}}] ) 
\nonumber \\
\leq &
1+ 
|{\cal T}_{n,\mathcal{Z}}(Q_{VU})|
e^{n (R_{\mathrm{c}}-R_{0})- H_{2}(S_{\mathcal{I}^{c,*},n}|S_{\mathcal{I},n},S_{0,n})}
e^{n I(V;Z|U)[W^Z\times Q_{VU}]}  . 
\Label{eq100-b}
\end{align}
Since $\log (1+x) \le x$,
taking the logarithm in (\ref{eq100-b}), we have
\begin{align}
& 
\rE_{\Phi_{a,n}}
I(S_{\mathcal{I},n};Z^n|S_{0,n})[\Upsilon_n(W^Z),\Phi_{a,n},P_{S_{{\cal T+1},n}}]
\nonumber \\
\leq &
 \log (1+ 
|{\cal T}_{n,\mathcal{Z}}(Q_{VU})|
e^{n (R_{\mathrm{c}}-R_{0}) -H_{2}(S_{\mathcal{I}^{c,*},n}|S_{\mathcal{I},n},S_{0,n})}
e^{n I(V;Z|U)[W^Z\times Q_{VU}]}   ) \nonumber \\
\le &
|{\cal T}_{n,\mathcal{Z}}(Q_{VU})|
e^{n (R_{\mathrm{c}}-R_{0}) - H_{2}(S_{\mathcal{I}^{c,*},n}|S_{\mathcal{I},n},S_{0,n})}
e^{n I(V;Z|U)[W^Z\times Q_{VU}]}   
\nonumber.
\end{align}
Since $\log | \mathcal{Z}^n|=n \log |\mathcal{Z}| \le  |{\cal T}_{n,\mathcal{Z}}(Q_{VU})|$,
we have
\begin{align}
& 
\rE_{\Phi_{a,n}}
I(S_{\mathcal{I},n};Z^n|S_{0,n})[\Upsilon_n(W^Z),\Phi_{a,n},P_{S_{{\cal T+1},n}}]
\le |{\cal T}_{n,\mathcal{Z}}(Q_{VU})|.
\end{align}
Hence, 
\begin{align}
& 
\rE_{\Phi_{a,n}}
I(S_{\mathcal{I},n};Z^n|S_{0,n})[\Upsilon_n(W^Z),\Phi_{a,n},P_{S_{{\cal T+1},n}}]
\nonumber \\
\leq &
|{\cal T}_{n,\mathcal{Z}}(Q_{VU})|
e^{- [
H_{2}(S_{\mathcal{I}^{c,*},n}|S_{\mathcal{I},n},S_{0,n})
- n (R_{\mathrm{c}}-R_{0}+I(V;Z|U)[W^Z\times Q_{VU}]) ]_+}   
\Label{eq100-1}.
\end{align}

Next, related to the decomposition (\ref{1-12-1}), we focus on a fixed arbitrary $\Theta_n\in {\cal W}_{n,\mathcal{Z}}(Q_V)$.
Similar to (\ref{eq100-bc}),
Lemmas \ref{lem2-1} and \ref{l12-3-2} 
yield that
\begin{align}
&
\rE_{\Phi_{a,n}}
\exp (\rho I(S_{\mathcal{I},n};Z^n|S_{0,n})[\Theta_n,\Phi_{a,n},P_{S_{{\cal T+1},n}}] ) 
\nonumber\\
\leq &
1+ 
e^{n \rho (R_{\mathrm{c}}-R_{0}) -\rho H_{1+\rho}(S_{\mathcal{I}^{c,*},n}|S_{\mathcal{I},n},S_{0,n})}
e^{E_0(\rho| \Theta_n, P_{V^n|U^n, \mix, T_n(Q_{VU})},P_{\mix, T_n(Q_{U})})} \nonumber \\
= &
1+ 
\varepsilon_{n,\rho,\mathcal{I}}(\Theta_n, P_{\mix, T_n(Q_{VU})})
\Label{eq100}.
\end{align}

Observe that we have shown that
the averages over $\Phi_{a,n}$ of 
$\exp(
\rho I(S_{\mathcal{I},n};Z^n|S_{0,n})[\Upsilon_n(W^Z),\Phi_{a,n},P_{S_{{\cal T+1},n}}] )$
and
$I(S_{\mathcal{I},n};Z^n|S_{0,n})[\Theta_n,\Phi_{a,n},P_{S_{{\cal T+1},n}}]$
are smaller than 
(\ref{eq100-1}) and (\ref{eq100}) , respectively.

Choosing $p_1(n):= 2^{T}
(|{\cal T}_{n,\mathcal{Z}}(Q_{VU}) |+|{\cal W}_{n,\mathcal{Z}}(Q_V)|)
+1$,
thanks to the Markov inequality in the same as (\ref{12-28-3}) and (\ref{12-28-4}),
given a fixed $\rho\in (0,1)$,
we can see that there exists at least 
one code $\varphi_{n}$ such that
the relations
\begin{align}
& I(S_{\mathcal{I},n};Z^n|S_{0,n})[\Upsilon_n(W^Z),\varphi_{a,n},P_{S_{{\cal T+1},n}}]
\nonumber \\
\le & 
p_1(n)
\rE_{\Phi_{a,n}}
I(S_{\mathcal{I},n};Z^n|S_{0,n})[\Upsilon_n(W^Z),\Phi_{a,n},P_{S_{{\cal T+1},n}}] 
\nonumber \\
\le & 
p_1(n)
|{\cal T}_{n,\mathcal{Z}}(Q_{VU})|
e^{n (R_{\mathrm{c}}-R_{0}) - H_{2}(S_{\mathcal{I}^{c,*},n}|S_{\mathcal{I},n},S_{0,n})}
e^{nI(V;Z|U)[W^Z\times Q_{VU}]}   
 \Label{Haya-13} \\
& \exp (\rho I(S_{\mathcal{I},n};Z^n|S_{0,n})[\Theta_n,\varphi_{a,n},P_{S_{{\cal T+1},n}}] )
\nonumber \\
\le & p_1(n) 
\rE_{\Phi_{a,n}}
\exp (\rho I(S_{\mathcal{I},n};Z^n|S_{0,n})[\Theta_n, \Phi_{a,n},P_{S_{{\cal T+1},n}}] ) 
\nonumber\\
\le & p_1(n) (1+ \varepsilon_{n,\rho,\mathcal{I}}(\Theta_n, P_{\mix, T_n(Q_{VU})}) ) \Label{Haya-4}.
\end{align}
hold for any $W^Z \in {\cal T}_{n,\mathcal{Z}}(Q_{VU})$ and $\Theta_n\in {\cal W}_{n,\mathcal{Z}}(Q_V)$.

\noindent{\it Step (5): Evaluation of the leaked information when the channel is given by discrete memoryless channel:\quad}

Using (\ref{Haya-13}),
we obtain
\begin{align}
&
I(S_{\mathcal{I},n};Z^n|S_{0,n})[\overline{W}^{Z,n},\varphi_{a,n},P_{S_{{\cal T+1},n}}]
\nonumber \\
\le &
\sum_{W^Z \in {\cal T}_{n,\mathcal{Z}}(Q_{VU})}
\lambda_{n,T} (W^Z) 
I(S_{\mathcal{I},n};Z^n|S_{0,n})[\Upsilon_n(W^Z), \varphi_{a,n},P_{S_{{\cal T+1},n}}]
\Label{1-10-1}\\
\le &
\sum_{W^Z \in {\cal T}_{n,\mathcal{Z}}(Q_{VU})}
\Bigl[
\lambda_{n,T} (W^Z) 
 p_1(n)
|{\cal T}_{n,\mathcal{Z}}(Q_{VU})|
\nonumber \\
& \hspace{8ex}\cdot e^{
- [H_{2}(S_{\mathcal{I}^{c,*},n}|S_{\mathcal{I},n},S_{0,n}) 
-n (R_{\mathrm{c}}-R_{0}+ I(V;Z|U)[W^Z\times Q_{VU}]) ]_+}
\Bigr]
\Label{1-10-2} \\
\le &
\sum_{W^Z \in {\cal T}_{n,\mathcal{Z}}(Q_{VU})}
\Bigl[
p_1(n)
|{\cal T}_{n,\mathcal{Z}}(Q_{VU})|
\nonumber \\
& \hspace{2ex}\cdot e^{-n D(W^Z\|\overline{W}^Z|Q_{VU}) 
- [H_{2}(S_{\mathcal{I}^{c,*},n}|S_{\mathcal{I},n},S_{0,n}) 
-n (R_{\mathrm{c}}-R_{0}+ I(V;Z|U)[W^Z\times Q_{VU}]) ]_+
}   
\Bigr]
\Label{1-10-3} \\
\le &
\sum_{W^Z \in {\cal T}_{n,\mathcal{Z}}(Q_{VU})}
p_1(n)
|{\cal T}_{n,\mathcal{Z}}(Q_{VU})|
e^{- K_n(\overline{W}^Z,Q_{VU},R_{\mathrm{c}},R_{0}|S)}
\Label{1-10-3b} \\
= &
p_1(n)
|{\cal T}_{n,\mathcal{Z}}(Q_{VU})|^2
e^{- K_n(\overline{W}^Z,Q_{VU},R_{\mathrm{c}},R_{0}|S)}
\Label{Haya-23},
\end{align}
where $K_n(\overline{W}^Z,Q_{VU},R_{\mathrm{c}},R_{0}|S)
$ is defined as 
\begin{align*}
&K_n(\overline{W}^Z,Q_{VU},R_{\mathrm{c}},R_{0}|S)
\nonumber \\
:=&
\min_{W^Z} \Biggl[
n D(W^Z\|\overline{W}^Z|Q_{VU}) 
+\Bigl[H_{2}(S_{\mathcal{I}^{c,*},n}|S_{\mathcal{I},n},S_{0,n}) 
\nonumber \\
&\hspace{8ex} -n (R_{\mathrm{c}}-R_{0}+ I(V;Z|U)[W^Z\times Q_{VU}]) 
\Bigr]_+
\Biggr],
\end{align*}
and
(\ref{1-10-1}),
(\ref{1-10-2}),
and
(\ref{1-10-3}) 
follow from
(\ref{Haya-21}),
(\ref{Haya-13}),
and 
\eqref{11-29-2}, respectively.

Hence,
\begin{align}
& 
\liminf_{n \to \infty} 
\frac{-1}{n}
\log 
I(S_{\mathcal{I},n};Z^n|S_{0,n})[\overline{W}^{Z,n},\varphi_{a,n},P_{S_{{\cal T+1},n}}]
\nonumber \\
\ge &
\liminf_{n \to \infty} 
\frac{1}{n}
\min_{W^Z}
\biggl[
n D(W^Z\|\overline{W}^Z|Q_{VU}) 
+
\Bigl[H_{2}(S_{\mathcal{I}^{c,*},n}|S_{\mathcal{I},n},S_{0,n}) 
\nonumber \\
&\hspace{18ex}
-n (R_{\mathrm{c}}-R_{0}+ I(V;Z|U)[W^Z\times Q_{VU}]) 
\Bigr]_+ \biggr]
\nonumber \\
=&
\min_{W^Z}
\biggl[
D(W^Z\|\overline{W}^Z|Q_{VU}) 
\nonumber \\
&\hspace{9ex}+
\Bigl[
\underline{H}_2(\mathcal{I}^{c,*})
- R_{\mathrm{c}}+R_{0} - I(V;Z|U)[W^Z\times Q_{VU}]) 
\Bigr]_+ \biggr]\nonumber \\
= &
E_+^{\mathcal{I}}
\Label{Haya-03-b}
\end{align}

Next, defining 
\begin{align}
p_2(n)
:=p_1(n)(n+1)^{|\mathcal{U}|^2|\mathcal{V}|}
|{\cal W}_{n,\mathcal{Z}}(Q_V)|,\Label{3-24-19eq}
\end{align}
we obtain the following inequalities, in which,
the first, second, and third inequalities follow from
the convexity of function $x \mapsto \exp(x)$ and
(\ref{Haya-21-b}), (\ref{Haya-4}), and (\ref{Haya-22-c}), 
respectively.
The final equation follows from \eqref{3-24-19eq}.
\begin{align}
& \exp (\rho 
I(S_{\mathcal{I},n};Z^n|S_{0,n})[\overline{W}^{Z,n},\varphi_{a,n},P_{S_{{\cal T+1},n}}]
) \nonumber \\
\le &
\sum_{\Theta_n\in {\cal W}_{n,\mathcal{Z}}(Q_V)}
\lambda_{n,W} (\Theta_n) 
 \exp (\rho 
I(S_{\mathcal{I},n};Z^n|S_{0,n})[\overline{W}_n,\varphi_{a,n},P_{S_{{\cal T+1},n}}]
) 
\nonumber \\
\le &
\sum_{\Theta_n\in {\cal W}_{n,\mathcal{Z}}(Q_V)}
\lambda_{n,W} (\Theta_n) 
p_1(n)(1+  \varepsilon_{n,\rho,\mathcal{I}}(\Theta_n,P_{\mix, T_n(Q_{VU})}))
\nonumber \\
\le &
\sum_{\Theta_n\in {\cal W}_{n,\mathcal{Z}}(Q_V)}
p_1(n)
(n+1)^{|\mathcal{U}|^2|\mathcal{V}|}
(1+  \varepsilon_{n,\rho,\mathcal{I}}
(\overline{W}^{Z,n},Q_{V,U})) \nonumber \\
= &
p_1(n)|{\cal W}_{n,\mathcal{Z}}(Q_V)|
(n+1)^{|\mathcal{U}|^2|\mathcal{V}|}
(1+  \varepsilon_{n,\rho,\mathcal{I}}(\overline{W}^{Z,n},Q_{V,U})) \nonumber\\
= &
p_2(n)(1+  
\varepsilon_{n,\rho,\mathcal{I}}(\overline{W}^{Z,n},Q_{V,U})).
\Label{1-13-6}
\end{align}

Taking the logarithm,
we have
\begin{align}
& 
I(S_{\mathcal{I},n};Z^n|S_{0,n})[\overline{W}^{Z,n},\varphi_{a,n},P_{S_{{\cal T+1},n}}]
\nonumber \\
\le &
\frac{
\log p_2(n)(1+  
\varepsilon_{n,\rho,\mathcal{I}}(\overline{W}^{Z,n},Q_{V,U}))
}{\rho}
\nonumber \\
\le &
\frac{\log (2 p_2(n))}{\rho}
+
\frac{[\log \varepsilon_{n,\rho,\mathcal{I}}(\overline{W}^{Z,n},Q_{V,U})]_+ }{\rho} 
\Label{Haya-24}.
\end{align}
Now, we have
\begin{align}
\lim_{n \to \infty}  \frac{\log (2 p_2(n))}{ n \cdot \frac{\delta\log n}{n}} 
=
\lim_{n \to \infty}  \frac{\log (2 p_2(n))}{\delta \log n}
=
\frac{\deg (p_2)}{\delta}
 \Label{12-26-7},
\end{align}
where $\deg (p_2)$ is the degree of the polynomial $p_2$.
Due to (\ref{11-25-5}) in Lemma \ref{l-12-26-2}, (\ref{Haya-24}), and (\ref{12-26-7}), 
choosing $\rho_n=\frac{\delta \log n}{n}$,
we obtain
\begin{align*}
& 
\limsup_{n\to \infty}
\frac{1}{n}I(S_{\mathcal{I},n};Z^n|S_{0,n})[\overline{W}^{Z,n},\varphi_{a,n},P_{S_{{\cal T+1},n}}]
\le 
\frac{\deg (p_2)}{\delta}+E_-^{\mathcal{I}}.
\end{align*}
Since $\delta>0$ is arbitrary, we have
\begin{align}
& 
\limsup_{n\to \infty}
\frac{1}{n}I(S_{\mathcal{I},n};Z^n|S_{0,n})[\overline{W}^{Z,n},\varphi_{a,n},P_{S_{{\cal T+1},n}}]
\le 
E_-^{\mathcal{I}}
\Label{Haya-02}.
\end{align}
Therefore, using (\ref{Haya-03-b}) and (\ref{Haya-02}),
we can see that $(E^{b}$, $E^{e}$, $E_+^{\mathcal{I}}$, $E_-^{\mathcal{I}})$
is a universally attainable quadruple of exponents in the sense of Definition \ref{def:univexp}.
\end{proof}

\begin{remark}
One might consider that if we apply the random coding of Theorem \ref{lem1}
to the uniform distribution $P_{\mix, T_n(Q_{VU})}$,
we obtain a better exponent.
However, this method yields the same exponent
because
$\psi(\rho| \Upsilon_n(W^Z), P_{V^n|U^n, \mix, T_n(Q_{VU})},P_{\mix, T_n(Q_{U})})$
is the same as
$E_0(\rho| \Upsilon_n(W^Z), P_{V^n|U^n, \mix, T_n(Q_{VU})},P_{\mix, T_n(Q_{U})})$,
which is shown as
\begin{align*}
& e^{\psi(\rho| \Upsilon_n(W^Z), P_{V^n|U^n, \mix, T_n(Q_{VU})},P_{\mix, T_n(Q_{U})})} \\
=&
\sum_{u \in T_n(Q_{U})}
\frac{1}{N(U)}
\sum_{v \in T_n(Q_{V|U=u})}
\\
&\hspace{10ex}
\Biggl[
\frac{1}{N(V|U)}
\sum_{z \in T_n(Q_{Z|VU=(u,v)})}
(\frac{1}{N(Z|VU)})^{1+\rho}
(\frac{1}{N(Z|U)})^{-\rho} 
\Biggr]
\\
=&
 \frac{N(Z|U)^{\rho}}{ N(Z|VU)^{\rho} } .
\end{align*}
\end{remark}

\section{Source-Channel Universal Coding for BCC}
\Label{s10}
Now, we introduce the concept of 
``source-channel universal code for BCC''
for the $n$-fold discrete memoryless extension of a discrete channel.
In a realistic setting,
we do not have statistical knowledge 
of the sources and the channel, precisely. 
In order to treat such a case, 
we have to make a code whose performance is 
guaranteed independently of the statistical properties of the sources and the channel.
Such a kind of universality is called source-channel universality,
and studied for the case of BCD \cite{korner80}.
For the case of wire-tap channel,
the source universality is divided into two parts.
One is the source universality for decoding error probability
and the other is that for the leaked information.
The paper \cite{Vardy-source-univ} studied the latter part.
Although 
the transmission rates are characterized by the pair $(R_0,R_1)$,
in order to make a code achieving the capacity region of BCC,
we employ other two parameters $R_{\mathrm{c}}$ and $R_{\mathrm{p}}$
that satisfy 
$R_0 \le R_{\mathrm{c}}$ and
$R_0+R_1\le R_{\mathrm{c}}+R_{\mathrm{p}}$.
Hence, in the following definition of
a universally attainable quadruple of exponents and 
leaked information rate,
we focus on the set 
$\mathbf{R}^4_{\BCC}:=
\{(R_{\mathrm{p}},R_{\mathrm{c}},R_0,R_1)\in (\mathbf{R}^+)^{4}
|R_0 \le R_{\mathrm{c}},~R_0+R_1\le R_{\mathrm{c}}+R_{\mathrm{p}}\}$.

\begin{definition}\Label{def:scunivexp-b}
A set of functions $({E}^{b}$, ${E}^{e}$,
${E}_+$, ${E}_-)$ 
from $
\mathbf{R}^4_{\BCC}
\times \mathcal{W}(\mathcal{X}$, $\mathcal{Y} \times \mathcal{Z})$
to $\mathbf{R}_{\ge 0}^{4}$ is said to be
a universally attainable quadruple 
of exponents and leaked information rate
for the family of channels
$\mathcal{W}(\mathcal{X}$, $\mathcal{Y}\times \mathcal{Z})$
and for sources
if 
for $\epsilon>0$ and $(R_{\mathrm{p}},R_{\mathrm{c}},R_0,R_1)\in \mathbf{R}^4_{\BCC}$,
there exist 
a sufficiently large integer $N$
and
a sequence of codes $\Phi_n$ of length $n$ satisfying the following conditions. 
(1)
The confidential message set  $\mathcal{S}_{n}$ 
of the code $\Phi_n$
has cardinality $e^{n R_1}$ 
and the common message set $\mathcal{E}_{n}$ of the code $\Phi_n$
has cardinality $e^{n R_0}$. 
(2)
The inequalities
\begin{align}
P_b[W^n,\Phi_n,P_{S_n,E_n}] 
 \leq & \exp(-n
[ {E}^{b}(R_{\mathrm{p}},R_{\mathrm{c}},R_0,R_1,W)-\epsilon]),\Label{Haya-51-b}\\
P_e[W^n,\Phi_n,P_{S_n,E_n}] 
\leq & \exp(-n
[ {E}^{e}(R_{\mathrm{p}},R_{\mathrm{c}},R_0,R_1,W)-\epsilon]),\Label{Haya-52-b}
\end{align}
and
\begin{align}
& I(S_{n}; Z^n|E_n)[W^n,\Phi_n,P_{S_n,E_n}] \nonumber \\
\leq & 
\max \Biggl[
\exp(-n
[ {E}_+^{l}(R_{\mathrm{p}},R_{\mathrm{c}},R_0,R_1,W)-\epsilon]),\nonumber \\
& \hspace{15ex} 
n [{E}_-^{l}(R_{\mathrm{p}},R_{\mathrm{c}},R_0,R_1,W)+\epsilon ]
\Biggr]
\Label{eq:logexpb}
\end{align}
hold for any sequence of joint distributions $P_{S_n,E_n}$
for the confidential message $S_{n}$ on $\mathcal{S}_{n}$ 
and the common message $E_{n}$ on $\mathcal{E}_{n}$,
and 
the $n$-th memoryless extension $W^n$ of any channel $W \in \mathcal{W}(\mathcal{X}$, $\mathcal{Y}\times \mathcal{Z})$
and $n \ge N$.
\end{definition}

Then, 
given a distribution $Q_{VU}$ on $\mathcal{U}\times \mathcal{V}$
and a channel (probability transition matrix) 
$\Xi: \mathcal{V}\rightarrow \mathcal{X}$,
we present 
a universally attainable quadruple of exponents and leaked information rate
as follows.
Given rates $(R_{\mathrm{p}}, R_{\mathrm{c}},R_0,R_1) 
\in (\mathbf{R}^+)^{4}$
and a broadcast $W \in \mathcal{W}(\mathcal{X}$, $\mathcal{Y}\times\mathcal{Z})$,
the quadruple 
$E^{b}$, $E^{e}$, $E_+^l$ and $E_-^l$
are given as
\begin{align}
E^{b} =& E^{b}(R_{\mathrm{p}},R_{\mathrm{c}},R_{0}, R_{1},W)
:= 
\tilde{E}^{b} (R_{\mathrm{p}},R_{\mathrm{c}}, (W\circ \Xi) \times Q_{VU}),
\Label{eq:scunivs1}\\
E^{e} =& E^{e}(R_{\mathrm{p}},R_{\mathrm{c}},R_{0}, R_{1},W)
:= 
\tilde{E}^{e} ( R_{\mathrm{c}}, (W\circ \Xi) \circ  Q_{VU}), 
\Label{eq:scunivs2}\\
E_+^{l}=& E_+^{l}(R_{\mathrm{p}},R_{\mathrm{c}},R_{0}, R_{1},W)
:= \tilde{E}^{l}(R_{\mathrm{p}}-R_1,(W\circ \Xi) \times Q_{VU} ), 
\Label{eq:scunivs3} \\
E_-^{l}=& E_-^{l}(R_{\mathrm{p}},R_{\mathrm{c}},R_{0}, R_{1},W)
:= I(V;Z|U) -R_{\mathrm{p}}+R_1 
\Label{eq:scunivs4}.
\end{align}

\begin{theorem}[Extension of {\cite[Theorem 1, part (a)]{korner80}}]
\Label{univ:scthm}
Eqs.\ (\ref{eq:scunivs1})--(\ref{eq:scunivs4}) are 
source-channel universally attainable rates of exponents and information leakage rate
in the sense of Definition \ref{def:scunivexp-b}.
\end{theorem}

Therefore, our source-channel universal code attaining 
Eqs.\ (\ref{eq:scunivs1})--(\ref{eq:scunivs4}) depends on 
$R_{\mathrm{p}}$, $R_{\mathrm{c}}$,
the distribution $Q_{VU}$ on $\mathcal{U}\times \mathcal{V}$,
and the channel $\Xi: \mathcal{V}\rightarrow \mathcal{X}$.

We prove Theorem \ref{univ:scthm} 
by expurgating the messages in the code given in Theorem \ref{univ:thm}. 
The outline of the proof is as follows:
First, in Step (1),
similar to Theorem \ref{univ:thm},
we evaluate the leaked information when the channel is given by the conditional types and the source obeys the uniform distribution. 
Then, for a given code in Step (1),
we expurgate the common message $E_n$ in Step (2) and 
the secret message $S_n$ in Step (3).
We evaluate the leaked information of the expurgated code for
an arbitrary source distribution and an arbitrary conditional type in Step (4).
Based on this evaluation, we evaluate 
the leaked information of the expurgated code for
an arbitrary source distribution and an arbitrary discrete memoryless channel in Step (5).

In the following proof, 
we assume that
the secret message $S_n$ and the common message $E_n$ 
obey the uniform distributions on
$\mathcal{S}_n$ and $\mathcal{E}_n$.
However,
expurgations $S_n'$ and $E_n'$
of the secret message $S_n$ and the common message $E_n$ 
are allowed to obey arbitrary distributions.

\noindent{\it Step (1): Evaluation of the leaked information when the channel is given as the uniform distribution on a fixed conditional type:\quad}

Recall the fixed code $\varphi_{\mathrm{p},n}$ for BCD given in Theorem \ref{lem-11-25-3-b}.
The code $\varphi_{\mathrm{p},n}$ has 
the private message set $\mathcal{S}_{0,n}\times \mathcal{B}_{1,n}$ and 
the common message set $\mathcal{B}_{2,n}$.
We attach 
the random coding $\Lambda_{F,G,n}$ for message $S_{1,n},\ldots,S_{T,n}$ given as Second Step of Code Ensemble \ref{con1} in Subsection \ref{s5-2}
to the code $\varphi_{\mathrm{p},n}$
when $T=2$, 
$S_{1,n}=S_n$, $S_{0,n}=E_n$, and 
$S_{2,n}$ is the random number subject to the uniform distribution,
which is used as the dummy for making $S_n$ secret for Eve.
The uniformity of the distribution guarantees that
\begin{align}
H_{1+\rho}(S_{2,n}|S_{1,n},S_{0,n})
= n (R_{\mathrm{c}}+ R_{\mathrm{p}}- R_1 - R_2)
\Label{12-26-10}
\end{align}
for any $\rho \in (0,1]$.
Then, the encoder is given by $\Phi_{a,n}=(\varphi_{\mathrm{p},n},\Lambda_{F,G,n})$.
In the following,
Bob's decoder $\Phi_{b,n}$
and
Eve's decoder $\Phi_{e,n}$
are given as the maximum mutual information decoder.
We treat the ensemble of 
codes $\Phi_n:=(\Phi_{a,n},\Phi_{b,n},\Phi_{e,n})$.

For an arbitrary $\Theta_n\in {\cal W}_{n,\mathcal{Z}}(Q_V)$ and
an arbitrary $\rho \in (0,1)$,
the combination of Lemmas \ref{lem2-1} and \ref{l12-3-2} yields that
\begin{align}
& \rE_{\Phi_{a,n}}
\sum_{e} P_{E_n}(e) \sum_s P_{S_n|E_n}(s|e) \nonumber \\
& \quad \cdot \exp (\rho D(P_{Z^n|S_n=s,E_n=e,\Phi_{a,n}}
\| P_{Z^n|E_n=e,\Phi_{a,n}})[\Theta_n] )
\nonumber\\
\leq &
1+ 
e^{n \rho (R_1 -R_{\mathrm{p}})}
e^{E_0(\Theta_n, P_{V^n|U^n, \mix, T_n(Q_{VU})},P_{\mix, T_n(Q_{U})})} \nonumber \\
= &
1+ 
\varepsilon_{n,\rho,\{1\}}(\Theta_n, P_{\mix, T_n(Q_{VU})}),
\Label{1-13-4-b} 
\end{align}
where
$D(P_{Z^n|S_n=s,E_n=e,\varphi_{a,n}}
\| P_{Z^n|E_n=e, \varphi_{a,n}})[\Theta_n] $
denotes the relative entropy
$D(P_{Z^n|S_n=s,E_n=e, \varphi_{a,n}}\| P_{Z^n|E_n=e, \varphi_{a,n}})$
when the channel is $\Theta_n \in {\cal W}_{n,\mathcal{Z}}(Q_V)$.

The relations 
(\ref{12-26-10}) and (\ref{eq100-1}) with $T=2$
yield
\begin{align}
& \rE_{\Phi_{a,n}}
I(S_{\mathcal{I},n};Z^n|S_{0,n})[\Upsilon_n(W^Z),\Phi_{a,n},P_{S_{{\cal T},n}}]
\nonumber \\
\le & 
|{\cal T}_{n,\mathcal{Z}}(Q_{VU})|
e^{-n [ R_{\mathrm{p}}-R_1-I(V;Z|U)[W^Z\times Q_{VU}] ]_+}   .
\Label{12-26-12}
\end{align}
Thanks to the Markov inequality in the same way as (\ref{12-28-3}) and (\ref{12-28-4}),
given a fixed $\rho\in (0,1)$,
due to (\ref{1-13-4-b}) and 
(\ref{12-26-12}),
we can see that there exists at least 
one code $\varphi_{a,n}$ such that
the relations
\begin{align}
& I(S_{\mathcal{I},n};Z^n|S_{0,n})[\Upsilon_n(W^Z),\varphi_{a,n},P_{S_{{\cal T},n}}]
\nonumber \\
\le & 
p_1(n)
|{\cal T}_{n,\mathcal{Z}}(Q_{VU})|
e^{-n [ R_{\mathrm{p}}-R_1-I(V;Z|U)[W^Z\times Q_{VU}] ]_+}   ,
 \Label{Haya-13a} \\
&
\sum_{e} P_{E_n}(e) \sum_s P_{S_n|E_n}(s|e) \nonumber \\
& \quad \cdot \exp (\rho D(P_{Z^n|S_n=s,E_n=e,\varphi_{a,n}}
\| P_{Z^n|E_n=e,\varphi_{a,n}})[\Theta_n] )
\nonumber \\
\le & p_1(n) (1+ \varepsilon_{n,\rho,\{1\}}(\Theta_n, P_{\mix, T_n(Q_{VU})}) ) \Label{Haya-4b}
\end{align}
hold for any $W^Z \in {\cal T}_{n,\mathcal{Z}}(Q_{VU})$
and $\Theta_n\in {\cal W}_{n,\mathcal{Z}}(Q_V)$.

\noindent{\it Step (2): Expurgation for common message $E_n$:\quad}

We choose $p_3(n):= 2 p_1(n)$.
When $e$ is randomly chosen from $\mathcal{E}_n$ subject to the uniform distribution,
the element $e$ satisfies all of the following conditions at least with probability of $1-p_1(n)/p_3(n)=\frac{1}{2}$.
The relations
\begin{align}
& \sum_s P_{S_n|E_n}(s|e)
\exp (\rho D(P_{Z^n|S_n=s,E_n=e,\varphi_{a,n}}
\| P_{Z^n|E_n=e,\varphi_{a,n}})[\Theta_n])
\nonumber\\
\le &
p_1(n)p_3(n)(1+ 
\varepsilon_{n,\rho,\{1\}}(\Theta_n, P_{\mix, T_n(Q_{VU})})) ,
\nonumber\\
& \sum_s P_{S_n|E_n}(s|e)
D(P_{Z^n|S_n=s,E_n=e,\varphi_{a,n}}
\| P_{Z^n|E_n=e,\varphi_{a,n}})[\Upsilon_n(W^Z)] \nonumber\\
=&
I(S_{n};Z^n)[\Upsilon_n(W^Z),
\varphi_{a,n},P_{\mix,\mathcal{S}_n|E_{n}=e}]
\nonumber \\
\le &
p_1(n)p_3(n)
|{\cal T}_{n,\mathcal{Z}}(Q_{VU})|
e^{-n [ R_{\mathrm{p}}-R_1-I(V;Z|U)[W^Z\times Q_{VU}] ]_+}   
\Label{1-13-21}
\end{align}
hold for any elements 
$W^Z \in \mathcal{T}_{n,\mathcal{Z}}(Q_{VU})$
and $\Theta_n \in {\cal W}_{n,\mathcal{Z}}(Q_V)$,
and $n \ge N$.
Thus, there exist $|\mathcal{E}_n|/2$
elements $e \in \mathcal{E}_n$ satisfies the above conditions.
So, we denote the set of such elements by $\mathcal{E}_n'$.

\noindent{\it Step (3): Expurgation for secret message $S_n$:\quad}

Then, when $s$ is randomly chosen from $\mathcal{S}_n$ subject to the uniform distribution,
the element $s$ satisfies all of 
the following conditions
at least with probability of $1-p_1(n)/p_3(n)\ge \frac{1}{2}$:
The relations
\begin{align}
& \exp (\rho D(P_{Z^n|S_n=s,E_n=e',\varphi_{a,n}}
\| P_{Z^n|E_n=e',\varphi_{a,n}})[\Theta_n])
\nonumber\\
\le &
p_1(n)p_3(n)^2(1+ 
\varepsilon_{n,\rho,\{1\}}(\Theta_n, P_{\mix, T_n(Q_{VU})}), 
\Label{2-26-2}
\\
& D(P_{Z^n|S_n=s,E_n=e',\varphi_{a,n}}
\| P_{Z^n|E_n=e',\varphi_{a,n}})[\Upsilon_n(W^Z)]
\nonumber\\
\le &
p_1(n)p_3(n)^2
|{\cal T}_{n,\mathcal{Z}}(Q_{VU})|
e^{-n [ R_{\mathrm{p}}-R_1-I(V;Z|U)[W^Z\times Q_{VU}] ]_+}   
\Label{2-26-3}
\end{align}
hold for any elements 
$e' \in \mathcal{E}_n'$,
$W^Z \in \mathcal{T}_{n,\mathcal{Z}}(Q_{VU})$,
$\Theta_n \in {\cal W}_{n,\mathcal{Z}}(Q_V)$,
and $n \ge N$.
Thus, there exist $|\mathcal{S}_n|/2$
elements $s \in \mathcal{S}_n$ satisfies the above conditions.
So, we denote the set of such elements by $\mathcal{S}_n'$.

\noindent{\it Step (4): Universal code that works for all sources when the channel is given as
the uniform distribution on a fixed conditional type:\quad}

In the following discussion, 
$P_{S_n',E_n'}$ is an arbitrary joint distribution 
of the random variables $S_n'$ and $E_n'$ 
on $\mathcal{S}_n'\times \mathcal{E}_n'$.
For a given $e \in \mathcal{E}_n'$,
we consider two kinds of marginal distributions of $Z^n$
as follows.
\begin{align*}
P_{Z^n|E_n'=e,\varphi_{a,n}}
&=
\sum_{s\in \mathcal{S}_n}P_{S_n}(s)
P_{Z^n|S_n=s, E_n'=e,\varphi_{a,n}} \\
P_{Z^n|E_n'=e,\varphi_{a,n}}'
&:=
\sum_{s'\in \mathcal{S}_n}P_{S_n'|E_n'}(s'|e)
P_{Z^n|S_n=s, E_n'=e,\varphi_{a,n}} .
\end{align*}
The former marginal distribution is discussed in Steps (1), (2), and (3).
Hence, using (\ref{2-26-1}) and (\ref{2-26-3}), 
we obtain
\begin{align}
& I(S_n';Z^n|E_n')[\Upsilon_n(W^Z),\varphi_{a,n},P_{S_n',E_n'}]
\nonumber \\
= &
\sum_{e \in \mathcal{E}_n'}P_{E_n'}(e)
D(P_{Z^n,S_n'|E_n'=e,\varphi_{a,n}}
\| P_{Z^n|E_n'=e,\varphi_{a,n}}' \times P_{S_n'|E_n'=e})
[\Upsilon_n(W^Z)]  \nonumber \\
\le &
\sum_{e \in \mathcal{E}_n'}P_{E_n'}(e)
D(P_{Z^n,S_n'|E_n'=e,\varphi_{a,n}}
\| P_{Z^n|E_n'=e,\varphi_{a,n}} \times P_{S_n'|E_n'=e})
[\Upsilon_n(W^Z)]  \nonumber \\
=&
\sum_{e \in \mathcal{E}_n'}P_{E_n'}(e)
\sum_{s \in \mathcal{S}_n}
\biggl[
P_{S_n'|E_n'}(s|e)
\nonumber \\
& \hspace{15ex} \cdot D(P_{Z^n|S_n'=s,E_n'=e,\varphi_{a,n}}
\| P_{Z^n|E_n'=e,\varphi_{a,n}})[\Upsilon_n(W^Z)]  
\biggr]
\nonumber \\
\le &
p_1(n)p_3(n)^2
|{\cal T}_{n,\mathcal{Z}}(Q_{VU})|
e^{-n [ R_{\mathrm{p}}-R_1-I(V;Z|U)[W^Z\times Q_{VU}] ]_+}   ,
\Label{1-13-13}
\end{align}
for any elements 
$W^Z \in \mathcal{T}_{n,\mathcal{Z}}(Q_{VU})$,
$\Theta_n \in {\cal W}_{n,\mathcal{Z}}(Q_V)$,
and $n \ge N$.
Similarly, 
using the convexity of $x \mapsto e^x$,
(\ref{2-26-1}), (\ref{2-26-2}), and (\ref{2-26-3}),
we obtain
\begin{align}
& e^{\rho I(S_n';Z^n|E_n')[\Theta_n,\varphi_{a,n},P_{S_n',E_n'}]}
\nonumber \\
\le &
\sum_{e \in \mathcal{E}_n'}P_{E_n'}(e)
e^{\rho D(P_{Z^n,S_n'|E_n'=e,\varphi_{a,n}}
\| P_{Z^n|E_n'=e,\varphi_{a,n}}' \times P_{S_n'|E_n'=e})
[\Theta_n]}
 \nonumber \\
\le &
\sum_{e \in \mathcal{E}_n'}P_{E_n'}(e)
e^{\rho D(P_{Z^n,S_n'|E_n'=e,\varphi_{a,n}}
\| P_{Z^n|E_n'=e,\varphi_{a,n}} \times P_{S_n'|E_n'=e})
[\Theta_n]}
 \nonumber \\
\le &
\sum_{e \in \mathcal{E}_n'}P_{E_n'}(e)
\sum_{s \in \mathcal{S}_n}P_{S_n'|E_n'}(s|e)
e^{\rho D(P_{Z^n|S_n'=s,E_n'=e,\varphi_{a,n}}
\| P_{Z^n|E_n'=e,\varphi_{a,n}})[\Theta_n] }\nonumber \\
\le & p_1(n)p_3(n)^2(1+ 
\varepsilon_{n,\rho,\{1\}}(\Theta_n, P_{\mix, T_n(Q_{VU})})) 
\Label{1-13-14}
\end{align}
for any elements 
$W^Z \in \mathcal{T}_{n,\mathcal{Z}}(Q_{VU})$,
$\Theta_n \in {\cal W}_{n,\mathcal{Z}}(Q_V)$,
and $n \ge N$.

\noindent{\it Step (5): Evaluation of leaked information for all sources and all discrete memoryless channels:\quad}

Similar to (\ref{Haya-23}) and (\ref{1-13-6}),
defining 
$p_4(n):= p_1(n)p_3(n)^2|{\cal T}_{n,\mathcal{Z}}(Q_{VU})|^2$
and
$p_5(n):= p_2(n)p_3(n)^2$
and using (\ref{1-13-13}) and (\ref{1-13-14}),
we obtain
\begin{align}
I(S_n';Z^n|E_n')[\overline{W}^{Z,n},\varphi_{a,n},P_{S_n',E_n'}]
\le &
p_4(n)
e^{-n E_+^{l}(R_{\mathrm{p}}, R_{\mathrm{c}},R_0,R_1,W)}
\Label{1-13-17},
\end{align}
and
\begin{align}
& \exp (\rho 
I(S_n';Z^n|E_n')[\overline{W}^{Z,n},\varphi_{a,n},P_{S_n',E_n'}]
) \nonumber \\
\le &
p_5(n)(1+  
\varepsilon_{n,\rho,\{1\}}(\overline{W}^{Z,n},Q_{V,U}^n))\nonumber \\
= &
p_5(n)(1+  
\varepsilon_{1,\rho,\{1\}}(\overline{W}^{Z},Q_{V,U})^n)
\Label{1-13-7}
\end{align}
for any sequence of joint distributions 
$P_{S_n',E_n'}$ and $n \ge N$.

Using (\ref{1-13-17}),
for an arbitrary $\epsilon>0$,
we can choose an integer $N_1$ such that
\begin{align}
& \log 
I(S_n';Z^n|E_n')[\overline{W}^{Z,n},\varphi_{a,n},P_{S_n',E_n'}]
\nonumber \\
\le &
-n ( E_+^l (R_{\mathrm{p}},R_{\mathrm{c}}, R_{0},R_1, W) - \epsilon)
\Label{scHaya-03}
\end{align}
for $n \ge N_1$.
Due to (\ref{1-13-7}),
we obtain
\begin{align}
&\frac{1}{n}I(S_n';Z^n|E_n')[\overline{W}^{Z,n},\varphi_{a,n},P_{S_n',E_n'}]
\nonumber \\
\le &
\frac{\log p_5(n)+
\log (1+  
\varepsilon_{1,\rho,\{1\}}(\overline{W}^{Z},Q_{V,U})^n) }{n\rho} \nonumber \\
\le &
\frac{\log p_5(n)+
\log 2+
\log \varepsilon_{1,\rho,\{1\}}(\overline{W}^{Z},Q_{V,U})^n) }{n\rho}\nonumber\\
\le &
\frac{\log 2 p_5(n)}{n \rho}+
\frac{\log \varepsilon_{1,\rho,\{1\}}(\overline{W}^{Z},Q_{V,U})) }{\rho}
\Label{1-13-7b}.
\end{align}
When $\rho=\frac{1}{\sqrt{n}}$, 
as is mentioned in Lemma \ref{l-12-26-2},
the RHS of (\ref{1-13-7b}) converges 
$E_-^l(R_{\mathrm{p}},R_{\mathrm{c}}, R_{0},R_1, W)$ uniformly.
Hence,
for an arbitrary $\epsilon>0$,
we can choose an integer $N_2$ such that
\begin{align}
& 
I(S_n';Z^n|E_n')[\overline{W}^{Z,n},\varphi_{a,n},P_{S_n',E_n'}]
\nonumber \\
\le &
n (E_-^l(R_{\mathrm{p}},R_{\mathrm{c}}, R_{0},R_1, W)+ \epsilon ) \Label{scHaya-02}
\end{align}
for $n \ge N_2$.

Therefore, 
since the original code $\varphi_{\mathrm{p},n}$
satisfies (\ref{Haya-51-d}) and (\ref{Haya-52-d}),
using 
(\ref{scHaya-03}) and (\ref{scHaya-02}),
we can see that $(E^{b}$, $E^{e}$, $E_+^l$,
$E_-^l)$
is a universally attainable quadruple of exponents
in the sense of Definition \ref{def:scunivexp-b}.
\qed

\begin{remark}\Label{R1-2-1}
In this section, we treat the leaked information 
asymptotically as (\ref{eq:logexpb}).
However, in Section \ref{s9}, we have treated it 
non-asymptotically
as (\ref{eq:logexp}) and (\ref{eq:logexp2}).
The difference is caused by the condition for the sequence of joint distributions
$P_{\mathcal{S}_{\mathcal{T},n}}$.
In Section \ref{s9}, we do not assume the uniformity.
However, in this section, we can use uniform distribution of $S_{2,n}$.
Hence, we can calculate the relative R\'{e}nyi entropy 
as (\ref{12-26-10}) non-asymptotically.
\end{remark}

\begin{remark}
Here, we remark the relation with 
the discussion for secure multiplex coding in 
\cite[Section IV-D]{yamamoto05}.
The preceding paper \cite{yamamoto05} showed the existence of 
the code $\varphi_n$ satisfying that
\begin{align}
\max_{s} D(P_{Z^n|S_i=s_i,\varphi_n} \| P_{Z^n,\varphi_n}) \to 0
\end{align}
when there is no common message $E_n$
and the random variables $S_1, \ldots, S_T$ obey the uniform distribution.
However, to show the source universality for leaked information
in secure multiplex coding 
we need to evaluate the above value 
when the random variables $S_1, \ldots, S_T$ do not necessarily 
obey the uniform distribution.
In this section, we show the source universality for leaked information for $S_1$
by assuming the uniformity of the other random variable $S_2$.
Although this method brings us the source universality for BCC,
it cannot derive the source universality for secure multiplex coding. 
\end{remark}

\section{Comparison of Exponents of Leaked Information}\Label{s12}
In this section, 
we compare the exponent of leaked information
given in Sections \ref{s9} and \ref{s10} and the exponents of leaked information given in Subsection \ref{s7-3}
when the source distribution $P_{S_{{\cal T},n}}$ is uniform.
First, in Subsection \ref{s12-2}, we compare the exponent given in Sections \ref{s9} and \ref{s10} with the above mentioned exponent.
Then, we clarify that the exponent in Sections \ref{s9} and \ref{s10}
is greater than one of exponents in Subsection \ref{s7-3},
which is the same as that in \cite{hayashimatsumoto2011allerton}.
Next, 
in Subsection \ref{s12-2b},
we give equality conditions between two exponents.
In the remaining subsections, 
we give proofs of Lemmas used in Subsections \ref{s12-2} and \ref{s12-2b}.

\subsection{Comparison between Two Exponents
$\tilde{E}^{l}(R, \overline{W}^Z\times Q_{VU})$ and $\tilde{E}^{E_0}(R, \overline{W}^Z\times Q_{VU})$}\Label{s12-2}
First, we characterize the exponent 
$\tilde{E}^{E_0}(R, \overline{W}^Z\times Q_{VU})
=\sup_{\rho \in (0,1)}
\rho R - E_0(\rho| \overline{W}^Z ,Q_{V|U},Q_U)$,
which describes the exponent of leaked information
when $R$ is $R_{\mathrm{p}} -\sum_{i\in \mathcal{I}}R_i$
and the source distribution $P_{S_{{\cal T},n}}$ is uniform,
as is shown in Subsection \ref{s7-3}.
The exponent can be attained by the code constructed in the second construction (Subsection \ref{s5-2}).
Since $E_0(\rho| \overline{W}^Z ,Q_{V|U},Q_U)$ is convex with respect to $\rho$ \cite{gallager68},
$F_\rho(Q_{V|U},Q_U ):=\frac{d}{d\rho} E_0(\rho| \overline{W}^Z ,Q_{V|U},Q_U)$
is monotonically increasing with respect to $\rho$.
As limits, 
we define 
\begin{align}
F_1(Q_{V|U},Q_U )
&:=\lim_{\rho\to 1-0}
F_\rho(Q_{V|U},Q_U ) \\
E_0(1| \overline{W}^Z ,Q_{V|U},Q_U )
&:= 
\lim_{\rho \to 1-0}E_0(\rho| \overline{W}^Z ,Q_{V|U},Q_U ) .
\end{align}
In particular, when 
$Q_{VU}$ equal $Q_V \times Q_U$,
$\tilde{E}^{l}(R, \overline{W}^Z\times Q_{VU})$,
$\tilde{E}^{E_0}(R, \overline{W}^Z\times Q_{VU})$,
and the above values depend only on $Q_V$. 
Then, 
$\tilde{E}^{l}(R, \overline{W}^Z\times Q_{VU})$,
$\tilde{E}^{E_0}(R, \overline{W}^Z\times Q_{VU})$,
$E_0(1| \overline{W}^Z ,Q_{V|U},Q_U )$,
$F_1(Q_{V|U},Q_U )$, 
and 
$F_\rho(Q_{V|U},Q_U )$ 
are simplified to 
$\tilde{E}^{l}(R, \overline{W}^Z\times Q_{V})$,
$\tilde{E}^{E_0}(R, \overline{W}^Z\times Q_{V})$,
$E_0(1| \overline{W}^Z ,Q_{V})$, $F_1(Q_{V})$,
and $F_\rho(Q_{V})$.
Then, we obtain the following lemma.
\begin{lemma}\Label{l-12-20-2}
(1) Case of $R < F_1(Q_{V|U},Q_U )$.
There uniquely exists $\rho \in (0,1)$ such that $R=F_\rho(Q_{V|U},Q_U )$.
Then, the exponent $\tilde{E}^{E_0}(R, \overline{W}^Z\times Q_{VU})$
can be characterized as
\begin{align}
\tilde{E}^{E_0}(R, \overline{W}^Z\times Q_{VU}) =
\rho_0 R - 
E_0(\rho_0| \overline{W}^Z ,Q_{V|U},Q_U ) .\Label{12-20-10}
\end{align}
(2) Case of $R \ge F_1(Q_{V|U},Q_U )$.
The exponent $\tilde{E}^{E_0}(R, \overline{W}^Z\times Q_{VU})$
can be characterized as
\begin{align}
\tilde{E}^{E_0}(R, \overline{W}^Z\times Q_{VU}) 
&=
R - E_0(1| \overline{W}^Z ,Q_{V|U},Q_U ) \Label{12-20-11} .
\end{align}
\end{lemma}

The quantities appearing in Lemma \ref{l-12-20-2}
can be characterized by Lemma \ref{L-2-23-9}, which is 
displayed in the wide space in the next page.

\begin{figure*}[!t]
\begin{lemma}\Label{L-2-23-9}
The quantities 
$F_\rho(Q_{V|U},Q_U)$,
$F_1(Q_{V|U},Q_U)$,
and $E_0(1| \overline{W}^Z ,Q_{V|U},Q_U )$
are calculated as
\begin{align}
F_\rho(Q_{V|U},Q_U)
=&
\frac{
\sum_u Q_U(u) \sum_z 
(\sum_v \frac{1}{1-\rho } (\log \overline{W}^Z(z|v) )
Q_{V|U}(v|u) \overline{W}^Z(z|v)^{\frac{1}{1-\rho}})
( \sum_v Q_{V|U}(v|u) \overline{W}^Z(z|v)^{\frac{1}{1-\rho}} )^{-\rho}
}{\sum_u Q_U(u)\sum_z (\sum_v Q_{V|U}(v|u) \overline{W}^Z(z|v)^{\frac{1}{1-\rho}} )^{1-\rho}}\nonumber\\
&-\frac{
\sum_u Q_U(u) \sum_z
\log (\sum_v Q_{V|U}(v|u) \overline{W}^Z(z|v)^{\frac{1}{1-\rho}})
( \sum_v Q_{V|U}(v|u) \overline{W}^Z(z|v)^{\frac{1}{1-\rho}} )^{1-\rho}
}{\sum_u Q_U(u)\sum_z (\sum_v Q_{V|U}(v|u) \overline{W}^Z(z|v)^{\frac{1}{1-\rho}} )^{1-\rho}}. \Label{2-23-8}
\\
F_1(Q_{V|U},Q_U)
=&
-\frac{
\sum_u Q_U(u) \sum_z
\log (\sum_{v \in \mathcal{V}_z} Q_{V|U}(v|u))
\max_{v'} \overline{W}^Z(z|v') 
}
{\sum_z 
\max_{v'} \overline{W}^Z(z|v') 
}\Label{12-5-10} \\
E_0(1| \overline{W}^Z ,Q_{V|U},Q_U )
= &
\log 
\sum_{u}Q_U(u)
\sum_{z} \max_{v \in \supp(Q_{V|U=u}) }\overline{W}^Z(z|v)
\Label{12-20-12}.
\end{align}
In particular,
$F_\rho(Q_{V})$,
$F_1(Q_{V})$, and 
$E_0(1| \overline{W}^Z ,Q_{V})$
are simplified to
\begin{align}
F_\rho(Q_{V})
=&
\frac{
\sum_z 
(\sum_v \frac{1}{1-\rho } (\log \overline{W}^Z(z|v) )
Q_{V}(v) \overline{W}^Z(z|v)^{\frac{1}{1-\rho}})
( \sum_v Q_{V}(v) \overline{W}^Z(z|v)^{\frac{1}{1-\rho}} )^{-\rho}
}{\sum_z (\sum_v Q_{V}(v) \overline{W}^Z(z|v)^{\frac{1}{1-\rho}} )^{1-\rho}}\nonumber\\
&-\frac{
\sum_z
\log (\sum_v Q_{V}(v) \overline{W}^Z(z|v)^{\frac{1}{1-\rho}})
( \sum_v Q_{V}(v) \overline{W}^Z(z|v)^{\frac{1}{1-\rho}} )^{1-\rho}
}{\sum_z (\sum_v Q_{V}(v) \overline{W}^Z(z|v)^{\frac{1}{1-\rho}} )^{1-\rho}}. 
\nonumber\\
=&
\frac{
\sum_{z,v} 
(\frac{1}{1-\rho } (\log \overline{W}^Z(z|v) ) -\log (\sum_{v''} Q_{V}(v'') \overline{W}^Z(z|v'')^{\frac{1}{1-\rho}}))
( Q_{V}(v) \overline{W}^Z(z|v)^{\frac{1}{1-\rho}} )
( \sum_{v'} Q_{V}(v') \overline{W}^Z(z|v')^{\frac{1}{1-\rho}} )^{-\rho}
}{\sum_z (\sum_v Q_{V}(v) \overline{W}^Z(z|v)^{\frac{1}{1-\rho}} )^{1-\rho}}
\Label{2-23-9b}
\\
F_1(Q_{V})
=&
-\frac{
\sum_z
\log (\sum_{v \in \mathcal{V}_z} Q_{V}(v))
\max_{v'} \overline{W}^Z(z|v') 
}
{\sum_u Q_U(u)\sum_z 
\max_{v'} \overline{W}^Z(z|v') 
}\Label{12-5-10b} \\
E_0(1| \overline{W}^Z ,Q_{V})
= &
\log 
\sum_{z} \max_{v \in \supp(Q_{V}) }\overline{W}^Z(z|v)
\Label{12-20-12b}.
\end{align}
Further, the map $Q_V \mapsto F_1(Q_{V})$ is concave.
\end{lemma}
\end{figure*}

The proof of Lemma \ref{L-2-23-9} will be given in Subsection \ref{s12-d}.
For a detail analysis for the exponent $\tilde{E}^{E_0}(R, \overline{W}^Z\times Q_{VU})$,
we define
\begin{align}
F_\rho &:= \frac{d}{d\rho} E_{0,\max}(\rho|\overline{W}^Z), \quad 
F_1 :=\lim_{\rho \to 1-0} F_\rho, \\
\mathcal{K} &:= 
\{(z,v)\in \mathcal{Z}\times \mathcal{V} |
\overline{W}^{Z}(z|v)= 
\max_{v'}\overline{W}^Z(z|v')\} \nonumber \\
\mathcal{Z}_v &:= \{z \in \mathcal{Z}| (z,v)\in \mathcal{K} \} , \quad
\mathcal{V}_z:= \{v \in \mathcal{V}| (z,v)\in \mathcal{K} \}. \Label{6-4-1}
\end{align}

Due to the compactness of the set $\mathcal{P}(\mathcal{U})$, 
we have
\begin{align*}
\lim_{\rho \to 1-0} \max_{Q_{V}'}E_0(1| \overline{W}^Z ,Q_{V}') 
=\max_{Q_{V}'} \lim_{\rho \to 1-0} E_0(1| \overline{W}^Z ,Q_{V}') .
\end{align*}
Hence, we obtain the following lemma for characterization of the quantity 
$E_{0,\max}(1|\overline{W}^Z)$ defined in \eqref{eq10001}. 
\begin{lemma}
We have 
\begin{align}
E_{0,\max}(1|\overline{W}^Z)
=\log \sum_{z} \max_{v}\overline{W}^Z(z|v)
= \lim_{\rho \to 1-0} 
E_{0,\max}(\rho|\overline{W}^Z)
\Label{12-20-13}.
\end{align}
\end{lemma}

Then, we have the following characterization for a special case of Case (2) of Lemma \ref{l-12-20-2}.

\begin{lemma}\Label{L-2-23-10}
Assume that
$\cup_{v \in \supp(Q_u)} \mathcal{Z}_v = \mathcal{Z}$
for any $u \in \supp (Q_U)$.
When $R \ge F_1(Q_{V|U},Q_U )$,
we have 
\begin{align}
E_{0,\max}(1|\overline{W}^Z)
=E_0(1| \overline{W}^Z ,Q_{V|U},Q_U ) 
\Label{12-4-20}
\end{align}
and 
\begin{align}
\tilde{E}^{E_0}(R, \overline{W}^Z\times Q_{VU}) =
R-E_{0,\max}(1|\overline{W}^Z)
\Label{12-4-20-B}.
\end{align}
\end{lemma}
The proof of Lemma \ref{L-2-23-10} will be given in Subsection \ref{s12-e}.

For comparison between two exponential decreasing rates
$\tilde{E}^{E_0}(R, \overline{W}^Z\times Q_{VU})$
and $\tilde{E}^l (R, \overline{W}^Z\times Q_{VU})$, 
we prepare the following lemma.

\begin{lemma}\Label{l12-4-3}
Any channel 
$\overline{W}^Z \in \mathcal{W}(\mathcal{V},\mathcal{Z})$ satisfies
\begin{align}
&\min_{{W}^Z \in \mathcal{W}(\mathcal{U}\times \mathcal{V},\mathcal{Z})}
D({W}^Z\| \overline{W}^Z| Q_{VU})
-
\rho I(V;Z|U)[{W}^Z\times Q_{VU}]  \nonumber \\
\ge & -{E_0}(\rho | \overline{W}^Z, Q_{V|U}, Q_U)
\Label{12-4-1}
\end{align}
for any $\rho \in (0,1)$.
\end{lemma}
The proof of Lemma \ref{l12-4-3} will be given in Subsection \ref{s12-3}.
Since the inequalities
\begin{align*}
& \tilde{E}^l (R, \overline{W}^Z\times Q_{VU}) \nonumber \\
= & 
\min_{{W}^Z \in \mathcal{W}(\mathcal{U}\times \mathcal{V},\mathcal{Z})}
\! D({W}^Z\| \overline{W}^Z| Q_{VU})
\! +\! 
[R-I(V;Z|U)[{W}^Z\times Q_{VU}] ]_+ \nonumber \\
\ge & 
\min_{{W}^Z \in \mathcal{W}(\mathcal{U}\times \mathcal{V},\mathcal{Z})}
\! D({W}^Z\| \overline{W}^Z| Q_{VU})
\! +\! 
\rho [R-I(V;Z|U)[{W}^Z\times Q_{VU}] ]_+ \nonumber \\
\ge & 
\min_{{W}^Z \in \mathcal{W}(\mathcal{U}\times \mathcal{V},\mathcal{Z})}
\! D({W}^Z\| \overline{W}^Z| Q_{VU})
\! +\! 
\rho (R-I(V;Z|U)[{W}^Z\times Q_{VU}] ) 
\end{align*}
hold for any $\rho \in (0,1)$,
we obtain the following theorem, which is \eqref{3-22-1}. 
\begin{theorem}\Label{t-12-20-2}
\begin{align}
& \tilde{E}^l (R, \overline{W}^Z\times Q_{VU}) \nonumber \\
\ge & 
\sup_{\rho \in (0,1)}
\rho R - E_0(\rho| \overline{W}^Z ,Q_{V|U},Q_U ) 
= \tilde{E}^{E_0}(R, \overline{W}^Z\times Q_{VU}) 
\Label{12-4-17}.
\end{align}
\end{theorem}

\subsection{Equality Conditions of (\ref{12-4-17})}\Label{s12-2b}
In this subsection, we derive equality conditions of (\ref{12-4-17}).
For this purpose, we prepare two lemmas.
\begin{lemma}\Label{L2-22-1}
For a fixed $\rho \in (0,1)$, the following three conditions 
for a distribution $Q_V$ are equivalent.
\begin{itemize}
\item[(i)] 
The following value does not depend on $v \in \mathcal{V}$. 
\begin{align*}
\sum_z \overline{W}^Z(z|v)^{\frac{1}{1-\rho}}
(\sum_{v'} Q_{V}(v') 
\overline{W}^Z(z|v')^{\frac{1}{1-\rho}})^{-\rho}
\end{align*}
\item[(ii)] 
The following relation holds.
\begin{align}
E_0(\rho| \overline{W}^Z ,Q_{V})= 
E_{0,\max}(\rho|\overline{W}^Z)=\max_{Q_{V}'} E_0(\rho| \overline{W}^Z ,Q_{V}').
\Label{2-23-3}
\end{align}
\item[(iii)]
The following relations hold for any $v \in \mathcal{V}$.
\begin{align*}
&\sum_z \overline{W}^Z(z|v)^{\frac{1}{1-\rho}}
(\sum_{v'} Q_{V}(v') 
\overline{W}^Z(z|v')^{\frac{1}{1-\rho}})^{-\rho}
\\
=&
\max_{Q_{V}'}
\sum_{z}(\sum_{v'} Q_{V}'(v') 
\overline{W}^Z(z|v')^{\frac{1}{1-\rho}})^{1-\rho}
\\ =& \max_{Q_{V}'} e^{E_0(\rho| \overline{W}^Z ,Q_{V}')}
= e^{E_{0,\max}(\rho|\overline{W}^Z)}.
\end{align*}
\end{itemize}
\end{lemma}
The proof of Lemma \ref{L2-22-1} will be given in Subsection \ref{s12-3-5}.

\begin{lemma}\Label{L2-21-1}
The following three conditions for a distribution $Q_V$ are equivalent.
\begin{itemize}
\item[(i)] 
The following value does not depend on $v \in \mathcal{V}$.
\begin{align*}
\sum_{z \in \mathcal{Z}_v}
\frac{
\max_{v'\in \mathcal{V}}\overline{W}^Z(z|v')}
{\sum_{v''\in \mathcal{V}_z }  Q_{V}(v'') 
}
=
\sum_{z \in \mathcal{Z}_v}
\frac{\overline{W}^Z(z|v)}
{\sum_{v''\in \mathcal{V}_z }  Q_{V}(v'') 
}.
\end{align*}
\item[(ii)] 
The following relation holds.
\begin{align*}
F_1(Q_V) =\min_{Q_V'}F_1(Q_V') .
\end{align*}
\item[(iii)]
The following relations hold for any $v \in \mathcal{V}$.
\begin{align}
\sum_{z \in \mathcal{Z}_v}
\frac{
\max_{v'\in \mathcal{V}}\overline{W}^Z(z|v')}
{\sum_{v''\in \mathcal{V}_z }  Q_{V}(v'') 
}
=&
\sum_{z \in \mathcal{Z}_v}
\frac{\overline{W}^Z(z|v)}
{\sum_{v''\in \mathcal{V}_z }  Q_{V}(v'') 
}
\nonumber \\
=&\sum_{z} \max_{v'}\overline{W}^Z(z|v')\Label{2-19-1}.
\end{align}
\end{itemize}
\end{lemma}
The proof of Lemma \ref{L2-22-1} will be given in Subsection \ref{s12-7}.

Then, we introduce two conditions for a distribution $Q_V$.
\begin{condition}\Label{c-2-22-1}
Given a fixed $\rho \in (0,1)$,
the distribution $Q_V$ satisfies the condition given in Lemma \ref{L2-22-1}
\end{condition}

\begin{condition}\Label{c-2-21-2}
The distribution $Q_V$ satisfies the condition given in Lemma \ref{L2-21-1}
\end{condition}

Since Condition \ref{c-2-22-1} depends on $\rho$,
we describe it by ``Condition \ref{c-2-22-1} with $\rho$''
when we need to clarify the dependence on $\rho$.

\begin{lemma}\Label{L2-24-1}
When distribution $Q_V$ and $Q_V'$ satisfy Condition \ref{c-2-22-1} with $\rho$,
the relation 
$\sum_{v}Q_{V}(v)\overline{W}^Z(z|v)^{\frac{1}{1-\rho}}
=
\sum_{v}Q_{V}'(v)\overline{W}^Z(z|v)^{\frac{1}{1-\rho}}$
holds for any $z \in \mathcal{Z}$.
That is the value 
$\sum_{v}Q_{V}(v)\overline{W}^Z(z|v)^{\frac{1}{1-\rho}}$
does not depend on the choice of $Q_V$
as long as the distribution $Q_V$ satisfies Condition \ref{c-2-22-1} with $\rho$.
\end{lemma}
The proof of Lemma \ref{L2-24-1} will be given in Subsection \ref{s12-3-5}.

\begin{lemma}\Label{L2-24-2}
When distribution $Q_V$ and $Q_V'$ satisfy Condition \ref{c-2-21-2} with $\rho$,
the relation 
$\sum_{v''\in \mathcal{V}_z } Q_{V}(v'')=\sum_{v''\in \mathcal{V}_z } Q_{V}'(v'')$
holds for any $z \in \mathcal{Z}$.
That is the value 
$\sum_{v}Q_{V}(v)\overline{W}^Z(z|v)^{\frac{1}{1-\rho}}$
does not depend on the choice of $Q_V$
as long as the distribution $Q_V$ satisfies Condition \ref{c-2-21-2}.
\end{lemma}
The proof of Lemma \ref{L2-24-2} will be given in Subsection \ref{s12-7}.
Hence, we can define 
the transition matrices $W^{Z,\rho}$ and $W^{Z,1}$ from $\mathcal{V}$ to $\mathcal{Z}$ by
\begin{align*}
W^{Z,\rho}(z|v)
:=&
\frac
{
\overline{W}^Z(z|v)^{\frac{1}{1-\rho}}
(\sum_{v}Q_{V,\rho}(v)\overline{W}^Z(z|v)^{\frac{1}{1-\rho}})^{-\rho}
}
{\sum_z
\overline{W}^Z(z|v)^{\frac{1}{1-\rho}}
(\sum_{v}Q_{V,\rho}(v)\overline{W}^Z(z|v)^{\frac{1}{1-\rho}})^{-\rho}
},\nonumber \\\
W^{Z,1}(z|v)
:=&
\left\{
\begin{array}{ll}
\frac{\overline{W}^Z(z|v)}
{\sum_{v''\in \mathcal{V}_z }  Q_{V,1}(v'') \sum_{z'} \max_{v'}\overline{W}^Z(z'|v')}
& z \in \mathcal{Z}_v \\
0 & z \in \mathcal{Z}_v^c,
\end{array}
\right.
\end{align*}
where 
the distributions $Q_{V,\rho}$ and $Q_{V,1}$
satisfy 
Condition \ref{c-2-22-1} with $\rho$ and Condition \ref{c-2-21-2}, respectively.
These definitions do not depend on the choices of $Q_{V,\rho}$ and $Q_{V,1}$.

\begin{lemma}\Label{L2-22-4}
When $Q_{V,\rho}$ satisfies Condition \ref{c-2-22-1} with $\rho$,
we have
\begin{align}
&F_\rho=F_\rho(Q_{V,\rho}) 
=I(V;Z)[W^{Z,\rho} \times Q_{V,\rho}] \Label{2-23-2} \\
&D(W^{Z,\rho}\|\overline{W}^Z|Q_{V,\rho})
=\rho F_\rho-E_{0,\max}(\rho|\overline{W}^Z).
 \Label{2-23-1}
\end{align}
\end{lemma}
The proof of Lemma \ref{L2-22-4} will be given in Subsection \ref{s12-3-5}.
\begin{lemma}\Label{L2-22-5}
When $Q_{V,1}$ satisfies Condition \ref{c-2-21-2},
we have
\begin{align}
&F_1
=F_1(Q_{V,1}) 
=I(V;Z)[W^{Z,1} \times Q_{V,1}] \Label{2-23-4}\\
&D(W^{Z,1}\|\overline{W}^Z|Q_{V,1})
=F_1-E_{0,\max}(1|\overline{W}^Z).
\Label{2-23-5}
\end{align}
\end{lemma}
The proof of Lemma \ref{L2-22-5} will be given in Subsection \ref{s12-7}.

\begin{lemma}\Label{L2-21-3}
For any $\rho \in (0,1)$, we choose the distribution $Q_{V,\rho}$ 
satisfying Condition \ref{c-2-22-1} with $\rho$.
We choose a sequence $\rho_n$ such that $\rho_n \to 0$ as $n \to \infty$ and 
the limit distribution $\lim_{n \to \infty} Q_{V,\rho_n}$ exists.
(Since the set of distributions over $\mathcal{V}$ is compact, such a sequence $\rho_n$ exists.)
Then, 
the limit distribution $\lim_{n \to \infty} Q_{V,\rho_n}$
satisfies Condition \ref{c-2-21-2}.
\end{lemma}
The proof of Lemma \ref{L2-21-3} will be given in Subsection \ref{s12-h}.


Then, using the above lemmas,
we can characterize equality conditions of (\ref{12-4-17})
for the case $Q_{UV}=Q_{U}\times Q_{V}$
in the following way.
\begin{theorem}\Label{t-12-20-1b}
(1) Case of $R < F_1$.
We choose $\rho \in (0,1)$ such that $R=F_{\rho}$.
When $Q_{V,\rho}$ satisfies Condition \ref{c-2-22-1} with $\rho$, 
the relations
\begin{align}
&\min_{Q_V}\tilde{E}^l (R, \overline{W}^Z\times Q_{V})
=\min_{Q_V}\tilde{E}^{E_0} (R, \overline{W}^Z\times Q_{V})
\nonumber \\
=&\tilde{E}^l (R, \overline{W}^Z\times Q_{V,\rho}) 
=\tilde{E}^{E_0} (R, \overline{W}^Z\times Q_{V,\rho}) 
= \rho R-E_{0,\max}(\rho|\overline{W}^Z) \Label{12-5-21}
\end{align}
hold, which implies the equality in (\ref{12-4-17}).

(2) Case of $R \ge F_1$.
When 
$Q_{V,1}$ satisfies Condition \ref{c-2-21-2},
the relations
\begin{align}
&\min_{Q_V}\tilde{E}^l (R, \overline{W}^Z\times Q_{V}) 
=\min_{Q_V}\tilde{E}^{E_0} (R, \overline{W}^Z\times Q_{V})
\nonumber \\
=&\tilde{E}^l (R, \overline{W}^Z\times Q_{V,1}) 
=\tilde{E}^{E_0} (R, \overline{W}^Z\times Q_{V,1}) 
= R-E_{0,\max}(1 |\overline{W}^Z)
\Label{12-5-21b}
\end{align}
hold, which implies the equality in (\ref{12-4-17}).
\end{theorem}

Combining the discussions in both cases in Theorem \ref{t-12-20-1b}, we obtain
\begin{align}
\min_{Q_V}\tilde{E}^l (R, \overline{W}^Z\times Q_{V}) 
=& \min_{Q_V}\tilde{E}^{E_0} (R, \overline{W}^Z\times Q_{V})
\nonumber \\
=&\max_{\rho \in [0,1]} \rho R- E_{0,\max}(\rho|\overline{W}^Z),
\end{align}
which is \eqref{3-22-2}. 

\begin{proofof}{Theorem \ref{t-12-20-1b}}
First, we show (\ref{12-5-21}).
Since $I(V;Z)[W^{Z,\rho} \times Q_{V,\rho}]=F_\rho=R$ follows from (\ref{2-23-2}),
we have
\begin{align}
& \tilde{E}^l (R, \overline{W}^Z\times Q_{V,\rho}) \nonumber\\
\stackrel{(a)}{\le} &
D(W^{Z,\rho} \| \overline{W}^Z| Q_{V,\rho})  +
[R - I(V;Z)[W^{Z,\rho} \times Q_{V,\rho}] ]_+ \nonumber\\
\stackrel{(b)}{=} & 
\rho F_\rho - E_{0,\max}(\rho|\overline{W}^Z)
\stackrel{(c)}{=} 
\rho R- E_0(\rho|\overline{W}^Z,Q_{V,\rho}) 
\nonumber\\
\stackrel{(d)}{=} &
\tilde{E}^{E_0}(R, \overline{W}^Z\times Q_{V,\rho}) ,
\Label{3-23-1eq}
\end{align}
where 
$(a)$,
$(b)$,
$(c)$, and
$(d)$ follow from
the Definition \eqref{1-31-1-k} of $\tilde{E}^l (R, \overline{W}^Z\times Q_{V,\rho})$,
(\ref{2-23-1}),
(\ref{2-23-3}),
and Item (1) of Lemma \ref{l-12-20-2}, respectively.

Any distribution $Q_V$ satisfies
\begin{align*}
& \rho R- E_{0,\max}(\rho|\overline{W}^Z)
\le \rho R- E_0(\rho|\overline{W}^Z,Q_{V})
\le \tilde{E}^{E_0}(R, \overline{W}^Z\times Q_{V}) ,
\end{align*}
which implies
\begin{align}
& \rho R- E_{0,\max}(\rho|\overline{W}^Z)
\le \min_{Q_V} \tilde{E}^{E_0}(R, \overline{W}^Z\times Q_{V}) .
\Label{3-23-2eq}
\end{align}
Combining the above relations and 
we obtain
\begin{align}
& \tilde{E}^l (R, \overline{W}^Z\times Q_{V,\rho}) 
\stackrel{(a)}{\le} 
 \rho R- E_{0,\max}(\rho|\overline{W}^Z)
\nonumber \\
\stackrel{(b)}{\le}  & 
\min_{Q_V} \tilde{E}^{E_0}(R, \overline{W}^Z\times Q_{V}) 
\stackrel{(c)}{\le} 
\min_{Q_V} \tilde{E}^{l}(R, \overline{W}^Z\times Q_{V}) ,
\Label{3-23-3eq}
\end{align}
where 
$(a)$, $(b)$, and $(c)$ follow from
\eqref{3-23-1eq},
\eqref{3-23-2eq}, and
Theorem \ref{t-12-20-2}, respectively.
Hence, the combination of \eqref{3-23-3eq} and $(d)$ of \eqref{3-23-1eq}
leads (\ref{12-5-21}).

Next, we show (\ref{12-5-21b}).
The relations (\ref{2-23-4}) and (\ref{2-23-5}) imply 
\begin{align*}
& \tilde{E}^l (R, \overline{W}^Z\times Q_{V,1}) 
\nonumber\\
\le &
D(W^{Z,1} \| \overline{W}^Z| Q_{V,1})  +
[R - I(V;Z)[W^{Z,1} \times Q_{V,1}] ]_+ \nonumber\\
=&  F_1 - E_{0,\max}(1|\overline{W}^Z)
 +[R - F_1 ]_+ 
\nonumber\\
=&  F_1 - E_{0,\max}(1|\overline{W}^Z) +R - F_1  
=  R- E_{0,\max}(1|\overline{W}^Z)
\nonumber\\
=&  R- E_0(1|\overline{W}^Z,Q_{V,1}) 
=\tilde{E}^{E_0}(R, \overline{W}^Z\times Q_{V,1}) .
\end{align*}
Any distribution $Q_V$ satisfies
\begin{align*}
& R- E_{0,\max}(1|\overline{W}^Z)
\le R- E_0(1|\overline{W}^Z,Q_{V})
\le \tilde{E}^{E_0}(R, \overline{W}^Z\times Q_{V}) ,
\end{align*}
which implies
\begin{align*}
& R-  E_{0,\max}(1|\overline{W}^Z)
\le \min_{Q_V} \tilde{E}^{E_0}(R, \overline{W}^Z\times Q_{V,\rho}) .
\end{align*}
Combining the above relations and Lemma \ref{t-12-20-2}, we obtain
\begin{align*}
& \tilde{E}^l (R, \overline{W}^Z\times Q_{V,\rho}) 
\le R- E_{0,\max}(1|\overline{W}^Z)
=\tilde{E}^{E_0}(R, \overline{W}^Z\times Q_{V,\rho}) \\
\le & \min_{Q_V} \tilde{E}^{E_0}(R, \overline{W}^Z\times Q_{V}) 
\le \min_{Q_V} \tilde{E}^{l}(R, \overline{W}^Z\times Q_{V}) ,
\end{align*}
which implies (\ref{12-5-21b}).
\end{proofof}

For the general case, we prepare the generalizations of Lemmas \ref{L2-22-4} and \ref{L2-22-5}.
The following lemmas follow from Lemmas \ref{L2-22-4} and \ref{L2-22-5}.
\begin{lemma}\Label{L2-23-4}
When 
$Q_{V|U=u}$ satisfies Condition \ref{c-2-22-1} with $\rho$, 
for any $u \in \supp(Q_U)$,
\begin{align*}
&F_\rho=F_\rho(Q_{V|U}, Q_U) 
=I(V;Z|U)[W^{Z,\rho} \times Q_{VU}] \\
&D(W^{Z,\rho}\|\overline{W}^Z|Q_{VU})
=F_\rho-E_{0,\max}(\rho|\overline{W}^Z).
\end{align*}
\end{lemma}
\begin{lemma}\Label{L2-23-5}
When 
$Q_{V|U=u}$ satisfies Condition \ref{c-2-21-2} for any $u \in \supp(Q_U)$,
\begin{align*}
&F_1
=F_1(Q_{V|U}, Q_U)
=I(V;Z|U)[W^{Z,1} \times Q_{VU}] \\
&D(W^{Z,1}\|\overline{W}^Z|Q_{VU})
=F_1-E_{0,\max}(1|\overline{W}^Z).
\end{align*}
\end{lemma}


Then, we can characterize equality conditions for (\ref{12-4-17}) in the general case.
That is, similar to Theorem \ref{t-12-20-1b},
using Lemmas \ref{L2-23-4} and \ref{L2-23-5},
we can show the following theorem.
\begin{theorem}\Label{t-12-20-1}
(1) Case of $R < F_1$.
We choose $\rho \in (0,1)$ such that $R=F_{\rho}$.
When 
$Q_{V|U=u}$ satisfies Condition \ref{c-2-22-1} with $\rho$ 
for any $u \in \supp(Q_U)$,
the relations
\begin{align}
&\min_{Q_{VU}'}\tilde{E}^l (R, \overline{W}^Z\times Q_{VU}')
=\min_{Q_V'}\tilde{E}^l (R, \overline{W}^Z\times Q_{V}')\nonumber \\
=&\min_{Q_{VU}'}\tilde{E}^{E_0} (R, \overline{W}^Z\times Q_{VU}')
=\min_{Q_V'}\tilde{E}^{E_0} (R, \overline{W}^Z\times Q_{V}')\nonumber \\
=&\tilde{E}^l (R, \overline{W}^Z\times Q_{VU}) 
=\tilde{E}^{E_0} (R, \overline{W}^Z\times Q_{VU}) 
= \rho R-
E_{0,\max}(\rho|\overline{W}^Z)
\Label{12-5-21f}
\end{align}
hold, which implies the equality in (\ref{12-4-17}).

(2) Case of $R \ge F_1$.
When $Q_{V|U=u}$ satisfies Condition \ref{c-2-21-2} for any $u \in \supp(Q_U)$,
the relations
\begin{align}
&\min_{Q_{VU}'}\tilde{E}^l (R, \overline{W}^Z\times Q_{VU}')
=\min_{Q_V'}\tilde{E}^l (R, \overline{W}^Z\times Q_{V}') \nonumber \\
=&\min_{Q_{VU}'}\tilde{E}^{E_0} (R, \overline{W}^Z\times Q_{VU}')
=\min_{Q_V'}\tilde{E}^{E_0} (R, \overline{W}^Z\times Q_{V}')\nonumber \\
=&\tilde{E}^l (R, \overline{W}^Z\times Q_{VU}) 
=\tilde{E}^{E_0} (R, \overline{W}^Z\times Q_{VU}) 
= R-
E_{0,\max}(1 |\overline{W}^Z)
\Label{12-5-21bf}
\end{align}
hold, which implies the equality in (\ref{12-4-17}).
\end{theorem}
Then, we obtain the following two corollaries.
\begin{corollary}
When the channel $W^Z$ is regular
and $Q_{V}$ is the uniform distribution, 
the equality in (\ref{12-4-17}) holds.
\end{corollary}

\begin{proof}
When the channel $W^Z$ is regular,
the uniform distribution over $\mathcal{V}$ satisfies Condition \ref{c-2-22-1} with $\rho$. 
Hence, when $Q_{V}$ is the uniform distribution, 
the equality in (\ref{12-4-17}) holds.
\end{proof}

\begin{corollary}\Label{c-12-24-1}
When 
$R=F_{\rho}$
and $Q_{V|U=u}$ satisfies Condition \ref{c-2-21-2} for any $u \in \supp(Q_U)$,
we have
\begin{align}
\tilde{E}^l (R, \overline{W}^Z\times Q_{VU}) 
= & 
 \tilde{E}^{E_0}(R, \overline{W}^Z\times Q_{VU}) \nonumber \\
\le & 
\tilde{E}^{\psi} (R, \overline{W}^Z\times Q_{VU}).
\nonumber
\end{align}
\end{corollary}

In the above case of Corollary \ref{c-12-24-1}, 
the exponent $\tilde{E}^l (R, \overline{W}^Z\times Q_{VU})$
 cannot improve the exponent
$\tilde{E}^{\psi} (R, \overline{W}^Z\times Q_{VU})$,
which is the exponent 
of the code constructed in the first construction (Subsection \ref{s5-3})
and is given in Subsection \ref{s7-3}.
However, the relation between
$\tilde{E}^l (R, \overline{W}^Z\times Q_{VU})$
and
$\tilde{E}^{\psi} (R, \overline{W}^Z\times Q_{VU})$
remains unknown up to now.

\subsection{Examples}\Label{s12-c}
In this subsection, 
we numerically compare
\begin{align*}
&\tilde{E}^l (R, \overline{W}^Z\times Q_{V}) \\
=&
\min_{{W}^Z \in \mathcal{W}(\mathcal{V},\mathcal{Z})}
D({W}^Z\| \overline{W}^Z| Q_{V})
+
[R-I(V;Z)[W^Z\times Q_{V}] ]_+ 
\end{align*}
and
\begin{align*}
\tilde{E}^{E_0} (R, \overline{W}^Z\times Q_{V})
= &
\max_{0 \le \rho \le 1}
\rho R - E_0(\rho| \overline{W}^Z ,Q_{V}) \\
\tilde{E}^{\psi} (R, \overline{W}^Z\times Q_{V})
&=
\max_{0 \le \rho \le 1} \rho R - \psi(\rho|\overline{W}^Z,Q_{V})
\end{align*}
in the following two examples.

\begin{example}\Label{e1}
In this example,
we address the channel given by a $2\times 2$ general transition matrix. 
Consider the case when 
$\mathcal{Z}=\mathcal{V}=\{1,2\}$.
Define the transition matrix $\overline{W}^{Z}$
by
\begin{align}
\overline{W}^{Z}:=
\left(
\begin{array}{cc}
1-p& q \\
p  & 1-q
\end{array}
\right)
\end{align}
with $p > q \in (0,1/2)$.
When 
$Q_V(1)=1/2$ and
$Q_V(2)=1/2$,
we have
\begin{align}
& E_0(\rho|\overline{W}^{Z},Q_V) 
\nonumber \\
=& \log
(
(\frac{1}{2}(1-p)^{\frac{1}{1-\rho}}
+
\frac{1}{2}q^{\frac{1}{1-\rho}})^{1-\rho}
+
(\frac{1}{2}p^{\frac{1}{1-\rho}}
+
\frac{1}{2}(1-q)^{\frac{1}{1-\rho}})^{1-\rho}
),\\
&\psi(\rho|\overline{W}^{Z},Q_V)
\nonumber \\
=&\log
(
\frac{1}{2}(1-p)^{1+\rho} (\frac{1-p+q}{2})^{-\rho}
+
\frac{1}{2}p^{1+\rho} (\frac{1-q+p}{2})^{-\rho} \nonumber \\
&+
\frac{1}{2}q^{1+\rho} (\frac{1-p+q}{2})^{-\rho}
+
\frac{1}{2}(1-q)^{1+\rho} (\frac{1-q+p}{2})^{-\rho}
).
\end{align}
Fig. \ref{f1} suggests that 
$\tilde{E}^{\psi} (R, \overline{W}^Z\times Q_{V})$
is larger than 
$\tilde{E}^l (R, \overline{W}^Z\times Q_{V})$.
In Fig. \ref{f2},
we numerically calculate 
$\argmax_{0 \le \rho \le 1} \rho R - E_0 (\rho| \overline{W}^Z, Q_{V})$
and
$\argmax_{0 \le \rho \le 1} \rho R - \psi (\rho| \overline{W}^Z, Q_{V})$
which realize
$\tilde{E}^{E_0} (R, \overline{W}^Z\times Q_{V})$
and
$\tilde{E}^{\psi}(R, \overline{W}^Z\times Q_{V})$, respectively.
\end{example}

\begin{figure}[htbp]
\begin{center}
\scalebox{0.7}{\includegraphics[scale=2.5]{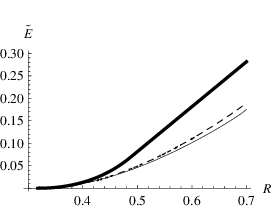}}
\end{center}
\caption{
Lower bounds of exponent in Example \ref{e1} with $p=0.01$ and $q=0.3$. 
In this case, 
$I(V;Z)[\overline{W}^Z\times Q_{V}]=0.317054$.
Thick line, Dashed line, and Normal line
plot 
$\tilde{E}^{\psi} (R, \overline{W}^Z\times Q_{V})$,
$\tilde{E}^l (R, \overline{W}^Z\times Q_{V})$,
and
$\tilde{E}^{E_0} (R, \overline{W}^Z\times Q_{V})$
as functions of $R$ from $R= 0.317054$ to $R= \log 2=0.693147$
with the origin (0.3,0).}
\Label{f1}
\end{figure}%

\begin{figure}[htbp]
\begin{center}
\scalebox{0.7}{\includegraphics[scale=2.5]{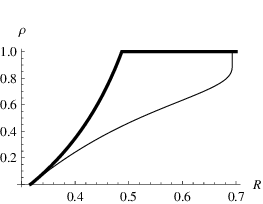}}
\end{center}
\caption{
Relation between $R$ and $\rho$ realizing the optimal value. 
in Example \ref{e1} with $p=0.01$ and $q=0.3$. 
Thick line expresses
$\argmax_{0 \le \rho \le 1} \rho R - \psi (\rho| \overline{W}^Z, Q_{V})$,
which realizes $\tilde{E}^{\psi} (R, \overline{W}^Z\times Q_{V})$.
Normal line expresses  
$\argmax_{0 \le \rho \le 1} \rho R - E_0 (\rho| \overline{W}^Z, Q_{V})$,
which realizes $\tilde{E}^{E_0} (R, \overline{W}^Z\times Q_{V})$.
There is no graph corresponding to $\tilde{E}^l (R, \overline{W}^Z\times Q_{V})$
because 
$\tilde{E}^l (R, \overline{W}^Z\times Q_{V})$ is not given as maximization with respect to
$\rho$.
The origin is (0.3,0).}
\Label{f2}
\end{figure}%

\begin{example}\Label{e2}
In this example, 
we consider the case when
states satisfying Conditions \ref{c-2-22-1} and \ref{c-2-21-2} are not unique.
Consider the case when 
$\mathcal{Z}=\mathcal{V}=\{1,2,3,4\}$.
Define the transition matrix $\overline{W}^{Z}$
by
\begin{align}
\overline{W}^{Z}:=
\left(
\begin{array}{cccc}
\frac{1}{2}-p& p& \frac{1}{2}-p& p\\
p& \frac{1}{2}-p& p& \frac{1}{2}-p\\
\frac{1}{2}-p& p& p& \frac{1}{2}-p\\
p& \frac{1}{2}-p& \frac{1}{2}-p& p
\end{array}
\right)
\end{align}
with $p \in (0,1/4)$.
When 
$Q_V(1)=q$,
$Q_V(2)=q$,
$Q_V(3)=\frac{1}{2}-q$,
and 
$Q_V(4)=\frac{1}{2}-q$,
we have
\begin{align}
& \sum_z \overline{W}^Z(z|v)^{\frac{1}{1-\rho}}
(\sum_{v'} Q_{V}(v') 
\overline{W}^Z(z|v')^{\frac{1}{1-\rho}})^{-\rho}
\nonumber \\
=&
4 
(
\frac{1}{2}
(\frac{1}{2}-p)^{\frac{1}{1-\rho}}
+
\frac{1}{2}
p^{\frac{1}{1-\rho}}
)^{1-\rho} 
=
2^{1+\rho}
(
(\frac{1}{2}-p)^{\frac{1}{1-\rho}}
+
p^{\frac{1}{1-\rho}}
)^{1-\rho}.
\end{align}
for all $v \in \mathcal{V}$, which implies Condition \ref{c-2-22-1}.
Hence, 
\begin{align}
& E_{0,\max}(\rho |\overline{W}^Z) 
=
E_0(\rho|\overline{W}^{Z},Q_V)
\nonumber \\
=&
(1+\rho)\log 2
+(1-\rho)\log(
(\frac{1}{2}-p)^{\frac{1}{1-\rho}}
+
p^{\frac{1}{1-\rho}}
) ,\\
& F_\rho = F_\rho(Q_V)
\nonumber \\
=&
\log 2
-\log(
(\frac{1}{2}-p)^{\frac{1}{1-\rho}}
+
p^{\frac{1}{1-\rho}}
) 
\nonumber \\
&+
\frac{1}{1-\rho}
\frac
{
(\frac{1}{2}-p)^{\frac{1}{1-\rho}}
\log(\frac{1}{2}-p) 
+
p^{\frac{1}{1-\rho}}
\log p}
{
(\frac{1}{2}-p)^{\frac{1}{1-\rho}}
+
p^{\frac{1}{1-\rho}}} ,\\
&\psi(\rho|\overline{W}^{Z},Q_V) =
(2\rho+1) \log 2
+\log ((\frac{1}{2}-p)^{1+\rho}+p^{1+\rho} ).
\end{align}

Next, we check Condition \ref{c-2-21-2}.
For this purpose, we check Condition (i) in Lemma \ref{L2-21-1}
by treating $\mathcal{V}_z$ given in \eqref{6-4-1}.
Since
$\mathcal{V}_1=\{1,3 \}$,
$\mathcal{V}_2=\{2,4 \}$,
$\mathcal{V}_3=\{1,4 \}$,
and
$\mathcal{V}_4=\{2,3 \}$,
in the above choice of $Q_V$,
we have
$\sum_{v''\in \mathcal{V}_z }  Q_{V}(v'') 
=\frac{1}{2}$,
which implies
\begin{align}
\sum_{z \in \mathcal{Z}_v}
\frac{
\max_{v'\in \mathcal{V}}\overline{W}^Z(z|v')}
{\sum_{v''\in \mathcal{V}_z }  Q_{V}(v'') 
}
=
2 \frac{\frac{1}{2}-p}{\frac{1}{2}}
=4(\frac{1}{2}-p)
\end{align}
for all $v \in \mathcal{V}$.
Thus, Condition \ref{c-2-21-2} holds.
Hence,
\begin{align}
E_{0,\max}(1|\overline{W}^Z)
&=
\log 4(\frac{1}{2}-p) \\
F_1 &=
\log 2.
\end{align}
Further, Theorem \ref{t-12-20-1} guarantees that
$\tilde{E}^{E_0} (R, \overline{W}^Z\times Q_{V})
=\tilde{E}^l (R, \overline{W}^Z\times Q_{V})$.
So, we numerically compare only
$\tilde{E}^{\psi} (R, \overline{W}^Z\times Q_{V})$
and
$\tilde{E}^{E_0} (R, \overline{W}^Z\times Q_{V})$
in Fig. \ref{f3}.
Since
$\tilde{E}^{E_0} (R, \overline{W}^Z\times Q_{V})$ attains the
minimum value due to Theorem \ref{t-12-20-1},
$\tilde{E}^{E_0} (R, \overline{W}^Z\times Q_{V})$ does not depend on $q$.
Further,
$\tilde{E}^{\psi}(R, \overline{W}^Z\times Q_{V})$
also does not depend on $q$
due to the form of $\tilde{E}^{\psi}(R, \overline{W}^Z\times Q_{V})$.
Similar to Fig. \ref{f2}, 
Fig. \ref{f4} suggests that
the parameter $\rho$ realizing 
$\tilde{E}^{E_0} (R, \overline{W}^Z\times Q_{V})$
has a behavior different from 
the parameter $\rho$ realizing 
$\tilde{E}^{\psi}(R, \overline{W}^Z\times Q_{V})$.
\end{example}

\begin{figure}[htbp]
\begin{center}
\scalebox{0.7}{\includegraphics[scale=2.5]{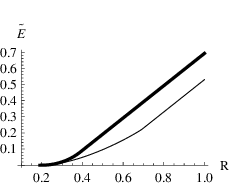}}
\end{center}
\caption{Lower bounds of exponent in Example \ref{e2} with $p=0.1$. 
In this case, 
$I(V;Z)[\overline{W}^Z\times Q_{V}]=0.192745$.
Thick line 
and
Normal line
express 
$\tilde{E}^{\psi} (R, \overline{W}^Z\times Q_{V})$
and $\tilde{E}^{E_0} (R, \overline{W}^Z\times Q_{V})=\tilde{E}^l (R, \overline{W}^Z\times Q_{V})$
as functions of $R$ from $R=0.192745$ to $R=1.0$
with the origin (0.1,0).
Thick line is straight when $R \ge 0.4$ because $\argmax_{0 \le \rho \le 1} \rho R - \psi (\rho| \overline{W}^Z, Q_{V})$ is 1 when $R \ge 0.4$,
as in Fig \ref{f4}.
Normal line is straight when $R \ge 0.7$ because $\argmax_{0 \le \rho \le 1} \rho R - E_0 (\rho| \overline{W}^Z, Q_{V})$ is 1 when $R \ge 0.7$,
as in Fig \ref{f4}.}
\Label{f3}
\end{figure}%

\begin{figure}[htbp]
\begin{center}
\scalebox{0.7}{\includegraphics[scale=2.5]{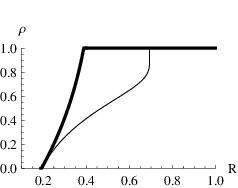}}
\end{center}
\caption{
Relation between $R$ and $\rho$ realizing the optimal value in Example \ref{e2} with $p=0.1$.
Normal line expresses  
$\argmax_{0 \le \rho \le 1} \rho R - E_0 (\rho| \overline{W}^Z, Q_{V})$,
which realizes $\tilde{E}^{E_0} (R, \overline{W}^Z\times Q_{V})$.
Thick line expresses
$\argmax_{0 \le \rho \le 1} \rho R - \psi (\rho| \overline{W}^Z, Q_{V})$,
which realizes $\tilde{E}^{\psi} (R, \overline{W}^Z\times Q_{V})$.
There is no graph corresponding to $\tilde{E}^l (R, \overline{W}^Z\times Q_{V})$
because 
$\tilde{E}^l (R, \overline{W}^Z\times Q_{V})$ is not given as maximization with respect to
$\rho$.
The origin is (0.1,0).}
\Label{f4}
\end{figure}%

\subsection{Proof of Lemma \ref{L-2-23-9}}\Label{s12-d}
\begin{proof}
We can show (\ref{2-23-8}) and (\ref{12-20-12}) by direct calculations.
Now, we show (\ref{12-20-12}).
In general,
when $b_i >0$ and $a_1=a_2=\ldots=a_l > a_i>0$ for $i=l+1,\ldots, k$,
the relation 
\begin{align}
&\lim_{\rho \to 1-0}(\sum_{i=1}^k b_i a_{i}^{\frac{1}{1-\rho}})^{1-\rho}
\nonumber \\
=&
\lim_{\rho \to 1-0}( (\sum_{i=1}^l b_i) a_{1}^{\frac{1}{1-\rho}})^{1-\rho}
(1+ \sum_{i=l+1}^k \frac{b_i}{\sum_{i=1}^l b_i} \frac{a_{i}}{a_1}^{\frac{1}{1-\rho}})^{1-\rho} \nonumber \\
=&
\lim_{\rho \to 1-0}( (\sum_{i=1}^l b_i) a_{1}^{\frac{1}{1-\rho}})^{1-\rho}
=
a_1 \Label{2-13-1}
\end{align}
holds. 
That is,
the difference 
$(\sum_{i=1}^k b_i a_{i}^{\frac{1}{1-\rho}})^{1-\rho}
-( (\sum_{i=1}^l b_i) a_{1}^{\frac{1}{1-\rho}})^{1-\rho}$
behaves as $O(\exp(-\frac{a}{1-\rho}))$
with a constant $a$.
Applying the above general discussion, 
we have
\begin{align*}
& 
\lim_{\rho \to 1-0}
\sum_{u}Q_U(u) \sum_z 
\biggl[\sum_{v } Q_{V|U}(v|u) \overline{W}^Z(z|v)^{\frac{1}{1-\rho}}\biggr]^{1-\rho} \\
= &
\lim_{\rho \to 1-0}
\sum_{u}Q_U(u) \sum_z 
\Biggl[
\sum_{v \in \mathcal{V}_z(Q_{V|U=u})} Q_{V|U}(v|u) 
\\
& \hspace{25ex} \cdot \biggl(
\max_{v \in \supp (Q_{V|U=u}) } \overline{W}^Z(z|v)
\biggr)^{\frac{1}{1-\rho}}
\Biggr]^{1-\rho} \\
=&
\lim_{\rho \to 1-0}
\sum_{u}Q_U(u) \sum_z 
\Biggl[
\biggl(
\sum_{v \in \mathcal{V}_z(Q_{V|U=u})} Q_{V|U}(v|u) 
\biggr)^{1-\rho} 
\\
&\hspace{25ex} \cdot 
\biggl(
\max_{v \in \supp (Q_{V|U=u}) } \overline{W}^Z(z|v)
\biggr)\Biggr] \\
= & 
\sum_{u}Q_U(u) \sum_z 
(\max_{v \in \supp (Q_{V|U=u}) } \overline{W}^Z(z|v)) .
\end{align*}
where $\mathcal{V}_z(Q_{V|U=u})
:= \{v \in \supp (Q_{V|U=u})|  \max_{v \in \supp (Q_{V|U=u}) } \overline{W}^Z(z|v)\}$.
Hence, we obtain (\ref{12-20-12}).

Further, since $x \mapsto - \log x$ is concave, 
the map $Q_V \mapsto F_1(Q_{V})$ is concave.
The remaining task is the poof of the equation (\ref{12-5-10}), will be shown in the wide space style in the next page.
\end{proof}

\begin{figure*}[!t]
\begin{proofof}{(\ref{12-5-10})}
We have
\begin{align}
&\frac{d}{d\rho} E_0(\rho| \overline{W}^Z ,Q_{V|U},Q_U) \nonumber\\
=&
\frac{
\sum_u Q_U(u) \sum_z 
(\sum_v \frac{1}{1-\rho } (\log \overline{W}^Z(z|v) )
Q_{V|U}(v|u) \overline{W}^Z(z|v)^{\frac{1}{1-\rho}})
( \sum_v Q_{V|U}(v|u) \overline{W}^Z(z|v)^{\frac{1}{1-\rho}} )^{-\rho}
}{\sum_u Q_U(u)\sum_z (\sum_v Q_{V|U}(v|u) \overline{W}^Z(z|v)^{\frac{1}{1-\rho}} )^{1-\rho}}\nonumber\\
&-\frac{
\sum_u Q_U(u) \sum_z
\log (\sum_v Q_{V|U}(v|u) \overline{W}^Z(z|v)^{\frac{1}{1-\rho}})
( \sum_v Q_{V|U}(v|u) \overline{W}^Z(z|v)^{\frac{1}{1-\rho}} )^{1-\rho}
}{\sum_u Q_U(u)\sum_z (\sum_v Q_{V|U}(v|u) \overline{W}^Z(z|v)^{\frac{1}{1-\rho}} )^{1-\rho}}. \nonumber
\end{align}
When $\rho$ approaches $1$,
$\sum_v Q_{V|U}(v|u) \overline{W}^Z(z|v)^{\frac{1}{1-\rho}}$
approaches
$(\sum_{v\in \mathcal{V}_z} Q_{V|U}(v|u)) 
(\max_{v'}\overline{W}^Z(z|v'))^{\frac{1}{1-\rho}}$.
Hence, 
\begin{align}
&
\lim_{\rho \to 1-0}
\frac{d}{d\rho} E_0(\rho| \overline{W}^Z ,Q_{V|U},Q_U) \nonumber\\
= &
\lim_{\rho \to 1-0}
\Bigl(
\frac{
\sum_u Q_U(u) \sum_z 
(
\frac{1}{1-\rho } 
\log \max_{v'} \overline{W}^Z(z|v') 
(\sum_{v \in \mathcal{V}_z} Q_{V|U}(v|u))^{1-\rho}
\max_{v'} \overline{W}^Z(z|v') )
}
{\sum_u Q_U(u)\sum_z 
(\sum_{v \in \mathcal{V}_z} Q_{V|U}(v|u))^{1-\rho}
\max_{v'} \overline{W}^Z(z|v') 
}
\nonumber\\
&-\frac{
\sum_u Q_U(u) \sum_z
(\frac{1}{1-\rho } \log \max_{v'} \overline{W}^Z(z|v') 
+\log (\sum_{v \in \mathcal{V}_z} Q_{V|U}(v|u)))
(\sum_{v \in \mathcal{V}_z} Q_{V|U}(v|u))^{1-\rho}
\max_{v'} \overline{W}^Z(z|v') 
}
{\sum_u Q_U(u)\sum_z 
(\sum_{v \in \mathcal{V}_z} Q_{V|U}(v|u))^{1-\rho}
\max_{v'} \overline{W}^Z(z|v') 
}
\Bigr)
\nonumber\\
=&
\lim_{\rho \to 1-0}
\frac{
-\sum_u Q_U(u) \sum_z
\log (\sum_{v \in \mathcal{V}_z} Q_{V|U}(v|u))
(\sum_{v \in \mathcal{V}_z} Q_{V|U}(v|u))^{1-\rho}
\max_{v'} \overline{W}^Z(z|v') 
}
{\sum_u Q_U(u)\sum_z 
(\sum_{v \in \mathcal{V}_z} Q_{V|U}(v|u))^{1-\rho}
\max_{v'} \overline{W}^Z(z|v') 
}\nonumber\\
= &
\lim_{\rho \to 1-0}
-\frac{
\sum_u Q_U(u) \sum_z
\log (\sum_{v \in \mathcal{V}_z} Q_{V|U}(v|u))
\max_{v'} \overline{W}^Z(z|v') 
}
{\sum_u Q_U(u)\sum_z 
\max_{v'} \overline{W}^Z(z|v') 
},\Label{12-5-15}
\end{align}
which implies (\ref{12-5-10}).
\end{proofof}
\end{figure*}

\subsection{Proof of Lemma \ref{L-2-23-10}}\Label{s12-e}
\begin{proof}
Due to (\ref{12-20-12}), we have
\begin{align*}
E_{0,\max}(1|\overline{W}^Z)
=&
\max_{Q_{VU}'}\lim_{\rho \to 1-0} 
E_0(\rho| \overline{W}^Z ,Q_{V|U}',Q_U ') \\
=&
\max_{Q_{VU}}
\log 
\sum_{u}Q_U(u)
\sum_{z} \max_{v \in \supp(Q_{V|U=u}) }\overline{W}^Z(z|v) \\
=&
\log \sum_{z} \max_{v}\overline{W}^Z(z|v),
\end{align*}
which implies (\ref{12-20-13}).

Assume that
the support of $Q_{V|U=u}$ contains
$\{ v \in \mathcal{V}|
\min_{z}
\frac{\max_{v'}\overline{W}^Z(z|v')}{\overline{W}^Z(z|v)}=1
\}$ for any $u \in \supp (Q_U)$.
Due to (\ref{12-20-12}), we have
\begin{align}
E_0(1| \overline{W}^Z ,Q_{V|U},Q_U ) 
=\log \sum_{z} \max_{v}\overline{W}^Z(z|v).
\end{align}
Combining (\ref{12-20-13}), we obtain (\ref{12-4-20}).
Hence, as a special case of (\ref{12-20-11}), we obtain (\ref{12-4-20-B}).
\end{proof}

\subsection{Proofs of Lemmas \ref{L2-22-1}, \ref{L2-24-1}, and \ref{L2-22-4}}
\Label{s12-3-5}

\begin{lemma}\Label{L2-24-3}
Let $f$ be a concave $C^1$ function from $\bR^d$ to $\bR$
and $\mathcal{P}(d)$ be the subset $\{(x_1,\ldots,x_d)\in\bR^d| x_i \ge 0, \sum_{i=1}^d x_i=1 \}$.
The following two conditions for $x=(x_1,\ldots,x_d)\in \mathcal{P}(d)$ are equivalent.
\begin{itemize}
\item[(i)]
\begin{align}
f(x)=\max_{x' \in \mathcal{P}(d)} f(x').
\end{align}
\item[(ii)]
The following relation holds for any $i \neq j$.
\begin{align}
\frac{\partial}{\partial x^i} f(x)=\frac{\partial}{\partial x^i} f(x).
\end{align}
\end{itemize}
\end{lemma}
\begin{proofof}{Lemma \ref{L2-24-3}}
We choose variable $y=(y_1, \ldots y_{d-1}) \in \bR^{d-1}$,
and define a function $\tilde{f}(y):=f(y_1, \ldots, y_{d-1}, 1- \sum_{i=1}^{d-1}y_i)$.
Due to the concavity, the condition (i) holds if and only if
$\frac{\partial }{\partial y_i}\tilde{f}(y)=0$ for $i=1, \ldots, d-1$.
This condition is equivalent to the condition (ii)
because 
$\frac{\partial }{\partial y_i}\tilde{f}(y)=
\frac{\partial }{\partial x_i}f(y_1, \ldots, y_{d-1}, 1- \sum_{i=1}^{d-1}y_i)
-\frac{\partial }{\partial x_d}f(y_1, \ldots, y_{d-1}, 1- \sum_{i=1}^{d-1}y_i)$.
\end{proofof}

\begin{proofof}{Lemma \ref{L2-22-1}}
In order to apply Lemma \ref{L2-24-3},
we regard all of probabilities $Q_V(v)$ as independent parameters
by removing the constraint $\sum_{v}Q_V(v)=1$.
The partial derivatives are calculated as
\begin{align*}
&\frac{\partial }{\partial Q_V(v)}
\sum_z
(\sum_{v'} Q_{V}(v') \overline{W}^Z(z|v')^{\frac{1}{1-\rho}})^{1-\rho}\\
=&
\sum_z
(1-\rho)
(\sum_{v'} Q_{V}(v') \overline{W}^Z(z|v')^{\frac{1}{1-\rho}})^{-\rho}
\overline{W}^Z(z|v)^{\frac{1}{1-\rho}}.
\end{align*}
Hence, Lemma \ref{L2-24-3} guarantees the equivalence between (i) and (ii).
Condition (iii) trivially implies Condition (i).

The remaining task is showing Condition (i) implies Condition (iii).
Assume Condition (i).
Since 
$\sum_z \overline{W}^Z(z|v)^{\frac{1}{1-\rho}}
(\sum_{v'} Q_{V}(v') 
\overline{W}^Z(z|v')^{\frac{1}{1-\rho}})^{-\rho}$ does not depend on $v$
and Condition (ii) holds, 
\begin{align*}
& \sum_z \overline{W}^Z(z|v)^{\frac{1}{1-\rho}}
(\sum_{v'} Q_{V}(v') 
\overline{W}^Z(z|v')^{\frac{1}{1-\rho}})^{-\rho}
\\
=&
\sum_{v}Q_V(v)
\sum_z \overline{W}^Z(z|v)^{\frac{1}{1-\rho}}
(\sum_{v'} Q_{V}(v') 
\overline{W}^Z(z|v')^{\frac{1}{1-\rho}})^{-\rho} \\
= &
\sum_z
(\sum_{v} Q_{V}(v) 
\overline{W}^Z(z|v)^{\frac{1}{1-\rho}})^{1-\rho}
=
e^{E_0(\rho| \overline{W}^Z ,Q_{V})} \\
=& \max_{Q_{V}'} 
e^{E_0(\rho| \overline{W}^Z ,Q_{V}')}
= e^{E_{0,\max}(\rho |\overline{W}^Z)}.
\end{align*}
\end{proofof}

\begin{proofof}{Lemma \ref{L2-24-1}}
Assume that
\begin{align}
\sum_{v} Q_{V}(v) 
\overline{W}^Z(z|v)^{\frac{1}{1-\rho}}
\neq
\sum_{v} Q_{V}'(v) 
\overline{W}^Z(z|v)^{\frac{1}{1-\rho}}
\end{align}
for any $z \in \mathcal{Z}$.
Due to the strict concavity of $x \mapsto x^{1-\rho}$,
we have
\begin{align}
& \frac{1}{2}
(\sum_{v} Q_{V}(v) 
\overline{W}^Z(z|v)^{\frac{1}{1-\rho}})^{1-\rho}
+
\frac{1}{2}(\sum_{v} Q_{V}'(v) 
\overline{W}^Z(z|v)^{\frac{1}{1-\rho}})^{1-\rho}
\nonumber \\
<&
(\sum_{v}(\frac{1}{2}Q_{V}(v) +\frac{1}{2}Q_{V}'(v) )
\overline{W}^Z(z|v)^{\frac{1}{1-\rho}})^{1-\rho}.
\end{align}
Hence, 
\begin{align}
&\frac{1}{2}
\sum_z
(\sum_{v} Q_{V}(v) 
\overline{W}^Z(z|v)^{\frac{1}{1-\rho}})^{1-\rho}
+
\frac{1}{2}
\sum_z
(\sum_{v} Q_{V}'(v) 
\overline{W}^Z(z|v)^{\frac{1}{1-\rho}})^{1-\rho} 
\nonumber \\
&<
\sum_z
(\sum_{v} 
(\frac{1}{2} Q_{V}(v)+\frac{1}{2} Q_{V}'(v) )
\overline{W}^Z(z|v)^{\frac{1}{1-\rho}})^{1-\rho} .
\Label{2-24-3}
\end{align}
However,
Lemma \ref{L2-22-1} guarantees that 
\begin{align}
\sum_z
(\sum_{v} Q_{V}(v) 
\overline{W}^Z(z|v)^{\frac{1}{1-\rho}})^{1-\rho}
&=
\sum_z
(\sum_{v} Q_{V}'(v) 
\overline{W}^Z(z|v)^{\frac{1}{1-\rho}})^{1-\rho}
\nonumber \\
&=\max_{Q_{V}'} 
e^{E_0(\rho| \overline{W}^Z ,Q_{V}')}.
\Label{2-24-4}
\end{align}
Since (\ref{2-24-3}) contradicts (\ref{2-24-4}),
we obtain the desired argument.
\end{proofof}

\begin{proofof}{Lemma \ref{L2-22-4}}
As
\begin{align}
W^{Z,\rho} \circ Q_{V,\rho} (z)=
\frac
{
(\sum_{v}Q_{V,\rho}(v)\overline{W}^Z(z|v)^{\frac{1}{1-\rho}})^{1-\rho}
}
{
\sum_z (\sum_{v}Q_{V,\rho}(v)\overline{W}^Z(z|v)^{\frac{1}{1-\rho}})^{1-\rho}
},\nonumber
\end{align}
we can calculate the mutual information $
I(V;Z)[W^{Z,\rho} \times Q_{V,\rho}]$ as 
\begin{align}
&I(V;Z)[W^{Z,\rho} \times Q_{V,\rho}] \nonumber\\
=&\sum_{v,z}
\frac
{Q_{V,\rho}(v)
\overline{W}^Z(z|v)^{\frac{1}{1-\rho}}
(\sum_{v}Q_{V,\rho}(v)\overline{W}^Z(z|v)^{\frac{1}{1-\rho}})^{-\rho}
}
{\sum_z
\overline{W}^Z(z|v)^{\frac{1}{1-\rho}}
(\sum_{v}Q_{V,\rho}(v)\overline{W}^Z(z|v)^{\frac{1}{1-\rho}})^{-\rho}
}
\nonumber\\
&\hspace{4ex} \cdot \Biggl[
\log \Bigl[\overline{W}^Z(z|v)^{\frac{1}{1-\rho}}
\bigl(\sum_{v}Q_{V,\rho}(v)\overline{W}^Z(z|v)^{\frac{1}{1-\rho}}\bigr)^{-\rho}\Bigr]
\nonumber\\
&\hspace{20ex}
-\log \Bigl[\Bigl(\sum_{v}Q_{V,\rho}(v)\overline{W}^Z(z|v)^{\frac{1}{1-\rho}}\Bigr)^{1-\rho}
\Bigr]
 \Biggr]\nonumber\\
=&\sum_{v,z}
\frac
{Q_{V,\rho}(v)
\overline{W}^Z(z|v)^{\frac{1}{1-\rho}}
(\sum_{v}Q_{V,\rho}(v)\overline{W}^Z(z|v)^{\frac{1}{1-\rho}})^{-\rho}
}
{\sum_z
\overline{W}^Z(z|v)^{\frac{1}{1-\rho}}
(\sum_{v}Q_{V,\rho}(v)\overline{W}^Z(z|v)^{\frac{1}{1-\rho}})^{-\rho}
}\nonumber \\
& \hspace{5ex}\cdot
\Biggl[\frac{1}{1-\rho} \log \overline{W}^Z(z|v)
- \log \Bigl[\sum_{v}Q_{V,\rho}(v)\overline{W}^Z(z|v)^{\frac{1}{1-\rho}}\Bigr] \Biggr]\nonumber\\
=&
F_\rho(Q_{V,\rho}),
\end{align}
where the final equation follows from (\ref{2-23-9b}).
We obtain the second equation of (\ref{2-23-2}).

Since the constraint (i) in Lemma \ref{L2-22-1} for $Q_{V,\rho}$ is differentiable with respect to $\rho$,
for a given $\rho_0 \in (0,1)$,
we can choose $Q_{V,\rho}$ such that 
the map $\rho \mapsto Q_{V,\rho}$ is differentiable at least in an enough small neighborhood of $\rho_0$.
Since
\begin{align}
\frac{d}{d \rho}
E_0(\rho_0 |\overline{W}^Z,Q_{V,\rho})|_{\rho=\rho_0} =0,
\end{align}
we have
\begin{align}
& F_{\rho_0}
=
\frac{d}{d \rho}
E_0(\rho |\overline{W}^Z,Q_{V,\rho})|_{\rho=\rho_0} \nonumber \\
=&
\frac{d}{d \rho}
E_0(\rho |\overline{W}^Z,Q_{V,\rho_0})|_{\rho=\rho_0} 
+
\frac{d}{d \rho}
E_0(\rho_0 |\overline{W}^Z,Q_{V,\rho})|_{\rho=\rho_0} \nonumber \\
=&
\frac{d}{d \rho}
E_0(\rho |\overline{W}^Z,Q_{V,\rho_0})|_{\rho=\rho_0} 
=
F_{\rho_0}(Q_{V,\rho_0}).
\Label{1-24-6}
\end{align}
Hence, we obtain the first equation of (\ref{2-23-2}).

The conditional divergence $D(W^Z \| \overline{W}^Z| Q_{V,\rho})$
is calculated to
\begin{align}
& D(W^{V,\rho} \| \overline{W}^Z| Q_{V,\rho}) \nonumber\\
=&\sum_{v,z}
\frac
{Q_{V,\rho}(v)
\overline{W}^Z(z|v)^{\frac{1}{1-\rho}}
(\sum_{v}Q_{V,\rho}(v)\overline{W}^Z(z|v)^{\frac{1}{1-\rho}})^{-\rho}
}
{\sum_z
\overline{W}^Z(z|v)^{\frac{1}{1-\rho}}
(\sum_{v'}Q_{V,\rho}(v')\overline{W}^Z(z|v)^{\frac{1}{1-\rho}})^{-\rho}
}
\nonumber \\
& \hspace{2ex}\cdot
 \Biggl(\log 
\Bigl[\overline{W}^Z(z|v)^{\frac{1}{1-\rho}}
(\sum_{v}Q_{V,\rho}(v)\overline{W}^Z(z|v)^{\frac{1}{1-\rho}})^{-\rho}\Bigr]
-
\log \overline{W}^Z(z|v) \Biggr)\nonumber\\
&- \sum_{v}
Q_{V,\rho}(v)
\log \Biggl[ \sum_z
\overline{W}^Z(z|v)^{\frac{1}{1-\rho}}
\Bigl(\sum_{v'}Q_{V,\rho}(v')\overline{W}^Z(z|v)^{\frac{1}{1-\rho}}\Bigr)^{-\rho}
\Biggr] \nonumber\\
=&
\sum_{v,z}
\frac
{Q_{V,\rho}(v)
\overline{W}^Z(z|v)^{\frac{1}{1-\rho}}
(\sum_{v}Q_{V,\rho}(v)\overline{W}^Z(z|v)^{\frac{1}{1-\rho}})^{-\rho}
}
{\sum_z
\overline{W}^Z(z|v)^{\frac{1}{1-\rho}}
(\sum_{v}Q_{V,\rho}(v)\overline{W}^Z(z|v)^{\frac{1}{1-\rho}})^{-\rho}
}\nonumber \\
&\hspace{5ex} \cdot \Biggl(
\frac{\rho}{1-\rho}\log \overline{W}^Z(z|v) 
-\rho \log \Bigl[ 
\sum_{v}Q_{V,\rho}(v)\overline{W}^Z(z|v)^{\frac{1}{1-\rho}}
\Bigr]
\Biggr)
\nonumber\\
&- \sum_{v}
Q_{V,\rho}(v)
\log \Biggl[\sum_z
\overline{W}^Z(z|v)^{\frac{1}{1-\rho}}
\Bigl(\sum_{v'}Q_{V,\rho}(v')\overline{W}^Z(z|v)^{\frac{1}{1-\rho}}\Bigr)^{-\rho}
\Biggr] \nonumber\\
=& \rho F_\rho(Q_{V,\rho})
- \sum_{v}
Q_{V,\rho}(v)
\log \Biggl[
\sum_z
\Bigl(\sum_{v'}Q_{V,\rho}(v')\overline{W}^Z(z|v)^{\frac{1}{1-\rho}}\Bigr)^{1-\rho}
\Biggr] \nonumber\\
=& \rho F_\rho - E(\rho|\overline{W}^Z,Q_{V,\rho}).
\nonumber
\end{align}
We obtain (\ref{2-23-1}).
\end{proofof}

\subsection{Proofs of Lemmas \ref{L2-21-1}, \ref{L2-24-2}, and \ref{L2-22-5}}
\Label{s12-7}
\begin{proofof}{Lemma \ref{L2-21-1}}
In order to apply Lemma \ref{L2-24-3},
we regard all of probabilities $Q_V(v)$ as independent parameters
by removing the constraint $\sum_{v}Q_V(v)=1$.
The partial derivatives are calculated as
\begin{align*}
&\frac{\partial }{\partial Q_V(v)}
-\frac{\sum_z \log (\sum_{v \in \mathcal{V}_z} Q_{V}(v))
\max_{v'} \overline{W}^Z(z|v') }
{\sum_z \max_{v'} \overline{W}^Z(z|v') }
\\
=&-
\sum_{z \in \mathcal{Z}_v}
\frac{
\max_{v'\in \mathcal{V}}\overline{W}^Z(z|v')}
{\sum_{v''\in \mathcal{V}_z }  Q_{V}(v'') 
}.
\end{align*}
Hence, Lemma \ref{L2-24-3} guarantees the equivalence between (i) and (ii).
Condition (iii) trivially implies Condition (i).

The remaining task is showing Condition (i) implies Condition (iii).
Assume Condition (i).
Since 
$\sum_z \overline{W}^Z(z|v)^{\frac{1}{1-\rho}}
(\sum_{v'} Q_{V}(v') 
\overline{W}^Z(z|v')^{\frac{1}{1-\rho}})^{-\rho}$ does not depend on $v$
and Condition (ii) holds, 
we have
\begin{align*}
&\sum_{z \in \mathcal{Z}_v}
\frac{\overline{W}^Z(z|v)}
{\sum_{v''\in \mathcal{V}_z }  Q_{V}(v'') 
}
=
\sum_{z \in \mathcal{Z}_v}
\frac{
\max_{v'\in \mathcal{V}}\overline{W}^Z(z|v')}
{\sum_{v''\in \mathcal{V}_z }  Q_{V}(v'') 
} \\
=&
\sum_{v}Q_V(v)
\sum_{z \in \mathcal{Z}_v}
\frac{
\max_{v'\in \mathcal{V}}\overline{W}^Z(z|v')}
{\sum_{v''\in \mathcal{V}_z }  Q_{V}(v'') 
} \\
=&
\sum_{(z,v) \in \mathcal{K}}
Q_V(v)
\frac{
\max_{v'\in \mathcal{V}}\overline{W}^Z(z|v')}
{\sum_{v''\in \mathcal{V}_z }  Q_{V}(v'') 
}\\
=&
\sum_{z}
\sum_{v \in \mathcal{V}_z} Q_V(v)
\frac{
\max_{v'\in \mathcal{V}}\overline{W}^Z(z|v')}
{\sum_{v''\in \mathcal{V}_z }  Q_{V}(v'') 
}
=\sum_{z} \max_{v'}\overline{W}^Z(z|v').
\end{align*}
\end{proofof}

\begin{proofof}{Lemma \ref{L2-24-2}}
We focus on the function
$
\{\sum_{v''\in \mathcal{V}_z }  Q_{V}(v'')\}_{z}
\mapsto 
-\frac{\sum_z \log (\sum_{v \in \mathcal{V}_z} Q_{V}(v))
\max_{v'} \overline{W}^Z(z|v') }
{\sum_z \max_{v'} \overline{W}^Z(z|v') }$,
which is strictly concave.
Hence, 
when 
there exists an element $z \in \mathcal{Z}$ 
such that
$\sum_{v''\in \mathcal{V}_z }  Q_{V}(v'')
\neq \sum_{v''\in \mathcal{V}_z }  Q_{V}'(v'')$
for two distributions $Q_V$ and $Q_V'$,
the convex combination $\frac{Q_V+Q_V'}{2}$
gives a strictly greater value for the above function,
which contradicts (ii) of Lemma \ref{L2-21-1}.
Hence, 
$\sum_{v''\in \mathcal{V}_z }  Q_{V}(v'')
= \sum_{v''\in \mathcal{V}_z }  Q_{V}'(v'')$ for all $z \in \mathcal{Z}$.
\end{proofof}

\begin{proofof}{Lemma \ref{L2-22-5}}
Since
\begin{align}
W^{Z,1} \times Q_{V,1} (v,z)
=&
\left\{
\begin{array}{ll}
\frac{Q_{V,1}(v)\overline{W}^Z(z|v)}
{\sum_{v''\in \mathcal{V}_z }  Q_{V,1}(v'') \sum_{z'} \max_{v'}\overline{W}^Z(z'|v')}
& z \in \mathcal{Z}_v \\ 
0& z \in \mathcal{Z}_v ^c,
\end{array}
\right.
\Label{2-13-2}
\end{align}
the mutual information $ I(V;Z)[W^{Z,1} \times Q_{V,1}] $ 
is calculated as
\begin{align}
I(V;Z)[W^{Z,1} \times Q_{V,1}] 
=&
-\frac{
\sum_z
\log (\sum_{v \in \mathcal{V}_z} Q_{V,1}(v))
\max_{v'} \overline{W}^Z(z|v') 
}
{\sum_z \max_{v'} \overline{W}^Z(z|v') } \nonumber \\
=&
F_1(Q_{V,1}),
\end{align}
where the final equation follows from (\ref{12-5-10b}).
Hence, we obtain the second equation in (\ref{2-23-4}).
The first equation in (\ref{2-23-4}) follows from the limit $\rho\to 1-0$ at (\ref{1-24-6}).

When $Q_V$ satisfies Condition \ref{c-2-21-2},
\begin{align}
& D(W^{Z,1}\|\overline{W}^Z|Q_{V} )
\nonumber \\
=&
- 
\sum_{z,v}
W^{Z,1} \times Q_{V,1} (v,z)
\log 
\Bigl[
\sum_{v''\in \mathcal{V}_z }  Q_{V}(v'') \sum_{z'} \max_{v'}\overline{W}^Z(z'|v')\Bigr] \nonumber\\
=&
- \log \Bigl[\sum_{z'} \max_{v'}\overline{W}^Z(z'|v')\Bigr]
\nonumber\\
&\hspace{15ex}
-\sum_{z}
\log 
\Bigl[\sum_{v''\in \mathcal{V}_z }  Q_{V}(v'') \Bigr]
W^{Z,1} \circ Q_{V} (z)
\nonumber\\
=&
- \log \Bigl[\sum_{z'} \max_{v'}\overline{W}^Z(z'|v')\Bigr]
\nonumber\\
&\hspace{15ex}
-\frac{
\sum_z
\log \Bigl[\sum_{v \in \mathcal{V}_z} Q_{V}(v)\Bigr]
\max_{v'} \overline{W}^Z(z|v') 
}
{\sum_z \max_{v'} \overline{W}^Z(z|v') }\nonumber \\
=& F_1- E_{0,\max}(1|\overline{W}^Z),\nonumber
\end{align}
which implies (\ref{2-23-5}).
\end{proofof}

\subsection{Proof of Lemma \ref{L2-21-3}}\Label{s12-h}
\begin{proofof}{Lemma \ref{L2-21-3}}
Due to Condition \ref{c-2-22-1} with $\rho$,
we can choose a constant $C_\rho$ in the following way:
the relation 
\begin{align}
C_\rho=
\sum_z \overline{W}^Z(z|v)^{\frac{1}{1-\rho}}
(\sum_{v'} Q_{V,\rho}(v') 
\overline{W}^Z(z|v')^{\frac{1}{1-\rho}})^{-\rho} 
\end{align}
holds for all $v$.
Due to the general relation as (\ref{2-13-1}),
we have
\begin{align*}
C:=& \lim_{\rho \to 1-0}C_\rho \\
= & 
\lim_{\rho \to 1-0}
\sum_z \overline{W}^Z(z|v)^{\frac{1}{1-\rho}}
(\sum_{v'} Q_{V,\rho}(v') 
\overline{W}^Z(z|v')^{\frac{1}{1-\rho}})^{-\rho} \\
= &
\lim_{\rho \to 1-0}
\sum_{z\in \mathcal{Z}_v}
(\sum_{v''\in \mathcal{V}_z} Q_{V,\rho}(v''))^{-\rho} 
\max_{v'}\overline{W}^Z(z|v')  \\
=&
\sum_{z\in \mathcal{Z}_v}
\frac{\max_{v'}\overline{W}^Z(z|v') }{\sum_{v''\in \mathcal{V}_z} (\lim_{n \to \infty} Q_{V,\rho_n}(v''))}.
\end{align*}
Since $C$ does not depend on $v$,
the distribution $\lim_{n \to \infty} Q_{V,\rho_n}$ satisfies Condition \ref{c-2-21-2}.
\end{proofof}

\subsection{Proof of Lemma \ref{l12-4-3}}\Label{s12-3}
We show the inequality in (\ref{12-4-1}).
First, we obtain the inequality \eqref{12-4-5}, 
which is displayed in the wide space in the next page.

\begin{figure*}[!t]
\begin{align}
& \min_{{W}^Z \in \mathcal{W}(\mathcal{U}\times \mathcal{V},\mathcal{Z})}
D({W}^Z\| \overline{W}^Z| Q_{VU})
-\rho I(V;Z|U)[{W}^Z\times Q_{VU}]  \nonumber \\
= & 
\min_{{W}^Z \in \mathcal{W}(\mathcal{U}\times \mathcal{V},\mathcal{Z})}
\Bigl(
\sum_{u} Q_U(u) (\sum_{v} Q_{V|U}(v|u) 
\sum_{z} W^Z(z|u,v) \log \frac{W^Z(z|u,v)}{\overline{W}^Z(z|v)}
\nonumber \\
&-
\rho 
\min_{\tilde{Q}  \in \mathcal{P}(\mathcal{Z})}
\sum_{v} Q_{V|U}(v|u) 
\sum_{z} W^Z(z|u,v) \log \frac{W^Z(z|u,v)}{\tilde{Q}(z)} ) 
\Bigr)
\nonumber \\
= & 
\min_{{W}^Z \in \mathcal{W}(\mathcal{U}\times \mathcal{V},\mathcal{Z})}
\max_{\tilde{W}^Z \in \mathcal{W}(\mathcal{U},\mathcal{Z})}
\sum_{u} Q_U(u) \sum_{v} Q_{V|U}(v|u) 
(\sum_{z} W^Z(z|u,v) \log \frac{W^Z(z|u,v)}{\overline{W}^Z(z|v)}
-
\rho \sum_{z} W^Z(z|u,v) \log \frac{W^Z(z|u,v)}{\tilde{W}^Z(z|u)} )\nonumber \\
= & 
\min_{{W}^Z \in \mathcal{W}(\mathcal{U}\times \mathcal{V},\mathcal{Z})}
\max_{\tilde{W}^Z \in \mathcal{W}(\mathcal{U},\mathcal{Z})}
\sum_{u} Q_U(u) \sum_{v} Q_{V|U}(v|u) 
\sum_{z} W^Z(z|u,v) 
\log \frac{W^Z(z|u,v)^{1-\rho}\tilde{W}^Z(z|u)^{\rho}}{\overline{W}^Z(z|v)} 
\nonumber \\
= & 
\max_{\tilde{W}^Z \in \mathcal{W}(\mathcal{U},\mathcal{Z})}
\min_{{W}^Z \in \mathcal{W}(\mathcal{U}\times \mathcal{V},\mathcal{Z})}
\sum_{u} Q_U(u) \sum_{v} Q_{V|U}(v|u)
\sum_{z} W^Z(z|u,v) 
 \log \frac{W^Z(z|u,v)^{1-\rho}\tilde{W}^Z(z|u)^{\rho}}{\overline{W}^Z(z|v)} 
\Label{12-5-14} \\
= & 
(1-\rho)
\max_{\tilde{W}^Z \in \mathcal{W}(\mathcal{U},\mathcal{Z})}
\sum_{u} Q_U(u) \sum_{v} Q_{V|U}(v|u) 
\min_{\tilde{P}_Z \in \mathcal{P}(\mathcal{Z})}
\sum_{z} \tilde{P}_Z(z) \log \frac{\tilde{P}_Z(z)}{
\overline{W}^Z(z|v)^{\frac{1}{1-\rho}}\tilde{W}^Z(z|u)^{\frac{-\rho}{1-\rho}}
} \Label{12-4-3}\\
= & 
-(1-\rho)
\min_{\tilde{W}^Z \in \mathcal{W}(\mathcal{U},\mathcal{Z})}
\sum_{u} Q_U(u) \sum_{v} Q_{V|U}(v|u) 
\log \sum_z \overline{W}^Z(z|v)^{\frac{1}{1-\rho}}\tilde{W}^Z(z|u)^{\frac{-\rho}{1-\rho}} \nonumber \\
\ge & 
-(1-\rho)
\min_{\tilde{W}^Z \in \mathcal{W}(\mathcal{U},\mathcal{Z})}
\sum_{u} Q_U(u) 
\log 
\sum_{v} Q_{V|U}(v|u) 
\sum_z \overline{W}^Z(z|v)^{\frac{1}{1-\rho}}\tilde{W}^Z(z|u)^{\frac{-\rho}{1-\rho}} \Label{12-4-4}\\
= & 
-(1-\rho)
\sum_{u} Q_U(u) 
\log 
\min_{\tilde{Q}_Z \in \mathcal{P}(\mathcal{Z})}
\sum_z
(\sum_{v} Q_{V|U}(v|u) 
 \overline{W}^Z(z|v)^{\frac{1}{1-\rho}})
\tilde{Q}_Z(z)^{\frac{-\rho}{1-\rho}} .\Label{12-4-5}
\end{align}
The above derivation can be shown in the following way.
The equality (\ref{12-5-14}) follows from the minimax theorem \cite[Chap. IV Prop. 2.3]{Eke-Tem} because 
the function is concave for $\tilde{W}^Z$ and is convex for ${W}^Z$.
The equality (\ref{12-4-3}) holds
because the minimum is attained with 
$\tilde{P}_Z(z)= \overline{W}^Z(z|v)^{\frac{1}{1-\rho}}\tilde{W}^Z(z|u)^{\frac{-\rho}{1-\rho}}/
\sum_z \overline{W}^Z(z|v)^{\frac{1}{1-\rho}}\tilde{W}^Z(z|u)^{\frac{-\rho}{1-\rho}}$.
The inequality (\ref{12-4-4}) follows from the concavity of $x \mapsto \log x$.
\end{figure*}

Since $\frac{1}{1-\rho}+\frac{-\rho}{1-\rho}=1$,
the reverse H\"{o}lder inequality yields that
\begin{align*}
& \sum_z
(\sum_{v} Q_{V|U}(v|u) 
 \overline{W}^Z(z|v)^{\frac{1}{1-\rho}})
\tilde{Q}_Z(z)^{\frac{-\rho}{1-\rho}} 
\\
\ge &
(\sum_z
(\sum_{v} Q_{V|U}(v|u) 
 \overline{W}^Z(z|v)^{\frac{1}{1-\rho}})^{1-\rho})^{\frac{1}{1-\rho}}
(\sum_z
(\tilde{Q}_Z(z)^{\frac{-\rho}{1-\rho}} )^{-\frac{1-\rho}{\rho}})^\frac{-\rho}{1-\rho}
\\
\ge &
\min_{\tilde{Q}_Z \in \mathcal{P}(\mathcal{Z})}
(\sum_z
(\sum_{v} Q_{V|U}(v|u) 
 \overline{W}^Z(z|v)^{\frac{1}{1-\rho}})^{1-\rho})^{\frac{1}{1-\rho}}
(\sum_z \tilde{Q}_Z(z))^\frac{-\rho}{1-\rho} 
\\
= &
(\sum_z
(\sum_{v} Q_{V|U}(v|u) 
 \overline{W}^Z(z|v)^{\frac{1}{1-\rho}})^{1-\rho})^{\frac{1}{1-\rho}}.
\end{align*}
The equality holds only when 
$(\sum_{v} Q_{V|U}(v|u) 
 \overline{W}^Z(z|v)^{\frac{1}{1-\rho}})^{1-\rho}
= C\tilde{Q}_Z(z) $ with a constant $C$.
Hence, 
\begin{align*}
& \min_{\tilde{Q}_Z \in \mathcal{P}(\mathcal{Z})}
\sum_z
(\sum_{v} Q_{V|U}(v|u) 
 \overline{W}^Z(z|v)^{\frac{1}{1-\rho}})
\tilde{Q}_Z(z)^{\frac{-\rho}{1-\rho}} 
\\
= &
(\sum_z
(\sum_{v} Q_{V|U}(v|u) 
 \overline{W}^Z(z|v)^{\frac{1}{1-\rho}})^{1-\rho})^{\frac{1}{1-\rho}}.
\end{align*}
Thus,
\begin{align}
&-(1-\rho)
\sum_{u} Q_U(u) 
\log \Biggl[
\nonumber \\
&\hspace{10ex}
\min_{\tilde{Q}_Z \in \mathcal{P}(\mathcal{Z})}
\sum_z
\Bigl(\sum_{v} Q_{V|U}(v|u) 
 \overline{W}^Z(z|v)^{\frac{1}{1-\rho}}\Bigr)
\tilde{Q}_Z(z)^{\frac{-\rho}{1-\rho}} 
\Biggr]
\nonumber \\
=&
-(1-\rho)
\sum_{u} Q_U(u) 
\log \Biggl[
\Biggl(\sum_z
\Bigl(\sum_{v} Q_{V|U}(v|u) 
 \overline{W}^Z(z|v)^{\frac{1}{1-\rho}}
\!\Bigr)^{1-\rho}
\!\Biggr)^{\frac{1}{1-\rho}} 
\!\Biggr]\nonumber \\
=&
-
\sum_{u} Q_U(u) 
\log 
(\sum_z
(\sum_{v} Q_{V|U}(v|u) 
 \overline{W}^Z(z|v)^{\frac{1}{1-\rho}})^{1-\rho}) \nonumber \\
\ge &
-
\log 
\sum_{u} Q_U(u) 
(\sum_z
(\sum_{v} Q_{V|U}(v|u) 
 \overline{W}^Z(z|v)^{\frac{1}{1-\rho}})^{1-\rho}) \Label{12-4-6}\\
=& -{E_0}(\rho | \overline{W}^Z, Q_{V|U}, Q_U)\Label{12-4-7},
\end{align}
where (\ref{12-4-6}) follows from the concavity of $x \mapsto \log x$.
The combination of (\ref{12-4-5}) and (\ref{12-4-7}) yields (\ref{12-4-1}).

The equality in (\ref{12-4-4}) holds if and only if
for an arbitrary fixed $u$,
$\sum_z \overline{W}^Z(z|v)^{\frac{1}{1-\rho}}\tilde{W}^Z(z|u)^{\frac{-\rho}{1-\rho}} $
does not depend on $v$ 
with 
$\tilde{W}^Z(z|u)
=(\sum_{v} Q_{V|U}(v|u) \overline{W}^Z(z|v)^{\frac{1}{1-\rho}})^{1-\rho}/
\sum_z (\sum_{v} Q_{V|U}(v|u) 
\overline{W}^Z(z|v)^{\frac{1}{1-\rho}})^{1-\rho}$, i.e.,
the quantity 
$\sum_z \overline{W}^Z(z|v)^{\frac{1}{1-\rho}}
(\sum_{v} Q_{V|U}(v|u) \overline{W}^Z(z|v)^{\frac{1}{1-\rho}})^{-\rho}$
does not depend on $v$ for an arbitrary fixed $u$.
The condition holds 
when $Q_{V|U=u}$ is 
$\argmin_{Q_{V}} 
E_0(\rho| \overline{W}^Z ,Q_{V} )$
because of Lemma \ref{L2-22-1}.
Further, the equality in (\ref{12-4-6}) holds in this case.
Hence, 
when $Q_{V|U=u}$ is 
$\argmin_{Q_{V}} 
E_0(\rho| \overline{W}^Z ,Q_{V} )$,
the equality holds in the inequality (\ref{12-4-1}).

\section{Conclusion}\Label{sec4}
In order to treat the secure multiplex coding with 
dependent and non-uniform multiple messages and common messages,
we have generalized resolvability to the case when input random variable
is subject to a non-uniform distribution.
Two kinds of generalization have been given.
The first one (Theorem \ref{lem-01})
is a simple extension of Han-Verd\'u's
channel resolvability coding \cite{han93}
with the non-uniform inputs.
The second one (Theorem \ref{Lee3})
uses randomly chosen affine mapping
satisfying Condition \ref{C2-b} with the non-uniform inputs.

We have constructed two kinds of codes for the above type of SMC.
Similar to BCC in \cite{csiszar78},
the second construction has two steps.
In the first step, similar to the BCD encoder,
we apply superposition random coding.
In the second step, as is illustrated in Fig. \ref{fig:encoder},
we split the confidential message into the private message
$B_2$ and a part $B_1$ of the common message encoded by the BCD encoder.
Employing the second type of channel resolvability,
we have derived a non-asymptotic formula for 
the average leaked information under this kind of code construction. 
On the other hand, 
in the first construction, 
the confidential message is simply sent as the private message
encoded by the BCD encoder.
Hence, it has only one step.
Employing the first type of channel resolvability,
we have derived a non-asymptotic formula for 
the average leaked information under this kind of code construction. 

For asymptotic treatment for the non-uniform and dependent sources,
we have introduced three kinds of asymptotic conditional uniformity conditions.
Then, we have clarified the relation among three conditions, especially, that two of them are equivalent.
Further, we have shown that these conditions can be satisfied by data compressed by Slepian-Wolf compression,
in the respective senses.
Extending the above formula for the second construction 
to the asymptotic case,
we have derived the capacity region of SMC defined in our general setting,
in which, 
the message is allowed to be dependent and non-uniform
while it has to satisfy the weaker asymptotic conditional uniformity condition.
We have shown the strong security when the the leaked information rate is zero
and the message satisfies the stronger asymptotic conditional uniformity condition.
Using the both formulas,
we have also derived the exponential decreasing rate of leaked information.
While the first formula gives an upper bound in any case,
the second one gives a better upper bound in some specific cases.

We have also given two kinds of practical constructions for SMC
by using ordinary linear codes.
Following our constructions, we can make a code satisfying a required security level.
Further, we have given a universal code for SMC, 
which does not depend on the channel.
Extending this result, we have derived
a source-channel universal code for BCC,
which does not depend on the channel or the source distribution.

\section*{Acknowledgment}
RM would like to thank Prof.\ H.\  Yamamoto to teach him
the secure multiplex coding.
The authors are grateful to Prof. Alexander Vardy for
pointing out the importance for the non-independent case for 
the multiple secret messages.
The authors are grateful to Dr. Shun Watanabe for
informing the references \cite{e1,e2,e3,e4}.
They also would like to express their appreciation to the referee of
this paper for his/her helpful comments.
A part of this research was done during RM's stay
at the Institute of Network Coding, the Chinese University
of Hong Kong, and Department of Mathematical Sciences, Aalborg University.
He greatly appreciates the hospitality by Prof.\ R.\ Yeung and
Prof.\ O.\ Geil.

This research was partially supported by 
the MEXT Grant-in-Aid for Young Scientists (A) No.\ 20686026 and
(B) No.\ 22760267, Grant-in-Aid for Scientific Research (A) No.\ 23246071,
and the ImPACT Program of Council for Science, Technology and
Innovation (Cabinet Office, Government of Japan).
The Center for Quantum Technologies is funded
by the Singapore Ministry of Education and the National Research
Foundation as part of the Research Centres of Excellence programme.

\appendices

\section{Inequality between R\'{e}nyi Entropy and Conditional R\'{e}nyi Entropy}\Label{as4}
In this appendix, we derive a useful inequality between R\'{e}nyi entropy and conditional R\'{e}nyi entropy,
which was used in Subsection \ref{as2}.
For this purpose, we prepare the following lemma.

\begin{lemma}\Label{lem-11-25-1}
Any two distributions $P_{XY}$ and $Q_{XY}$ over $\mathcal{X}\times \mathcal{Y}$
satisfy
\begin{align}
\psi(\rho|P_{X,Y}\|Q_{X,Y})
\ge 
\frac{1}{1-\rho}\psi(\rho(1-\rho)|P_{X,Y}\| Q_{Y|X} \times P_X  )
\end{align}
for $\rho>0$,
where $P_X$ is the marginal distribution of $P_{X,Y}$ on $\mathcal{X}$,
and $Q_{Y|X}$ is the conditional distribution of $Q_{X|Y}$
on $\mathcal{Y}$ conditioned with $\mathcal{X}$. 
\end{lemma}

When $Q_{XY}$ is the uniform distribution, 
 $\frac{1}{\rho}\psi(\rho|P_{X,Y}\|Q_{X,Y})
=\log (|\mathcal{X}| |\mathcal{Y}|)-H_{1+\rho}(X,Y)$
and
$\frac{1}{\rho(1-\rho)}\psi(\rho(1-\rho)|P_{X,Y}\| Q_{Y|X} \times P_X  )
=\log |\mathcal{Y}| -H_{1+\rho(1-\rho)}(Y|X)$,
which implies the following corollary of the above lemma as
an inequality between R\'{e}nyi entropy and conditional R\'{e}nyi entropy.

\begin{corollary}\Label{lem-11-23-2}
For $\rho>0$, 
arbitrary random variables $X$ and $Y$ over $\mathcal{X}$ and $\mathcal{Y}$
satisfy
\begin{align}
\log (|\mathcal{X}| |\mathcal{Y}|)
-H_{1+\rho}(X,Y)
\ge
\log |\mathcal{Y}| -H_{1+\rho(1-\rho)}(Y|X),
\end{align}
which implies
\begin{align}
\log |\mathcal{X}|+
H_{1+\rho(1-\rho)}(Y|X)
\ge
H_{1+\rho}(X,Y).
\end{align}
\end{corollary}

\begin{proofof}{Lemma \ref{lem-11-25-1}}
Applying H\"{o}lder inequality 
$\sum_{x}P_X(x) |A(x)B(x)|
\le (\sum_{x}P_X(x) |A(x)|^{\frac{1}{1-\rho}})^{1-\rho} 
(\sum_{x}P_X(x) |B(x)|^{\frac{1}{\rho}})^{\rho}$,
to the case 
$A(x)= 
P_X(x)^{\rho} Q_X(x)^{-\rho}
(\sum_y 
P_{Y|X}(y|x)^{1+\rho(1-\rho)}Q_{Y|X}(y|x)^{-\rho(1-\rho)})^{\frac{1}{1-\rho}} $
and
$B(x)= P_X(x)^{-\rho} Q_X(x)^{\rho}$,
we obtain the following.
In the following derivation, 
we employ the above H\"{o}lder inequality in (\ref{11-24-2}),
and the Jensen inequality for the convex function $x\mapsto x^{\frac{1}{1-\rho}}$
in (\ref{11-24-1-b}), (\ref{11-24-3-b}), and (\ref{11-24-4}).
\begin{align}
&e^{\frac{1}{1-\rho}\psi(\rho (1-\rho)|P_{X,Y}\| Q_{Y|X} \times P_X  )}\nonumber\\
=&
(\sum_{x}P_X(x)\sum_y 
P_{Y|X}(y|x)^{1+\rho(1-\rho)}Q_{Y|X}(y|x)^{-\rho(1-\rho)})^{\frac{1}{1-\rho}} \nonumber\\
\le &
\sum_{x}P_X(x)
(\sum_y 
P_{Y|X}(y|x)^{1+\rho(1-\rho)}Q_{Y|X}(y|x)^{-\rho(1-\rho)})^{\frac{1}{1-\rho}} \Label{11-24-1-b}\\
= &
\sum_{x}P_X(x)
\Biggl[(P_X(x)^{\rho} Q_X(x)^{-\rho}
\nonumber\\
&\hspace{1ex}\cdot
\sum_y 
\Bigl(P_{Y|X}(y|x)^{1+\rho(1-\rho)}Q_{Y|X}(y|x)^{-\rho(1-\rho)} \Bigr)^{\frac{1}{1-\rho}} 
\!\Bigl(\! P_X(x)^{-\rho} Q_X(x)^{\rho} \Bigr)
\!\Biggr] 
\nonumber\\
\le &
\Biggl[
\sum_{x}P_X(x)
P_X(x)^{\frac{\rho}{1-\rho}} Q_X(x)^{-\frac{\rho}{1-\rho}}
\nonumber\\
& \hspace{3ex}
\cdot 
\Bigl(
\sum_y 
P_{Y|X}(y|x)^{1+\rho(1-\rho)}Q_{Y|X}(y|x)^{-\rho(1-\rho)}
\Bigr)^{\frac{1}{(1-\rho )^2}} 
\Biggr]^{1-\rho}
\nonumber\\
& \hspace{22ex}
\cdot 
\Bigl(\sum_{x}P_X(x) P_X(x)^{-1} Q_X(x)\Bigr)^{\rho} \Label{11-24-2}\\
= &
\Biggl[ 
\sum_{x}P_X(x)
P_X(x)^{\frac{\rho}{1-\rho}} Q_X(x)^{-\frac{\rho}{1-\rho}}
\nonumber\\
& \hspace{1ex}
\cdot 
\biggl(\sum_y 
P_{Y|X}(y|x) \Bigl(P_{Y|X}(y|x)^{\rho(1-\rho)}Q_{Y|X}(y|x)^{-\rho(1-\rho)} \Bigr)
\biggr)^{\frac{1}{(1-\rho)^2}} 
\Biggr]^{1-\rho}
\cdot 
1^{\rho} \nonumber\\
\le &
\sum_{x}P_X(x)
P_X(x)^{\rho} Q_X(x)^{-\rho}
\Biggl[
\nonumber\\
& \hspace{6ex}
\sum_y 
P_{Y|X}(y|x) \biggl(P_{Y|X}(y|x)^{\rho(1-\rho)}Q_{Y|X}(y|x)^{-\rho(1-\rho)} \biggr)
\Biggr]^{\frac{1}{1-\rho}} \Label{11-24-3-b}\\
\le &
\sum_{x}P_X(x)
P_X(x)^{\rho} Q_X(x)^{-\rho}
\Biggl[
\nonumber\\
& \hspace{4ex}
\sum_y 
P_{Y|X}(y|x) \biggl(P_{Y|X}(y|x)^{\rho(1-\rho)}Q_{Y|X}(y|x)^{-\rho(1-\rho)}
\biggr)^{\frac{1}{1-\rho}}
\Biggr] \Label{11-24-4}\\
= &
\sum_{x}P_X(x)
P_X(x)^{\rho} Q_X(x)^{-\rho}
\biggl[
\nonumber\\
& \hspace{14ex}
\sum_y 
P_{Y|X}(y|x) \Bigl(P_{Y|X}(y|x)^{\rho}Q_{Y|X}(y|x)^{-\rho}
\Bigr)
\biggr] \nonumber\\
= &
\sum_{x,y}P_{X,Y}(x,y)^{1+\rho} Q_{X,Y}(x,y)^{-\rho}
=
e^{\psi(\rho|P_{X,Y}\|Q_{X,Y})}.
\end{align}
\end{proofof}

\section{Existence of Code Required in Theorem \ref{th-12-23-1} with $\epsilon=0$}\Label{as5}

In this appendix, 
we show the existence of Slepian-Wolf data compression code satisfying the condition (\ref{12-26-1}) required in Theorem \ref{th-12-23-1} with $\epsilon=0$
in the two-terminal and i.i.d. case.
For this purpose, we assume that
the random variables $(S_1^n,S_2^n)$ are subject to 
the $n$-fold i.i.d. distribution
of a given non-uniform joint distribution of $S_1$ and $S_2$.
For this purpose, we recall the definition of achievable rate pair for Slepian-Wolf compression.
\begin{definition}
A rate pair $(R_1,R_2)$ is called {\it achievable}
when there exists a sequence of encoders $\varphi^n=(\varphi^n_1,\varphi^n_2)$
($\varphi^n_i:\mathcal{S}_i^n \to \{1, \ldots, \lceil e^{nR_i}\rceil\}$)
and decoders $\hat{\varphi}^n$ 
($\hat{\varphi}^n:\{1, \ldots, \lceil e^{nR_1}\rceil\} \times \{1, \ldots, \lceil e^{nR_2}\rceil\} \to
\mathcal{S}_1^n \times \mathcal{S}_2^n$) such that
the decoding error probability $\varepsilon(\varphi^n,\hat{\varphi}^n)$
satisfies 
\begin{align}
\lim_{n \to \infty}\varepsilon(\varphi^n,\hat{\varphi}^n) =0.
\end{align}
\end{definition}

Then, we prepare the following lemma.

\begin{lemma}\Label{12-30-1}
Let $(R_1,R_2)$ be a pair of 
achievable rates for Slepian-Wolf compression satisfying $R_1+R_2=H(S_1,S_2)$.
When the compression rate pair $(R_{1,n},R_{2,n})$ behaves as
$R_{1,n}=R_1+\frac{c_1}{n^t}$
and
$R_{2,n}=R_2+\frac{c_2}{n^t}$
with $0< t <1/2$ and $c_1>,c_2>0$,
there exists a sequence of Slepian-Wolf codes 
$({\varphi}^n,\hat{\varphi}^n)
=(({\varphi}^n_1,{\varphi}^n_2),\hat{\varphi}^n)$ 
for any positive integer $n$
such that
${\varphi}^n_i$ is a map from $\mathcal{S}_i^n$ to $\{1, \ldots, \lceil e^{nR_{i,n}}\rceil\}$
for $i=1,2$
and
the decoding error probability
$\varepsilon ({\varphi}^n,\hat{\varphi}^n)$ 
satisfies
\begin{align}
&\liminf_{n \to \infty}
-{n^{2t-1}}\log \varepsilon ({\varphi}^n,\hat{\varphi}^n)
\nonumber \\
\ge &
\min \Biggl( 
\lambda \frac{c^2_1}{2 V(S_1)} ,
\lambda \frac{c^2_2}{2 V(S_2|S_1)} , \nonumber \\
&\hspace{15ex}
(1-\lambda) \frac{c^2_2}{2 V(S_2)} ,
(1-\lambda) \frac{c^2_1}{2 V(S_1|S_2)} 
\Biggr),\Label{1-27-11}
\end{align}
where
$V(S_2|S_1):=
\sum_{s_1,s_2} P_{S_1,S_2}(s_1,s_2)( \log P_{S_2|S_1}(s_2|s_1) - H(S_2|S_1))^2$
and $\lambda\in [0,1]$ is the real number satisfying that
\begin{align}
(R_1,R_2)= \lambda (H(S_1),H(S_2|S_1))+ (1-\lambda) (H(S_1|S_2), H(S_2)).
\end{align}
Further, when $R_1=H(S_1)$ and $R_2=H(S_2|S_1)$
and the compression rates $(R_{1,n},R_{2,n})$
behaves as
$R_{1,n}=H(S_1)+\frac{c_1}{n^t}$
and
$R_{2,n}=H(S_2|S_1) +\frac{c_2}{n^t}$
with $0<t <1/2$ and $c_1>,c_2>0$,
there exists a sequence of Slepian-Wolf codes $({\varphi}^n,\hat{\varphi}^n)$ 
such that 
the decoding error probability
$\varepsilon ({\varphi}^n,\hat{\varphi}^n)$ 
satisfies
\begin{align}
\liminf_{n \to \infty}
-{n^{2t-1}}\log \varepsilon ({\varphi}^n,\hat{\varphi}^n)
\ge
\min \Biggl( 
\frac{c^2_1}{2 V(S_1)} ,
\frac{c^2_2}{2 V(S_2|S_1)} 
\Biggr).\Label{1-27-10}
\end{align}
\end{lemma}

We will prove Lemma \ref{12-30-1} after preparing several lemmas.
Using Lemma \ref{12-30-1}, 
we make a Slepian-Wolf compression whose compressed data 
satisfies the SACU condition.
Let $(R_1,R_2)$ be a pair of 
achievable rates for Slepian-Wolf compression satisfying $R_1+R_2=H(S_1,S_2)$.
Then, 
let 
$\varphi^n=(\varphi_{1}^n,\varphi_{2}^n) $
and $\hat{\varphi}^n$
be the Slepian-Wolf encoders
and the Slepian-Wolf decoder given in 
Lemma \ref{12-30-1} with the case of $c_1= R_1 c$ and $c_2= R_2 c$.
We choose
the integer $m_n:=
\lfloor \frac{n}{1+\frac{c}{n^t}}\rfloor
=\lfloor\frac{R_1 n}{R_1+R_1\frac{c}{n^t}}\rfloor
=\lfloor\frac{R_2 n}{R_2+R_2\frac{c}{n^t}}\rfloor
=\lfloor\frac{R_1 n}{R_{1,n}}\rfloor
=\lfloor\frac{R_2 n}{R_{2,n}}\rfloor$
for $0<t<\frac{1}{2}$ and $c>0$.
Then, we obtain 
the  Slepian-Wolf encoders
$\varphi_{i}^{m_n}: \mathcal{S}_{i}^{m_n} \to 
\{1,\ldots, \lceil e^{n R_i} \rceil\}$
and the Slepian-Wolf decoder 
$\hat{\varphi}^{m_n}: 
\{1,\ldots, \lceil e^{n R_1} \rceil\}\times 
\{1,\ldots, \lceil e^{n R_2} \rceil\} 
\to \mathcal{S}_{1}^{m_n} \times \mathcal{S}_{2}^{m_n}$.
Using the code, we define the  Slepian-Wolf encoders
$\varphi_{i,u}^n: \mathcal{S}_{i}^{m_n} \to 
\{1,\ldots, \lceil e^{n R_i} \rceil\}$
and the Slepian-Wolf decoder 
$\hat{\varphi}^n_u: 
\{1,\ldots, \lceil e^{n R_1} \rceil\}\times 
\{1,\ldots, \lceil e^{n R_2} \rceil\} 
\to \mathcal{S}_{1}^{m_n} \times \mathcal{S}_{2}^{m_n}$
by 
\begin{align}
\varphi_{i,u}^n(s^{m_n})
&:=\varphi_{i}^{m_n}(s^{m_n}) \\
\hat{\varphi}^n_u(x_1,x_2)
&:=\hat{\varphi}^{m_n}(x_1,x_2) .
\end{align}
Then, due to Lemma \ref{12-30-1},
since 
$m_n (R_1+R_1\frac{c}{n^t})= n R_1$ and
$m_n (R_2+R_2\frac{c}{n^t})= n R_2$,
the code $((\varphi_{1,u}^n,\varphi_{2,u}^n),\hat{\varphi}^n_u)$
satisfies the condition (\ref{12-26-1}) in Theorem \ref{th-12-23-1} with $\epsilon=0$.
Theorem \ref{th-12-23-1} guarantees that the compressed data satisfies the SACU condition.

Now, in order to show Lemma \ref{12-30-1}, we prepare several lemmas.

\begin{lemma}[\cite{e1,e2,e3}]\Label{12-28-10}
For a given compression rate $R_2>0$,
there exists a pair of the encoder $\varphi^n$ and the decoder 
$\hat{\varphi}^n$ of the random variable $S_2^n$ with the side information $S_1^n$
such that the decoding error probability
$\varepsilon ({\varphi}^n,\hat{\varphi}^n)$
satisfies 
\begin{align}
\varepsilon ({\varphi}^n,\hat{\varphi}^n)
\le
e^{- n(\rho R_2 - 
E_0(-\rho|S_2|S_1)
)}
\end{align}
for any $\rho \in (0,1]$, where
\begin{align}
E_0(\rho|S_2|S_1):= \log \sum_{s_1} 
(\sum_{s_2} P_{S_1,S_2}(s_1,s_2)^{\frac{1}{1-\rho}})^{1-\rho}.
\end{align}
Note that when there is no side information, we have
\begin{align}
E_0(-\rho|S_2)=\rho H_{\frac{1}{1+\rho}}(S_2) .
\end{align}
\end{lemma}

\begin{lemma}\Label{l1-27-1}
The quantity $E_0(-\rho|S_2|S_1)$
has the expansion 
\begin{align}
E_0(-\rho|S_2|S_1)= \rho H(S_2|S_1)+ \frac{\rho^2}{2} V(S_2|S_1)
\Label{1-3-1}
\end{align}
with small $\rho$. 
In particular,
the quantity $\rho H_{\frac{1}{1+\rho}}(S_1)$
has the expansion 
\begin{align}
\rho H_{\frac{1}{1+\rho}}(S_1) = \rho H(S_1)+ \frac{\rho^2}{2} V(S_1)
\Label{1-3-1b}
\end{align}
with small $\rho$ and
$V(S_1):=
\sum_{s_1} P_{S_1}(s_1)( \log P_{S_1}(s_1) - H(S_1))^2$.
\end{lemma}

\begin{proof}
Take the Taylor expansion of $e^{E_0(\rho|S_2|S_1)}$ as
\begin{align}
&e^{E_0(-\rho|S_2|S_1)} \nonumber \\
=& 1+\rho H(S_2|S_1)\nonumber \\
&+ \frac{\rho^2}{2}
\sum_{s_1,s_2} P_{S_1,S_2}(s_1,s_2)( \log P_{S_2|S_1}(s_2|s_1))^2
+o(\rho^2).
\end{align}
Taking the logarithm, we obtain (\ref{1-3-1}).
\end{proof}

\begin{lemma}\Label{l12-28-1}
Let $(R_1,R_2)$ belong to the Slepian-Wolf compression region of $(S_1^n,S_2^n)$.
We choose 
the rates
$R_1'$,
$R_2'$,
$R_1''$,
and $R_2''$ and the real number 
$\lambda \in [0,1]$ such that
\begin{align}
(R_1,R_2)= \lambda (R_1',R_2')+(1-\lambda )(R_1'',R_2'').
\end{align}
Then, there exists a pair of the Slepian-Wolf encoder $\varphi^n$ and the decoder 
$\hat{\varphi}^n$ 
such that 
the decoding error probability
$\varepsilon ({\varphi}^n,\hat{\varphi}^n)$
satisfies 
\begin{align}
& \varepsilon ({\varphi}^n,\hat{\varphi}^n) \nonumber \\
\le &
\inf_{\rho \in (0,1]} e^{- \lambda n(\rho R_1' - \rho H_{\frac{1}{1+\rho}}(S_1) )}
+
\inf_{\rho \in (0,1]} e^{- \lambda n(\rho R_2' - E_0(-\rho|S_2|S_1) )}\nonumber \\
& +
\inf_{\rho \in (0,1]} e^{- (1-\lambda) n(\rho R_1'' - E_0(-\rho|S_1|S_2) )}
+
\inf_{\rho \in (0,1]} e^{- (1-\lambda) n(\rho R_2''- \rho H_{\frac{1}{1+\rho}}(S_2) )},\Label{12-28-11}
\end{align}
Also, there exists a pair of the Slepian-Wolf encoder $\varphi^n$ and the decoder 
$\hat{\varphi}^n$ 
such that 
the decoding error probability
$\varepsilon ({\varphi}^n,\hat{\varphi}^n)$
satisfies 
\begin{align}
& \varepsilon ({\varphi}^n,\hat{\varphi}^n) \nonumber \\
\le &
\inf_{\rho \in (0,1]} e^{- n(\rho R_1 - \rho H_{\frac{1}{1+\rho}}(S_1) )}
+
\inf_{\rho \in (0,1]} e^{- n(\rho R_2 - E_0(-\rho|S_2|S_1) )},\Label{12-28-11b}.
\end{align}
\end{lemma}

\begin{proof}
First, we show the existence of a sequence of codes satisfying (\ref{12-28-11b}).
We apply the usual data compression for $S_2^n$,
and the data compression given in Lemma \ref{12-28-10} for $S_1^n$.
The decoder is given by combination of the respective decoders.
Since the decoding error probability is bounded by the sum of 
the decoding error probabilities of $S_1^n$ and $S_2^n$,
we obtain (\ref{12-28-11b}).

Next, we show the existence of a sequence of codes satisfying (\ref{12-28-11}).
We divide $n$ symbols into two parts, $\lambda n$ symbols and $(1-\lambda) n$ symbols.
We apply the construction given in the previous paragraph with the rates $(R_1',R_2')$ to the first part,
and 
apply the same construction with the rates $(R_1'',R_2'')$ to the second part.
Due to Lemma \ref{12-28-10},
the decoding error probability of the first part is less than
$\inf_{\rho \in (0,1]} e^{- \lambda n(\rho R_1' - \rho H_{\frac{1}{1+\rho}}(S_1) )}
+
\inf_{\rho \in (0,1]} e^{- \lambda n(\rho R_2' - E_0(-\rho|S_2|S_1) )}$,
and 
the decoding error probability of the second part is less than
$\inf_{\rho \in (0,1]} e^{- (1-\lambda) n(\rho R_1'' - E_0(-\rho|S_1|S_2) )}
+
\inf_{\rho \in (0,1]} e^{- (1-\lambda) n(\rho R_2''- \rho H_{\frac{1}{1+\rho}}(S_2) )}$.
Then, we obtain (\ref{12-28-11}).
\end{proof}

\begin{proofof}{Lemma \ref{12-30-1}}
First, we consider the case when $R_1=H(S_1)$ and $R_2=H(S_2|S_1)$.
Since $R_{1,n}= H(S_1)+ \frac{c_1}{n^t}$ and $R_{2,n}:= H(S_2|S_1)+ \frac{c_2}{n^t}$,
we can show that
\begin{align}
\lim_{n \to \infty} -{n^{2t-1}}\log 
\inf_{\rho \in (0,1]} e^{-  n(\rho R_{1,n} - \rho H_{\frac{1}{1+\rho}}(S_1) )}
&=
\frac{c^2_1}{2 V(S_1)} \Label{12-28-12f}\\
\lim_{n \to \infty} -{n^{2t-1}}\log 
\inf_{\rho \in (0,1]} e^{- n(\rho R_{2,n} - E_0(-\rho|S_2|S_1) )}
&=
\frac{c^2_2}{2 V(S_2|S_1)} \Label{12-28-13f}.
\end{align}
Since the proof of (\ref{12-28-12f})
is similar to those of (\ref{12-28-13f}), 
we show only (\ref{12-28-12f}).
When $\rho$ is sufficiently small, 
due to Lemma \ref{l1-27-1}, we have
\begin{align}
& \rho R_{1,n} - \rho H_{\frac{1}{1+\rho}}(S_1)
\cong \rho \frac{c_1}{n^t} - \frac{\rho^2}{2} V(S_1)
\nonumber\\
=& -\frac{V(S_1)}{2}(\rho - \frac{c_1}{ V(S_1) n^t} )^2
+\frac{c^2_1}{2 V(S_1) n^{2t}}.
\end{align}
Hence,
$\inf_{\rho \in (0,1]} e^{- n(\rho R_{1,n}' - \rho H_{\frac{1}{1+\rho}}(S_1) )}
\cong e^{- n \frac{c^2_1}{2 V(S_1) n^{2t}}}$,
which implies (\ref{12-28-12f}).
Then, we apply the evaluation (\ref{12-28-11b}) for the decoding error probability
in Lemma \ref{l12-28-1} to the case
when $R_{1}$, $R_{2}$ are $R_{1,n}$, $R_{2,n}$.
Combining the relations (\ref{12-28-12f}) and (\ref{12-28-13f}),
we obtain (\ref{1-27-10}).

Next, we show the general case.
We choose
$R_{1,n}':= H(S_1)+ \frac{c_1}{n^t}$,
$R_{2,n}':= H(S_2|S_1)+ \frac{c_2}{n^t}$,
$R_{1,n}'':= H(S_1|S_2)+ \frac{c_1}{n^t}$, 
$R_{2,n}'':= H(S_2)+ \frac{c_2}{n^t}$.
Then, we obtain
\begin{align}
(R_{1,n},R_{2,n})= \lambda (R_{1,n}',R_{2,n}')
+ (1-\lambda) (R_{1,n}'',R_{2,n}'').
\end{align}
Then, similar to (\ref{12-28-12f}) and (\ref{12-28-13f}), we can show that
\begin{align}
&\lim_{n \to \infty} -{n^{2t-1}}\log 
\inf_{\rho \in (0,1]} e^{- \lambda n(\rho R_{1,n}' - \rho H_{\frac{1}{1+\rho}}(S_1) )}
=
\lambda \frac{c^2_1}{2 V(S_1)} \Label{12-28-12}\\
& \lim_{n \to \infty} -{n^{2t-1}}\log 
\inf_{\rho \in (0,1]} e^{- \lambda n(\rho R_{2,n}' - E_0(-\rho|S_2|S_1) )}
=
\lambda \frac{c^2_2}{2 V(S_2|S_1)} \Label{12-28-13}\\
& \lim_{n \to \infty} -{n^{2t-1}}\log 
\inf_{\rho \in (0,1]} e^{- (1-\lambda) n(\rho R_{1,n}'' - E_0(-\rho|S_1|S_2) )}
=
(1-\lambda) \frac{c^2_2}{2 V(S_2)} \Label{12-28-14}\\
& \lim_{n \to \infty} -{n^{2t-1}}\log 
\inf_{\rho \in (0,1]} e^{- (1-\lambda) n(\rho R_{2,n}''- \rho H_{\frac{1}{1+\rho}}(S_2) )} 
=
(1-\lambda) \frac{c^2_1}{2 V(S_1|S_2)} \Label{12-28-15}.
\end{align}
We apply the evaluation (\ref{12-28-11}) for the decoding error probability
in Lemma \ref{l12-28-1} to the case
when 
$R_{1}'$, 
$R_{2}'$, 
$R_{1}''$, 
$R_{2}''$, 
are
$R_{1,n}'$, 
$R_{2,n}'$, 
$R_{1,n}''$, 
$R_{2,n}''$.
Combining the relations (\ref{12-28-12}), (\ref{12-28-13}), (\ref{12-28-14}) and (\ref{12-28-15}),
we obtain (\ref{1-27-11}).
\end{proofof}

\section{Equivalence between the SWACU Condition and the WACU Condition}\Label{as1}
In Subsection \ref{s6-2-2},
we have introduced three asymptotic conditional uniformity conditions.
The aim of this appendix is to show the equivalence between the SWACU condition and the WACU condition,
which was used in our proof of Theorem \ref{th:smc}.
\begin{lemma}\Label{lem22}
Let $A_n$ be a random variable on the set ${\cal A}_n$ with the cardinality $e^{nR}$
and $B_n$ be another random variable for any positive inter $n$.
Then, the relation
\begin{align}
\lim_{n\to \infty}
\frac{1}{n}H(A_n|B_n) = R \Label{lem22-1}
\end{align}
holds, if and only if
\begin{align}
\lim_{n\to \infty}
\frac{1}{n}H_{1+\alpha/n}(A_n|B_n) = R\Label{lem22-2}
\end{align}
for any $\alpha >0$.
\end{lemma}

Lemma \ref{lem22} will be shown after Lemma \ref{lem11}, which is used in the proof of Lemma \ref{lem22}.
Thanks to Lemma \ref{lem22}, we can replace the WACU condition (\ref{11-9-3})
by the SWACU condition (\ref{11-9-3-d}).
Indeed, in order to apply our results in Section \ref{s5}
to the proof of Theorem \ref{th:smc}, 
we need evaluation 
conditional R\'{e}nyi entropy instead of conditional entropy, 
as is discussed around (\ref{11-27-1}).
Lemma \ref{lem22} provides 
the evaluation of conditional R\'{e}nyi entropy (\ref{lem22-2})
from the evaluation of conditional entropy (\ref{lem22-1}).
Hence, Lemma \ref{lem22} is useful for the application of our results in Section \ref{s5} to the asymptotic setting.

\begin{lemma}\Label{lem11}
Let $A$ be a random variable on the set ${\cal A}$ with the cardinality $M$
and $B$ be another random variable.
For arbitrary $\epsilon_1>0$ and $1 \ge \epsilon_2>0$,
we define the subset of joint distributions for $A$ and $B$ as
\begin{align}
{\cal P}_{\epsilon_1,\epsilon_2,M}^{A|B}
:=
\{ P_{A,B}|
P_{A,B} \{ (a,b)|- \log P_{A|B}(a|b) \le \log M - \epsilon_1 \} \le \epsilon_2
\}.
\end{align}
Then,
\begin{align}
\max_{P_{A,B}\in {\cal P}_{\epsilon_1,\epsilon_2,M}^{A|B}}
H(A|B)
\le & 
\log M - \epsilon_2 (e^{-\epsilon_1}-1+\epsilon_1) 
\Label{eq-a-11} \\
\min_{P_{A,B}\in {\cal P}_{\epsilon_1,\epsilon_2,M}^{A|B}}
H_{1+\rho}(A|B)
\ge &
-\frac{1}{\rho}\log (
(1-\epsilon_2) \frac{e^{\rho \epsilon_1}}{M^\rho}
+\epsilon_2 ).
\Label{eq-a-12}
\end{align}
Here, since the region
${\cal P}_{\epsilon_1,\epsilon_2,M}^{A|B}$
is compact,
the above maximum and the above minimum exist.
\end{lemma}

\begin{proofof}{Lemma \ref{lem11}}
For an arbitrary integer $k$, we define the set 
\begin{align*}
{\cal P}_{\epsilon_1,\epsilon_2,M,k}^{A}
& :=
\left\{ P_{A} \left|
\begin{array}{l}
P_{A} \{ a|- \log P_{A}(a) \le \log M - \epsilon_1 \} \le \epsilon_2, \\
|\{ a|- \log P_{A}(a) \le \log M - \epsilon_1 \}|=k
\end{array}
\right. \right\} \\
{\cal P}_{\epsilon_1,\epsilon_2,M}^{A}
& :=
\{ P_{A}|
P_{A} \{ a|- \log P_{A}(a) \le \log M - \epsilon_1 \} \le \epsilon_2 \},
\end{align*}
and define the function
\begin{align*}
f(x):=
\epsilon_2 (\log x- \log \epsilon_2)
+(1-\epsilon_2) (\log (M-x)- \log (1-\epsilon_2) )
\end{align*}
for $\epsilon_2 \in (0,1)$.
The set ${\cal P}_{\epsilon_1,\epsilon_2,M,k}^{A}$
is a non-empty set 
only when 
the integer $k$ belongs to $[0, \epsilon_2 M e^{-\epsilon_1}]$.
Under the above choice of $k$, 
we have
\begin{align*}
\max_{P_{A}\in {\cal P}_{\epsilon_1,\epsilon_2,M,k}^A}
H(A) = f(k) 
\end{align*}
and
\begin{align*}
\max_{P_{A}\in {\cal P}_{\epsilon_1,\epsilon_2,M}^A}
H(A) = 
\max_{k \in [0, \epsilon_2 M e^{-\epsilon_1}]} f(k), 
\end{align*}
where $k$ is restricted to an integer in the maximum.
Taking the derivative, we have
\begin{align*}
f'(x)= \frac{\epsilon_2}{x}- \frac{1-\epsilon_2}{M-x},
\end{align*} 
which is positive when $x < M\epsilon_2$.
Hence,
\begin{align*}
&
\max_{P_{A}\in {\cal P}_{\epsilon_1,\epsilon_2,M}^A}
H(A) \\
\le & f(\epsilon_2 M e^{-\epsilon_1}) \\
= &\epsilon_2 (\log M - \epsilon_1)
+(1-\epsilon_2) (\log M \!+ \!\log (1\!-\!\epsilon_2 e^{-\epsilon_1})-\log (1\!-\!\epsilon_2) ) \\
= & \log M - \epsilon_2 \epsilon_1
+(1-\epsilon_2 ) \log [1+ \frac{\epsilon_2 (1-e^{-\epsilon_1})}{1-\epsilon_2 }] \\
\le & \log M - \epsilon_2 \epsilon_1
+(1-\epsilon_2 ) \frac{\epsilon_2 (1-e^{-\epsilon_1})}{1-\epsilon_2 } \\
= & \log M - \epsilon_2 (e^{-\epsilon_1} -1+\epsilon_1 ).
\end{align*}
Since $\log M - \epsilon_2 (e^{-\epsilon_1} -1+\epsilon_1 )$ is 
an affine function of $\epsilon_2$,
we obtain (\ref{eq-a-11}).

On the other hand,
using the set $\Omega:=
\{ a|- \log P_{A}(a) \le \log M - \epsilon_1 \}$,
we have
\begin{align*}
& \max_{P_{A}\in {\cal P}_{\epsilon_1,\epsilon_2,M}^A}
e^{-\rho H_{1+\rho}(A)} 
=
\sum_{a \in \Omega^c}(P_A(a))^{1+\rho} 
+
\sum_{a \in \Omega}(P_A(a))^{1+\rho} \\
\le &
(1-\epsilon_2) \frac{e^{\rho \epsilon_1}}{M^\rho}
+\epsilon_2^{1+\rho} 
\le 
(1-\epsilon_2) \frac{e^{\rho \epsilon_1}}{M^\rho}
+\epsilon_2 .
\end{align*}
Since $(1-\epsilon_2) \frac{e^{\rho \epsilon_1}}{M^\rho}+\epsilon_2 $ is a linear function of $\epsilon_2$,
we obtain 
\begin{align*}
 \max_{P_{A|B}\in {\cal P}_{\epsilon_1,\epsilon_2,M}^{A|B}}
e^{-\rho H_{1+\rho}(A|B)} 
\le 
(1-\epsilon_2) \frac{e^{\rho \epsilon_1}}{M^\rho}
+\epsilon_2 ,
\end{align*}
which implies (\ref{eq-a-12}).
\end{proofof}

\begin{proofof}{Lemma \ref{lem22}}
Since (\ref{lem22-2}) implies (\ref{lem22-1}),
we only show (\ref{lem22-2}) from (\ref{lem22-1}).
For an arbitrary small number $\epsilon>0 $,
we define the probability
\begin{align*}
\delta_n:=
P_{A^n,B^n} \{ (a,b) |- \frac{1}{n}\log P_{A^n|B^n}(a|b) \le R- \epsilon \}.
\end{align*}
Applying Eq.~(\ref{eq-a-11}) of Lemma \ref{lem11} to the case when 
$\epsilon_1= n \epsilon$ and $\epsilon_2=\delta_n$, 
we obtain
\begin{align*}
H(A_n|B_n) \le n R -  \delta_n (e^{-n \epsilon}-1+n \epsilon)  .
\end{align*}
That is,
\begin{align}
\delta_n \le \frac{R- \frac{1}{n}H(A_n|B_n)}{\frac{e^{-n \epsilon}-1}{n}+ \epsilon}.
\Label{11-22-8}
\end{align}
Thus, $\lim_{n \to \infty} \delta_n =0$.
Hence, Eq.~(\ref{eq-a-12}) of Lemma \ref{lem11} guarantees that
\begin{align}
H_{1+\alpha/n}(A_n|B_n)
\ge 
-\frac{n}{\alpha}\log ( (1-\delta_n)e^{\alpha(\epsilon -R)} +\delta_n ).
\Label{11-22-6}
\end{align}
Thus,
\begin{align*}
\liminf_{n \to \infty} \frac{1}{n}H_{1+\alpha/n}(A_n|B_n)
\ge &
\liminf_{n \to \infty} -\frac{1}{\alpha}\log ( (1-\delta_n)e^{\alpha(\epsilon -R)} +\delta_n ) \\
= &
R-\epsilon.
\end{align*}
Since $\epsilon >0$ is arbitrary,
\begin{align*}
\liminf_{n \to \infty} \frac{1}{n}H_{1+\alpha/n}(A_n|B_n)
\ge R.
\end{align*}
Since the cardinality of $\mathcal{A}_n$ is $e^{nR}$,
we have $\frac{1}{n}H_{1+\alpha/n}(A_n|B_n) \le R$.
Hence, 
\begin{align*}
\lim_{n \to \infty} \frac{1}{n}H_{1+\alpha/n}(A_n|B_n)= R.
\end{align*}
Combining relation (\ref{eq-1-13}), we obtain the desired argument. 
\end{proofof}

\section{Extension to general measurable spaces}\Label{as14}
\subsection{Information quantities}
Our results has been obtained based on discrete sets, i.e., sets with countable elements.
Here, we explain how our results are extended to the case of measurable spaces, which contain continuous sets.
Firstly, we state the assumptions used in Appendix \ref{as14}.
As before, $\mathcal{X}$ is the input alphabet of the channel
and $\mathcal{Z}$ is the output alphabet to Eve.
In general, a channel from $\mathcal{X}$ to $\mathcal{Z}$ is
described as a collection of conditional probability measures
$\mu_{Z|X=x}$ on $\mathcal{Z}$ for all inputs $x \in \mathcal{X}$,
and $\mu_{Z|X=x}$ might not have a probability density for some $x \in \mathcal{X}$.
In this appendix, however, we assume that there exists
a finite measure $\nu_{\mathcal{Z}}$ on $\mathcal{Z}$ such that
for all $x \in \mathcal{X}$, $\mu_{Z|X=x}$ is absolutely continuous with respect 
$\nu_{\mathcal{Z}}$. In the following $P_{Z|X}(\cdot| x)$ denotes the Radon-Nikodym derivative
$d \mu_{Z|X=x} / d \nu_{\mathcal{Z}}$.
We also make the same assumption on the channel from Alice to Bob.

In addition, as before, we consider probability measures $\eta$
on $\mathcal{U} \times \mathcal{V} \times \mathcal{X}$.
We assume that there exist finite mesures 
$\nu_{\mathcal{U}}$ on $\mathcal{U}$,
$\nu_{\mathcal{V}}$ on $\mathcal{V}$ and
$\nu_{\mathcal{X}}$ on $\mathcal{X}$
such that $\eta$ is absolutely continuous with
respect to the product measure $\nu_{\mathcal{U}} \times
\nu_{\mathcal{V}} \times \nu_{\mathcal{X}}$.
Under this assumption we can denote  by $P_{UVX}$
the Radon-Nikodym derivative
$d \eta / d (\nu_{\mathcal{U}} \times
\nu_{\mathcal{V}} \times \nu_{\mathcal{X}})$,
and marginal probability densities $P_U$, etc.\ and
conditional probability densities $P_{V|U}$, etc.\
can be computed from $P_{UVX}$.
In the following, $dv$, $dz$, etc.\  denote
$d \nu_{\mathcal{V}}$, $d \nu_{\mathcal{Z}}$, etc.\ 
assumed above.

Firstly, we give the definition of the information quantities 
in the general measurable case.
Although $E_0(\rho| P_{Z|V},P_V)$ and
$E_0(\rho| P_{Z|V}, P_{V|U},P_U)$ are defined 
for distributions $P_V$ and $P_U$ and conditional distributions
$P_{Z|V}$ and $P_{V|U}$ with discrete sets in 
\eqref{phid},
they can be defined as follows 
even when
${\cal Z}$, ${\cal V}$, and ${\cal U}$ are 
measurable spaces in the sense of \cite[Theorem 32.2]{Billingsley}.
Then, we define 
\begin{align}
&E_0(\rho| P_{Z|V},P_V) \nonumber \\
:=& \log \int_{{\cal Z}} d z
\left(
\int_{{\cal V}} d v
P_{V}(v) (P_{Z|V}(z|v)^{1/(1-\rho)})\right)^{1-\rho},
\Label{phid2} \\
&E_0(\rho| P_{Z|V}, P_{V|U},P_U) \nonumber \\
:=& \log \int_{{\cal U}} d u
\int_{{\cal Z}} d z
\left(
\int_{{\cal V}} d v
 P_{V|U}(v|u) (P_{Z|V}(z|v)^{1/(1-\rho)})\right)^{1-\rho}.\nonumber 
\end{align}
The above definition formally depends on the choices of 
the measures $dz,du,dv$.
But in the next paragraph
we will explain the above values are
independent of the choice of measures $dz,du,dv$.

Now, suppose that we choose other measures $dz',du',dv'$
so that the measures $dz',du',dv'$ and 
the original measures $dz,du,dv$ are absolutely continuous with respect to each other, respectively.
As is shown in the left hand side of \cite[p.7740]{H-tight},
even when these information quantities are defined with the measures $dz',du',dv'$,
these information quantities have the same values as those defined with the original measures $dz,du,dv$.
So, these information quantities do not depend on the choice of the measures $dz,du,dv$ whenever the measures and the original measures are absolutely continuous with respect to each other.

When $Q$ and $P$ are probability density functions
on a measurable space ${\cal Z}$ with respect to
a common finite measure $dz$,
$\psi(\rho|Q\|P)$ is defined as
\begin{align}
\psi(\rho|Q\|P) & :=\log \int_{{\cal Z}} d z
Q(z)^{1+\rho}P(z)^{-\rho} .\nonumber
\end{align}

Further,
$\psi(\rho| P_{Z|V}, P_{V}) $ 
and
$\psi(\rho| P_{Z|V}, P_{V|U},P_U) $ 
are defined as follows.
\begin{align}
& \psi(\rho| P_{Z|V}, P_{V|U},P_U) \nonumber \\
=& 
\log 
\int_{{\cal V}} d v P_{V}(v)
\int_{{\cal Z}} d z
 P_{Z|V}(z|v)^{1+\rho} P_{Z}(z)^{-\rho}, \\
& \psi(\rho| P_{Z|V}, P_{V|U},P_U) \nonumber \\
=& 
\log 
\int_{{\cal U}} d u P_U(u) 
\int_{{\cal V}} d v P_{V|U}(v|u)
\int_{{\cal Z}} d z
 P_{Z|V}(z|v)^{1+\rho} P_{Z|U}(z|u)^{-\rho}.
\end{align}
Similar to the information quantities 
$E_0(\rho| P_{Z|V},P_V)$ and $E_0(\rho| P_{Z|V}, P_{V|U},P_U)$,
we can show that
the information quantities $\psi(\rho| P_{Z|V}, P_{V|U},P_U)$
and $\psi(\rho| P_{Z|V}, P_{V|U},P_U)$ 
do not depend on the choice of the measures $dz,du,dv$ whenever the measures and the original measures are absolutely continuous with respect to each other.

The above quantities can be defined for a channel.
When the input and output systems ${\cal Z}$ and ${\cal V}$ 
are measurable spaces,
a channel $W$ is defined as 
a set of probability density functions $\{W_v\}_{v \in {\cal V}}$ on $Z$.
That is, substituting $W$ into a conditional probability density function $P_{Z|V}$
as $P_{Z|V}(z|v)=W_v(z)$, we define the above information quantities for the channel $W$.
So, when the channels $W^Z$ and $W^Y$ 
satisfy the above conditions,
the code construction and security evaluation given in the next subsection work well. 
Note that the above generalization works well even when ${\cal V}$ is a finite set because a finite set is also a measurable space.

\subsection{Code construction and security evaluation}
Under the above extension, our results can be extended as follows.
Firstly, we focus on Theorem \ref{lem-01}.
Assume that $W$ is a channel from a measurable space ${\cal X}$ to a measurable space ${\cal Y}$
and that $A$ is a discrete random variable on a finite set ${\cal A}$
subject to the distribution $P_A$.
Theorem \ref{lem-01} holds even under this assumption,
whose proof can be done by replacing $\sum_x$ and $\sum_y$
by $\int_{{\cal X}} dx$ and $\int_{{\cal Y}} dy$.
Theorem \ref{Lee3} and Corollary \ref{cor1}
also hold with a slightly different extension.
Assume that $W$ is a channel from a finite-dimensional vector space ${\cal X}$ over $\bF_q$ to a measurable space ${\cal Y}$
and that $A$ is a discrete random variable on a finite-dimensional vector space ${\cal A}$ over $\bF_q$ subject to $P_A$.
Then, 
Theorem \ref{Lee3} and Corollary \ref{cor1}
hold even under this assumption,
whose proof can be done by replacing $\sum_y$
by $\int_{{\cal Y}} dy$.

Now, we consider the extension of Code Ensemble \ref{con0}.
Assume that
${\cal X}={\cal V}$, ${\cal Y}$, ${\cal Z}$, and ${\cal U}$ are measurable,
and that the private and common messages $S_{\mathrm{p}}$ 
and $S_{\mathrm{c}}$ take values in finite sets.
Then, we can apply Code Ensemble \ref{con0} to the above situation.
Hence, Lemma \ref{lem0} holds even under this assumption
because the proof by Kaspi and Merhav \cite[Section II]{kaspi11}
is still valid under this assumption.

Next, we proceed to the extension of Code Ensemble \ref{con2}.
Assume that
${\cal X}$, ${\cal Y}$, ${\cal Z}$, ${\cal V}$, and ${\cal U}$ are measurable, and
that all messages $S_0, S_1, \ldots, S_T$ take values in finite sets.
Then, we can apply Code Ensemble \ref{con2} to the above situation.
Hence, Theorem \ref{lem1} holds even under this assumption
because \eqref{ineq-12} holds under this assumption.

Then, we extend the contents of Section \ref{s5}.
We consider the extension of Code Ensemble \ref{con1}.
Assume that
${\cal X}$, ${\cal Y}$, ${\cal Z}$, ${\cal V}$, and ${\cal U}$ are measurable,
and that
${\cal B}_1$ and ${\cal B}_2$ are finite Abelian groups.
In this case, all messages $S_0, S_1, \ldots, S_T$ take values in finite sets.
Then, we can apply Code Ensemble \ref{con1} to the above situation.
First, notice that Theorem \ref{lem0} still holds in the above situation.
Hence, Lemma \ref{lem2-1} and Theorem \ref{lem2} 
hold even under this assumption,
whose proof can be done by applying the extension of Theorems \ref{lem0} and \ref{Lee3}.
Lemma \ref{l12-3-2} holds with a slightly different extension.
That is, Lemma \ref{l12-3-2} holds 
when the sets ${\cal U}$ and ${\cal V}$ are finite set,
i.e., only the set ${\cal Z}$ is allowed to be a general measurable space.
This is because we need to consider the cardinalities of the subsets in ${\cal U}$ and ${\cal V}$.
Since the contents of Sections V and VI are extended to 
the case of measurable spaces in the above way,
the contents of Sections VIII and IX also can be extended to the case of measurable spaces in the same way.

In Section \ref{s8},
we have proposed several types of practical code constructions.
Code Constructions \ref{con3} and \ref{con4} can be applied to the channel 
$P_{Z|V}$ from a measurable space ${\cal V}$ to a measurable space ${\cal Z}$.
In these constructions, since the code $\varphi_{\rm p}$ is given,
we can restrict the set ${\cal V}$ to the finite subset given as the image of the map $\varphi_{\rm p}$.
Hence, we can apply Lemma \ref{l12-3-2} with the above extension 
in this context.

When the above discussion is applied to the wire-tap channel model,
we obtain an extension of existing results to the case of the asymptotic uniform dummy message.
That is, we consider 
the case with no common messages and $T=2$ when
${S}_1$ corresponds to the message to be secretly sent to Bob,
and ${S}_2$ does to the dummy message making $S_1$ ambiguous to Eve.
For a given rate $R_1$ of secret message
and a given rate $R_2$ of dummy message,
the RHS of (\ref{Haya-51-w}) coincides with the Gallager exponents,
the RHS of (\ref{bound-b}) coincides with the RHS of (59) in
\cite{hay-wire},
and 
the RHS of (\ref{bound-1}) coincides with the exponents of the RHS of (15) in
\cite{hayashi11}.

\subsection{Gaussian case}
Finally, when the channel $P_{YZ|X}$ is a degraded Gaussian channel
as \eqref{11-29-a},
we demonstrate how the strong security can be shown 
for the wire-tap channel, which is given as
the case with no common messages and $T=2$ when
${S}_1$ corresponds to the message $S$ to be secretly sent to Bob,
and ${S}_2$ does to the dummy message $A$ making $S$ ambiguous to Eve.
Assume that ${\cal X}$, ${\cal Y}$, and ${\cal Z}$ are the set of real numbers.
So, we choose the measures $dx$, $dy$, and $dz$ to be the Lebesgue measure.
Then, we assume that the conditional probability density functions corresponding to the channels are 
\begin{align}
P_{Y|X}(y|x):= \frac{1}{\sqrt{2\pi v_1}} e^{-\frac{(y-x)^2}{2 v_1}} ,\quad
P_{Z|X}(z|x):= \frac{1}{\sqrt{2\pi v_2}} e^{-\frac{(z-x)^2}{2 v_2}} ,\Label{11-29-a}
\end{align}
where $v_2 > v_1$.
Since the channel is degraded, we do not need to introduce 
random variables $U$ and $V$.
Now, we choose 
the probability density function $P_X$ to be $P_X(x)=
\frac{1}{\sqrt{2\pi v_3}} e^{-\frac{x^2}{2 v_3}}$.
Then,
\begin{align}
E_0(\rho| P_{Z|X},P_X)
=& \frac{\rho}{2}\log (1+ \frac{v_3}{(1-\rho)v_2}) ,\\
\psi(\rho|P_{Z|X=x},P_Z) 
=& \frac{(1+\rho)\rho }{2 (v_2+(1+\rho) v_3)}x^2
-\frac{\rho}{2}\log v_2 \nonumber \\
&+ \frac{1+\rho}{2} \log (v_2+v_3)
-\frac{1}{2}\log (v_2+(1+\rho) v_3 ) ,\\
\psi(\rho|P_{Z|X},P_X) 
=&
\frac{1+\rho}{2} \log (v_2+v_3) \nonumber \\
&-\frac{1}{2}\log (v_2+(1-\rho^2)v_3 )
-\frac{\rho}{2}\log v_2 \nonumber \\
=&
\frac{\rho}{2} \log (1+\frac{v_3}{v_2})
-\frac{1}{2} \log (1-\frac{v_3}{v_2+v_3} \rho^2 ).
\end{align}
Hereafter, we denote 
the average leaked information under our code $\Phi$
by $I(S;E)[\Phi]$.
Assume that we use the Gaussian channel $P_{YZ|X}$ $n$ times,
and that the rates of secret message $S$ and dummy message $A$ are 
$R_1$ and $R_2$, respectively.
When the dummy message $A$ has the R\'{e}nyi entropy $H_{1+\rho}(A)$,
Theorem \ref{lem1} guarantees that 
\begin{align}
\rE_\Phi [e^{\rho I(S;E)}] 
\le &
1+ e^{ -\rho H_{1+\rho} +n  (
\frac{\rho}{2} \log (1+\frac{v_3}{v_2})
-\frac{1}{2} \log (1-\frac{v_3}{v_2+v_3} \rho^2 )
)} 
\end{align}
i.e.,
\begin{align}
\rE_\Phi [I(S;E)]
\le &
\frac{1}{\rho}
e^{ -\rho H_{1+\rho} +n  (
\frac{\rho}{2} \log (1+\frac{v_3}{v_2})
-\frac{1}{2} \log (1-\frac{v_3}{v_2+v_3} \rho^2 )
)} 
\Label{11-29-c}
\end{align}
for $\rho \in (0,1]$.
Since there is no common messages,
the cardinality of ${\cal B}_1$ is $1$ in Code Ensemble \ref{con1}.
Theorem \ref{lem2} guarantees that 
\begin{align}
\rE_\Phi [e^{\rho I(S;E)[\Phi]}] 
\le &
1+ e^{ -\rho H_{1+\rho}(A) +n \frac{\rho}{2}\log (1+ \frac{v_3}{(1-\rho)v_2}) },
\end{align}
i.e.,
\begin{align}
\rE_\Phi [I(S;E)]
\le &
\frac{1}{\rho}
e^{ -\rho H_{1+\rho}(A) +n  \frac{\rho}{2}\log (1+ \frac{v_3}{(1-\rho)v_2}) }
\Label{11-29-b}
\end{align}
for $\rho \in (0,1]$.
When the dummy message $A$ is uniform,
\eqref{11-29-b} and \eqref{11-29-c} are simplified as follows
\begin{align}
\rE_\Phi [I(S;E)]
\le &
\frac{1}{\rho}
e^{ -n (\rho R_2 -(
\frac{\rho}{2} \log (1+\frac{v_3}{v_2})
-\frac{1}{2} \log (1-\frac{v_3}{v_2+v_3} \rho^2 )) )} .
\Label{11-29-e} \\
\rE_\Phi [I(S;E)]
\le &
\frac{1}{\rho}
e^{ -n (\rho R_2-  \frac{\rho}{2}\log (1+ \frac{v_3}{(1-\rho)v_2}) )} .
\Label{11-29-f} 
\end{align}
Since $
\lim_{\rho \to 0}\frac{1}{\rho}
(\frac{\rho}{2}\log (1+ \frac{v_3}{(1-\rho)v_2}) )
=\lim_{\rho \to 0}\frac{1}{\rho}
(
\frac{\rho}{2} \log (1+\frac{v_3}{v_2})
-\frac{1}{2} \log (1-\frac{v_3}{v_2+v_3} \rho^2 ))=
\frac{1}{2} \log (1+ \frac{v_3}{v_2})$,
both \eqref{11-29-e} and \eqref{11-29-f} yield
the strong security when $R_2 > \frac{1}{2} \log (1+ \frac{v_3}{v_2})$.



\begin{IEEEbiographynophoto}{Masahito Hayashi}(M'06--SM'13) was born in Japan in 1971.
He received the B.S. degree from the Faculty of Sciences in Kyoto 
University, Japan, in 1994 and the M.S. and Ph.D. degrees in Mathematics from 
Kyoto University, Japan, in 1996 and 1999, respectively.

He worked in Kyoto University as a Research Fellow of the Japan Society of the 
Promotion of Science (JSPS) from 1998 to 2000,
and worked in the Laboratory for Mathematical Neuroscience, 
Brain Science Institute, RIKEN from 2000 to 2003,
and worked in ERATO Quantum Computation and Information Project, 
Japan Science and Technology Agency (JST) as the Research Head from 2000 to 2006.
He also worked in the Superrobust Computation Project Information Science and Technology Strategic Core (21st Century COE by MEXT) Graduate School of Information Science and Technology, The University of Tokyo as Adjunct Associate Professor from 2004 to 2007. 
In 2006, he published the book ``Quantum Information: An Introduction'' from Springer. 
He worked in the Graduate School of Information Sciences, Tohoku University as Associate Professor from 2007 to 2012. 
In 2012, he joined the Graduate School of Mathematics, Nagoya University as Professor. 
He also worked in Centre for Quantum Technologies, National University of Singapore as Visiting Research Associate Professor from 2009 to 2012 
and as Visiting Research Professor from 2012 to now. 
In 2011, he received Information Theory Society Paper Award (2011) for Information-Spectrum Approach to Second-Order Coding Rate in Channel Coding.
In 2016, he received the Japan Academy Medal from the Japan Academy
and the JSPS Prize from Japan Society for the Promotion of Science.

He is on the Editorial Board of {\it International Journal of Quantum Information}
and {\it International Journal On Advances in Security}. 
His research interests include classical and quantum information theory and classical and quantum statistical inference.  
\end{IEEEbiographynophoto}

\begin{IEEEbiographynophoto}{Ryutaroh Matsumoto}(M'00) was born in Nagoya, Japan, on November 29, 1973. 
He received the B.E. degree in computer science, the M.E. degree
in information processing, and the Ph.D. degree in electrical and electronic
engineering, all from Tokyo Institute of Technology, Japan, in 1996, 1998 and
2001, respectively. He was an Assistant Professor from 2001 to 2004, and has
been an Associate Professor since 2004 in the Department of Communications
and Computer Engineering, Tokyo Institute of Technology.
He also served as a Velux visiting professor for the Department
of Mathematical Sciences, Aalborg University, Denmark
during 2011--2012 and 2014.
His research interests
include error-correcting codes, quantum information theory, information
theoretic security, and communication theory. Dr. Matsumoto received the
Young Engineer Award from IEICE and the Ericsson Young Scientist Award
from Ericsson Japan in 2001. He received the Best Paper Awards from IEICE
in 2001, 2008, 2011 and 2014.
\end{IEEEbiographynophoto}

\end{document}